\documentclass[11pt, oneside]{article}

\usepackage{amsmath}
\usepackage{amsthm}
\usepackage{amssymb}
\usepackage{mathrsfs}
\usepackage{complexity}
\usepackage{enumitem}
\usepackage{tikz}
\usepackage{tkz-berge}
\usepackage{float}
\usepackage{mathtools}
\usepackage{xspace}

\usepackage{microtype}
\usepackage{comment}

\usepackage{caption}
\usepackage{subcaption}

\usepackage{tabularx}
\usepackage{listofitems}
\usepackage{array}
\newcolumntype{Y}{>{\centering\arraybackslash}X}

\usepackage{thmtools, thm-restate}

\usepackage{titlesec}
\usepackage{autobreak}

\usepackage[hypertexnames=false]{hyperref}
\usepackage[capitalize]{cleveref}
\crefname{proposition}{Prop.}{Props.}
\crefname{section}{Sec.}{Secs.}
\crefname{corollary}{Cor.}{Cors.}
\crefname{lemma}{Lem.}{Lems.}
\crefname{definition}{Def.}{Defs.}

\usepackage[margin=1in]{geometry}

\usetikzlibrary {arrows.meta} 

\usepackage[style=alphabetic, 
url=false, 
maxalphanames=5,
maxbibnames=10
]{biblatex}
\newtheorem{theorem}{Theorem}[section]
\newtheorem{lemma}[theorem]{Lemma}
\newtheorem{corollary}[theorem]{Corollary}
\newtheorem{proposition}[theorem]{Proposition}

\theoremstyle{definition}
\newtheorem{definition}[theorem]{Definition}
\newtheorem{remark}{Remark}
\newtheorem{example}{Example}

\DeclareMathOperator{\hol}{Holant}
\DeclareMathOperator{\hols}{Holant^*}
\DeclareMathOperator{\holbs}{Holant^*_2}
\DeclareMathOperator{\holts}{Holant^*_3}
\DeclareMathOperator{\supp}{supp}
\DeclareMathOperator{\rank}{rank}
\DeclareMathOperator{\spn}{span}

\DeclareMathOperator{\typei}{I}
\DeclareMathOperator{\typeii}{II}
\DeclareMathOperator{\proj}{proj}

\DeclareMathOperator{\ann}{Ann}

\DeclareMathOperator{\sym}{Sym}

\DeclareMathOperator{\tr}{Tr}
\DeclareMathOperator{\gl}{GL}
\DeclareMathOperator{\speciallinear}{SL}
\DeclareMathOperator{\psl}{PSL}
\DeclareMathOperator{\ortho}{O}
\DeclareMathOperator{\so}{SO}

\DeclareMathOperator{\stab}{Stab}

\DeclareMathOperator{\ag}{Fr}
\DeclareMathOperator{\agx}{Fr_\X}

\newcommand{\db}{\mathsf{B}}
\newcommand{\dg}{\mathsf{G}}
\newcommand{\dr}{\mathsf{R}}

\newcommand{\sbto}{\boldsymbol{\delta}_1}
\newcommand{\sbtt}{\boldsymbol{\delta}_2}

\newcommand{\sepbgmath}{\db\dg|\dr}
\newcommand{\swbgmath}{\mathsf{Swap}_{\db\dg;\dr}}
\newcommand{\sepbg}{$\db\dg|\dr$\xspace}
\newcommand{\sepbr}{$\db\dr|\dg$\xspace}
\newcommand{\sepgr}{$\db|\dg\dr$\xspace}

\newcommand{\suppbg}{\{\db, \dg\}^*}

\newcommand{\geneq}{\textsf{GenEQ}\xspace}

\newcommand{\swhelper}[1]{$\mathsf{Swap}_{#1}$\xspace}
\newcommand{\swbg}{\swhelper{\db \dg; \dr}}
\newcommand{\swbr}{\swhelper{\db \dr; \dg}}
\newcommand{\swgr}{\swhelper{\dg \dr; \db}}

\newcommand{\symspace}[2]{\mathsf{Sym}^{#1}(\mathbb{C}^#2)}

\newcommand{\sph}{\#\P-hard\xspace}
\newcommand{\sphness}{\#\P-hardness\xspace}
\newcommand{\scsp}{\#\mathrm{CSP}}

\newcommand{\teh}{^{\otimes 3}}
\newcommand{\tew}{^{\otimes 2}}
\newcommand{\transpose}{^\mathsf{T}}

\newcommand{\domres}[1]{
  ^{*\to\{#1\}}
}

\newcommand{\casea}{$\mathfrak{A}$\xspace}
\newcommand{\caseatwo}{$\mathfrak{A}_2$\xspace}
\newcommand{\caseathree}{$\mathfrak{A}_3$\xspace}
\newcommand{\caseb}{$\mathfrak{B}$\xspace}
\newcommand{\casebtwo}{$\mathfrak{B}_2$\xspace}
\newcommand{\casebthree}{$\mathfrak{B}_3$\xspace}

\newcommand{\casec}{$\mathfrak{C}$\xspace}
\newcommand{\casecz}{$\mathfrak{C}^0$\xspace}
\newcommand{\caseco}{$\mathfrak{C}^1$\xspace}
\newcommand{\casecosep}{$\mathfrak{C}^1_{\mathrm{s}}$\xspace}
\newcommand{\casecobg}{$\mathfrak{C}^1_{2}$\xspace}
\newcommand{\casecogen}{$\mathfrak{C}^1_{\mathrm{g}}$\xspace}
\newcommand{\casecodz}{$\mathfrak{C}^1_{\mathrm{z}}$\xspace}

\newcommand{\caseatwomath}{\mathfrak{A}_2}
\newcommand{\caseathreemath}{\mathfrak{A}_3}
\newcommand{\casebtwomath}{\mathfrak{B}_2}
\newcommand{\casebthreemath}{\mathfrak{B}_3}

\newcommand{\casecosepmath}{\mathfrak{C}^1_{\mathrm{s}}}
\newcommand{\casecobgmath}{\mathfrak{C}^1_{2}}
\newcommand{\casecogenmath}{\mathfrak{C}^1_{\mathrm{g}}}
\newcommand{\casecodzmath}{\mathfrak{C}^1_{\mathrm{z}}}

\newcommand{\casecsupp}{\mathcal{C}}

\newcommand{\blank}{{-}}
\newcommand{\bne}{(\mathord{\ne_2})}
\newcommand{\beq}{(\mathord{=_2})}

\newcommand{\iny}[1]{\widehat{#1}}

\renewcommand{\i}{\mathfrak{i}}
\newcommand{\bzp}{\boldsymbol{\beta}_+}
\newcommand{\bzm}{\boldsymbol{\beta}_-}

\newcommand{\Abold}{\mathbf{A}}
\newcommand{\X}{\mathbf{X}}
\newcommand{\Y}{\mathbf{Y}}
\newcommand{\Yinv}{\mathbf{Y}^{-1}}
\newcommand{\Z}{\mathbf{Z}}

\renewcommand{\phi}{\varphi}

\tikzset{
mycirc/.style={circle,fill=black, minimum size=0.05cm}
}
\tikzset{
mycircempty/.style={circle, draw=black, minimum size=0.05cm}
}
\tikzset{
mysq/.style={rectangle,fill=black, minimum size=0.05cm}
}

\newcommand{\tersig}[2]{
  \setsepchar{ }
  \readlist\arg{#2}
      \begin{tikzpicture}[xscale=#1, yscale=0.7, every node/.style={scale=1}]
      \node at (0, 0) {$\arg[1]$};
      \node at (-1, -1) {$\arg[2]$};
      \node at (1, -1) {$\arg[3]$};
      \node at(-2, -2) {$\arg[4]$};
      \node at (0, -2) {$\arg[5]$};
      \node at (2, -2) {$\arg[6]$};
      \node at (-3, -3) {$\arg[7]$};
      \node at (-1, -3) {$\arg[8]$};
      \node at (1, -3) {$\arg[9]$};
      \node at (3, -3) {$\arg[10]$};
    \end{tikzpicture}
}

\newcommand{\binsig}[2]{
  \setsepchar{ }
  \readlist\arg{#2}
      \begin{tikzpicture}[scale=#1, every node/.style={scale=1}]
      \node at (0, 0) {$\arg[1]$};
      \node at (-1, -1) {$\arg[2]$};
      \node at (1, -1) {$\arg[3]$};
      \node at(-2, -2) {$\arg[4]$};
      \node at (0, -2) {$\arg[5]$};
      \node at (2, -2) {$\arg[6]$};
    \end{tikzpicture}
}

\title{Group Action and  Spin Representation 
Lead to a $\hols$ Complexity Dichotomy on Domain Size $3$ 
}

  \author{Jin-Yi Cai\thanks{University of Wisconsin-Madison. {\tt jyc@cs.wisc.edu}}
  \and Jin Soo Ihm\thanks{University of Wisconsin-Madison.
  {\tt ihm2@wisc.edu}}}
 
\date{}

\renewcommand{\subset}{\subseteq}

\begin{document}
\maketitle
\thispagestyle{empty}
\begin{abstract}
We prove a complexity dichotomy theorem 
for $\hols$ problems over complex-valued symmetric constraint functions $\mathcal{F}$ on domain size $3$.
We give a decidable tractability criterion and prove that if $\mathcal{F}$ satisfies the criterion, then $\hols(\mathcal{F})$ is solvable in polynomial time, and otherwise it is \sph. 
This is the first Holant dichotomy for a set of complex-valued constraint functions on higher domains. 
We show that complex-valued constraint functions have a rich structure not observed in real-valued constraint functions. 
This structure is only revealed when we analyze them in a bipartite Holant setting with a non-standard bilinear form
and 
provides the backbone to the proof of the dichotomy.
We use group actions and the spin representation of $\speciallinear(2, \mathbb{C})$ 
in $\so(3, \mathbb{C})$ and the generalized orthogonal group
to facilitate this analysis. 
We also introduce 
\emph{frames} and \emph{shells}. Frames linearize the group action and provide a unified framework for proving \sphness when used in conjunction with shells. 
Furthermore, we characterize the lower dimensional constraint functions by \emph{annihilators}, 
which makes it possible to analyze \emph{essentially Boolean domain} functions in domain size $3$. 
\end{abstract}

\newpage
\tableofcontents
\newpage

\section{Introduction}

How many independent sets does a given graph have? 
How many perfect matchings does a given graph have?
Many counting problems  can be formulated as 
a partition function
on a graph.
For example, for a given graph $G = (V, E)$, counting the number of independent sets is computing $\sum_{\sigma: V \to \{0, 1\}} \prod_{(u, v) \in E} f(\sigma(u), \sigma(v))$, where $f(1, 1) = 0$ and $1$ otherwise. 
Similarly, counting the number of perfect matchings is computing 
$\sum_{\sigma: E \to \{0, 1\}} \prod_{v \in V} f_v(\sigma\restriction_{E(v)})$, where $E(v)$ denotes the incident edges of $v$ and $f_v$ is the \textsc{ExactOne} function that outputs $1$ if exactly one of its inputs is $1$ and $0$ otherwise.
The first problem 
is an example of counting graph homomorphisms (GH), or more generally, counting CSP (\#CSP).
The second problem 
is a Holant problem. 

GH, introduced by Lov\'asz~\cite{Lovsz1967OperationsWS,Hell2004GraphsAH}, is a special case of \#CSP  with only a single binary constraint function. In turn, \#CSP is a special case of Holant
\cite{Creignou-csp,bulatov-dalmau-csp,bulatov-counting-csp,dyer-goldberg-jerrum-weighted-csp,holant-star-boolean,dyer-richerby,nonneagitve-csp,complex-csp}. 
It is known that Holant is
strictly more expressive than GH; for example, the problem of counting perfect matchings cannot be modeled within GH~\cite{freedman-lovasz,cai-govorov-pm-rank-connection-tensor-gh}.

Holant problems, also known as edge models, are essentially equivalent to
tensor networks. 
They can express a broad class of counting problems, e.g., weighted matchings, cycle covers, edge colorings, and many others.
Meanwhile, it is also restrictive enough to prove complete classifications, 
known as complexity dichotomy theorems.
We give an informal definition of Holant problems.
Let $\mathcal{F}$ be a set of (real or complex valued) constraint functions defined on some domain set $D$.
It defines a Holant problem $\hol(\mathcal{F})$ as follows.
An input is a graph $G = (V, E)$ with each $v \in V$ associated with some $\mathbf{F}_v \in \mathcal{F}$. $\mathbf{F}_v$ takes incident edges of $v$ as its input variables. 
The goal is to compute the sum of product of evaluations of $\mathbf{F}_v$ over all 
edge assignments $\sigma: E \to D$.
In the example of counting perfect matchings,
 $D = \{0, 1\}$ and $\mathcal{F}$ is the set of \textsc{ExactOne} functions of all arities. 
 It turns out that for many 
 classes  $\mathcal{F}$, 
 the problem $\hol(\mathcal{F})$ is either polynomial time solvable or \sph, and the goal of the complexity classification program for Holant is to achieve this classification for
 ever broader classes of $\mathcal{F}$.

The complexity classification of GH and \#CSP has seen tremendous progress, culminating in a dichotomy of complex-valued functions over any finite domain~\cite{dyer-graph-hom,bulatov_graph-hom,goldberg-partition,cai-complex-graph-hom,
Creignou-csp,dyer-goldberg-jerrum-weighted-csp,holant-star-boolean,bulatov-counting-csp,dyer-richerby,nonneagitve-csp,complex-csp}.
A lot of work has been done on Holant problems as well, particularly on the Boolean domain ($\lvert D \rvert = 2$)
\cite{holant-star-boolean-journal,cai-lu-xia-fibonacci,real-holant-huang-lu,complex-symmetric-holant-dichotomy,fkt-universal,backens-holant-plus,nonnegative-holant,real-holant-dichotomy,backens-holant-c}. 
This includes the dichotomy for all complex-valued symmetric constraint functions~\cite{complex-symmetric-holant-dichotomy} and for all real-valued not necessarily symmetric constraint functions~\cite{real-holant-dichotomy}. 
Counting weighted Eulerian orientations (\#EO), which is a subclass of complex-valued Boolean domain Holant problems, received a lot of attention recently in an effort to completely classify Holant problems with complex-valued constraint functions on the Boolean domain~\cite{shao-tang-eo,meng-wang-xia-fp-np,meng-wang-xia-eo,meng-wang-xia-zheng-odd-arity}. 
There is also a grand program to unify all Boolean domain dichotomy theorems using a group theoretic framework~\cite{xia2026frameworkunifycomplexitydichotomy}.
However, 
extending the complexity classification to higher domain size for Holant problems
has been difficult.

To make progress, one often starts by 
assuming that some additional natural set of constraint functions is present. 
This helps in understanding the fundamental structure that determines the complexity.
In fact, the class of all $\scsp$ problems 
can be defined as  assuming 
the set of \textsc{Equality}  functions $\mathcal{EQ}$ of all arities are present, and we have the
equivalence $\scsp(\mathcal{F}) = \hol(\mathcal{F} \cup \mathcal{EQ})$. 
Another choice is the set of all unary functions $\mathcal{U}$. 
In this paper, we study $\holts(\mathcal{F}) := \hol(\mathcal{F} \cup \mathcal{U})$ for an arbitrary set $\mathcal{F}$ of symmetric complex-valued constraint functions on domain size $3$.

Previously, there were very few
Holant dichotomy theorems on higher domains.
Two of them are from a decade ago
and only consider a single constraint function: 
(1) a $\hol(f)$ dichotomy where $f$ is a ternary constraint function that has a strong symmetry property called domain permutation invariance~\cite{edge-coloring-higher-domain}; 
(2) a $\hols(f)$ dichotomy  where $f$ is a \textit{complex-valued} symmetric ternary function on domain size 3~\cite{domain3-single-signature}. 
It was observed in~\cite{domain4-01-dichotomy}, which also proves a dichotomy on domain size $4$ of a single $\{0, 1\}$-valued symmetric function, that certain
intricate 
yet tractable functions of~\cite{domain3-single-signature} do not occur when the functions are restricted to be \textit{real-valued}.
Building upon that observation, 
very recently  a $\holts(\mathcal{F})$ dichotomy is proved in which $\mathcal{F}$ is a set of \textit{real-valued} symmetric 
functions on domain size $3$~\cite{domain3-real-set}.

In this paper, we finally bring to fruition the program initiated in~\cite{domain3-single-signature}, and prove a $\holts(\mathcal{F})$ dichotomy in which $\mathcal{F}$ is an
arbitrary \emph{set of complex-valued} symmetric
functions on domain size $3$. This theorem resolves a long running open problem and makes the first progress in 
Holant complexity classification for a \emph{set} of complex-valued functions on a higher domain. 

The main difficulty of proving a  dichotomy  for a  \emph{set} of constraint
functions is how  individual constraint functions interact with each other.
A natural approach after obtaining a dichotomy for a single ternary function
is to broaden it to a pair of ternary and  binary functions. 
This approach was very successful in proving Boolean domain dichotomies earlier.
However, repeated efforts along this line failed. It turns out, counterintuitively, 
that for higher domain Holant problems, binary functions are very special and difficult
to deal with. Instead, we decided to take the 
unconventional approach of first proving a dichotomy
for a pair of ternary  functions. 

Three types 
of tractable functions \casea, \caseb and  \casec are  identified in~\cite{domain3-single-signature}. While they are individually tractable, it turns out that 
they behave very differently, even \emph{within} a single type. We gain 
insight by a penetrating
combinatorial analysis on how each subtype interacts with others,
and this analysis is only possible under a suitable change of basis, which reveals a rich hierarchy.
We then introduce new notions called \emph{frames} and \emph{shells}.
With these notions we can carry out our analysis naturally in the language of
group actions. 

When there is only a single constraint function,  we may always apply a change of basis to bring the function to a \emph{canonical form}, which does not change the complexity of the problem by the Holant theorem~\cite{valiant-accidental-alg}.
Extending the dichotomy of a single function in~\cite{domain3-single-signature} to a set of functions amounts to analyzing how the individually tractable functions interact.
These interactions occur under 
different changes of basis 
depending on  the  functions. 
Mathematically, we can formulate this problem as analyzing a group action on the functions. 
However, these actions are difficult 
to analyze in general. 
In the spirit of representation theory, 
the problem becomes easier if we \emph{linearize}: a frame associates a linear object (or a collection of linear objects)
with each function,
and we instead analyze the action of the group on 
these linear spaces.

We borrow tools from representation theory, in particular the spin representation
of (generalized) orthogonal groups with respect to some 
bilinear forms other than the 
standard dot product, and elementary algebraic geometry. 
It turns out that an algebraic characterization facilitated by the spin representation  %
closely corresponds to the aforementioned combinatorial analysis.
Also, we introduce the notion of a \emph{signature annihilator}, which 
simultaneously simplifies and unifies the analysis of \emph{lower dimensional functions} residing in domain size $3$. 
This also leads to an algorithm that simultaneously subsumes and simplifies different tractable cases given in \cite{domain3-real-set}.

One may wonder why study complex-valued functions? 
Holant problems correspond to computing a contraction of tensor networks, and in studying many-body quantum problems, complex-valued tensors are the right language~\cite{montangero-tensor-network}.
While a $\hols$ dichotomy for real-valued  functions is known~\cite{domain3-real-set},
it does not include some of the more intricate functions shown to have a tractable algorithm in~\cite{domain3-single-signature}. 
Also, as the dichotomy~\cite{domain3-real-set} shows,
real-valued tractable functions of arity $3$ or higher can always be \textit{domain separated},
and the most essential intricacies of domain size $3$ do not come into play. 
In this paper, we show that by understanding the dichotomy of complex-valued functions, we can introduce new techniques and build a coherent theory that shed 
light on Holant problems on higher domains.
Also, some real-valued functions appear more natural after a change of basis that requires using a complex-valued transformation.
We additionally see that complex-valued signatures are more intricately related to the bipartite Holant problem in two ways.
Some problems are easier to understand after a non-orthogonal holographic transform, $\holts(\mathcal{F}) =_T \holts(\X | \iny{\mathcal{F}})$, by changing the dot product to a different bilinear form $\X$. 
We make use of a generalized orthogonal Holant theorem combined with the spin representation to facilitate the study of such problems. 
Additionally, for some $\mathbf{F}$ and $\mathbf{G}$, 
there is actually an \textit{equivalence} to a bipartite problem
$\holts(\mathbf{F}, \mathbf{G}) =_T \holts(\mathbf{F} | \mathbf{G})$. While reduction in one direction is clear,
this equivalence is not 
obvious; it is only made possible
by interactions of \emph{isotropic} vectors (necessarily complex). 
To analyze this,
we need a dichotomy of bipartite $\hols$ problems on the Boolean domain.
This has not been 
proved before, and we prove it in this paper along the way.

\subsection{Examples}
\label{subsec:tractable-gadget-meet}

We 
describe some significant difficulties on the road toward achieving a
complexity dichotomy for a set of ternary signatures by giving some examples of tractable cases.
We quickly find that 
the approach of~\cite{domain3-real-set}
becomes infeasible, and results 
in an explosion  of possible cases here.
In fact, a deeper reason exists geometrically.
Unlike the 
characterization given in~\cite{domain3-real-set}, there is no uniform way to
bring all signatures to a ``nice'' basis. 
At the same time, 
this explicit 
analysis actually obscures the right 
view of how 
multiple complex-valued lower dimensional signatures interact as a whole. 

We show some examples of how the approach 
of~\cite{domain3-real-set} can fail here. 
The algorithmic idea of~\cite{domain3-real-set} is that in all tractable cases, we can transform all signatures to be  in a ``nice'' basis.
Then, the problem could be decomposed into Boolean domain problems, allowing us to use the Boolean domain algorithms directly. 
However, for complex-valued signatures the complexity of the interaction increases
enormously, and there are cases where one cannot bring all signatures 
to multiple pairs of Boolean domains. 
Allowing complex-valued signatures reveals \textit{truly domain $3$} tractable 
interactions.
We use the symbols $\{\db, \dg, \dr\}$ to denote the $3$ domain elements. 

\begin{example}
Consider
 $\mathbf{F} = \mathbf{e}_1\teh + \mathbf{e}_2\teh
=\begin{bsmallmatrix}
1 \\ 0 \\0
\end{bsmallmatrix}\teh
+
\begin{bsmallmatrix}
0 \\ 1 \\0
\end{bsmallmatrix}\teh$ and  $\mathbf{G} = \begin{bsmallmatrix}
1 \\ 0 \\ \i
\end{bsmallmatrix}\teh
+
\begin{bsmallmatrix}
0 \\ 1 \\ \i
\end{bsmallmatrix}\teh$.
They will be called
\caseatwo and \casebtwo signatures.
This case is not covered in~\cite{domain3-real-set} because under any orthogonal $T$, $T \mathbf{F}$ and $T \mathbf{G}$ cannot be both real-valued. 
However, this is a new tractable case.

The following algorithm solves $\holts(\mathbf{F}, \mathbf{G})$.
We can show explicitly  that (the signature of) any connected $\holts(\mathbf{F})$ gadget 
belongs to $\mathcal{F} = \{a \mathbf{e}_1^{\otimes n} + b \mathbf{e}_2^{\otimes n} : a, b \in \mathbb{C}, n \ge 1\}$.
Similarly, writing in row vector form,  connected $\holts(\mathbf{G})$ gadgets 
are in $\mathcal{G} = \{a (1, 0, \i)^{\otimes n} \otimes (0, 1, \i)^{\otimes m} + b (0, 1, \i)^{\otimes n} \otimes (1, 0, \i)^{\otimes m} : a, b \in \mathbb{C}, n, m \ge 0 \}$. 
Then,
to solve $\holts(\mathbf{F}, \mathbf{G})$, 
we first form connected components by $\mathbf{F}$ and 
$\mathbf{G}$ alone, and as $\mathbf{F}$ is supported on $\db\dg$, we  get 
\begin{align*}
\holts(\mathbf{F}, \mathbf{G})
&\le_T \holts(\mathcal{F} \mid \mathcal{G}) \\
&=_T \holbs(\mathcal{F}\domres{\db\dg} \mid \mathcal{G}\domres{\db\dg})
\end{align*}
which is a Boolean domain  P-time tractable problem.
\end{example}

One problem of this algorithm is that it does not generalize nicely when we have three signatures, say $\holts(\mathbf{F}, \mathbf{G}, \mathbf{H})$, because the second step of the above reduction crucially uses the fact that $\supp \mathcal{F} \subseteq \{\db, \dg\}^*$ and that $\mathcal{G}$ is relatively easy to determine.
We will see from the next example that dealing with
three signatures 
with this method grinds to a halt.

\begin{example}\label{ex:b2-a2-a2}
Let
$\mathbf{F}_1 
=  \begin{bsmallmatrix}
1 \\ \i \\0
\end{bsmallmatrix}\teh
+
\begin{bsmallmatrix}
1 \\ -\i \\0
\end{bsmallmatrix}\teh$,
$\mathbf{F}_2 = 
\begin{bsmallmatrix}
1 \\ \i \\p_2
\end{bsmallmatrix}\teh
+
\begin{bsmallmatrix}
1 \\ -\i \\q_2
\end{bsmallmatrix}\teh$
 and
$\mathbf{F}_3 = 
\begin{bsmallmatrix}
1 \\ \i \\p_3
\end{bsmallmatrix}\teh
+
\begin{bsmallmatrix}
1 \\ -\i \\q_3
\end{bsmallmatrix}\teh$,
such that $p_2 q_2 = p_3 q_3 = -2$ and $p_2/p_3 = e^{\pi \i / 3}$. 
We claim that $\holts(\mathbf{F}_1, \mathbf{F}_2, \mathbf{F}_3)$ is tractable.
We observe that any edge between $\mathbf{F}_2$ and $\mathbf{F}_3$ can be replaced with two (separate) unary signatures,
because 
there exists some orthogonal $T$ such that $T \mathbf{F}_2$ is supported on $\dg\dr$ and $T \mathbf{F}_3$ is supported on $\db\dr$.  So such edges 
can be removed since  unary signatures are free in $\holts$.
Writing in row vector form, let $\mathcal{F}_1 = \{ a (1, \i, 0)^{\otimes n} \otimes (1, -\i, 0)^{\otimes m} + b (1, -\i, 0)^{\otimes n} \otimes (1, \i, 0)^{\otimes m} : a, b \in \mathbb{C}, n, m \ge 0  \}$,
and
$\mathcal{F}_k = \{ a(1, \i, p_k)^{\otimes n} + b(1, -\i, q_k)^{\otimes n} : a, b \in \mathbb{C} , n \ge 1\}$, for $k=2, 3$.
Then, 
\begin{align*}
\holts(\mathbf{F}_1, \mathbf{F}_2, \mathbf{F}_3)
&\le_T \holts(\mathcal{F}_1 \mid \mathcal{F}_2, \mathcal{F}_3) \\
&=_T \holbs(\mathcal{F}_1\domres{\db\dg} \mid \mathcal{F}_2\domres{\db\dg}, \mathcal{F}_3\domres{\db\dg}) \,,
\end{align*}
which is a Boolean domain   P-time tractable problem.

One can also show that there is no orthogonal transform that can put all three signatures in separate domains.
The main reason is that the 
\emph{annihilators} (see \cref{def:annihilator}) 
of the three signatures are not pairwise orthogonal. 
Under the  orthogonal $T$ above, we can bring $T \mathbf{F}_2$ to $\dg\dr$ and $T \mathbf{F}_3$ to $\db\dr$, but  then it becomes 
mysterious how $T \mathbf{F}_1$ interacts with the others.
For example, choosing $p_3 = \sqrt{2} \i$ gives
\begin{align*}
T &= \begin{bsmallmatrix}
    -1/\sqrt{2} & -\sqrt{3/2} & -\i \\
    -\sqrt{2} & 0 & -\i \\
    -\i \sqrt{3/2} & -\i/\sqrt{2} & \sqrt{3}
\end{bsmallmatrix} &&
    T \mathbf{F}_1 \sim 
    \begin{bsmallmatrix}
    -1 - \i \sqrt{3} \\ -2 \\ 1 - \i \sqrt{3}
    \end{bsmallmatrix}\teh + 
    \begin{bsmallmatrix}
    -1 + \i \sqrt{3} \\ -2 \\ -1 - \i \sqrt{3}
    \end{bsmallmatrix}\teh \\
    T \mathbf{F}_2 &\sim 
    \begin{bsmallmatrix}
    0 \\ -1 + e^{\pi \i / 3} \\ -1
    \end{bsmallmatrix}\teh + 
    \begin{bsmallmatrix}
    0 \\ -1 - e^{2 \pi \i / 3} \\ 1
    \end{bsmallmatrix}\teh
    &&
    T \mathbf{F}_3 \sim 
    \begin{bsmallmatrix}
    1 - \i \sqrt{3} \\ 0 \\ 1 + \i \sqrt{3}
    \end{bsmallmatrix}\teh +
    \begin{bsmallmatrix}
    1 + \i \sqrt{3} \\ 0 \\ -1 + \i \sqrt{3}
    \end{bsmallmatrix}\teh.
\end{align*}
\end{example}

\begin{example}
There is another example that can be solved using yet a 
different algorithm.
Let $\mathbf{G}_1 = \mathbf{e}_1\teh + \mathbf{e}_2\teh$, 
$\mathbf{G}_2 = \begin{bsmallmatrix}
    5 \\ 0 \\ 3 \i
\end{bsmallmatrix}\teh + \begin{bsmallmatrix}
    0 \\ 5 \\ 4 \i
\end{bsmallmatrix}\teh$,
and
$\mathbf{H} = \sym\left(\begin{bsmallmatrix}
    3 \\ 4 \\ 5 \i
\end{bsmallmatrix}, \begin{bsmallmatrix}
    3 \\ 4 \\ 5 \i
\end{bsmallmatrix}, \begin{bsmallmatrix}
    3 \\ 4 \\ - 5 \i
\end{bsmallmatrix}\right)$.
One can show that under orthogonal $T_1 = 
\begin{bsmallmatrix}
0 & 0 & 1 \\
4/5 & -3/5 & 0 \\
-3/5 & -4/5 & 0
\end{bsmallmatrix}$, 
$T_1 \mathbf{G}_1$ is supported on $\dg\dr$
and $T_1 \mathbf{H}$ is supported on $\db\dr$ so an edge between them can be replaced by two unary signatures $\mathbf{e}_3 \otimes \mathbf{e}_3$.
Also, for non-orthogonal $T_2 = \begin{bsmallmatrix}
3/10 & 2/5 & -\i/2 \\
3/5 & 4/5 & \i \\
4/5 & - 3/5 & 0
\end{bsmallmatrix}
$ we have $T_2 T_2\transpose = \X$, thus 
$(T_2^{-1})\transpose  (T_2^{-1})= \X^{-1} = \X 
= \begin{bsmallmatrix}
  0& 1  & 0 \\ 1 &0 & 0  \\ 0& 0 &1
\end{bsmallmatrix}$, 
$T_2 \mathbf{G}_2$ is supported on $\db\dr$ and $T_2 \mathbf{H}$ is supported on $\db\dg$.
Then,
any connection between $T_2 \mathbf{G}_2$ and $T_2 \mathbf{H}$ via $\X$ can be replaced as $\mathbf{e}_1 \otimes \mathbf{e}_2$.
Note that we cannot separate $\mathbf{G}_2$ and $\mathbf{H}$ using orthogonal transformations. 
Also, similar to the previous case, $T_2 \mathbf{G}_1$ and $T_1 \mathbf{G}_2$ do not 
have nice expressions;
 there is no single 
transformation that organizes all signatures neatly. 
But from the separate transformations, we know that all connections between $\mathbf{H}$ 
and $\mathbf{G}_1$ (resp. $\mathbf{G}_2$) can be severed, and replaced by  unary signatures. 
Further, any connection between two $T_2 \mathbf{G}_2$ via $\X$ can be replaced with $\mathbf{e}_3 \otimes \mathbf{e}_3$ because $T_2 \mathbf{G}_2$ is supported on $\db\dr$ and $\X$ swaps between $\db$ and $\dg$ while keeping $\dr$ constant. 
So, we can replace any connection between two $\mathbf{G}_2$ with unary signatures as well. 
Given a $\holts(\mathbf{G}_1, \mathbf{G}_2, \mathbf{H})$ instance,
we can replace these edges, and get connected components of $\holts(\mathbf{G}_1, \mathbf{G}_2)$ and $\holts(\mathbf{H})$ instances, with the further property that the 
$\holts(\mathbf{G}_1, \mathbf{G}_2)$ instance does not have edges between two $\mathbf{G}_2$ signatures. 
Then,
\[\holts(\mathbf{H}) =_T \holts(T_1 \mathbf{H}) \le_T \holbs( (T_1 \mathbf{H})\domres{\db\dr}) \,,\]
\begin{align*}
    \holts(\mathbf{G}_1, \mathbf{G}_2) 
    &\le_T \holts(\mathcal{G}_1 \mid \mathbf{G}_2) \\
    &=_T \holbs(\mathcal{G}_1\domres{\db\dg} \mid \mathbf{G}_2 \domres{\db\dg}) \,, 
\end{align*}
where $\mathcal{G}_1 = \{a \mathbf{e}_1^{\otimes n} + b \mathbf{e}_2^{\otimes n} : a, b \in \mathbb{C}, n \ge 1 \}$.  These are Boolean domain tractable problems.
Thus $\holts(\mathbf{G}_1, \mathbf{G}_2, \mathbf{H})$ can be solved in polynomial time.
\end{example}

There are many more tractable cases, and we can exhibit for each
a different 
ad hoc algorithm.
What about four or more signatures? Can they fit together in more myriad ways?
These examples suggest that performing a case by case analysis 
is not going to work.

\subsection{Some Proof Highlights}\label{subsec:proof-outline}

We highlight some ideas used in this paper.

\paragraph{Refinement of  Tractable Types}
The main difficulty in generalizing the dichotomy of a single 
signature
to a dichotomy of a set of signatures is analyzing the interaction between two or more signatures. In~\cite{domain3-single-signature}, three different tractable types of signatures \casea, \caseb and  \casec  were given together with separate algorithms for each type.
However, we discovered that this  classification is too coarse
when  analyzing the interactions of multiple signatures.
Even within each type, there is a genuine difference in what the signatures can do.
For example, $\mathbf{e}_1\teh + \mathbf{e}_2\teh + \mathbf{e}_3\teh$ and $\mathbf{e}_1\teh + \mathbf{e}_2\teh$ are both of type \casea,
but we will see that the first signature imposes a lot of constraint on other signatures
while the second imposes much less.
The situation with \casea is actually more involved
because
$\holts(\mathbf{F}) =_T \holts(T \mathbf{F})$ for an arbitrary orthogonal $T$,
and thus $\mathbf{F}$ takes infinitely many forms.
The situation with  \caseb is similar. The  situation with  \casec is 
even more intricate.

We overcome these difficulties by refining the $3$ tractable types into $9$ (\cref{def:types-ternary}). 
For \casec signatures, the internal structure is only revealed under a 
 non-orthogonal holographic transform  $\Yinv$,
 where $\Y = 
\begin{bsmallmatrix}
  \Z & 0 \\ 0 & 1
\end{bsmallmatrix}$ and $\Z = \frac{1}{\sqrt{2}}  \begin{bsmallmatrix}
  1 & 1 \\ \i & - \i
\end{bsmallmatrix}$.
Then, the \casec signatures transform to have a very interesting support pattern (\cref{fig:casec-general}), which can be analyzed combinatorially. The analysis  reveals that there are genuinely different subtypes hidden within \casec signatures. 

\paragraph{Analysis of Type \texorpdfstring{\casec}{C}}
We analyze type \casec in \cref{sec:casec}.
Let $\X = \Y\transpose \Y = \begin{bsmallmatrix}
  0& 1  & 0 \\ 1 &0 & 0  \\ 0& 0 &1
\end{bsmallmatrix}$. As $\Y$ is not orthogonal, after the $\Yinv$ transform, 
we need to analyze the signatures in the bipartite Holant setting $\holts(\X \mid \iny{\mathbf{F}})$, where the edges between copies of $\iny{\mathbf{F}} = \Yinv \mathbf{F}$  are connected by a binary function $\X$.
Although it may seem more complicated, we will see that for  \casec signatures the support pattern of $\iny{\mathbf{F}}$ allows us to give a combinatorial analysis.
This allows us to come up with a simpler algorithm for $\holts(\mathbf{F})$ (than the one in~\cite{domain3-single-signature}) (\cref{subsec:case-c-algorithm}).

More importantly, we discover that there is a rich hierarchy of  \casec signatures that 
behave very differently, which was not observed previously. 
We identify 5 subtypes of \casec, and show that their combinatorial behaviors are  very different, which only becomes apparent in the $\holts(\X \mid \iny{\mathbf{F}})$ setting (\cref{subsec:casecz-characterization,subsec:characterize-caseco}).
The refinement is also intimately related to (complex) spin representation, and we will see that certain subgroups of the spin representation yield subrepresentations that
go hand in hand 
with the subtypes (\cref{sec:case-c-spin-rep-calc,subsec:stabc-rep}).
We also give an analysis of what signatures are realizable for each. 
It turns out that most of them can realize a certain linear subspace of signatures, leading to the next idea. 

\paragraph{Frames, Shells, and Algebraic Geometry}
We define \emph{frames} and \emph{shells} in \cref{sec:two-pieces}.
After a change of basis, it becomes easier to analyze the set of gadgets that a particular signature can realize. 
Here, we use the power of the unary signatures that are available to use to the fullest.
For example, it can be easily shown that if we have a signature $\mathbf{e}_1\teh + \mathbf{e}_2\teh + \mathbf{e}_3\teh$, then with unary signatures, we can make a gadget that realizes any signature of the form $a \mathbf{e}_1^{\otimes n} + b \mathbf{e}_2^{\otimes n} + c \mathbf{e}_3^{\otimes n}$, for  $a, b, c \in \mathbb{C}$ and $n$. 
This suggests that the advantage of working in $\hol^*$ is that we can realize a linear space of signatures for which we can identify a basis. 
We 
define a linear space
of realizable signatures for each type,
called a \textit{frame}  of the type. 
One advantage is that while we need to analyze signatures after holographic transformations, 
\emph{frames} work well with spans of tensors of standard basis vectors, and this makes it transparent their behavior under holographic transformations. 

We further formulate the conditions of the known Boolean domain dichotomy theorem in terms of algebraic varieties. 
It turns out that the conditions 
for ternary signatures 
can be expressed as requiring the
signatures to be contained in the union of certain lines and a twisted cubic curve in the projective space, which we call a \emph{shell}. 
Combined with the linear sets in a frame, 
this makes the analysis of checking the known \sphness condition equivalent to whether the frame corresponding to one signature is contained in the shell of the other signature, and the latter is a purely algebraic analysis (\cref{prop:subspace-e-bg,prop:subspace-m-bg}). 

Imagine trying to analyze $\holts(\mathbf{F}, \mathbf{G})$ directly, where $\mathbf{F} = \mathbf{e}_1\teh + \mathbf{e}_2\teh + \mathbf{e}_3\teh$.
We may assume that $\holts(\mathbf{G})$ is tractable, so $\mathbf{G} = T \mathbf{G}'$ for an orthogonal $T$ and a tractable signature $\mathbf{G}'$ in some canonical form.
Then, 
if $\holts(\mathbf{F}, \mathbf{G})$ is not \#P-hard, 
$\mathbf{G}\domres{\db\dg}$ must escape the Boolean domain hardness conditions
imposed by $\mathbf{F}$.
However, actually carrying out this proof becomes very difficult, as the entries of $\mathbf{G}$ are degree $4$ polynomials in the entries of $T$ and $\mathbf{G}'$ with complex dependencies.
Frames and shells greatly simplify this, making the proof possible.

\paragraph{Pair of Ternary Signatures}
With \textit{mise en place} complete, we finally proceed to prove dichotomy theorems. 
A natural starting point, 
which was successful in earlier dichotomies,  is to first prove a (3, 2)-dichotomy, i.e., $\holts(\mathbf{F}, \mathbf{G})$ for a ternary $\mathbf{F}$ and a binary $\mathbf{G}$. 
The reason is that the interaction between a ternary signature and a binary signature is seemingly the simplest when both tractability and \sphness  occur.
However, binary signatures are hard to analyze precisely because they are simple: there are not many constructions one can do with binary signatures. 
There are unexpected cancellations, and it becomes very difficult  even to analyze the following seemingly simple situation:
what kind of matrix $M \in \mathbb{C}^{3 \times 3}$ has the property that for all 
$k \in \mathbb{N}$, 
the $2 \times 2$ upper left submatrix of $M^k$ is rank $1$,  diagonal, or antidiagonal?
One can show that $\holts(\mathbf{e}_1\teh + \mathbf{e}_2\teh, M) =_T \holbs([1, 0, 0, 1], \{M^k : k \in \mathbb{N} \}\domres{\db\dg})$ is tractable, where $M$ represents a binary signature, if $M$ has such property and \sph otherwise. 
In order to prove a (3, 2)-dichotomy, many similar questions must be answered. 

Instead, we take a different approach: we go for a  (3, 3)-dichotomy first! We prove the (3,3)-dichotomy in \cref{sec:pair-dicho}.
We will see that the proof strategy is crucially dependent on whether a signature is \textit{essentially Boolean domain} or \textit{truly domain 3}.
Some signatures, while living in domain 3, are essentially Boolean domain signatures,
e.g., all orthogonal transformations of $\mathbf{e}_1\teh + \mathbf{e}_2\teh$.
Signatures such as  $\mathbf{e}_1\teh + \mathbf{e}_2\teh + \mathbf{e}_3\teh$
are truly domain $3$ signatures. 
They can realize many signatures and impose a 
strong constraint.
In other words, its frame is large and its shell is small.
Frames and shells make the analysis of the interaction of two such signatures
feasible.

We note that the concept of frame relieves us from the need to construct any additional gadgets in the proof of a (3, 3)-dichotomy.
The frame captures the essence of a signature. 
Consider for example if we want to analyze the interaction between 
$\mathbf{e}_1\teh + \mathbf{e}_2\teh + \mathbf{e}_3\teh$ and a subtype of \casec 
signatures which we will call \casecogen.
By definition, this means 
there exists some orthogonal $T$ such that $\Yinv T \mathbf{G} = [a; \lambda, b; 0, 0, c; 0, 0, 0, d]$ such that $\lambda, d \ne 0$ and $c^2 - c \lambda - bd \ne 0$~(see \cref{def:types-ternary}; these conditions 
all have intrinsic meanings under spin representation when the theory is fully developed).
We may normalize so that $\lambda = 1$, and  writing $T = (t_{ij})$, we get the values of $\mathbf{G}$, e.g.,
\begin{align*}
    4\mathbf{G}(\db\db\db) &= t_{11} \left(-3 \sqrt{2} (a-1) t_{12}^2+12 i b t_{12} t_{13}+6 \sqrt{2} c t_{13}^2\right)+t_{11}^2 \left(6 b t_{13}+3 i
   \sqrt{2} (a+1) t_{12}\right) \\
   &\quad+ 
   \sqrt{2} (a+3) t_{11}^3-i \sqrt{2} (a-3) t_{12}^3-6 b t_{12}^2 t_{13}+6 i \sqrt{2} c t_{12} t_{13}^2+4 d
   t_{13}^3
\end{align*}

and similarly complicated expressions for other values.
We get expressions quite difficult to handle, and it would be very hard to apply the known Boolean domain dichotomy criteria. 
Cancellations are
possible 
because $T$ is an orthogonal matrix.
Even if one were able to apply the Boolean domain dichotomy on $\holbs([1, 0, 0, 1], \mathbf{G}\domres{\db\dg}) \le_T \holts(\mathbf{F}, \mathbf{G})$ it could happen that $\holbs([1, 0, 0, 1], \mathbf{G}\domres{\db\dg})$ escapes the known hardness conditions but it is still not clear whether we can conclude $\holts(\mathbf{F}, \mathbf{G})$ is tractable (it is actually \sph). 
Then, one would have to start again with more complicated constructions.
By contrast, frames and shells, combined with spin representation, give an elegant and simple
way to analyze the above situation.

\paragraph{Spin Representation}
However, there are still difficulties because we need to analyze the interaction of signatures up to some orthogonal transform: the frames are also transformed. 
Here, we use the \emph{complex spin representation}  $\speciallinear(2, \mathbb{C}) 
\overset{\phi}{\rightarrow}  \so(3, \mathbb{C})$ (\cref{eqn:phiM-parameterization}). 
If we directly
analyze up to orthogonal transform, there are $9$ parameters to consider in
the $3 \times 3$ matrix
$T = (t_{ij}) \in \mathbb{C}^{3 \times 3}$.
But 
there are nonlinear dependencies among $t_{ij}$ because $T$ is 
orthogonal.
We use the fact that $\speciallinear(2, \mathbb{C})$ is a double cover of $\so(3, \mathbb{C})$, and ${\rm dim}_\mathbb{C} \so(3, \mathbb{C}) = 3$.
This is much
easier to use than the implicit relations among $t_{ij}$. 
One unexpected benefit is that because quadratic forms over $\mathbb{C}$
are isomorphic, we have $\so(3, \mathbb{C}) \cong \so(3, \mathbb{C}, \X)$, which is the group of transformations $\iny{T}$ such that $\iny{T}\transpose \X \iny{T} = \X$. Thus we can also 
get a  representation  $\speciallinear(2, \mathbb{C}) \overset{\rho}{\rightarrow}  \so(3, \mathbb{C}, \X)$ (\cref{eq:paramaterize-so3cx}). 
We analyze \casec signatures after applying $\Yinv$ transform, which brings $\holts(\mathcal{F})$ to 
$\holts(\X \mid \Yinv \mathcal{F})$. It is
in this bipartite setting where all the intricacies of \casec signatures are  revealed.
Then, $\so(3, \mathbb{C}, \X)$ is the right group of transformations that keeps the left side invariant, similar to how $\so(3, \mathbb{C})$
keeps  $\beq$ invariant in the non-bipartite Holant setting. 
We generalize the orthogonal Holant theorem to  $\so(3, \mathbb{C}, \X)$. 
Another miracle is that the 
spin representation for $\so(3, \mathbb{C}, \X)$ reveals a wealth of hidden structures of \casec signatures that are crucial to our dichotomy (\cref{sec:case-c-spin-rep-calc,subsec:stabc-rep}).

\paragraph{Annihilators}
To analyze the interaction between essentially Boolean domain signatures, we 
define the notion of the \emph{annihilator} $\ann(\mathbf{F})$ of a signature $\mathbf{F}$ (\cref{def:annihilator}). 
We note that the main tool used in~\cite{domain3-real-set} 
to separate the subtypes of \casea and \caseb
is symmetric tensor rank 
 (\casec signatures do not occur there
 because they cannot be real-valued).
However, 
the \casec  signature $\sym(\bzp, \bzp, \bzm)$ has symmetric rank $3$
and yet it is supported on $\db\dg$ and thus is an essentially Boolean domain signature.
The notion of annihilator provides a more suitable tool for 
complex-valued signatures.

The existence of a nonzero annihilator is invariant under holographic transforms. 
Also, any gadget generated by the signature is still annihilated by the same unary signature. 
We will show that  $\ann(\mathbf{F}) \ne 0$
means the signature is \emph{weak}, that its frame is small and its shell is large.
We note that the map from ternary signatures $\mathbf{F} \mapsto \langle \mathbf{F}, \mathbf{u} \rangle$ for nonzero $\mathbf{u} \in \mathbb{C}^3$ is a full rank linear map $\mathbb{C}^{10} \to \mathbb{C}^6$.
We leverage the weakness of a signature $\mathbf{F}$ to our advantage by noticing that 
the kernel of the above map has dimension $4$, which also corresponds to the dimension of 
the set of symmetric ternary Boolean domain signatures as a vector space.
With this intuition, we show that if two signatures have non-orthogonal annihilators, then the \emph{shell} of one signature, which can be viewed as polynomial constraints on four entries, is sufficient to determine the form of the other signature uniquely 
(\cref{prop:annihilator-uniquely-propagates-values}). 
If the annihilators are orthogonal, then we show that the two signatures can only interact degenerately, which always results in tractability (\cref{prop:orthogonal-normal-vector-edge-factors}). 

An interesting consequence of this for complex-valued signatures is that there are signatures whose annihilator self-annihilates, i.e., the annihilator is 
isotropic (\cref{subsec:casecz-characterization}). 
Then, such a signature self-degenerates: it cannot interact with itself. 
If we have two such signatures $\mathbf{F}$ and $\mathbf{G}$ with different isotropic annihilators, then we can show that $\holts(\mathbf{F}, \mathbf{G}) =_T \holbs(\mathbf{F}' \mid \mathbf{G}')$ for some Boolean domain signatures $\mathbf{F}', \mathbf{G}'$. 
We see an interesting emergence of an equivalent bipartite Boolean domain problem, which has not been observed previously. Therefore, we also provide a proof of the bipartite Boolean domain dichotomy and use it in this paper (\cref{lem:bipartite-boolean-dichotomy,cor:bipartite-boolean-dichotomy-explicit}). 

We also note that this characterization with annihilators makes the proof of pairwise dichotomy a lot simpler than the proof given in~\cite{domain3-real-set} for similar cases. 
We will see that our proof and the theorem statements are both shorter with this new characterization.

This allows us to come up with a powerful 
unified algorithm for solving $\holts(\mathcal{F})$ where $\mathcal{F}$ is a set of signatures with annihilators (\cref{thm:ternary-set-ebd-algorithm}).
One unexpected benefit is that we can extend the result to include binary signatures. 
We notice that if $\mathcal{B}$ is a set of binary signatures, then $\holts(\mathcal{F} \cup \mathcal{B}) =_T \holts(\mathcal{M} \mid \mathcal{F})$ where $\mathcal{M}$ is the monoid generated by $\mathcal{B}$ under multiplication. 
Then, we may view a binary signature as a generalized edge defining a bilinear form that is not necessarily the dot product.
We generalize the idea that orthogonal annihilators lead to a degenerate interaction under the dot product to account for a general bilinear form (\cref{lem:edge-factors-general-case}). 
This idea 
generalizes to higher domains and asymmetric signatures  when two signatures interact degenerately. We  give  an algorithm that is always applicable when \sphness conditions are not met, resulting in a dichotomy (\cref{thm:algorithm-ebd-with-binary}). 

This is  surprising---we had expected  dealing with signatures with annihilators and adding binary signatures would be the hardest part of the proof, as illustrated by the examples in \cref{subsec:tractable-gadget-meet}.

\paragraph{Non-Mixing}
We prove a dichotomy of set of ternary signatures in \cref{sec:set-ternary}.
More challenges remain on the way to the dichotomy for a set of ternary signatures
even after the pairwise dichotomy is proved.
Unlike for the Boolean domain,
the (3, 3)-dichotomy is \emph{not} transitive.
In other words, in the Boolean domain, if $\holbs(\mathbf{F}, \mathbf{G})$ and $\holbs(\mathbf{F}, \mathbf{H})$ are tractable, then $\holbs(\mathbf{F}, \mathbf{G}, \mathbf{H})$ is tractable as well for ternary signatures since the forms of $\mathbf{G}$ and $\mathbf{H}$ are essentially uniquely determined. 
That is not the case in domain $3$ 
due to the richness of the tractable classes.
We show some (3, 3, 3)-non-mixing results, that  $\holts(\mathbf{F}, \mathbf{G}, \mathbf{H})$ can be \sph even when $\holts(\mathbf{F}, \mathbf{G})$ and $\holts(\mathbf{F}, \mathbf{H})$ are tractable (\cref{subsec:non-mixing-ternary}). The dichotomy for a set of ternary signatures is achieved after
more insight is gained on the interactions of the 9 tractable subtypes of  \casea, \caseb and
\casec signatures (\cref{thm:ternary-set-dichotomy}).

\paragraph{Binary Signatures}
Next, we prove a dichotomy of $\holts(\mathcal{F} \cup \mathcal{B})$ for a set of ternary signatures $\mathcal{F}$ and a set of binary signatures $\mathcal{B}$ in \cref{sec:binary}.
As mentioned above, proving the (3,2)-dichotomy directly seems very difficult. 
Having already handled the signatures with annihilators, 
it is 
finally possible 
to prove the (3, 2)-dichotomy using similar techniques with frames and the (3, 3)-dichotomy.
We can view the proof of the (3, 3)-dichotomy as analyzing what happens when we make a gadget by attaching an orthogonal matrix to the signatures in the frame. 
The (3, 2)-dichotomy can be viewed similarly, that we make a gadget by attaching a symmetric matrix to the signatures in the frame. 
Since the frame is expressed in terms of span of tensors of standard basis vectors, it allows us to directly analyze each column of the matrix. 
Similar to the (3, 3, 3)-non-mixing case, we show (3, 2, 2)-non-mixing to achieve the dichotomy for a set of ternary and binary signatures (\cref{thm:dichotomy-ternary-binary}).

\paragraph{Higher Arity}
Finally, we incorporate signatures of higher arity in \cref{sec:higher-arity}.
The proof for this part is similar to~\cite{domain3-real-set} but more involved.
Given an arity $4$ signature $\mathbf{F}$,
we look at the ternary gadgets that can be formed from $\mathbf{F}$ by attaching a unary signature, namely the linear set $\mathscr{F} = \{ \langle \mathbf{F}, \mathbf{u} \rangle : \mathbf{u} \in \mathbb{C}^3 \}$.
We can then use the (3, 3)-dichotomy to argue that for $\mathbf{G}, \mathbf{H} \in \mathscr{F}$,
there is a very strong constraint for $\holts(\mathbf{G}, \mathbf{H}, x \mathbf{G} + y \mathbf{H})$ to be tractable for all $x, y\in \mathbb{C}$.
This allows us to narrow down what the support of the signatures in $\mathscr{F}$ must be, identifying the support of $\mathbf{F}$. 
The result is that the tractable arity $4$ signatures must be of essentially the same form as the tractable ternary signatures (\cref{lem:arity-4}).
Then, we extend the definitions of tractable types to include higher arity signatures (\cref{def:types-higher-arity}), and show that a higher arity signature of a tractable type can realize a ternary signatures of the same type (\cref{prop:reduce-arity}). 
This allows us to induct on the arity to obtain a dichotomy encompassing signatures of any arity (\cref{thm:full-dichotomy}).

\section{Preliminaries and Notations}

\subsection{Definitions}
Let $D$ be a finite domain set, and $\mathcal{F}$ be a set of constraint functions, called signatures.
Each $\mathbf{F} \in \mathcal{F}$ is a mapping $D^k \to \mathbb{C}$ for some arity $k$.
If  $\mathbf{F}$ takes values in $\mathbb{R}$, we say $\mathbf{F}$ is real-valued.

A \textit{signature grid} $\Omega = (G, \mathcal{F}, \pi)$ consists of a graph $G = (V, E)$ where $\pi$ labels each vertex $v \in V$ with
a function $\mathbf{F}_v \in \mathcal{F}$ and 
the incident edges $E(v)$ 
as its inputs.
The arity of $\mathbf{F}_v$
matches the degree of $v$.
The Holant problem on instance $\Omega$ is to evaluate
\begin{equation}\label{eq:holant-value}
  \hol_{\Omega} = \sum_{\sigma} \prod_{v \in V} \mathbf{F}_v(\sigma |_{E(v)}) \, ,
\end{equation}
where the sum is over all edge assignments $\sigma : E \to D$
and $\mathbf{F}_v(\sigma |_{E(v)})$ is the evaluation of $\mathbf{F}_v$ on the ordered input tuple $\sigma |_{E(v)}$.

A signature $\mathbf{F}_v$ is listed by its values lexicographically as a table, or it can be expressed as a tensor in $(\mathbb{C}^{\lvert D \rvert})^{\otimes \deg(v)}$.
We can identify a unary function $\mathbf{F}(x): D \to \mathbb{C}$ with a vector $\mathbf{u} \in \mathbb{C}^{\lvert D \rvert}$.
Given two vectors $\mathbf{u} = (u_i)$ and $\mathbf{v} = (v_j)$ of dimension $\lvert D \rvert$, the tensor product $\mathbf{u} \otimes \mathbf{v}$ is a vector in 
$\mathbb{C}^{\lvert D \rvert^2}$, with entries $u_i v_j$ for $1 \le i, j \le \lvert D \rvert$.
For matrices $A = (a_{ij})$ and $B = (b_{kl})$ the tensor product (or Kronecker product) $A \otimes B$ is defined similarly;
it has entries $a_{ij}b_{kl}$ indexed by $((i, k), (j, l))$ lexicographically.
We write $\mathbf{u}^{\otimes k}$ for $\mathbf{u}\otimes \cdots \otimes \mathbf{u}$ with $k$ copies of $\mathbf{u}$.
$A^{\otimes k}$ is similarly defined.
We have $(A \otimes B)(A' \otimes B') = (A A' \otimes BB')$ whenever the matrix products are defined.
In particular, $A^{\otimes k}(\mathbf{u}_1 \otimes \cdots \otimes \mathbf{u}_k) = A \mathbf{u}_1 \otimes \cdots \otimes A \mathbf{u}_k$ when the matrix-vector 
products $A \mathbf{u}_i$ are defined.

A signature $\mathbf{F}$ of arity $k$ is \textit{degenerate} if $\mathbf{F} = \mathbf{u}_1 \otimes \cdots \otimes \mathbf{u}_k$ for some vectors $\mathbf{u}_i$.
Such a signature is very weak; there is no interaction between the variables.
If every signature in $\mathcal{F}$ is degenerate, then $\hol_\Omega$ for any $\Omega = (G, \mathcal{F}, \pi)$ is computable in polynomial time in a trivial way:
Simply split every vertex $v$ into $\deg(v)$ vertices each assigned a unary 
$\mathbf{u}_i$
and connected to the incident edge.
Then $\hol_{\Omega}$ becomes a product over each component of a single edge.
Thus degenerate signatures are weak and should be properly understood as made up by unary signatures.
To concentrate on the essential features that differentiate tractability from intractability, the class of problems $\hols$ was introduced in~\cite{holant-star-boolean, holant-star-boolean-journal}.
These are the problems where all unary signatures are assumed to be present, i.e. $\hols(\mathcal{F}) = \hol(\mathcal{F} \cup \mathcal{U})$ where $\mathcal{U}$ is the set of all unary signatures.

Our proof uses the  notion of a \textit{gadget}.
Consider a type of graphs $G = (V, I, E)$ where $I$ and $E$ are two kinds of edges.
Edges in $I$ are ordinary internal edges with two endpoints in $V$.
Edges in $E$ are external edges (also called dangling edges) which have only one end point in $V$.
Such a graph can be made into a part of a larger graph as follows: 
Given a graph $G'$ and a vertex $v$ of $G'$, we may replace $v$ with a graph $G$ by merging the external edges of $E$ to the incident edges of $v$.

An $\mathcal{F}$-gadget consists of a graph $G = (V, I, E)$ and a labeling $\pi$ where each vertex $v \in V$ is labeled by $\mathbf{F}_v \in \mathcal{F}$.
We may view $G$ as a function $\mathbf{F}_G$, such that if we replace a vertex $v$ of a graph $G'$ by $G$, the Holant value of the resulting instance is 
as if we assign $F_G$ to $v$.
For this to hold, $\mathbf{F}_G$ must be such that for an assignment $\tau: E \to D$, 
\[
  \mathbf{F}_G(\tau) = \sum_{\sigma}\prod_{u \in V} \mathbf{F}_u (\tau \sigma |_{E(u)}) \, ,
\]
where the sum is over all edge assignments $\sigma: I \to D$ and $\tau \sigma$ is the combined assignment on $E \cup I$.

\subsection{Holographic Transformation}
To describe the idea of holographic transformations, it is convenient to consider bipartite graphs.
For a general graph, we can always transform it into a bipartite graph while preserving the Holant value, as follows:
for each edge in the graph, we replace it by a path of length $2$, and assign to the new vertex the binary \textsc{Equality} function $\beq$.

We use the notation $\hol(\mathcal{R} \mid \mathcal{G})$ to denote the Holant problem on bipartite graphs $H = (U, V, E)$,
where each signature for a vertex in $U$ or $V$ is from $\mathcal{R}$ or $\mathcal{G}$, respectively.
An input instance for the bipartite Holant problem is a bipartite signature grid and is denoted as $\Omega = (H ; \mathcal{R} \mid \mathcal{G} ; \pi)$.
Signatures in $\mathcal{R}$ are considered as row vectors (or covariant tensors); signatures in $\mathcal{G}$ are considered as column vectors (or contravariant tensors).

For a $\lvert D \rvert \times \lvert D \rvert$ matrix $T$ and a signature set $\mathcal{F}$, define
\[
  T \mathcal{F} = \{\mathbf{G} : \exists \mathbf{F} \in \mathcal{F} \text{ of arity n, such that } \mathbf{G} = T^{\otimes n} \mathbf{F}\} \, ,
\]
and similarly for $\mathcal{F}T$.
Whenever we write $T^{\otimes n} \mathbf{F}$ or $T \mathcal{F}$, we view the signatures as column vectors; similarly $\mathbf{F} T^{\otimes n}$ or $\mathcal{F} T$ as row vectors.
For brevity we will omit the tensor power and write $T \mathbf{F}$ instead of $T^{\otimes n} \mathbf{F}$ when it is clear from the context.
A holographic transformation by $T$ is the following operation:
given a signature grid $\Omega = (H ; \mathcal{R} \mid \mathcal{G} ; \pi)$, for the same graph $H$,
we get a new grid $\Omega' = (H ; \mathcal{R} T \mid T^{-1} \mathcal{G}; \pi)$ by replacing each signature in $\mathcal{R}$ or $\mathcal{G}$
with the corresponding signature in $\mathcal{R}T$ or $T^{-1} \mathcal{G}$.

\begin{theorem}[Valiant's Holant Theorem~\cite{valiant-accidental-alg}]
\label{thm:holant-thm}
If there is a holographic transformation mapping signature grid $\Omega$ to $\Omega'$, then $\hol_{\Omega} = \hol_{\Omega'}$.
\end{theorem}

Therefore, an invertible holographic transformation does not change the complexity of the Holant problem in the bipartite setting.
Furthermore, if $T$ is orthogonal, then $\beq T^{\otimes 2}$ takes the matrix form $ T\transpose I T = I$, so it preserves binary equality.
This means that an orthogonal holographic transformation can be used freely in the standard setting.
\begin{corollary}
\label{thm:orthogonal-holant-thm}
Suppose $T$ is an orthogonal matrix, $T\transpose T = I$, and let $\Omega = (G, \mathcal{F}, \pi)$ be a signature grid.
Under a holographic transformation by $T$, we get a new signature grid $\Omega' = (G, T \mathcal{F}, \pi)$ and $\hol_{\Omega} = \hol_{\Omega'}$.
\end{corollary}

\subsection{Notations} \label{subsec:notations}

We define some common vectors and matrices that we use throughout this paper.
For  ease of notation, we do not distinguish between column vectors and row vectors in this paper when the intention is clear from the context.

We use $\sbto = (1, 0) \in \mathbb{C}^2$ and $\sbtt = (0, 1) \in \mathbb{C}^2$ to denote the standard basis vectors in $\mathbb{C}^2$.
We use $\mathbf{e}_1 = (1, 0, 0) \in \mathbb{C}^3$, $\mathbf{e}_2 = (0, 1, 0) \in \mathbb{C}^3$, and $\mathbf{e}_3 = (0, 0, 1) \in \mathbb{C}^3$ to denote the standard basis vectors in $\mathbb{C}^3$.
Let $\bzp = \frac{1}{\sqrt{2}} (1, \i, 0)$ and $\bzm = \frac{1}{\sqrt{2}} (1, -\i, 0)$,
often treated as column vectors, especially acted on the left by a matrix.
Let $\Z = \frac{1}{\sqrt{2}}  \begin{bsmallmatrix}
  1 & 1 \\ \i & - \i
\end{bsmallmatrix}$.
We have  $\Z^{-1} = \frac{1}{\sqrt{2}} \begin{bsmallmatrix}
  1 & - \i  \\ 1 & \i
\end{bsmallmatrix}$.
Let
$\Y = 
\begin{bmatrix}
  \Z & 0 \\ 0 & 1
\end{bmatrix}
$,

and  $\X = \Y\transpose \Y = \begin{bsmallmatrix}
  0 & 1 & 0 \\
  1 & 0 & 0 \\
  0 & 0 & 1
\end{bsmallmatrix}$.
Note that $\Y \mathbf{e}_1 = \bzp$, $\Y \mathbf{e}_2 = \bzm$ and $\Y \mathbf{e}_3 = \mathbf{e}_3$.

For two vectors $\mathbf{x}, \mathbf{y} \in \mathbb{C}^n$, 
we use $\mathbf{x} \sim \mathbf{y}$ to denote they are linearly dependent.
Note that
$0 \sim \mathbf{x}$ for any $\mathbf{x} \in \mathbb{C}^n$.
For two tuples of vectors $X = (\mathbf{x}_1, \ldots, \mathbf{x}_m)$ and $Y = (\mathbf{y}_1, \ldots, \mathbf{y}_m)$, we write $X \sim Y$ if $\mathbf{x}_i \sim \mathbf{y}_i$ for all $1 \le i \le m$ after some reordering.
Throughout this paper, the symbol $\langle \mathbf{u}, \mathbf{v} \rangle$ for $\mathbf{u}, \mathbf{v} \in \mathbb{C}^n$ denotes the dot product,
i.e. $\langle \mathbf{u}, \mathbf{v} \rangle = \mathbf{u}\transpose \mathbf{v} = \sum u_i v_i$.
We say $\mathbf{u}, \mathbf{v} \in \mathbb{C}^n$ are orthogonal if $\langle \mathbf{u}, \mathbf{v} \rangle = 0$.

We use the symbol $*$ to denote some number that does not matter for the argument, especially as an entry of a vector or a matrix.
So, $(*, *, 0)$ may be the vector $(0, 2, 0)$, $(1, 5, 0)$ or $(-3, 0, 0)$ and so on.

A signature $\mathbf{F}$ of arity $k$ is \textit{symmetric} if $\mathbf{F}(x_1, \ldots, x_k) = \mathbf{F}(x_{\sigma(1)}, \ldots, x_{\sigma(k)})$ for all $\sigma \in S_k$, the symmetric group.
It can be shown that a symmetric signature is degenerate if and only if $\mathbf{F} = \mathbf{u}^{\otimes k}$ for some unary $\mathbf{u}$.

We consider a signature $\mathbf{F}$ and its nonzero multiple $c \mathbf{F}$ as the same signature, since replacing $\mathbf{F}$ by $c \mathbf{F}$ only introduces a easily computable nonzero global factor in the Holant value.

For two vectors $\mathbf{u}, \mathbf{v} \in \mathbb{C}^n$,
we define an order-$k$ symmetric tensor 
\[\sym(\mathbf{u}, \ldots, \mathbf{u}, \mathbf{v}) = (\mathbf{u} \otimes \cdots \otimes \mathbf{u} \otimes \mathbf{v})
+  (\mathbf{u} \otimes \cdots \otimes \mathbf{u} \otimes \mathbf{v} \otimes \mathbf{u}) 
 +  \cdots +  (\mathbf{v} \otimes \mathbf{u} \otimes \cdots \otimes \mathbf{u}) \;, \]
the sum of $k$ terms where $\mathbf{v}$ is moving from right to left. 
For example, $\sym(\sbto, \sbto, \sbtt) = \sbto \otimes \sbto \otimes \sbtt + \sbto \otimes \sbtt \otimes \sbto + \sbtt \otimes \sbto \otimes \sbto$.
Note that there is no normalizing factor, and we are not summing up over all possible permutations.

A symmetric signature $\mathbf{F}$ on $k$ Boolean variables $\{0, 1\}$ can be expressed as $[f_0, f_1, \ldots, f_k]$ where $f_i$ is the value of $\mathbf{F}$ 
on inputs of Hamming weight $i$.
The Boolean domain signature $[0, 1, 0, 0]$ 
({\sc ExactOne} of arity 3)  can be written as $\sym(\sbto, \sbto, \sbtt)$.

In this paper, we focus on signatures on domain size $3$, and we use the symbols $\{\db, \dg, \dr\}$ to denote the domain elements.
We use the notation in \cref{fig:ternary-signature-notation} to denote a symmetric ternary signature on domain $3$,
\begin{figure}
 \centering
\tersig{0.7}{f_{\db\db\db} f_{\db\db\dg} f_{\db\db\dr} f_{\db\dg\dg} f_{\db\dg\dr} f_{\db\dr\dr} f_{\dg\dg\dg} f_{\dg\dg\dr} f_{\dg\dr\dr} f_{\dr\dr\dr}}
\caption{Triangular representation of a symmetric ternary domain $3$ signature.}
\label{fig:ternary-signature-notation}
\end{figure}
where $f_{i j k}$ is the value  $\mathbf{F}(i,j,k)$.
We also represent the signature as $\mathbf{F} = [f_{\db\db\db}; f_{\db\db\dg}, f_{\db\db\dr}; f_{\db\dg\dg}, f_{\db\dg\dr}, f_{\db\dr\dr}; f_{\dg\dg\dg}, f_{\dg\dg\dr}, f_{\dg\dr\dr}, f_{\dr\dr\dr}]$ by writing the values in the triangular representation row by row.
This triangular notation can be generalized to signatures of other arities, where an arity $k$ signature can be denoted by a triangle with side length $k + 1$.

A binary signature $\mathbf{F}$ (not necessarily symmetric) can be expressed as a $\lvert D \rvert \times \lvert D \rvert$ matrix $M_{\mathbf{F}}$, where the entry $(i, j) \in D \times D$
is the value  $\mathbf{F}(i, j)$.
For  ease of notation, we use the term matrix and binary signature interchangeably, and use $\mathbf{F}$ to refer to both a signature and its matrix $M_{\mathbf{F}}$.
To fix an ordering, binary signature on domain $3$ is expressed as the following:
\[
  \mathbf{F} = \begin{bmatrix}
    f_{\db\db} & f_{\db \dg} & f_{\db \dr} \\
    f_{\dg \db} & f_{\dg \dg} & f_{\dg \dr} \\
    f_{\dr \db} & f_{\dr \dg} & f_{\dr \dr} 
  \end{bmatrix} \, .
\]
A binary signature $\mathbf{F}$ is symmetric  iff its matrix is symmetric.

Let $\mathbf{G}$ be a binary signature and $\mathbf{F}$ be a symmetric signature of arity $k \ge 1$.
We use $\mathbf{G} \mathbf{F} := \mathbf{G}^{\otimes k} \mathbf{F}$ to denote the gadget constructed by attaching a copy of $\mathbf{G}$ at each edge of $\mathbf{F}$.
An example for $k = 3$ is shown in \cref{fig:gf}.
\begin{figure}
\centering
\begin{subfigure}[b]{0.45\textwidth}
\centering
      \begin{tikzpicture}[scale=0.7]
        \node[mycirc, label=right:{$\mathbf{F}$}] (f1) at (19, 0) {};
        \node[mycirc, label=above:{$\mathbf{G}$}] (f2) at (17, 0) {};

        \node[mycirc, label=above:{$\mathbf{G}$}] (g11) at (17, 2) {};
        \node[mycirc, label=above:{$\mathbf{G}$}] (g21) at (17, -2) {};

        \draw (f1) to[out=90, in=0] (g11) -- (16, 2);
        \draw (f1) to[out=-90, in=0] (g21) -- (16, -2);

        \draw (f1) -- (f2) -- (16, 0);
  \end{tikzpicture}
  \caption{Gadget $\mathbf{G} \mathbf{F}$}\label{fig:gf}
  \end{subfigure}
  \begin{subfigure}[b]{0.45\textwidth}
\centering
      \begin{tikzpicture}[scale=0.7]
        \node[mycirc, label=below:{$\mathbf{F}$}] (f1) at (0, 0) {};
        \node[draw,circle, label=above:{$\mathbf{u}$}] (f2) at (0, 2) {};
        \draw (f1) -- (f2);
        \draw (-2, 0) -- (f1) -- (2, 0);
  \end{tikzpicture}
  \caption{Gadget $\langle \mathbf{F}, \mathbf{u} \rangle$}\label{fig:fu}
  \end{subfigure}
  \caption{Common Gadget Constructions.}
\end{figure}
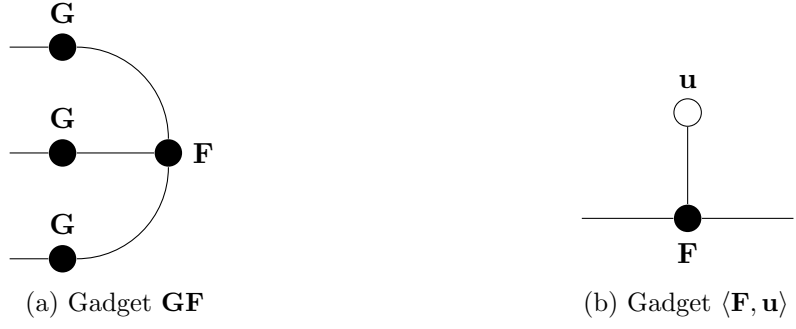
If $\mathbf{F}$ is written in a tensor form, 
$\mathbf{F} = \mathbf{v}_1^{\otimes k} + \cdots + \mathbf{v}_s^{\otimes k}$ for $\mathbf{v}_i \in \mathbb{C}^{\lvert D \rvert}$
we can easily check that $\mathbf{G} \mathbf{F} = (\mathbf{G} \mathbf{v}_1)^{\otimes k} + \cdots + (\mathbf{G} \mathbf{v}_s)^{\otimes k}$.

Another gadget construction common in this paper is connecting a unary signature, shown in \cref{fig:fu}.
For  ease of notation, we identify a vector $\mathbf{u} \in \mathbb{C}^{\lvert D \rvert}$ with a unary signature on domain $D$.
Let $\mathbf{u} \in \mathbb{C}^{\lvert D \rvert}$ and $\mathbf{F}$ be a symmetric signature on domain $D$ of arity $k$.
Then, $\langle \mathbf{F}, \mathbf{u} \rangle$ is the arity $k - 1$ gadget obtained by connecting $\mathbf{u}$ to any edge of $\mathbf{F}$.
Since $\mathbf{F}$ is symmetric, the choice of the edge does not matter. 
Given the triangular representation of $\mathbf{F}$, we can calculate $\langle \mathbf{F}, \mathbf{u} \rangle$ easily as described in \cref{fig:computing-f-u}. 
Each entry of $\mathbf{G} = \langle \mathbf{F}, (x, y, z) \rangle$ is obtained by a linear combination $\mathbf{G}(i,j) = x \mathbf{F}(\db, i, j) + y \mathbf{F}(\dg, i, j) + z \mathbf{F}(\dr, i, j)$;
i.e.\ we start at any entry on the first three rows in the triangular representation of $\mathbf{F}$, and then form a linear combination with coefficients $x, y, z$ in a counterclockwise order involving the three entries forming a small triangle. 
\begin{figure}
\centering
  \begin{tikzpicture}[scale=1, every node/.style={scale=1}]
  \node (a) at (0, 0) {$f_{\db\db\db}$};
  \node (b) at (-1, -1) {$f_{\db\db\dg}$};
  \node (c) at (1, -1) {$f_{\db\db\dr}$};
  
  \node at(-2, -2) {$f_{\db\dg\dg}$};
  \node at (0, -2) {$f_{\db\dg\dr}$};
  \node at (2, -2) {$f_{\db\dr\dr}$};
  
  \node at (-3, -3) {$f_{\dg\dg\dg}$};
  \node at (-1, -3) {$f_{\dg\dg\dr}$};
  \node at (1, -3) {$f_{\dg\dr\dr}$};
  \node at (3, -3) {$f_{\dr\dr\dr}$};

  \draw[->] (a) -- (b);
  \draw[->] (b) -- (c);

  \draw[dashed] (0, 0.7) -- (-2, -1.3) -- (2, -1.3) -- cycle;

  \node (aa) at (8, -0.5) {$g_{\db\db}$};

\node[anchor=west] at (8.3, -0.5) {$ = x f_{\db\db\db} + yf_{\db\db\dg} + zf_{\db\db\dr}$};

  \node at (7, -1.5) {$g_{\db\dg}$};
  \node at (9, -1.5) {$g_{\db\dr}$};
  
  \node at(6, -2.5) {$g_{\dg\dg}$};
  \node (bb) at (8, -2.5) {$g_{\dg\dr}$};
  \node at (10, -2.5) {$g_{\dr\dr}$};

  \draw[->] (2, -0.5) -- (aa);
\end{tikzpicture}
\caption{$\mathbf{F}$ and $\mathbf{G} = \langle \mathbf{F}, (x, y, z) \rangle$.} \label{fig:computing-f-u}
\end{figure}
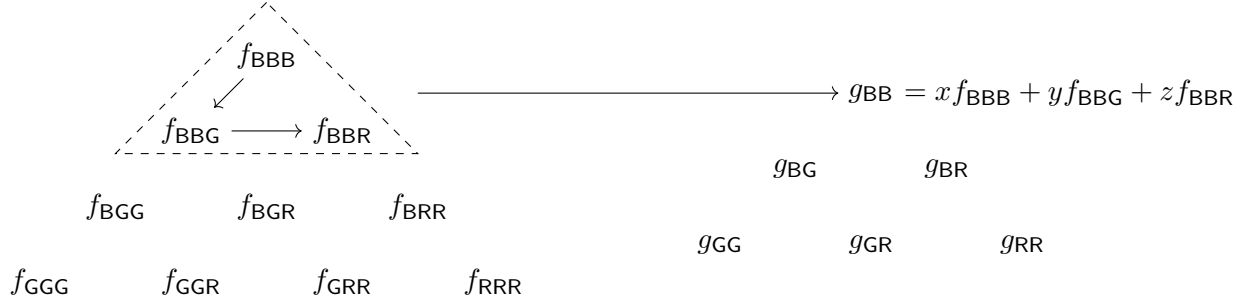

We use $\hol_2$ to denote the Holant problem on Boolean domain $\{0, 1\}$,
and $\hol_3$ to denote the Holant problem on domain $\{\db, \dg, \dr\}$.

For a domain $3$ signature, we use the symbol $\mathbf{F}\domres{i, j}$ to denote the Boolean domain signature obtained by restricting the inputs of $\mathbf{F}$ to be from $\{i, j\} \subset \{\db, \dg, \dr\}$.
We extend this notation to a set of signatures $\mathcal{F}$: $\mathcal{F}\domres{ij} := \{ \mathbf{F}\domres{ij} : \mathbf{F} \in \mathcal{F} \}$.
When we take a domain restriction of a domain $3$ signature to Boolean domain, we will identify the domains in the following way:
\begin{itemize}
  \item $\mathbf{F}\domres{\db \dg}$ is viewed as identifying $\db$ as $0$ and $\dg$ as $1$.
  \item $\mathbf{F}\domres{\db \dr}$ is viewed as identifying $\db$ as $0$ and $\dr$ as $1$.
  \item $\mathbf{F}\domres{\dg \dr}$ is viewed as identifying $\dg$ as $0$ and $\dr$ as $1$.
\end{itemize}
For example, $\mathbf{F}\domres{\db\dg} = [f_{\db\db\db}, f_{\db\db\dg}, f_{\db\dg\dg}, f_{\dg\dg\dg}]$ for the signature $\mathbf{F}$ in \cref{fig:ternary-signature-notation}, i.e.\ the signature formed by the four values at the left outer edge of the triangle.

Let $\mathbf{F}$ be a domain $3$ signature.
We define $\supp \mathbf{F}$ to be the set of inputs for which $\mathbf{F}$ is nonzero.
For brevity, we say $\mathbf{F}$ is supported on $\db\dg$ if $\supp \mathbf{F} \subseteq \{\db,\dg\}^*$, and similarly for $\db\dr$ and $\dg\dr$.

$\symspace{k}{n}$ is the set of order-$k$ symmetric tensor over $\mathbb{C}^n$.
In our setting, we may think of it as the set of symmetric complex-valued $k$-arity signatures over a domain of size $n$.
\begin{definition}[Definition 4.1 of~\cite{symmetric-tensors-comon}]
  The symmetric rank of $A \in \symspace{k}{n}$ is defined as 
  \[
    \rank(A) := \min \{s : A = \sum_{i = 1}^s \mathbf{y}_i^{\otimes k}\} \,.
  \]
\end{definition}
By Lemma 4.2 in~\cite{symmetric-tensors-comon}, symmetric rank always exists.

\subsection{Known Dichotomy Theorems}
\begin{definition}[Definition 2.9 in~\cite{holant-textbook}]\label{def:boolean-signature-type-i-ii}
  A signature $[x_0, x_1, \ldots, x_n]$, where $n \ge 2$, has type $\typei(a, b)$, if there exist $a$ and $b$ (not both $0$), such that $a x_k + b x_{k+1} = a x_{k+2}$ for $0 \le k \le n -2$.
  We say  it is of type $\typeii$, if $x_k = - x_{k+2}$ for $0 \le k \le n - 2$.
\end{definition}
\begin{theorem}[Theorem 2.12 in~\cite{holant-textbook}]\label{thm:dich-sym-Boolean}
Let $\mathcal{F}$ be a set of nondegenerate symmetric signatures over $\mathbb{C}$ in Boolean variables.
Then $\holbs(\mathcal{F})$ is computable in polynomial time for the following three classes of $\mathcal{F}$.
In all other cases, $\holbs(\mathcal{F})$ is \#\P-hard.
\begin{enumerate}
  \item Every signature in $\mathcal{F}$ is of arity $\le 2$.
  \item There exists $a$ and $b$ (not both $0$, depending only on $\mathcal{F}$), such that every signature in $\mathcal{F}$ either (1) has type $\typei(a,b)$ or (2) has arity $2$ and is of the form $[2a \lambda, b \lambda, -2 a \lambda]$.
  \item Every signature in $\mathcal{F}$ either (1) has type $\typeii$ or (2) has arity $2$ and is of the form $[\lambda, 0, \lambda]$.
\end{enumerate}
\end{theorem}
Let  $\mathbf{F}$ be a Boolean domain signature of arity $n$.
We can compute the following.
\begin{enumerate}
    \item If $\mathbf{F}$ is type $\typei(a, b)$ for $b \ne \pm 2\i a$, then there exists an orthogonal $T$ such that $T \mathbf{F} = x\sbto^{\otimes n} + y\sbtt^{\otimes n} = [x, 0, \ldots, 0, y]$ for some $x, y \in \mathbb{C}$. 
    Under this $T$, $T [2a\lambda, b\lambda - 2a\lambda] = [0, *, 0]$.
    \item If $\mathbf{F}$ is type $\typei(a, b)$ for $b = \pm 2\i a$, 
    then $\Z^{-1}\mathbf{F}$ is
    $x \sbto^{\otimes n} + y \sym(\sbto, \ldots, \sbto, \sbtt) = [x, y, 0, \ldots, 0]$ 
    or $[0, \ldots, 0, y, x]$ (switching $\sbto, \sbtt$),
    for some $x, y \in \mathbb{C}$.
    Note that $[2a\lambda, b\lambda - 2a\lambda]$ is degenerate.
    \item If $\mathbf{F}$ is type $\typeii$, then $\Z^{-1} \mathbf{F} = x\sbto^{\otimes n} + y\sbtt^{\otimes n} = [x, 0, \ldots, 0, y]$ for some $x, y \in \mathbb{C}$.
    Note that $\Z^{-1}[\lambda, 0, \lambda] = [0, *, 0]$.
\end{enumerate}
A $\holbs(\mathcal{F})$ for $\mathcal{F}$, a set of not necessary symmetric signatures, is also known \cite{holant-star-asymmetric}, but we do not use it in this paper. 

A $\holts(\mathbf{F})$ dichotomy for a single symmetric, complex-valued ternary signature is known.
\begin{theorem}[Theorem 3.1, 3.2 in~\cite{domain3-single-signature}]
\label{thm:single-ternary-dichotomy}
Let $\mathbf{F}$ be a symmetric, complex valued, ternary signature over domain $\{\db, \dg, \dr\}$.
Then $\holts(\mathbf{F})$ \sph unless $\mathbf{F}$ is one of the following forms, in which case the problem is solvable in polynomial time.
\begin{enumerate}
  \item There exists three vectors $\boldsymbol{\alpha}, \boldsymbol{\beta}, \boldsymbol{\gamma} \in \mathbb{C}^3$ such that $\langle \boldsymbol{\alpha}, \boldsymbol{\beta} \rangle = 0$, $\langle \boldsymbol{\alpha}, \boldsymbol{\gamma} \rangle = 0$, and $\langle \boldsymbol{\beta}, \boldsymbol{\gamma} \rangle = 0$, and 
    \[
      \mathbf{F} = \boldsymbol{\alpha}\teh + \boldsymbol{\beta}\teh + \boldsymbol{\gamma}\teh \,.
    \]
  \item There exists three vectors $\boldsymbol{\alpha}, \boldsymbol{\beta}_1, \boldsymbol{\beta}_2 \in \mathbb{C}^3$ such that 
    $\langle \boldsymbol{\alpha}, \boldsymbol{\beta}_1 \rangle = 0$, $\langle \boldsymbol{\alpha}, \boldsymbol{\beta}_2 \rangle = 0$, 
    $\langle \boldsymbol{\beta}_1, \boldsymbol{\beta}_1 \rangle = 0$,
    $\langle \boldsymbol{\beta}_2, \boldsymbol{\beta}_2 \rangle = 0$ and 
    \[
      \mathbf{F} = \boldsymbol{\alpha}\teh + \boldsymbol{\beta}_1\teh + \boldsymbol{\beta}_2\teh \,.
    \]
  \item There exists two vectors $\boldsymbol{\beta}, \boldsymbol{\gamma} \in \mathbb{C}^3$ and a signature $\mathbf{F}_{\boldsymbol{\beta}}$ of arity $3$, such that
    $\boldsymbol{\beta} \ne 0$, $\langle \boldsymbol{\beta}, \boldsymbol{\beta} \rangle = 0$, $\langle \mathbf{F}_{\boldsymbol{\beta}}, \boldsymbol{\beta} \rangle = 0$ and
    \[
      \mathbf{F} = \mathbf{F}_{\boldsymbol{\beta}} + 
      \sym(\boldsymbol{\beta}, \boldsymbol{\beta}, \boldsymbol{\gamma}) \,.
    \]
\end{enumerate}
These cases are equivalent to an existence of an orthogonal transformation $T$, such that 
\begin{enumerate}
  \item 
    For some $a, b, c \in \mathbb{C}$
    \[
      T\teh \mathbf{F} = a \mathbf{e}_1\teh + b \mathbf{e}_2\teh + c \mathbf{e}_3\teh \,.
    \]
  \item 
    For some $c \ne 0$ and $\lambda \in \mathbb{C}$,
    \[
      c T\teh \mathbf{F} =  \bzp \teh + \bzm\teh + \lambda \mathbf{e}_3 \teh  \,.
    \]
  \item 
    For $\epsilon \in \{0, 1\}$,
    \[
     T\teh \mathbf{F} = 
     \mathbf{F}_0 + \epsilon 
     \sym(\bzp, \bzp, \bzm)
    \]
    where $\mathbf{F}_0$ satisfies the annihilation condition $\langle \mathbf{F}_0,
    \bzp \rangle = 0$.
\end{enumerate}
\end{theorem}

\subsection{Generalized Orthogonal Group}\label{sec:spin-rep}
A fundamental technique in Holant theory is  holographic transformation (\cref{thm:holant-thm}).
Most often
we use orthogonal transformations,
as they keep the binary {\sc Equality} constraint invariant (represented by the identity matrix) (\cref{thm:orthogonal-holant-thm}).
In this paper, we will see that in domain size 3, often it is more illuminating to use
transformations from a more general group, in particular, the \textit{generalized orthogonal group} that
keeps the quadratic form $\mathbf{X}$ invariant,  and analyze the problem in 
the bipartite $\holts(\mathbf{X} \mid \blank)$ setting.

\begin{theorem}[Generalized Orthogonal Holant Theorem]
    Let $\Abold \in \mathbb{C}^{q \times q}$ be a symmetric matrix of rank $q$.
    Let $\ortho(q, \mathbb{C}, \Abold) = \{M \in \gl(q, \mathbb{C}) : M \transpose  \Abold M = \Abold \}$. 
    Let $\Omega = (G; \, \Abold \mid \mathcal{F} ;\, \pi)$ be a bipartite signature grid.
    Under a holographic transformation by $M \in \ortho(q, \mathbb{C}, \Abold)$,
    we get a new grid $\Omega' = (G; \, \Abold  \mid  M \mathcal{F}; \, \pi)$ such that $\hol_{\Omega} = \hol_{\Omega'}$.
\end{theorem}

Corresponding to $\Abold = I$ or $\X$, we have
two specific (generalized) orthogonal groups, $\ortho(3, \mathbb{C}) := \ortho(3, \mathbb{C}, I) = \{M \in \mathbb{C}^{3 \times 3}: M\transpose M = I\}$ and 
$\ortho(3, \mathbb{C}, \mathbf{X}) = \{M \in \mathbb{C}^{3 \times 3}: M\transpose \mathbf{X} M =  \mathbf{X}\}$.
(As both $I$ and $\X$ are of full rank, these $M$ are automatically in $\gl(3, \mathbb{C})$.)
One can also view them as the group of matrices that preserve the quadratic forms $\mathbf{u}\transpose I \mathbf{u}$ and $\mathbf{u}\transpose \X \mathbf{u}$ respectively.
Since a nonzero constant factor of signatures does not impact the computational complexity, we may 
use $\so(3, \mathbb{C}) = \{M \in \ortho(3, \mathbb{C}) : \det(M) = 1\}$ and
$\so(3, \mathbb{C}, \mathbf{X}) = \{M \in \ortho(3, \mathbb{C}, \mathbf{X}) : \det(M) = 1\}$.
It turns out that $\so(3, \mathbb{C}) \cong \so(3, \mathbb{C}, \mathbf{X})$ as abstract groups, but
by a nontrivial isomorphism. We will provide an explicit isomorphism below and it will be
important for our proof.

More crucially,
we will use the isomorphism $\psl(2, \mathbb{C}) \cong \so(3, \mathbb{C})$ to get a 
parametrization via $\speciallinear(2, \mathbb{C})$, which is also called \textit{spin representation} (but note that this is over $\mathbb{C}$; the analogous  spin representation over $\mathbb{R}$
is better known, e.g., in quantum mechanics).
One benefit of this parametrization is it facilitates the analysis of (generalized) orthogonal transformations.
We note that the $\mathbb{C}$-dimension of $\so(3, \mathbb{C})$ (and of
$\so(3, \mathbb{C}, \mathbf{X})$), is just 3 (as seen from its Lie algebra); this is a lot smaller than the number of entries in the matrix.

Consider $V = \{M \in \mathbb{C}^{2 \times 2} : \tr(M) = 0\}$.
$V$ is a $3$ dimensional vector space over $\mathbb{C}$ with bilinear form $\tr(AB)$.
One can construct a representation of $\speciallinear(2, \mathbb{C})$ (and of $\psl(2, \mathbb{C})$) on $V$.
Note that if $M \in \speciallinear(2, \mathbb{C})$, then $M^{-1} A M \in V$ for any $A \in V$.
  Choose a basis of $V$ as the Normalized Pauli matrices,
  \begin{align*}
\frac{1}{\sqrt{2}} X = \frac{1}{\sqrt{2}} \begin{bmatrix}
    0 & 1 \\ 1 & 0
  \end{bmatrix} \,, && 
  \frac{1}{\sqrt{2}} Y = \frac{1}{\sqrt{2}} \begin{bmatrix}
    0 & -\i \\ \i & 0
  \end{bmatrix} \,, && 
  \frac{1}{\sqrt{2}} Z = \frac{1}{\sqrt{2}} \begin{bmatrix}
    1 & 0 \\ 0 & -1
  \end{bmatrix} \, .
  \end{align*}
  We see that $M = \begin{bsmallmatrix} a & b \\ c & d \end{bsmallmatrix} \in \speciallinear(2, \mathbb{C})$ acts on $V$ by 
  \begin{equation}\label{eqn:phiM-parameterization}
      \phi(M) = \begin{bmatrix}
          \frac{1}{2}(a^2 - b^2 - c^2 + d^2) & \frac{1}{2} \i (a^2 - b^2 + c^2 - d^2) & -ac + bd \\
          -\frac{1}{2}\i(a^2 + b^2 - c^2 - d^2) & \frac{1}{2}(a^2 + b^2 + c^2 + d^2) & \i (ac + bd) \\
          -ab + cd & - \i (ab + c d) & bc + ad
        \end{bmatrix} \, .
      \end{equation}
      identifying $\frac{1}{\sqrt{2}} X$, $\frac{1}{\sqrt{2}} Y$ and $\frac{1}{\sqrt{2}} Z$
      with $\mathbf{e}_1$,   $\mathbf{e}_2$ and  $\mathbf{e}_3$, respectively.
      Under this identification, for $x_i, y_i, z_i \in \mathbb{C}$,
      the bilinear form
      \[
      \tr\left(\frac{1}{\sqrt{2}} \begin{bmatrix}
        z_1 & x_1 - \i y_1 \\
        x_1 + \i y_1 & -z_1
        \end{bmatrix} 
        \frac{1}{\sqrt{2}} \begin{bmatrix}
        z_2 & x_2 - \i y_2 \\
        x_2 + \i y_2 & -z_2
        \end{bmatrix}
        \right) = x_1 x_2 + y_1 y_2 + z_1 z_2
    \]
    can be identified as
    \[
      \begin{bmatrix}
        x_1 & y_1 & z_1
        \end{bmatrix} I \begin{bmatrix}
        x_2 \\ y_2 \\ z_2
      \end{bmatrix} = x_1 x_2 + y_1 y_2 + z_1 z_2 \, .
    \]
      One can show that $\phi$ is surjective,
      and hence we obtain a parametrization of $\so(3, \mathbb{C})$. 
      
    Choosing a different basis for $V$ gives a nicer description of $\so(3, \mathbb{C}, \X)$.
    With
    \begin{align*}
        P = \begin{bmatrix}
          0 & 1 \\ 0 & 0
          \end{bmatrix} \,,
          &&
          Q = \begin{bmatrix}
          0 & 0 \\ 1 & 0
        \end{bmatrix} \,,
        &&
        R = \frac{1}{\sqrt{2}}\begin{bmatrix}
          1 & 0 \\ 0 & -1
      \end{bmatrix} \, ,
    \end{align*}
    we see that $M = \begin{bsmallmatrix}
      a & b \\ c & d
    \end{bsmallmatrix} \in \speciallinear(2, \mathbb{C})$ acts on $V$ by 
    \begin{equation}\label{eq:paramaterize-so3cx}
      \rho(M) = 
      \begin{bmatrix}
        d^2 & -b^2 & \sqrt{2}bd \\
        -c^2 & a^2 & -\sqrt{2}ac \\
        \sqrt{2} cd & -\sqrt{2}ab & ad + bc
      \end{bmatrix}
    \end{equation}
    identifying $P$ with $\mathbf{e}_1$, $Q$ with $\mathbf{e}_2$ and $R$ with $\mathbf{e}_3$.
    Under this identification, the bilinear form 
    \[
      \tr\left(\begin{bmatrix}
        r_1/\sqrt{2} & p_1 \\
        q_1 & -r_1/\sqrt{2}
        \end{bmatrix} \begin{bmatrix}
        r_2/\sqrt{2} & p_2 \\ q_2 & -r_2/\sqrt{2}
      \end{bmatrix}
      \right) = p_1 q_2 + q_1 p_2 + r_1 r_2
    \]
    can be identified as
    \[
      \begin{bmatrix}
        p_1 & q_1 & r_1
        \end{bmatrix} \X \begin{bmatrix}
        p_2 \\ q_2 \\ r_2
      \end{bmatrix} = p_1 q_2 + q_1 p_2 + r_1 r_2 \, .
    \]
    We see that $\rho(M)\transpose \X \rho(M) = \X$ and thus $\rho(M) \in \so(3, \mathbb{C}, \X)$.
    Again, $\rho$ is surjective so we get a parametrization of $\so(3, \mathbb{C}, \X)$.
    
    Recall that $\X = \Y\transpose \Y$.
    Then, we have $\phi(M) = \Y \rho(M) \Yinv$, which gives us the isomorphism $\so(3, \mathbb{C}) \cong \so(3, \mathbb{C}, \X)$.

We remark that $V$ is also the Lie algebra $\mathfrak{sl}(2, \mathbb{C})$ and the action we defined is the adjoint representation. 
The basis $P, Q, R$ are scalar multiples of the basis chosen in Section 4.6 of~\cite{hall-lie-algebra}.
We also remark that $\X$ is corresponds to the bilinear form chosen in Section 18.1 in~\cite{fulton-rep-theory}.

\subsection{Types of Signatures}\label{subsec:types-of-sig}
Let $\mathbf{F}$ be a symmetric ternary signature on domain $3$. 
We classify $\mathbf{F}$ according to an existence of an orthogonal transform $T$.
In~\cite{domain3-single-signature}, three types of tractable signatures were identified. 
\begin{enumerate}
    \item \casea: $T \mathbf{F} = a \mathbf{e}_1\teh + b \mathbf{e}_2\teh + c \mathbf{e}_3\teh$ for some $a, b, c \in \mathbb{C}$.
    \item \caseb: $T \mathbf{F} = a \bzp\teh + b \bzm\teh + c \mathbf{e}_3\teh$ for some $a, b, c \in \mathbb{C}$.
    \item \casec: $T \mathbf{F} = \mathbf{F}_0 + \lambda \sym(\bzp, \bzp, \bzm)$ 
    for $\lambda \in \mathbb{C}$ and for some symmetric ternary $\mathbf{F}_0$ such that $\langle \mathbf{F}_0, \bzp \rangle = 0$.
\end{enumerate}
We will see that these three types warrant a refinement into 
finer categories as there are essential differences in the way they behave.
This will be fundamental in classifying how they interact with other signatures, which is
crucial to achieving a complexity dichotomy for a general {\it set} of signatures.
We define the following subtypes; they may seem arbitrary, especially the condition for \casecogen,
but it will become clearer later.
We will see in \cref{sec:casec} that the subtypes of \casec have distinct characteristics. 
We show in \cref{lem:types-distinct} that the following types are mutually exclusive. 
\begin{figure}
    \centering
    \tersig{0.7}{a {\lambda} b 0 0 c 0 0 0 d}
    \caption{A general \casec signature after $\Yinv T$ transform.}
    \label{fig:casec-general}
\end{figure}
\begin{definition}\label{def:types-ternary}
    Let $\mathbf{F}$ be a nondegenerate, symmetric, ternary signature.
We say $\mathbf{F}$ is one of the following types if there exists $T \in O(3, \mathbb{C})$ that satisfies the given condition. 
We say $\mathbf{F}$ is of the \emph{defining form} of the type if $T$ can be chosen to be the identity matrix. 
\begin{enumerate}
    \item \caseathree: $T \mathbf{F} = a \mathbf{e}_1\teh + b \mathbf{e}_2\teh + c \mathbf{e}_3\teh$ for nonzero $a, b, c \in \mathbb{C}$.
    \item \caseatwo: $T \mathbf{F} = a \mathbf{e}_1\teh + b \mathbf{e}_2\teh$ for nonzero $a, b \in \mathbb{C}$.
    \item \casebthree: $T \mathbf{F} = a \bzp\teh + b \bzm\teh + c \mathbf{e}_3\teh$ for nonzero $a, b, c \in \mathbb{C}$.
    \item \casebtwo: $T \mathbf{F} = a \bzp\teh + b \bzm\teh$ for nonzero $a, b \in \mathbb{C}$.
    \item \casecz: $T \mathbf{F}$ with the condition that $\langle T \mathbf{F}, \bzp \rangle = 0$.
    Equivalently, $\Yinv T \mathbf{F}$ is \cref{fig:casec-general} with $\lambda = 0$.
    \item \casecobg: 
    $T \mathbf{F} = a \bzp\teh + \lambda \sym(\bzp, \bzp, \bzm)$ for nonzero $\lambda \in \mathbb{C}$ and any $a \in \mathbb{C}$.
    Equivalently, 
    $\Yinv T \mathbf{F} = a \mathbf{e}_1\teh + \lambda \sym(\mathbf{e}_1, \mathbf{e}_1, \mathbf{e}_2)$ for nonzero $\lambda \in \mathbb{C}$ and any $a \in \mathbb{C}$;
    this is \cref{fig:casec-general} with $\lambda \ne 0$ and $b,c,d= 0$.
    \item \casecosep: 
    $T \mathbf{F} = a \bzp\teh + \lambda \sym(\bzp, \bzp, \bzm) + d \mathbf{e}_3\teh$ for nonzero $\lambda, d \in \mathbb{C}$ and any $a \in \mathbb{C}$.
    Equivalently, 
    $\Yinv T \mathbf{F} = a \mathbf{e}_1\teh + \lambda \sym(\mathbf{e}_1, \mathbf{e}_1, \mathbf{e}_2) + d \mathbf{e}_3\teh$, which is \cref{fig:casec-general} with $b, c = 0$ and $\lambda, d \ne 0$.
    \item \casecodz: $\Yinv T \mathbf{F}$ is \cref{fig:casec-general} with $\lambda, c \ne 0$ and $d = 0$.
    \item \casecogen: $\Yinv T \mathbf{F}$ is \cref{fig:casec-general} with $\lambda, d \ne 0$ and 
    $c^2 - c \lambda - bd \ne 0$,
    in particular $(b, c) \ne 0$.
\end{enumerate}
\end{definition}

We remark that the above conditions for type \casec signatures can be equivalently stated as 
an existence of $\iny{T} \in \so(3, \mathbb{C}, \X)$ such that 
$\iny{T} \Yinv \mathbf{F}$ taking the stated form, by the isomorphism
$\iny{T} = \Yinv T \Y$ for  $T \in \so(3, \mathbb{C})$ and thus
$\iny{T} \Yinv  \mathbf{F} =  \Yinv T \mathbf{F}$.

To differentiate between $\holts(\mathbf{F})$ and $\holts(\X \mid \Yinv \mathbf{F})$ and make the context clear,
we will use the notation $\iny{\mathbf{F}}$.
So, any signature
$\mathbf{F}$ (without $\iny{\;}$ ) is expected to occur in the setting $\holts(\mathbf{F})$ by default and
$\iny{\mathbf{F}}$ (with $\iny{\;}$ ) is expected to occur in the setting $\holts(\X \mid \iny{\mathbf{F}})$ by default unless otherwise specified (e.g., in \cref{prop:two-ternary-dichotomy-b-3} when we show the bipartite assumption can be removed).
We also denote a matrix in $\so(3, \mathbb{C}, \X)$ with $\iny{T}$.
We will often 
associate $\iny{\mathbf{F}} = \Yinv \mathbf{F}$ and $\iny{T} = \Yinv T \Y$.
Note that in this notation, $\iny{T \mathbf{F}} = \Yinv T \mathbf{F} = 
\iny{T} \Yinv \mathbf{F} = \iny{T} \iny{\mathbf{F}}$.

We define additional classes of signatures characterized by their support (not up to holographic transformation).
\begin{enumerate}
\item \geneq: We say $\mathbf{F}$ is \geneq if $\mathbf{F} = a \mathbf{e}_1^{\otimes n} + b \mathbf{e}_2^{\otimes n} + c \mathbf{e}_3^{\otimes n}$ for some $a, b, c \in \mathbb{C}$.
    \item $\mathcal{E}$: Let $\mathcal{E}$ denote the set of domain $3$ signatures $\mathbf{F}$  such that 
if $\mathbf{F}$ has arity $n$, then $\supp \mathbf{F} \subset \{\mathbf{x}, \mathbf{y}, \mathbf{z}\}$ for some $\mathbf{x} = (x_1, \ldots, x_n), \mathbf{y} = (y_1, \ldots, y_n), \mathbf{z} = (z_1, \ldots, z_n) \in \{\db, \dg, \dr\}^n$ such that 
for all $1 \le i \le n$, $x_i, y_i, z_i$ are all distinct.  
Note that \geneq signatures are in $\mathcal{E}$.
Conversely, if $\mathbf{F} \in \mathcal{E}$ is symmetric of arity $n \ge 3$, we can show that $\mathbf{F}$ is \geneq. 

\item $\casecsupp$: Let $\casecsupp$ denote the set of domain $3$ signatures $\mathbf{F}$ such that
if $\mathbf{F}$ has arity $n$, then $\supp \mathbf{F} \subseteq \{\dg \db^{n-1}, \db \dg \db^{n-2}, \ldots, \db^{n-1} \dg\} \cup \{\db, \dr\}^n$.
For example, if $\mathbf{F}$ is symmetric and of arity $3$, then its support has the form in 
\cref{fig:casec-general}.
\item $ij|k$ (\textit{domain separated}):  $\mathbf{F}$ is in this class if $\supp \mathbf{F} \subseteq \{i, j\}^* \cup \{k\}^*$.
In other words, $\mathbf{F}$ is zero on inputs that take values from both $\{i, j\}$ and $\{k\}$.
We say $\mathbf{F}$ is \textit{domain separated} to $\{i, j\}$ and $\{k\}$.
 Specific examples include \sepbg, \sepgr, and \sepbr.
We also refer to a matrix $\db \dg | \dr$ if it can be viewed as a $\db \dg | \dr$ binary signature.
For example, $M = \begin{bsmallmatrix}
  * & * & 0 \\
  * & * & 0 \\
  0 & 0 & *
\end{bsmallmatrix}$ is a $\db \dg | \dr$ matrix.
 
\item \swhelper{ij;k}: We call a binary signature/matrix of the form $\begin{bsmallmatrix}
    0 & 0 & a \\
    0 & 0 & b \\
    c & d & 0
\end{bsmallmatrix}$ for some $a, b, c, d \in \mathbb{C}$ as a \swbg signature/matrix.
The intuition is that a signature of this form swaps the inputs between the domains $\{\db, \dg\}$ and $\{\dr\}$.
Similarly, we call $\begin{bsmallmatrix}
    0 & a & 0 \\
    b & 0 & c \\
    0 & d & 0
\end{bsmallmatrix}$ as \swbr and $\begin{bsmallmatrix}
    0 & a & b \\
    c & 0 & 0 \\
    d & 0 & 0
\end{bsmallmatrix}$ as \swgr.
\end{enumerate}
We interchangeably use phrases such as ``$\mathbf{F}$ is \sepbg'' or ``$\mathbf{F} \in \sepbgmath$''. 

We note that \geneq, $\mathcal{E}$, $ij|k$ and $\casecsupp$ are linear spaces, and closed under connecting a unary signature, i.e.\ if $\mathbf{F} \in \geneq$, then $\langle \mathbf{F}, \mathbf{u} \rangle \in \geneq$ for any $\mathbf{u} \in \mathbb{C}^3$,  and similarly for
$\mathcal{E}$, $ij|k$ and $\casecsupp$.
We see that $\mathbf{F}$ is a  \casea signatures if there exists some $T \in \so(3, \mathbb{C})$ such that $T \mathbf{F} \in \mathcal{E}$;
$\mathbf{F}$ is a  \caseb signatures if there exists some $T \in \so(3, \mathbb{C})$ such that $\Yinv T \mathbf{F} \in \mathcal{E}$;
$\mathbf{F}$ is a  \casec signatures if there exists some $T \in \so(3, \mathbb{C})$ such that $\Yinv T \mathbf{F} \in \casecsupp$.

We will see later that one  crucial 
determining factor of the behavior of the signatures is the existence of a nonzero \textit{annihilator}, which we define shortly. 
We note that \caseatwo, \casebtwo, and \casecobg types correspond to \textit{essentially Boolean domain} signatures defined in~\cite{domain3-real-set}, 
as there exists an orthogonal transform that brings them to be supported on $\db\dg$, although \casecobg does not occur in~\cite{domain3-real-set} because it cannot be real valued.
Also, we note that a \casecz signature can be transformed to be supported on a domain of size 2, but only under a non-orthogonal transform.

The main factor that separates the types \caseatwo, \casebtwo, \casecobg, and \casecz from the others is the notion of the annihilator of a signature which we now introduce.
It turns out that annihilators also help us characterize the interaction between \caseatwo, \casebtwo, \casecobg, and \casecz signatures.

\begin{definition}\label{def:annihilator}
    Let $\mathbf{F}$ be a symmetric signature on domain $3$.
    We define $\ann(\mathbf{F}) = \{ \mathbf{u} \in \mathbb{C}^3 : \langle \mathbf{F}, \mathbf{u} \rangle = 0 \}$.
\end{definition}
Note that $\ann(\mathbf{F})$ is a linear subspace of $\mathbb{C}^3$.

We may also use the notion of symmetric rank to separate the types. 
In~\cite{domain3-real-set}, symmetric rank being $2$ or $3$ was the differentiating factor used to separate \caseatwo, \casebtwo from \caseathree, \casebthree.
However, as the following proposition shows, a signature that would be classified as essentially Boolean domain in spirit, specifically a \casecobg signature, has symmetric rank $3$. 
We will see that the existence of a nonzero annihilator captures the notion of a signature being essentially Boolean domain better than symmetric rank in our setting.

     \begin{proposition}[Proposition 5.6 in~\cite{symmetric-tensors-comon}]\label{prop:symmetric-rank-sym-sum}
        Let $\mathbf{u}, \mathbf{v} \in \mathbb{C}^n$ be  linearly independent. %
        Then for any $k > 1$, the following order-$k$ symmetric tensor
        $\sym(\mathbf{u}, \ldots, \mathbf{u}, \mathbf{v})$
        (where there are $k-1$ copies of $\mathbf{u}$)
        has symmetric rank $k$.
      \end{proposition}
      In particular, the Boolean domain signature
      $[0, 1, 0, 0] =\sym(
      \sbto, \sbto, \sbtt
            )$
      has symmetric rank $3$. 
      Similarly the domain $3$ signature $\sym(\mathbf{e}_1, \mathbf{e}_1, \mathbf{e}_2)$ has symmetric rank $3$. 
      By a similar argument as in~\cite{symmetric-tensors-comon} using a Vandermonde matrix, one can also show that $[x, 1, 0, 0]$ has symmetric rank $3$ for any $x \in \mathbb{C}$.
      \begin{proposition}\label{prop:two-annihilator-degenerate}
        Let $\mathbf{F}$ be a symmetric signature on domain $3$.
        If $\dim \ann(\mathbf{F}) \ge 2$, then $\mathbf{F}$ is degenerate.
      \end{proposition}
      \begin{proof}
      Let $\mathbf{u}_1, \mathbf{u}_2 \in \ann(\mathbf{F})$ be linearly independent.
        Then, there exists $M \in \gl(3, \mathbb{C})$ such that $\langle M \mathbf{F}, \mathbf{e}_1 \rangle = 0$ and $\langle M \mathbf{F}, \mathbf{e}_2\rangle = 0$. This can be only true if $M \mathbf{F} \sim \mathbf{e}_3^{\otimes n}$.
        This is best seen by operating directly on the triangular table for
        $M \mathbf{F}$ (see~\cref{fig:ternary-signature-notation}). It follows that $\mathbf{F}$ is degenerate.
      \end{proof}

By \cref{prop:two-annihilator-degenerate},
for  any nondegenerate symmetric signature
$\mathbf{F}$ on domain $3$, the annihilator space has
 $\dim \ann(\mathbf{F}) \le 1$, and so 
we can determine \textit{the annihilator}
$\ann(\mathbf{F}) = \spn(\mathbf{u})$ for a vector $\mathbf{u}
\in \mathbb{C}^3$ up to a scalar constant.
This includes the case when $\ann(\mathbf{F}) = \{0\}$ by setting $\mathbf{u} = 0$.

\begin{proposition}\label{prop:nonzero-annihilator-types}
    Let $\mathbf{F}$ be a nondegenerate symmetric ternary signature such that $\holts(\mathbf{F})$ is not \sph and
    $\ann(\mathbf{F}) = \spn(\mathbf{u}) \ne 0$.
    Then, $\mathbf{F}$ must be \caseatwo, \casebtwo, \casecobg, or \casecz.
\end{proposition}
\begin{proof}
    $\dim(\ann(\mathbf{F}))$ is invariant under holographic transformation by $\gl(3, \mathbb{C})$. 
    Since $\holts(\mathbf{F})$ is assumed to be not \sph, it must be one of the tractable types given above. 
    Thus, we will argue that the types \caseathree, \casebthree, \casecosep, \casecodz, and \casecogen do not have a nonzero annihilator by using the specific form given above for each type. 

    For \caseathree, it is easy to see that $\langle a \mathbf{e}_1\teh + b \mathbf{e}_2\teh + c \mathbf{e}_3\teh, (x, y, z) \rangle = ax \mathbf{e}_1 \tew + by \mathbf{e}_2 \tew + cz \mathbf{e}_3\tew$, which is a nonzero matrix if $x, y, z \ne 0$.
    For \casebthree, we note that $\Yinv (a \bzp\teh + b \bzm \teh + c \mathbf{e}_3\teh) = a \mathbf{e}_1\teh + b \mathbf{e}_2\teh + c \mathbf{e}_3\teh$.

    For \casecosep, \casecodz, and \casecogen, we look at the linear map $\mathbb{C}^3 \to \mathbb{C}$ defined by $\mathbf{v} \mapsto \langle \mathbf{F}, \mathbf{v} \rangle$. 
    In terms of \cref{fig:casec-general}, 
    $\ann(\mathbf{F}) \ne 0$ means
    \[
    \begin{bmatrix}
        a & \lambda & b \\
        \lambda & 0 & 0 \\
        b & 0 & c \\
        c & 0 & d
    \end{bmatrix} 
    \begin{bmatrix}
        x \\ y \\ z
    \end{bmatrix} = 0
    \]
    has a nontrivial solution.
    But since at least one of $c$ or $d$ is nonzero and $\lambda \ne 0$, the matrix is always rank $3$.
\end{proof}

We also need the following results on the representations of the signatures being well defined.
\begin{proposition}[Proposition 2.6 of (arXiv version)~\cite{domain3-real-set}]\label{prop:uniqueness-sum-lin-ind-tensor}
    Let $k \ge 3$. 
    Let $\mathbf{y}_1, \ldots, \mathbf{y}_s, \mathbf{z}_1, \ldots, \mathbf{z}_s \in \mathbb{C}^n$ such that 
    $\{\mathbf{y}_i\}_{1 \le i \le s}$ and $\{\mathbf{z}_i\}_{1 \le i \le s}$ are two sets of linearly independent vectors.
    Suppose 
    \[\sum_{i = 1}^s \mathbf{y}_i^{\otimes k} = \sum_{i = 1}^s \mathbf{z}_i^{\otimes k} \, .\] 
    Then, $(\mathbf{y}_1, \ldots, \mathbf{y}_s) \sim (\mathbf{z}_1, \ldots, \mathbf{z}_s)$.
\end{proposition}

      \begin{proposition}\label{prop:uniqueness-sym-form}
          Let nonzero $\mathbf{u}_1, \mathbf{u}_2, \mathbf{v}_1, \mathbf{v}_2 \in \mathbb{C}^n$ for $n \ge 2$ be such that 
          $\mathbf{u}_1 \not \sim \mathbf{v}_1$ and $\mathbf{u}_2 \not \sim \mathbf{v}_2$.
          Let $\mathbf{F}_i = \lambda_i \mathbf{u}_i^{\otimes k} + \sym(\mathbf{u}_i, \ldots, \mathbf{u}_i, \mathbf{v}_i)$ (where there are $k-1$ copies of $\mathbf{u}_i$) for some $\lambda_i \in \mathbb{C}$ for $i = 1, 2$.
          If $\mathbf{F}_1 = \mathbf{F}_2$, then $\mathbf{u}_1 \sim \mathbf{u}_2$.
      \end{proposition}
      \begin{proof}
          Suppose for a contradiction that $\mathbf{u}_1 \not \sim \mathbf{u}_2$. 
          Then, there exists some $\mathbf{w} \in \mathbb{C}^n$ such that $\langle \mathbf{w}, \mathbf{u}_1 \rangle = 0$
          and $\langle \mathbf{w}, \mathbf{u}_2 \rangle = 1$.
          If we connect $\mathbf{w}$ to the first $k - 1$ variables of $\mathbf{F}_1$ then we get $0$, because $k-1 \ge 2$ and $\mathbf{u}_1$ appears among the first $k-1$ tensor factors
          of each summand of  $\sym(\mathbf{u}_1, \ldots, \mathbf{u}_1, \mathbf{v}_1)$.
          On the other hand,  for $\mathbf{F}_2$ we get $\lambda_2 \mathbf{u}_2
          + (k-1) \langle \mathbf{w}, \mathbf{v}_2 \rangle  \mathbf{u}_2 +  \mathbf{v}_2
          = x \mathbf{u}_2 +  \mathbf{v}_2$ for some $x \in \mathbb{C}$.
          But this cannot be 0, since $\mathbf{u}_2 \not \sim \mathbf{v}_2$.
          This is a contradiction to $\mathbf{F}_1 = \mathbf{F}_2$. 
      \end{proof}

\begin{lemma}\label{lem:types-distinct}
    The types \caseathree, \caseatwo, \casebthree, \casebtwo, \casecz, \casecobg, \casecosep, \casecodz, and \casecogen are mutually exclusive.
\end{lemma}
\begin{proof}
By \cref{prop:nonzero-annihilator-types}, only \caseatwo, \casebtwo, \casecobg, and \casecz have a nontrivial annihilator. 
By \cref{prop:two-annihilator-degenerate}, the dimension of the annihilator space for a nondegenerate signature is at most $1$, so either all nonzero annihilators are isotropic or all are non-isotropic. 
Also, the annihilator being isotropic or not is invariant under orthogonal transform.
We see that \caseatwo, \casebtwo, and \casecobg are annihilated by a $T \mathbf{e}_3$ for some orthogonal $T$, while \casecz is annihilated by an isotropic vector $T \bzp$ for some orthogonal $T$. 
Therefore, \casecz is different from \caseatwo, \casebtwo and \casecobg.
\casecobg is different from \caseatwo and \casebtwo because \casecobg has symmetric rank $3$ by \cref{prop:symmetric-rank-sym-sum} while \caseatwo and \casebtwo have symmetric rank $2$.
\caseatwo and \casebtwo are distinct because the sum of tensor representation is essentially unique by \cref{prop:uniqueness-sum-lin-ind-tensor}.

Similarly, one can argue that \caseathree and \casebthree are distinct and different from the other \casec types. 
The fact that \casecosep, \casecodz and \casecogen are distinct is proven in \cref{cor:classification-case-c-mutually-exclusive}, which also explains the condition $c^2 - c - bd \ne 0$ in \casecogen.
\end{proof}

\section{Type \texorpdfstring{\casec}{C} signatures}\label{sec:casec}

As listed in \cref{thm:single-ternary-dichotomy}, 
\cite{domain3-single-signature} introduced three tractable cases with explicit algorithms.
However, the third case, which we named 
\casec in \cref{subsec:types-of-sig},
required an unexpected and complex algorithm unlike the other two cases.
This is in contrast to the other two cases which can be viewed as a natural generalization
of the tractable cases in the Boolean domain.
In~\cite{domain4-01-dichotomy}, it was observed that  case \casec does not occur when the signature is real-valued. 
This observation led to a dichotomy of a set of real-valued signatures in~\cite{domain3-real-set}.

In this paper, we need to consider the \casec signatures since we allow our signatures to be complex-valued. 
We first present a simple combinatorial algorithm for case  \casec by studying the problem in 
the $\holts(\mathbf{X} \mid \blank)$ setting and give further characterizations.
This allows us to view it as a generalization of the Boolean domain tractable case and discover a complex hierarchy within the set of \casec signatures.

We start by giving a 
combinatorial characterization of \casec signatures.
An algebraic characterization using the (complex) spin representation of $\speciallinear(2, \mathbb{C})$ parameterizing $\so(3, \mathbb{C}, \X)$ is given in \cref{sec:case-c-spin-rep-calc}.

Recall from
\cref{subsec:types-of-sig}
that $\mathbf{F}$ is in type \casec if there exists some orthogonal matrix $T$ such that 
$\mathbf{G} = T \mathbf{F}$ and $\mathbf{G} \sim \mathbf{G}_0 + \lambda \sym(\bzp, \bzp, \bzm)$ 
for some $\lambda \in \mathbb{C}$ and $\langle \mathbf{G}_0, \bzp \rangle = 0$.
We will further separate \casec into \casecz and \caseco, where \casecz corresponds to $\lambda = 0$ and \caseco corresponds to $\lambda \ne 0$.

Recall that $\langle \mathbf{G}_0, \bzp \rangle = 0$ means that when we  connect $\bzp$ to $\mathbf{F}_0$, we get a
binary signature that is identically zero.
We have $\holts(\mathbf{F}) =_T \holts( \beq \mid \mathbf{F})  =_T \holts( \beq  \mid  \mathbf{G})$.
Consider applying a holographic transformation by $\Yinv$ on the right side, mapping $\mathbf{G}$ to $\iny{\mathbf{G}} = \Yinv \mathbf{G}$.
Note that on the left side $(=_2)$ is transformed to $\X$, which has the matrix form $\Y\transpose \Y$. 
Then, we get $\holts(\mathbf{G}) =_T \holts(\X  \mid  \iny{\mathbf{G}})$ where $\iny{\mathbf{G}} = \iny{\mathbf{G}}_0 + \lambda \sym(\mathbf{e}_1, \mathbf{e}_1, \mathbf{e}_2)$ and $\iny{\mathbf{G}}_0 = \Yinv \mathbf{G}_0$.
We have $\Y \mathbf{e}_1 = \bzp$, $\Y \mathbf{e}_2 = \bzm$, and also as row vectors $\bzp \Y = \mathbf{e}_2$. 
From $\langle \mathbf{G}_0, \bzp \rangle$, we find 
that $\langle \iny{\mathbf{G}}_0, \mathbf{e}_2 \rangle = \langle \Yinv \mathbf{G}_0, \mathbf{e}_2 \rangle =  \langle \Yinv \mathbf{G}_0, \X \mathbf{e}_1 \rangle = \langle \mathbf{G}_0, (\Yinv)\transpose \X \mathbf{e}_1 \rangle = \langle \mathbf{G}_0, \Y \mathbf{e}_1\rangle = \langle \mathbf{G}_0, \bzp \rangle = 0$ (see also \cref{prop:annihilating-vector-under-gl-transform}).

This gives us a very clean picture of the signature $\iny{\mathbf{G}}$:
$\iny{\mathbf{G}}$ must have the form in \cref{fig:casec0-ternary}.
\begin{figure}
  \centering
  \begin{subfigure}[b]{0.45\textwidth}
    \centering
    \tersig{0.7}{a {\lambda} b 0 0 c 0 0 0 d}
    \caption{$\iny{\mathbf{G}}$}
    \label{fig:casec0-ternary}
  \end{subfigure}
  \begin{subfigure}[b]{0.45\textwidth}
    \centering
    \binsig{0.5}{* * * 0 0 *}
    \caption{$\langle \iny{\mathbf{G}}, \mathbf{u} \rangle$}
    \label{fig:casec0-binary}
  \end{subfigure}
  \caption{\casec signatures after $\Yinv$ transform.}
  \label{fig:general-case-c-signatures}
\end{figure}
for some $a, b, c, d \in \mathbb{C}$.
Note that if we connect any unary $\mathbf{u}$ to $\iny{\mathbf{G}}$, the signature $\langle \iny{\mathbf{G}}, \mathbf{u} \rangle$ has the form in \cref{fig:casec0-binary}.
Therefore, $\mathbf{F}$ has type \casec if and only if there exists some orthogonal transform $T$ such that $\Yinv T \mathbf{F} \in \casecsupp$.

\subsection{Algorithm}\label{subsec:case-c-algorithm}

Now we can give a simple combinatorial polynomial time algorithm for $\holts(\X  \mid  \iny{\mathbf{G}})$ for $\iny{\mathbf{G}} \in \casecsupp$.
Let $\Omega$ be a connected bipartite signature grid, with $\X$ on the left and $\iny{\mathbf{G}}$ on the right, 
with unary signatures connected to $\X$ or $\iny{\mathbf{G}}$.
We first recursively absorb all unary signatures to the adjacent signatures until there are no unary signatures left.
This gives a signature grid $\Omega'$ with no unary signature, such that every vertex  on the right side 
has degree 2 or 3, and has a signature of the form in \cref{fig:casec0-ternary} or \cref{fig:casec0-binary},
respectively.
If there are no degree $3$ vertices in $\Omega'$, we may compute the Holant value by matrix multiplication and taking the trace. Now suppose there are degree $3$ vertices in $\Omega'$.
Then $\Omega'$ cannot be a cycle.
Consider an assignment $\sigma:E \to \{\db, \dg, \dr\}$ on $\Omega'$.
Let $n_i$ be the number of edges assigned to $i \in \{\db, \dg, \dr\}$ by $\sigma$.
We see that having $\X$ on the left side forces $n_\db = n_\dg$.
We claim that this number is 0, in order that $\sigma$ 
has a nonzero evaluation.
Otherwise suppose  $n_\dg > 0$, then take any edge $e = (u, v)$ that is assigned $\dg$,
with $v$ belonging to the right side which is labeled by a signature 
in \cref{fig:casec0-ternary} or \cref{fig:casec0-binary}.
From both forms, for any vertex on the right side
if one incident edge is assigned $\dg$
then all other incident edges  must be  assigned $\db$, so the number of incident edges assigned $\db$
is greater than or equal to the number of incident edges assigned $\dg$.
Thus no incident edge to a vertex of degree 3 (they are all on the right side) can be assigned $\dg$,
 as this would imply  $n_\db > n_\dg$.
 Hence $v$
must have degree 2 and we find the other incident edge to $v$ assigned $\db$, which is turned to  $\dg$ again
by the left side $\X$. 
Repeating this as described in \cref{fig:case-c-nonzero-must-be-cycle}, we will trace out a cycle as $\Omega'$ is finite. Since  $\Omega'$ is  connected,
this cycle is all  $\Omega'$,  a
contradiction.
Therefore, $n_\db = n_\dg=0$. Thus a nonzero assignment can only use  $\dr$,
and so  $\hol_{\Omega'}$ is just the value of the single assignment of all $\dr$.
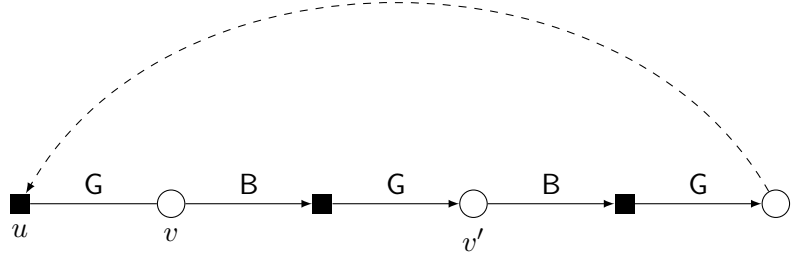
\begin{figure}
\centering
\begin{tikzpicture}

    \node[mysq, label=below:{$u$}] (x1) at (0, 0) {};
    \node[draw, circle, label=below:{$v$}] (g1) at (2, 0) {};
    \node[mysq] (x2) at (4, 0) {};
    \node[draw, circle, label=below:{$v'$}] (g2) at (6, 0) {};
    \node[mysq] (x3) at (8, 0) {};
    \node[draw, circle] (g3) at (10, 0) {};
    \draw (x1) -- (g1) node [midway, above] {$\dg$};
    \draw[-latex] (g1) -- (x2) node [midway, above] {$\db$};
    \draw[-latex] (x2) -- (g2) node [midway, above] {$\dg$};
    \draw[-latex] (g2) -- (x3) node [midway, above] {$\db$};
    \draw[-latex] (x3) -- (g3) node [midway, above] {$\dg$};

    \draw[-latex, dashed] (g3) to[out=120, in=60, looseness=1] (x1);
\end{tikzpicture}
\caption{$n_B = n_G > 0$ forces the graph to be a cycle. Squares are the vertices on the left side and circles are the vertices on the right side.}
\label{fig:case-c-nonzero-must-be-cycle}
\end{figure}

Note that this argument works for not only $\holts(\X  \mid  \iny{\mathbf{G}})$, 
but for $\holts(\X  \mid  \casecsupp)$ in general.
We can view this is a generalization of the Boolean domain Fibonacci gates of type $\typei(1, \pm 2\i)$, in which the vanishing part on $\db \dg$ annihilates the other possible assignments.

\subsection{Characterization of \texorpdfstring{\casecz}{C0}} 
\label{subsec:casecz-characterization}
Recall that $\mathbf{F}$ has type \casecz if and only if  for some nonzero isotropic $\mathbf{u} \in \mathbb{C}^3$, i.e.\ $\langle \mathbf{u}, \mathbf{u} \rangle = 0$, it holds that $\langle \mathbf{F}, \mathbf{u} \rangle = 0$.
Since $\mathbf{F}$ is nondegenerate, we may further assume that $\ann(\mathbf{F}) = \spn(\mathbf{u})$.
This means that there is some orthogonal matrix $T$ such that $\Yinv T \mathbf{F}$ has the form \cref{fig:casec0-ternary} with $\lambda = 0$.
We give another characterization of type \casecz.
\begin{proposition}\label{prop:casec0-characterization-rank-diminishing-matrix-symmetric}
Let $\mathbf{F}$ be a nondegenerate, symmetric, ternary, domain 3 signature.
$\mathbf{F}$ is in type \casecz if and only if $\mathbf{F} = M \mathbf{G}$ for some $M \in \mathbb{C}^{3 \times 3}$ and symmetric ternary signature $\mathbf{G}$ such that $\rank M = 2$ and $\rank M \transpose M = 1$.
\end{proposition}
\begin{proof}
Suppose $\mathbf{F}$ is type \casecz. 
After orthogonal transformation $T$, we may assume that $\mathbf{F}' = T \teh \mathbf{F}$ is such that $\langle \mathbf{F}', \bzp \rangle = 0$.
Then, $(\mathbf{F}')\domres{\db\dr} = [a, b, c, d]$ uniquely determines $\mathbf{F}'$ as in \cref{fig:casec0-determined-by-beta0}.
\begin{figure}
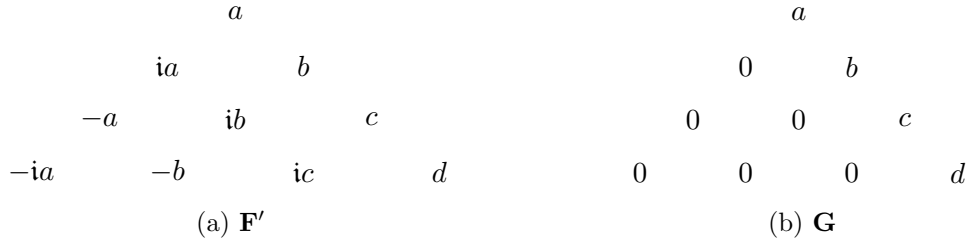

  \centering
  \begin{subfigure}[b]{0.5\textwidth}
\centering
\tersig{0.9}{a {\i a} b -a {\i b} c {-\i a} -b {\i c} d}
\caption{$\mathbf{F}'$}
\label{fig:casec0-determined-by-beta0}
\end{subfigure}
\begin{subfigure}[b]{0.4\textwidth}
\centering
\tersig{0.7}{a 0 b 0 0 c 0 0 0 d}
\caption{$\mathbf{G}$}
\label{fig:casec0-determined-by-beta0-G}
\end{subfigure}
\caption{A signature annihilated by $\bzp$ before and after a transform by $M$.}
\end{figure}
We can verify that $\mathbf{F'} = M\teh \mathbf{G}$ where $M = \begin{bsmallmatrix}
  1 & \i & 0 \\
  \i & -1 & 0 \\
  0 & 0 & 1
\end{bsmallmatrix}$ and $\mathbf{G}$ is the signature in \cref{fig:casec0-determined-by-beta0-G}.

For the other direction, suppose $M$ is such that $ M\transpose M = \mathbf{v} \mathbf{v}\transpose$ for some $\mathbf{v} \in \mathbb{C}^{3}$.
Consider the map $f: \mathbf{u} \mapsto \mathbf{u}\transpose \mathbf{v}$.
The kernel of $f$ has dimension $2$.
Since the kernel of $M$ is dimension $1$, we may choose some $\mathbf{u} \in \mathbb{C}^3$ such that $\mathbf{u}\transpose \mathbf{v} = 0$ while $M \mathbf{u} \ne 0$.
Then, we have $\mathbf{u}\transpose M \transpose M \mathbf{u} = \mathbf{u} \transpose \mathbf{v} \mathbf{v}\transpose \mathbf{u} = 0$, which means $\mathbf{u}' = M \mathbf{u}$ is an isotropic vector.
Also, $(\mathbf{u}')\transpose M = \mathbf{u}\transpose M\transpose M = \mathbf{u}\transpose \mathbf{v} \mathbf{v}\transpose = 0$.
Therefore, we have $\langle \mathbf{F}, \mathbf{u}' \rangle = \langle M \mathbf{G}, \mathbf{u}' \rangle = 0$.
\end{proof}

This characterization gives another algorithm for type \casecz.
If $M \mathbf{G}$ is a type \casecz signature, then in any signature grid, an edge between two $M \mathbf{G}$ can be decomposed into two unary signatures $\mathbf{v}$ where $M\transpose M = \mathbf{v} \otimes \mathbf{v}$.
This means that if the signature grid has $n$ vertices, the Holant value is $\lambda^n$ where $\lambda$ is the value of the signature grid made by connecting $\mathbf{v}$ to the three edges of $\mathbf{G}$.
We will see that this can be viewed as a special case of \cref{lem:edge-factors-general-case} and \cref{prop:pinning-vectors-iff-rank-1}.

\begin{remark}
One significance of this characterization of a \casecz signature $\mathbf{F}$ is that 
any edge between two copies of $\mathbf{F}$ signatures in a grid can be 
replaced with two unary signatures.
This means that $\holts(\mathbf{F}, \mathcal{F}) =_T \holts(\beq, \mathbf{F} \mid \mathcal{F})$. 
Then, if we have two such signatures $\mathbf{F}, \mathbf{G}$, 
$\holts(\mathbf{F}, \mathbf{G}) =_T \holts(\mathbf{F} \mid \mathbf{G})$. 
Interestingly, a computationally \textit{equivalent} bipartite Holant problem arises from a non-bipartite Holant problem.

\end{remark}

\subsection{Characterization of \texorpdfstring{\caseco}{C1}}\label{subsec:characterize-caseco}
\caseco signatures behave quite differently from \casecz signatures.
Suppose $\mathbf{G}$ is a \caseco signature such that $\iny{\mathbf{G}} = \Yinv \mathbf{G} = [a; 1, b; 0, 0, c;, 0, 0, 0, d]$ for some $a, b, c, d \in \mathbb{C}$ (see \cref{fig:casec1-ternary}). 
Note that we rescaled the $\lambda \ne 0$ to be $1$ from the definition given in \cref{def:types-ternary}.
\begin{figure}
\centering
\tersig{0.7}{a 1 b 0 0 c 0 0 0 d}
\caption{A \caseco signature.}\label{fig:casec1-ternary}
\end{figure}
There are several  subtypes of \caseco that behave differently.

If $b, c, d = 0$ in \cref{fig:casec1-ternary}, then $\supp \iny{\mathbf{G}} \subseteq \suppbg$, 
so any gadget we can make 
in the setting of $\holts(\X \mid \iny{\mathbf{G}})$ on the right side is also supported on $\db\dg$.
We call this case \casecobg, where $2$ comes from the signature essentially being \textit{Boolean domain}.
Next, suppose $b, c = 0$ but $d \ne 0$.
Then, $\iny{\mathbf{G}}$ is \sepbg, and since $\X$ is \sepbg as well, 
any gadget we can make on the right side
in the setting of $\holts(\X \mid \iny{\mathbf{G}})$  is also \sepbg.
We call this case \casecosep, where $\mathrm{s}$ stands for \textit{domain separated}.

An interesting phenomenon occurs when we have $c, d = 0$ with $b \ne 0$.
Setting 
$A = 
\begin{bsmallmatrix}
1 & -b^2/2 & b \\ 0 & 1 & 0 \\ 0 & -b & 1
\end{bsmallmatrix}$ 
and considering $\holts(\X (A^{-1})\tew  \mid  A \teh \iny{\mathbf{G}})$, 
we have $(A^{-1})\transpose \X A^{-1} = \X$ and $A \teh \iny{\mathbf{G}} = [a + 3 b^2/2; 1, 0; 0, 0, 0; 0, 0, 0, 0]$.
So, this case is actually equivalent to \casecobg.
Note that $A \in \so(3, \mathbb{C}, \mathbf{X})$.

\subsubsection{Characterization of \texorpdfstring{\casecodz}{C1dz}}\label{subsec:case-c-codz-analysis}
An unexpected interesting case is when $c \ne 0$ and $d = 0$. 
We call it type \casecodz, where $\mathrm{z}$ stands for having a zero in the corner. 
Suppose $\mathbf{G}$ is \casecodz, so
$\iny{\mathbf{G}} = [a; 1, b; 0, 0, c; 0, 0, 0, 0]$ for some $a, b, c \in \mathbb{C}$ such that $c \ne 0$. 
We characterize the set of signatures that $\iny{\mathbf{G}}$ can realize on the right side in the setting of 
$\holts(\X \mid \iny{\mathbf{G}})$.
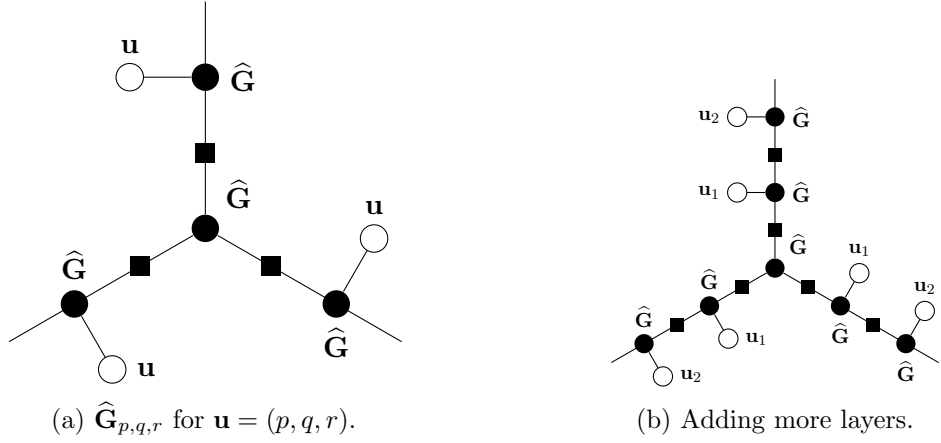
\begin{figure}
\centering
\begin{subfigure}[b]{0.45\textwidth}
\centering
\begin{tikzpicture}
  \node[mycirc, label=above right:{$\iny{\mathbf{G}}$}] (g1) at (0, 0) {};

  \draw (g1) -- (90:1cm) node[mysq] () {} -- ++(90:1cm) node[mycirc, label=right:{$\iny{\mathbf{G}}$}] (g2) {} -- ++(90:1cm);
  \draw (g1) -- (210:1cm) node[mysq] () {} -- ++(210:1cm) node[mycirc, label=above:{$\iny{\mathbf{G}}$}] (g3) {} -- ++(210:1cm);
  \draw (g1) -- (330:1cm) node[mysq] () {} -- ++(330:1cm) node[mycirc, label=below:{$\iny{\mathbf{G}}$}] (g4) {} -- ++(330:1cm);

  \draw (g2) -- ++(180:1cm) node[draw, circle, fill=white, label=above:{$\mathbf{u}$}] () {};
  \draw (g3) -- ++(300:1cm) node[draw, circle, fill=white, label=right:{$\mathbf{u}$}] () {};
  \draw (g4) -- ++(60:1cm) node[draw, circle, fill=white, label=above:{$\mathbf{u}$}] () {};
\end{tikzpicture}
\caption{$\iny{\mathbf{G}}_{p, q,r}$ for $\mathbf{u} = (p, q, r)$.}
\label{fig:grow-one-grass}
\end{subfigure}
\begin{subfigure}[b]{0.45\textwidth}
\centering
\begin{tikzpicture}[scale=0.5, every node/.style={scale=0.7}]
  \node[mycirc, label=above right:{$\iny{\mathbf{G}}$}] (g1) at (0, 0) {};

  \draw (g1) -- (90:1cm) node[mysq] () {} -- ++(90:1cm) node[mycirc, label=right:{$\iny{\mathbf{G}}$}] (g2) {} --++(90:1cm) node[mysq] () {} -- ++(90:1cm) node[mycirc, label=right:{$\iny{\mathbf{G}}$}] (g22) {} -- ++(90:1cm);
  \draw (g1) -- (210:1cm) node[mysq] () {} -- ++(210:1cm) node[mycirc, label=above:{$\iny{\mathbf{G}}$}] (g3) {} --++ (210:1cm) node[mysq] () {} -- ++(210:1cm) node[mycirc, label=above:{$\iny{\mathbf{G}}$}] (g32) {} -- ++(210:1cm);
  \draw (g1) -- (330:1cm) node[mysq] () {} -- ++(330:1cm) node[mycirc, label=below:{$\iny{\mathbf{G}}$}] (g4) {} --++ (330:1cm) node[mysq] () {} -- ++(330:1cm) node[mycirc, label=below:{$\iny{\mathbf{G}}$}] (g42) {} -- ++(330:1cm);

  \draw (g2) -- ++(180:1cm) node[draw, circle, fill=white, label=left:{$\mathbf{u}_1$}] () {};
  \draw (g3) -- ++(300:1cm) node[draw, circle, fill=white, label=right:{$\mathbf{u}_1$}] () {};
  \draw (g4) -- ++(60:1cm) node[draw, circle, fill=white, label=above:{$\mathbf{u}_1$}] () {};

  \draw (g22) -- ++(180:1cm) node[draw, circle, fill=white, label=left:{$\mathbf{u}_2$}] () {};
  \draw (g32) -- ++(300:1cm) node[draw, circle, fill=white, label=right:{$\mathbf{u}_2$}] () {};
  \draw (g42) -- ++(60:1cm) node[draw, circle, fill=white, label=above:{$\mathbf{u}_2$}] () {};

\end{tikzpicture}
\caption{Adding more layers.}
\label{fig:grow-more-grass}
\end{subfigure}
\caption{Layer gadgets.}
\end{figure}

\paragraph{Realizability}
If we make the gadget in \cref{fig:grow-one-grass} with $\mathbf{u} = (p, q, r) \in \mathbb{C}^3$, we get a signature
$\iny{\mathbf{G}}_{p,q,r}$ with 
\begin{itemize}
  \item $\iny{\mathbf{G}}_{p,q,r}(\db \db \db) =  (4 ap^2 + 3 b p r + 3 b p (b p + cr) + 3c(bp + cr)^2)p + 3 p^2 q$
  \item $\iny{\mathbf{G}}_{p,q,r}(\db \db \dg) = p^3$
  \item $\iny{\mathbf{G}}_{p,q,r}(\db \db \dr) = b(1 + c + 2c^2) p^3 + c(1 + 2c^2) p^2 r$
  \item $\iny{\mathbf{G}}_{p,q,r}(\db \dr \dr) = c^3 p^3$
\end{itemize}
and $0$ otherwise.
We see that $\iny{\mathbf{G}}_{p,q,r}$ is still \casecodz if $p \ne 0$.
Also, if $1 + 2c^2 \ne 0$, for any $p \ne 0$, we may choose $q, r$ such that $\iny{\mathbf{G}}_{p,q,r}(\db \db \db)$ and $\iny{\mathbf{G}}_{p,q,r}(\db \db \dr)$ values are arbitrary.
We show that $\iny{\mathbf{G}}$ can realize all signatures in the set 
\[
\spn\left(\mathbf{e}_1\teh, \quad \sym(\mathbf{e}_1, \mathbf{e}_1, \mathbf{e}_3), \quad \sym(\mathbf{e}_1, \mathbf{e}_1, \mathbf{e}_2) + c^{3} \sym(\mathbf{e}_1, \mathbf{e}_3, \mathbf{e}_3)\right) \,,
\]
i.e.\ a signature of the form $\iny{\mathbf{H}}_{x, y, z} = [x; y , z; 0, 0, yc^3; 0, 0, 0, 0]$.
\begin{proposition}\label{claim:casecodz-interpolation}
Let $\iny{\mathbf{G}} = [a; 1, b; 0, 0, c; 0, 0, 0, 0]$ such that $1 + 2 c^2 \ne 0$. 
Let $x, y, z \in \mathbb{C}$ be arbitrary and 
$\iny{\mathbf{H}}_{x, y, z} = [x; y , z; 0, 0, yc^3; 0, 0, 0, 0]$.
Then, $\holts(\X \mid \mathcal{G} \cup \{\iny{\mathbf{H}}_{x, y, z} \}) \le_T \holts(\X \mid \mathcal{G})$ by a gadget construction or a polynomial interpolation for any set of signatures $\mathcal{G}$ containing $\iny{\mathbf{G}}$.
\end{proposition}

\begin{proof}
If $y \ne 0$, then we can construct $\iny{\mathbf{H}}_{x, y, z}$ from $\iny{\mathbf{G}}$ by a direct gadget construction.
Otherwise, if $y = 0$, then consider a signature grid $\Omega$ such that $\iny{\mathbf{H}}_{x, 0, z}$ occurs $n$ times. 
We construct from $\Omega$ a sequence of instances $\Omega_s$ of $\holts(\X \mid \iny{\mathbf{G}}, \iny{\mathbf{H}}_{x, s, z})$ indexed by
  $s \in \mathbb{N}$  by replacing every occurrence of $\iny{\mathbf{H}}_{x, 0, z}$ in $\Omega$ with $\iny{\mathbf{H}}_{x, s, z}$. 
We stratify the assignments in $\Omega$ based on the assignments to the $n$ occurrences of the signature $\iny{\mathbf{H}}_{x, 0, z}$. 
We only need to consider the assignments that assign $i$ many times to $\db\db\db$ and $j$ many times to $\db\db\dr$, $k$ many times to $\db\db\dg$, and $\ell$ many times to  $\db\dr\dr$, including each of their permutations.
Let $c_{ijk\ell}$ be the sum over all such assignments of products of evaluation of all signatures in $\Omega$ except $\iny{\mathbf{H}}_{x, 0, z}$.
Then,
\[\hol_{\Omega} = \sum_{i + j + k + \ell = n} x^i z^j 0^k 0^\ell c_{ijk\ell} = \sum_{i + j = n} x^i z^j c_{ij00} \]
and
\[\hol_{\Omega_s} = \sum_{i + j + k + \ell = n} x^i z^j s^k (s c^3)^\ell c_{ijk\ell} = \sum_{i + j + k + \ell = n} x^i z^j s^{k + \ell} c^{3\ell} c_{ijkl} \,. \]
Let $d_{k\ell} = \sum_{i + j = n - (k + \ell)} x^i z^j c_{ijk\ell}$. 
We further define $e_{i} = \sum_{k + \ell = i} c^{3 \ell} d_{k \ell}$. 
Note that $\hol_{\Omega} = d_{00} = e_0$.
Then,
\[\hol_{\Omega_s} = \sum_{0 \le k + \ell \le n} s^{k + \ell} c^{3 \ell} d_{k \ell} = \sum_{i = 0}^n s^i e_i \]
which is a full rank Vandermonde system since we can choose $n+1$ distinct $s \in \mathbb{N}$.
\end{proof}

By attaching more layers of $\langle \iny{\mathbf{G}}, \mathbf{u}_i \rangle$ outwards as described in \cref{fig:grow-more-grass} for $\mathbf{u}_i = (1, *, *) \in \mathbb{C}^3$,
we can make signatures $\iny{\mathbf{G}}'$ with 
the value at $\db \dr \dr$  being $c' = c^3, c^5, c^7, \ldots$.
Then, connecting one more layer of $\langle \iny{\mathbf{G}}, \mathbf{u} \rangle$ for $\mathbf{u} = (p, q, r)$ on such $\iny{\mathbf{G}}'$ will give
the $\db \db \dr$ value being a polynomial with the coefficient of $r$ being $c p^2 (1 + 2 c c')$.
With $c' = c^{2k + 1}$ for $k \ge 0$, we see that there can be only one value of $k$ where $1 + 2 c c' = 0$.
Therefore, except for one possible $k_0$, we can realize any signature in the set
\begin{equation}
\label{eq:c1dz-gadeget-set}
  \bigcup_{k \ge 1, k \ne k_0} \spn\left(\mathbf{e}_1\teh, \quad \sym(\mathbf{e}_1, \mathbf{e}_1, \mathbf{e}_3), \quad \sym(\mathbf{e}_1, \mathbf{e}_1, \mathbf{e}_2) + c^{2k + 1} \sym(\mathbf{e}_1, \mathbf{e}_3, \mathbf{e}_3)\right)
\end{equation}
Note that if $c$ is not a root of unity, this is an infinite union of dimension $3$ subspaces in $\symspace{3}{3}$.

For any value of $c$,
we may realize $\sym(\mathbf{e}_1, \mathbf{e}_1, \mathbf{e}_3)$ on the right side. 
We can move it to the left by connecting $\X$, which becomes $\sym(\mathbf{e}_2, \mathbf{e}_2, \mathbf{e}_3)$, a signature supported on $\dg\dr$ as $[0, 1, 0, 0]$. 
We see that $\holts(\X  \mid \mathcal{F}, \iny{\mathbf{G}}) \ge_T \holbs([0,1,0,0] \mid \mathcal{F}\domres{\dg\dr})$.
\begin{remark}\label{rem:c-not-root-of-unity-interpolate-away}
    If $c$ is not a root of unity, then $c^{2k + 1}$ are all distinct. 
    By a similar interpolation argument as in \cref{claim:casecodz-interpolation}, we can make the values at $\db\db\dg$ and $\db\dr\dr$ independent of each other. 
    In other words, any signature in the following set is realizable on the right side in $\holts(\X \mid \iny{\mathbf{G}})$: 
    \[\spn\left(\mathbf{e}_1\teh, \quad \sym(\mathbf{e}_1, \mathbf{e}_1, \mathbf{e}_2), \quad \sym(\mathbf{e}_1, \mathbf{e}_1, \mathbf{e}_3), \quad \sym(\mathbf{e}_1, \mathbf{e}_3, \mathbf{e}_3)\right)  \, .\] 
    Then, using $\X$, we can get signatures of the form $[x, y, z, 0]$ for any $x, y, z \in \mathbb{C}$ that is supported on $\dg\dr$ on the left side. 
    By \cref{cor:bipartite-boolean-dichotomy-explicit},
    for any symmetric signature $\iny{\mathbf{H}}$, 
    $\iny{\mathbf{H}}\domres{\dg\dr}$ must be degenerate for $\holts(\X \mid \iny{\mathbf{G}}, \iny{\mathbf{H}}) \ge_T \holbs([x, y, z, 0] \mid \iny{\mathbf{H}}\domres{\dg\dr})$ to be not \sph.
\end{remark}

\paragraph{Un-realizability}
We can argue that any symmetric ternary signatures that can be realized in 
$\holts(\X  \mid  \iny{\mathbf{G}})$ on the right side by a connected graph
is of the form $[x; y, z; 0, 0, c^k y; 0, 0, 0, 0]$ for some $x, y, z \in \mathbb{C}$ and an odd $k$.
Let $\iny{\mathbf{H}}$ be a $\hol_3(\X \mid \iny{\mathbf{G}})$ connected gadget on the right side of arity $n$, not necessarily symmetric, that does not use unary signatures. 

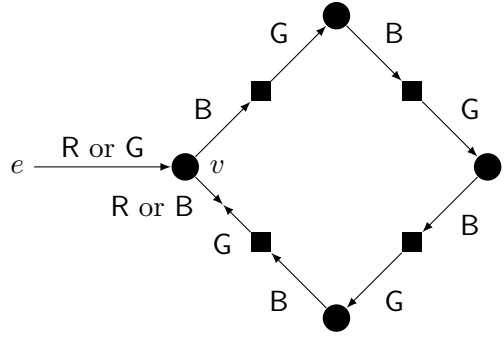
\begin{figure}
\centering
\begin{tikzpicture}
  \node[mycirc, label=right:{$v$}] (g1) at (-2, 0) {};
\node[mysq] (x1) at (-1, 1) {};
  \node[mycirc] (g4) at (0, 2) {};
  \node[mysq] (x4) at (1, 1) {};
  \node[mycirc] (g3) at (2, 0) {};
  \node[mysq] (x2) at (1, -1) {};
  \node[mycirc] (g2) at (0, -2) {};
  \node[mysq] (x3) at (-1, -1) {};

  \draw[-latex] (g1) -- (x1) node [midway, above left] {$\db$};
  \draw[-latex] (x1) -- (g4) node [midway, above left] {$\dg$};
  \draw[-latex] (g4) -- (x4) node [midway, above right] {$\db$};
  \draw[-latex] (x4) -- (g3) node [midway, above right] {$\dg$};
  \draw[-latex] (g3) -- (x2) node [midway, below right] {$\db$};
  \draw[-latex] (x2) -- (g2) node [midway, below right] {$\dg$};
  \draw[-latex] (g2) -- (x3) node [midway, below left] {$\db$};
  \draw[-latex] (x3) -- (-1.5, -0.5);
  \draw[-latex] (g1) -- (-1.5, -0.5);

  \node[below left] at (-1.75, -0.25) {$\dr$ or $\db$};

  \node[below left] at (-1.25, -0.75) {$\dg$};

  \draw[-latex] (-4, 0) -- (g1) node[midway, above] {$\dr$ or $\dg$};
  \node[left] at (-4, 0) {$e$};
\end{tikzpicture}
\caption{A cycle runs into a conflict.}
\label{fig:cycle-degenerate}
\end{figure}

First, if $\iny{\mathbf{H}}$ is non-degenerate, then we claim that the underlying graph of $\iny{\mathbf{H}}$ must be a tree. 
For a contradiction, assume  $\iny{\mathbf{H}}$ has a cycle $C$.  This cycle is reachable from an external edge.
Let $v \in C$ and $e$ be an edge connected to $v$ and $e \notin C$.
If $e$ is assigned  $\dr$, then  the two other edges of $v$ inside $C$ must 
either get $\db$ and $\dr$, or $\db$ and $\db$. Pick  an edge  assigned $\db$.
Then on the cycle $C$, the color $\db$ is changed to $\dg$ by $\X$, which must become  $\db$ again
going out of any vertex labeled $\iny{\mathbf{G}}$. Thus,  the color $\db$ gets propagated through $C$ by changing between $\db$ and $\dg$ via $\X$. This implies that when it returns to $v$ on $C$ it is assigned $\dg$, a contradiction.
If $e$ gets assigned to $\dg$, then the other two edges of $v$ on $C$ must be both assigned $\db$, which would result in 
the same conflict as before.
Therefore, any cycle $C$ must have all its incident edges  not on $C$ assigned to $\db$. 
(The argument is described in \cref{fig:cycle-degenerate}, where the circles are $\iny{\mathbf{G}}$ and the squares are $\X$.)
Then, if we follow a path from $C$ to any external edge, it must alternate between $\db$ and $\dg$ via $\X$ and 
vertices labeled by $\iny{\mathbf{G}}$, and thus $\db$ is propagated uniquely and ends up being $\db$ at the external edge.
Thus, the only nonzero assignment on the external edges of $\iny{\mathbf{H}}$ is the all $\db$ assignment, which shows that $\iny{\mathbf{H}}$ is degenerate.

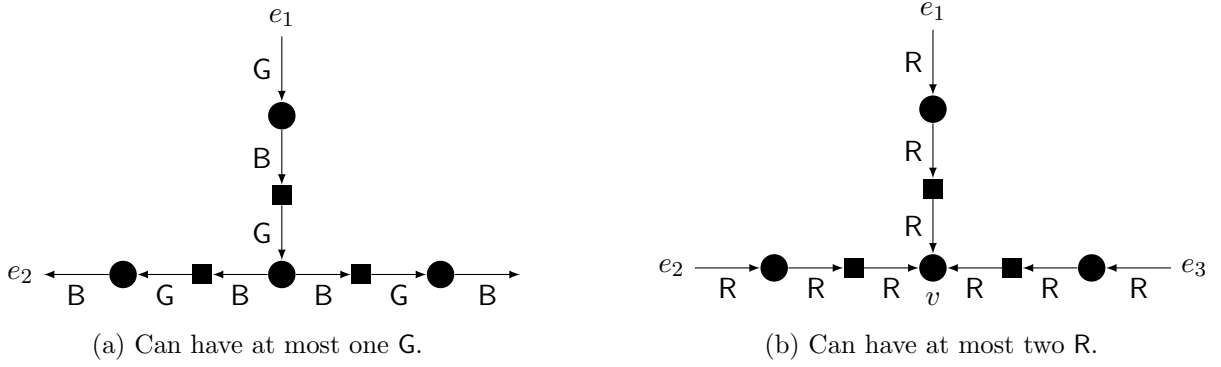
\begin{figure}
\centering
\begin{subfigure}[b]{0.46\textwidth}
   \centering 
\begin{tikzpicture}[scale=0.7]
  \node[mycirc] (g0) at (0, 3) {};

  \node[mycirc] (g1) at (0, 0) {};
  \node[mycirc] (g2) at (-3, 0) {};
  \node[mycirc] (g3) at (3, 0) {};

  \node[mysq] (x01) at (0, 1.5) {};
  \node[mysq] (x12) at (-1.5, 0) {};
  \node[mysq] (x13) at (1.5, 0) {};

  \node[above] at (0, 4.5) {$e_1$};
  \node[left] at (-4.5, 0) {$e_2$};

  \draw[-latex] (0, 4.5) -- (g0) node[midway, left] {$\dg$};
  \draw[-latex] (g0) -- (x01) node[midway, left] {$\db$};
  \draw[-latex] (x01) -- (g1) node[midway, left] {$\dg$};
  \draw[-latex] (g1) -- (x12) node[midway, below] {$\db$};
  \draw[-latex] (x12) -- (g2) node[midway, below] {$\dg$};
  \draw[-latex] (g2) -- (-4.5, 0) node[midway, below] {$\db$};
    \draw[-latex] (g1) -- (x13) node[midway, below] {$\db$};
  \draw[-latex] (x13) -- (g3) node[midway, below] {$\dg$};
  \draw[-latex] (g3) -- (4.5, 0) node[midway, below] {$\db$};
  
\end{tikzpicture}
\caption{Can have at most one $\dg$.}
\label{fig:supp-one-g-casecodz}
\end{subfigure}
\hfill
\begin{subfigure}[b]{0.46\textwidth}
\centering
\begin{tikzpicture}[scale=0.7]

  \node[mycirc] (g0) at (0, 3) {};

  \node[mycirc, label=below:{$v$}] (g1) at (0, 0) {};
  \node[mycirc] (g2) at (-3, 0) {};
  \node[mycirc] (g3) at (3, 0) {};

  \node[mysq] (x01) at (0, 1.5) {};
  \node[mysq] (x12) at (-1.5, 0) {};
  \node[mysq] (x13) at (1.5, 0) {};

    \node[above] at (0, 4.5) {$e_1$};
  \node[left] at (-4.5, 0) {$e_2$};
  \node[right] at (4.5, 0) {$e_3$};

  \draw[-latex] (0, 4.5) -- (g0) node[midway, left] {$\dr$};
  \draw[-latex] (g0) -- (x01) node[midway, left] {$\dr$};
  \draw[-latex] (x01) -- (g1) node[midway, left] {$\dr$};
  \draw[latex-] (g1) -- (x12) node[midway, below] {$\dr$};
  \draw[latex-] (x12) -- (g2) node[midway, below] {$\dr$};
  \draw[latex-] (g2) -- (-4.5, 0) node[midway, below] {$\dr$};
    \draw[latex-] (g1) -- (x13) node[midway, below] {$\dr$};
  \draw[latex-] (x13) -- (g3) node[midway, below] {$\dr$};
  \draw[latex-] (g3) -- (4.5, 0) node[midway, below] {$\dr$};
  
\end{tikzpicture}
\caption{Can have at most two $\dr$.}
\label{fig:supp-two-r-casecodz}
\end{subfigure}
\caption{Support of a \casecodz gadget.}
\end{figure}
We claim that $\iny{\mathbf{H}}$ is zero on any input outside $\db^n$, $\db^{n-1} \dg$, $\db^{n-1} \dr$, $\db^{n-2} \dr \dr$ and their permutations.
Suppose $e_1$ and $e_2$ are two external edges of $\iny{\mathbf{H}}$.
Let $P_{12}$ be a path connecting $e_1$ and $e_2$.
If $e_1$ gets $\dg$, then along the path $P_{12}$, we must alternate between $\dg \db \dg \db \cdots$, so it must end up at $e_2$ as $\db$ as shown in \cref{fig:supp-one-g-casecodz}.
Therefore, $\iny{\mathbf{H}}$ is $0$ when more than one of its inputs are $\dg$ or it takes both $\dg, \dr$.
If the arity $n$ of $\iny{\mathbf{H}}$ is 2, then we are done.  
Suppose $n \ge 3$.
Now, let $e_3$ be another external edge and let $P_{13}$ be a path connecting $e_1$ and $e_3$.
Let $v$ be the first vertex where the paths $P_{12}$ and $P_{13}$ diverge starting at $e_1$.
Suppose $e_1, e_2, e_3$ all get $\dr$.
This means that $P_{12}$ and $P_{13}$ must be $\dr$ all along the way, since as soon as we use $\db$ to go out of some $\iny{\mathbf{G}}$ vertex, it cannot escape the subsequent alternation of $\db, \dg$.
However, this means that at $v$, all its three edges must be assigned $\dr$, which results in $\iny{\mathbf{G}}$ at $v$ being $0$, as shown in \cref{fig:supp-two-r-casecodz}.
Therefore, $\iny{\mathbf{H}}$ has value $0$ when more than two of its inputs are $\dr$.
This proves the claim.

Making a $\holts(\X \mid \iny{\mathbf{G}})$ gadget can be viewed as adding unary signatures on a $\hol(\X \mid \iny{\mathbf{G}})$ gadget,
and we see that adding unary signatures cannot increase the number of $\dg$ or $\dr$ in the support, meaning that any $\holts(\X \mid \iny{\mathbf{G}})$ gadget
must have the same property.

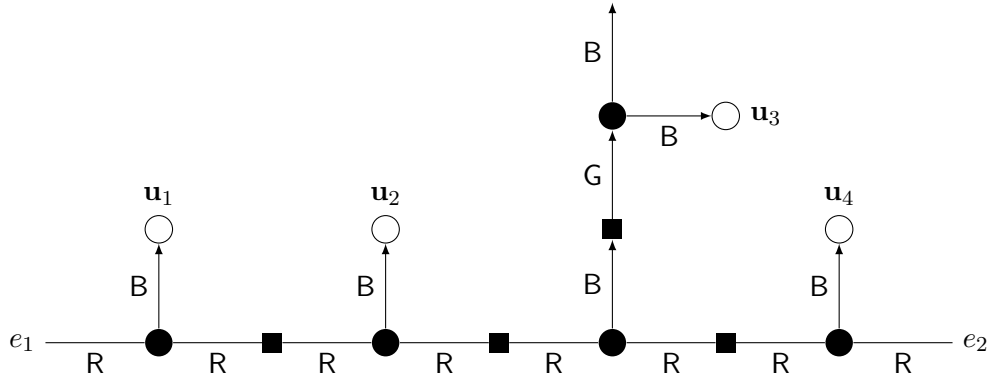
\begin{figure}
\centering
\begin{tikzpicture}
  \node[mycirc] (g1) at (0, 0) {};
  \node[mysq] (x12) at (1.5, 0) {};
  \node[mycirc] (g2) at (3, 0) {};
  \node[mysq] (x23) at (4.5, 0) {};
  \node[mycirc] (g3) at (6, 0) {};
  \node[mysq] (x34) at (7.5, 0) {};
  \node[mycirc] (g4) at (9, 0) {};

   \node[draw, circle, label=above:{$\mathbf{u}_2$}] (u1) at (3, 1.5) {};

   \node[mysq] (x35) at (6, 1.5) {};
   \node[mycirc] (g5) at (6, 3) {};
   \node[draw, circle, label=right:{$\mathbf{u}_3$}] (u2) at (7.5, 3) {};

   \draw (-1.5, 0) -- (g1) node [midway, below] {$\dr$};
   \draw (g1) -- (x12) node [midway, below] {$\dr$};
   \draw (x12) -- (g2) node [midway, below] {$\dr$};
   \draw (g2) -- (x23) node [midway, below] {$\dr$};
   \draw (x23) -- (g3) node [midway, below] {$\dr$};
   \draw (g3) -- (x34) node [midway, below] {$\dr$};
   \draw (x34) -- (g4) node [midway, below] {$\dr$};
   \draw (g4) -- (10.5, 0) node [midway, below] {$\dr$};

   \draw[-latex] (g2) -- (u1) node[midway, left] {$\db$};
   \draw[-latex] (g3) -- (x35) node[midway, left] {$\db$};
   \draw[-latex] (x35) -- (g5) node[midway, left] {$\dg$};
   \draw[-latex] (g5) -- (u2) node[midway, below] {$\db$};
   \draw[-latex] (g5) -- (6, 4.5) node[midway, left] {$\db$};

   \node[draw, circle, label=above:{$\mathbf{u}_1$}] (u3) at (0, 1.5) {};
   \node[draw, circle, label=above:{$\mathbf{u}_4$}] (u4) at (9, 1.5) {};

   \draw[-latex] (g1) -- (u3) node[midway, left] {$\db$};
   \draw[-latex] (g4) -- (u4) node[midway, left] {$\db$};

   \node[left] at (-1.5, 0) {$e_1$};
   \node[right] at (10.5, 0) {$e_2$};

\end{tikzpicture}
\caption{An assignment with two $\dr$.}
\label{fig:two-dr-and-others}
\end{figure}
Now, suppose $\iny{\mathbf{H}}$ is a non-degenerate connected gadget of arity $n$, so we may assume its underlying graph is a tree.
Let $\mathbf{u}_1, \ldots, \mathbf{u}_m$ be the unary signatures used in $\iny{\mathbf{H}}$.
We claim that the value of $\iny{\mathbf{H}}$ at $\db^{n-2} \dr \dr$ or any of its permutations must be $c^k \mathbf{u}_1(\db) \mathbf{u}_2(\db) \cdots \mathbf{u}_m(\db)$, where $c= \iny{\mathbf{G}}(\db \dr \dr)$
and  $k$ may depend on the permutation of $\db^{n-2} \dr \dr$.
Let $e_1, e_2$ be two external edges of $\iny{\mathbf{H}}$ assigned $\dr$.
Then, the unique path $P$ between $e_1, e_2$ must all take values $\dr$.
And all the edges going out from the vertices in $P$ must take value $\db$, which then uniquely propagate alternating between $\db$ and $\dg$.
So, all values of $\iny{\mathbf{G}}$ on $P$ is $\iny{\mathbf{G}}(\db \dr \dr) = c$, all others are $\iny{\mathbf{G}}(\db \db \dg) = 1$, while all edges to the unary signatures get $\db$. 
\cref{fig:two-dr-and-others} shows the argument in a diagram.
On the other hand, the value of $\iny{\mathbf{H}}$ on $\db^{n-1} \dg$ must be $\mathbf{u}_1(\db) \mathbf{u}_2(\db) \cdots \mathbf{u}_m(\db)$ by a similar argument.

If $\iny{\mathbf{H}}$ realizes a ternary symmetric function, 
then it must be of the form $[x; y, z; 0, 0, c^k y; 0, 0, 0, 0]$ for some $x, y, z \in \mathbb{C}$ and $k$.
We  further claim that $k$ must be odd.
Let $e_1, e_2, e_3$ be the external edges of $\iny{\mathbf{H}}$.
We may repeatedly reduce all unary signatures, until we get $\iny{\mathbf{H}}$ consisting of three paths $P_1, P_2, P_3$ from a central vertex $v$ to $e_1, e_2, e_3$ respectively, and a unary signature attached to each vertex in the paths $P_i$.
Let $\lambda = \mathbf{u}_1(\db) \cdots \mathbf{u}_m(\db)$ where $\mathbf{u}_i$ are the unary signatures in this reduced gadget.
If $\lambda = 0$, then we are done since $\iny{\mathbf{H}}(\db \db \dg) = \iny{\mathbf{H}}(\db \dr \dr) = 0$.
Let $\ell_i = \lvert P_i \rvert$.
Then, $\iny{\mathbf{H}}(\dr \dr \db) = c^{\ell_1 + \ell_2 + 1} \lambda$.
If $\ell_1  \ne\ell_2 $, then $\iny{\mathbf{H}}(\dr \db \dr) = c^{\ell_1 + \ell_3 + 1} \lambda \ne c^{\ell_2 + \ell_3 + 1}\lambda = \iny{\mathbf{H}}(\db \dr \dr)$, contradicting the symmetry of $\iny{\mathbf{H}}$.  Similarly $\ell_2 = \ell_3$.
Therefore $\ell_1 = \ell_2 = \ell_3$, showing that $\iny{\mathbf{H}}(\db \dr \dr) =  c^{2k + 1} \iny{\mathbf{H}}(\db \db \dg) $.

\subsubsection{Characterization of \texorpdfstring{\casecogen}{C1gen}}\label{subsec:case-c-gen-analysis}
\newcommand{\nonzerovalue}{+}
Below we use the symbol $\nonzerovalue$  to mean some nonzero value and $*$ to mean an arbitrary value.
So far, we have considered the cases where
\begin{enumerate}
    \item $(b, c, d) = (0, 0, 0)$: \casecobg
    \item $(b, c, d) = (\nonzerovalue, 0, 0)$: \casecobg after a $\so(3, \mathbb{C}, \X)$ transform
    \item $(b, c, d) = (0, 0, \nonzerovalue)$: \casecosep
    \item $(b, c, d) = (*, \nonzerovalue, 0)$: \casecodz
\end{enumerate}
When $d = 0$, we have considered all cases of $(b, c) \ne (0, 0)$ by $c \ne 0$ (case 4) and $c = 0$ (cases 1 and 2). 
When $d \ne 0$, we have already  considered the case when $(b, c) = (0, 0)$.
So the only remaining case is when $(b, c) \ne (0, 0)$ and $d \ne 0$. 
We will see that if $c^2 - c - bd = 0$, then such a signature is $\so(3, \mathbb{C}, \X)$ transformable to \casecosep. 
Therefore, to define \casecogen ($\mathrm{g}$ stands for \textit{generic}), we require that $c^2 - c - bd \ne 0$. 
We will see that this condition significantly impacts the behavior of the signature.
\cref{cor:classification-case-c-mutually-exclusive} shows in an alternative way that  \casecobg, \casecosep, \casecodz, and \casecogen exhaust the \caseco type.

Let $\iny{\mathbf{G}} = [a; 1, b; 0, 0, c; 0, 0, 0, d]$ with $d \ne 0$ and $c^2 - c - bd \ne 0$.
Using $\iny{\mathbf{G}}$, we show that we may realize a set of signatures $\mathcal{G}$ on the right side
of $\holts(\X  \mid  \iny{\mathbf{G}})$ that are supported on $\db\dr$, such that $\mathcal{G}$ is \textit{generic} enough.
Then, for another signature $\iny{\mathbf{H}}$, we have $\holts(\X  \mid  \iny{\mathbf{G}}, \iny{\mathbf{H}}) \ge_T \holts(\X \mathcal{G}  \mid  \mathcal{H})$.
Here, $\X \mathcal{G} $ denotes the set of $\mathcal{G}$ moved to the left side by attaching $\X$ on all ends,
and $\mathcal{H}$ is the set of signatures realizable on the right side in the setting of $\holts(\X \mid \iny{\mathbf{H}})$.
Since $\X \mathcal{G}$ consists of signatures supported on $\dg\dr$, we can analyze $\holbs( (\X \mathcal{G})\domres{\dg \dr}  \mid  \mathcal{H} \domres{\dg\dr})$.
If $\X \mathcal{G}$ is sufficiently generic, then we can argue from the Boolean domain dichotomy that the signatures on the right side must be degenerate to escape \sphness.
This will allow us to restrict the possible $\iny{\mathbf{H}}$ that can escape \sphness. 
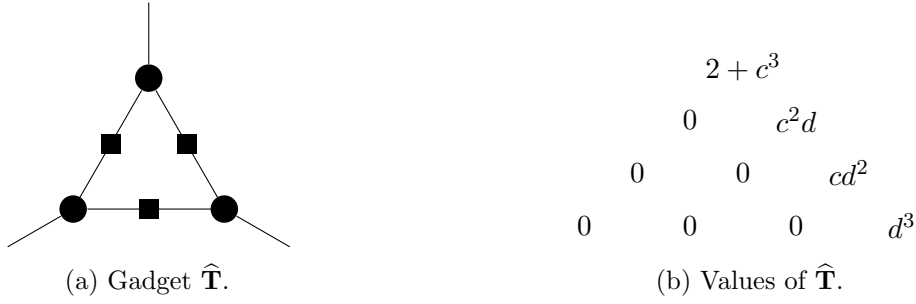
\begin{figure}
\centering
\begin{subfigure}[b]{0.35\textwidth}
  \centering
\begin{tikzpicture}
  \node[mycirc] (g1) at (0, 0) {};

  \draw (g1) -- (300:1cm) node[mysq] () {} -- ++(300:1cm) node[mycirc] (g2) {};
  \draw (g1) -- (240:1cm) node[mysq] () {} -- ++(240:1cm) node[mycirc] (g3) {};
  \draw (g3) -- ++ (0:1cm) node[mysq] () {} -- ++(0:1cm);

  \draw (g2) -- ++(330:1cm);
  \draw (g3) -- ++(210:1cm);
  \draw (g1) -- ++(90:1cm);
\end{tikzpicture}
\caption{Gadget $\iny{\mathbf{T}}$.}
\label{fig:triangle-gadget}
\end{subfigure}
\begin{subfigure}[b]{0.6\textwidth}
\centering
\tersig{0.7}{{2 + c^3} 0 {c^2 d} 0 0 {c d^2} 0 0 0 d^3}
\caption{Values of $\iny{\mathbf{T}}$.}
\label{fig:triange-with-x}
\end{subfigure}
\caption{Triangle gadget.}
\end{figure}

Consider the gadget in \cref{fig:triangle-gadget} where the square vertices are $\X$.
One possible way to make the gadget $\iny{\mathbf{T}}$ is by using $\iny{\mathbf{G}}$ for the circles.
Another possibility is to make the gadget $\iny{\mathbf{T}}_{p,q,r}$ by using $\iny{\mathbf{G}}_{p,q,r}$ from \cref{fig:grow-one-grass} for the circles.
The values of $\iny{\mathbf{T}}$ are shown in \cref{fig:triange-with-x}.
We see that $\iny{\mathbf{T}} = 2 \mathbf{e}_1\teh + (c, 0, d)\teh$.
Moreover, we see that $\iny{\mathbf{T}}_{1,q,r} =  2 \mathbf{e}_1\teh + (c + dr)^6 (x, 0, y)\teh$ where
$x = c + bd + cdr$ and $y = d(c + dr)$.
For $r$ such that $c + dr \ne 0$, we see that
$x/y = \frac{c + bd + cdr}{d(c+dr)}$.
In particular, we can choose $r$ such that $x/y \ne c/d$,
unless $x/y = c/d$ for all $r$, which can happen if $d(c + db + dcr) = c(dc + d^2 r)$, which only holds if $d(c^2 -c - bd) = 0$.
This gives the genericity we desired.
We can also view the possible set of $x/y$ as the image of a fractional linear transformation $r \mapsto \frac{c + db + dcr}{dc + d^2 r}$.
The determinant of the matrix corresponding to this map is exactly $d^2(c^2 -c - bd)$.
If it is $0$, then the map is constant. Otherwise, its image is $\mathbb{C} \setminus \{c/d\}$.

Alternatively,
if we make a triangle gadget $\iny{\mathbf{T}}_{1, y, z}$ with $\iny{\mathbf{G}}$ for the circles and $ \mathbf{X} \langle \iny{\mathbf{G}}, (1, y, z) \rangle \mathbf{X}$ for the squares,
we get
\[
  \iny{\mathbf{T}}_{1, y, z} = 2 \mathbf{e}_1\teh + (c + dz)^3 \begin{bmatrix}c \\ 0 \\ d\end{bmatrix}\teh \, .
\]
In general, if we make a triangle gadget $\iny{\mathbf{T}}_{1,q,r; 1, y, z}$ with $\iny{\mathbf{G}}_{1,q,r}$ for the circles and $\mathbf{X} \langle \iny{\mathbf{G}}, (1, y, z) \rangle \mathbf{X}$ for the squares,
we get a signature
\[
  \iny{\mathbf{T}}_{1, q, r; 1, y, z} = 2 \mathbf{e}_1\teh + (c + dz)^3 (c + dr)^6 \begin{bmatrix}c + bd + cdr \\ 0 \\ d(c + dr) \end{bmatrix}\teh \, .
\]
By combining the analysis from above and since $z$ is free, we can realize any signature of the form
$2 \mathbf{e}_1\teh + \lambda (u, 0, v)\teh$ for any $\lambda, u, v \in \mathbb{C}$.
By a polynomial interpolation similar to \cref{claim:casecodz-interpolation}, this allows us to realize any signature in the set
\[
  \bigcup_{u, v \in \mathbb{C}} \spn\left(\mathbf{e}_1\teh, \begin{bmatrix}
  u \\ 0 \\ v
  \end{bmatrix}\teh\right) \, .
\]

The criteria  $c^2 -c - bd = 0$ is fascinating.
It may seem to be problematic, since   adding more layers of unary signatures to the
gadget construction cannot escape that relation.
In other words, $\iny{\mathbf{G}}_{1, q,r}$ always has the property that $ \iny{\mathbf{G}}_{1, q,r}(\db \dr \dr)^2 - \iny{\mathbf{G}}_{1, q,r}(\db \dr \dr)  - \iny{\mathbf{G}}_{1, q,r}(\db \db \dr) \iny{\mathbf{G}}_{1, q,r}(\dr \dr \dr) = 0$, so repeating this construction cannot escape that algebraic relation.
For example, we can have $c = \frac{1 - \sqrt{5}}{2}$ and $b, d = 1$.
Also,
we 
are unable to realize domain restriction to $\db \dg$ because 
the $\db \db \dr$ part cannot be made to be $0$.
But miraculously, 
we can transform such cases.
Suppose $c^2 -c - bd = 0$ for $\iny{\mathbf{G}}$.
Let 
$\iny{A} = 
\begin{bsmallmatrix}
  1 & -x^2/2 & x \\ 0 & 1 & 0 \\ 0 & -x & 1
\end{bsmallmatrix}
 \in \so(3, \mathbb{C}, \X)$, 
with $x = -c/d$.
Let $\iny{\mathbf{G}}' = \iny{A} \iny{\mathbf{G}}$. Then 
\begin{itemize}
  \item $\iny{\mathbf{G}}'(\db \db \db) = *$
  \item $\iny{\mathbf{G}}'(\db \db \dg) = 1$ 
  \item $\iny{\mathbf{G}}'(\db \db \dr) = (c - c^2 + bd)/d$ 
  \item $\iny{\mathbf{G}}'(\db \dr \dr) = 0$ 
  \item $\iny{\mathbf{G}}'(\dr \dr \dr) = d$
\end{itemize}
and $0$ otherwise.
So, if $c^2 -c - bd = 0$, $\iny{\mathbf{G}}'$ is \sepbg.
Therefore in this case, $\mathbf{G}$ is actually \casecosep. 
An alternate justification is given in \cref{sec:case-c-spin-rep-calc}.

Therefore, a \casecogen signature must have $c^2 -c - bd \ne 0$.
Let $\iny{\mathbf{G}}$ be such.
The values of $\iny{\mathbf{G}}_{p, q,r}$ are
\begin{itemize}
  \item $\iny{\mathbf{G}}_{p, q, r}(\db \db \db) = *$ 
  \item $\iny{\mathbf{G}}_{p, q, r}(\db \db \dg) = p^3$
  \item $\iny{\mathbf{G}}_{p, q, r}(\db \db \dr) = p^2(bp + cr) + bp^2(cp + dr) + 2cp (bp + cr) (cp + dr) + d(bp + cr)^2 (cp + dr)$
  \item $\iny{\mathbf{G}}_{p, q, r}(\db \dr \dr) = cp (cp + dr)^2 + d (bp + cr)(cp + dr)^2$
  \item $\iny{\mathbf{G}}_{p, q, r}(\dr \dr \dr) = d(cp + dr)^3$
\end{itemize}
If $(b, c)$ and $(c, d)$ are linearly dependent, then we may choose $p = 1$ and $r = -c/d$ which will make $\iny{\mathbf{G}}_{p,q,r}$ supported on $\db \dg$ with the signature being $[*,1,0,0]$.
If $(b, c)$ and $(c, d)$ are linearly independent, so $c^2 - bd \ne 0$,
then we may choose
$r = \frac{(-c^3 + bd + bcd)p}{d(c^2 -c - bd)}$ 
which makes $\iny{\mathbf{G}}_{p,q,r}(\db\dr\dr)^2 - \iny{\mathbf{G}}_{p,q,r}(\db \db \dr) \iny{\mathbf{G}}_{p,q,r}(\dr \dr \dr) = 0$, with $\iny{\mathbf{G}}_{p,q,r}(\dr\dr\dr) = \frac{(bd - c^2)p}{c^2 -c - bd} \ne 0$.
So, by growing one more layer of unary signatures, we can realize a signature supported on $\db \dg$ of the form $[*,1,0,0]$, which is a \casecobg signature.

\subsection{Type \texorpdfstring{\casec}{C}, \texorpdfstring{$\so(3, \mathbb{C}, \X)$}{so3cx}, and  Spin Representation
}
\label{sec:case-c-spin-rep-calc}

Recall that $\mathbf{F}$ is a \casec signature if and only if there exists some $T \in \so(3, \mathbb{C})$ such that $\Yinv T \mathbf{F} = [x; \lambda, y; 0, 0, z; 0, 0, 0, w] \in \casecsupp$, 
or equivalently, if there exists some $\iny{T} \in \so(3, \mathbb{C}, \X)$ such that 
$\iny{T} \Yinv \mathbf{F}  \in \casecsupp$. 
We gave the definition of \casecz, \casecobg, \casecosep, \casecodz, and \casecogen as 
certain numerical conditions on $x, y, z, w, \lambda$.
We recall the definition of each type in \cref{fig:each-casec-type}.
It is conceivable that if for some other $S \in \so(3, \mathbb{C})$, it happens that $\Yinv S \mathbf{F} = [x'; \lambda', y'; 0, 0, z'; 0, 0, 0, w'] \in \casecsupp$, then
$x', y', z', w', \lambda'$ may take a different form.
In this section we characterize to what extent that can happen. 
It turns out that the spin representation of $\speciallinear(2, \mathbb{C})$ parameterizing $\so(3, \mathbb{C}, \X)$ simplifies the analysis.

\begin{figure}
  \centering
  \begin{subfigure}[b]{0.4\textwidth}
\centering
\tersig{0.7}{x 0 y 0 0 z 0 0 0 w}
\caption{\casecz}
\end{subfigure}
\begin{subfigure}[b]{0.4\textwidth}
\centering
\tersig{0.7}{x {\lambda} 0 0 0 0 0 0 0 0}
\caption{\casecobg: $\lambda \ne 0$}
\end{subfigure} \\
\hfill
\\
  \begin{subfigure}[b]{0.4\textwidth}
\centering
\tersig{0.7}{x {\lambda} 0 0 0 0 0 0 0 w}
\caption{\casecosep: $\lambda, w \ne 0$}
\end{subfigure}
\begin{subfigure}[b]{0.4\textwidth}
\centering
\tersig{0.7}{x {\lambda} y 0 0 z 0 0 0 0}
\caption{\casecodz: $\lambda, z \ne 0$}
\end{subfigure}
\\
  \begin{subfigure}[b]{0.4\textwidth}
\centering
\tersig{0.7}{x {\lambda} y 0 0 z 0 0 0 w}
\caption{\casecogen: $\lambda, w \ne 0$ and $z^2 - z \lambda - yw \ne 0$}
\end{subfigure}
\caption{The defining forms of \casec types.}
\label{fig:each-casec-type}
\end{figure}

\begin{figure}
  \centering
  \begin{subfigure}[b]{0.4\textwidth}
\centering
\tersig{0.7}{x 0 y 0 0 z 0 0 0 w}
\caption{\casecz}
\end{subfigure}
\begin{subfigure}[b]{0.4\textwidth}
\centering
\tersig{0.7}{x {\lambda} y 0 0 0 0 0 0 0}
\caption{\casecobg: $\lambda \ne 0$}
\end{subfigure} \\
\hfill
\\
  \begin{subfigure}[b]{0.4\textwidth}
\centering
\tersig{0.7}{x {\lambda} y 0 0 z 0 0 0 w}
\caption{\casecosep: $\lambda, w \ne 0$ and $z^2 - z \lambda - y w = 0$}
\end{subfigure}
\begin{subfigure}[b]{0.4\textwidth}
\centering
\tersig{0.7}{x {\lambda} y 0 0 z 0 0 0 0}
\caption{\casecodz: $\lambda, z \ne 0$}
\end{subfigure}
\\
  \begin{subfigure}[b]{0.4\textwidth}
\centering
\tersig{0.7}{x {\lambda} y 0 0 z 0 0 0 w}
\caption{\casecogen: $\lambda, w \ne 0$ and $z^2 - z \lambda - yw \ne 0$}
\end{subfigure}
\caption{\casec types after any $\Yinv T$ transform putting it in $\casecsupp$.}
\label{fig:each-casec-type-just-in-C}
\end{figure}

Let $\mathbf{F}$ be a signature such that $\Yinv T \mathbf{F} \in \casecsupp$ for some $T \in \so(3, \mathbb{C})$. 
We want to analyze what can happen when $\Yinv S \mathbf{F} \in \casecsupp$ as well for some other $S \in \so(3, \mathbb{C})$.
We can equivalently start with $\mathbf{F}'$ such that $\Yinv \mathbf{F}' \in \casecsupp$ and analyze what happens when $\Yinv T' \mathbf{F}' \in \casecsupp$ for some $T' \in \so(3, \mathbb{C})$ by renaming $\mathbf{F}' = T \mathbf{F}$ and $T' = S T^{-1}$.

We show the follwing in this section (see \cref{fig:each-casec-type-just-in-C}).
\begin{enumerate}
    \item If $\mathbf{F}$ is  \casecz,
    such that $\Yinv \mathbf{F} \in \casecsupp$, then $\Yinv \mathbf{F}$ is supported on $\db\dr$.
    \item If $\mathbf{F}$ is \casecobg such that $\Yinv \mathbf{F} \in \casecsupp$, then $\Yinv \mathbf{F} = [x; \lambda, y; 0, 0, 0; 0, 0, 0, 0]$ for $\lambda \ne 0$.
    \item If $\mathbf{F}$ is \casecosep such that $\Yinv \mathbf{F} \in \casecsupp$, then $\Yinv \mathbf{F} = [x; \lambda, y; 0, 0, z; 0, 0, 0, w]$ for $\lambda, w \ne 0$ and $z^2 - z \lambda - y w = 0$.
    \item If $\mathbf{F}$ is \casecodz such that $\Yinv \mathbf{F} \in \casecsupp$, then $\Yinv \mathbf{F} = [x; \lambda, y; 0, 0, z; 0, 0, 0, 0]$ for $\lambda, z \ne 0$.
    \item If $\mathbf{F}$ is \casecogen such that $\Yinv \mathbf{F} \in \casecsupp$, then $\Yinv \mathbf{F} = [x; \lambda, y; 0, 0, z; 0, 0, 0, w]$ for $\lambda, w \ne 0$ and $z^2 - z \lambda - y w \ne 0$.
\end{enumerate}

For \casecz signatures, it is easy to show 
that $(\Yinv T \mathbf{F})(\db\db\dg) = 0$ whenever $\Yinv T \mathbf{F} \in \casecsupp$ for any $T \in \so(3, \mathbb{C})$,
as shown in \cref{prop:if-casec0-in-c-then-supp-br}.

\begin{proposition}\label{prop:if-casec0-in-c-then-supp-br}
    Let $\mathbf{F}$ be a nondegenerate \casecz signature such that $\Yinv \mathbf{F} \in \casecsupp$.
    Then, $\Yinv \mathbf{F}$ is supported on $\db\dr$.
\end{proposition}
\begin{proof}
By \cref{prop:annihilating-vector-under-Y}, 
$\ann(\Yinv \mathbf{F}) = \spn(\mathbf{u}')$ for $\mathbf{u}' = (x, y, z)$ such that $2 x y + z^2 = 0$.
Suppose $\Yinv \mathbf{F}$ is not supported on $\db\dr$.
Then, $\Yinv \mathbf{F} \sim [a; 1, b; 0, 0, c; 0, 0, 0, d]$ for some $a, b, c, d \in \mathbb{C}$.
For $\langle \Yinv \mathbf{F}, \mathbf{u}' \rangle = 0$, we must have $1 \cdot x + 0 \cdot y + 0 \cdot z = 0$, which means $x = 0$.
Then, $z = 0$ as well, which means $\mathbf{u}' \sim \mathbf{e}_2$. 
This means $\Yinv \mathbf{F}$ is supported on $\db\dr$, which can be seen from the triangular representation.
\end{proof}

To analyze the other types, we show that the image of the upper triangular matrices of $\speciallinear(2, \mathbb{C})$ is the stabilizer of $\casecsupp$ in a strong sense.
We show below that $\stab(\casecsupp) = \{ \rho(M) : M = \begin{bsmallmatrix}
    a & b \\ 0 & d
\end{bsmallmatrix} \in \speciallinear(2, \mathbb{C}) \}$ is the stabilizer subgroup of $\casecsupp$ under $\so(3, \mathbb{C}, \X)$ action.
Also, not only $\stab(\casecsupp)$ fixes $\casecsupp$ as a set,
but if any $\iny{T} \in \so(3, \mathbb{C}, \X)$ has the property that $\iny{T} \iny{\mathbf{F}} \in \casecsupp$ for some nondegenerate $\iny{\mathbf{F}} \in \casecsupp$ of arity $\ge 3$,
then $\iny{T} \in \stab(\casecsupp)$. 
This allows us to only consider $\iny{T} \in \stab(\casecsupp)$
instead of the whole $\so(3, \mathbb{C}, \X)$,
when analyzing the situation where
$\Yinv \mathbf{F} \in \casecsupp$ and $\Yinv T \mathbf{F} = \iny{T} \Yinv \mathbf{F} \in \casecsupp$ for $\iny{T} = \Yinv T \Y$ and $T \in \so(3, \mathbb{C})$.
\begin{proposition}\label{prop:stabilizer-c-1-signature}
    Let $\iny{\mathbf{F}} \in \casecsupp$ be a signature of arity $n \ge 3$ such that $\iny{\mathbf{F}}(\db, \db, \ldots, \db, \dg) = 1$.
    Let $M = \begin{bsmallmatrix}
        a & b \\ c & d
    \end{bsmallmatrix} \in \speciallinear(2, \mathbb{C})$.
    If $\rho(M) \iny{\mathbf{F}} \in \casecsupp$, then $c = 0$.
\end{proposition}
\begin{proof}
    Consider $\mathscr{F} = \{ \langle \iny{\mathbf{F}}, \mathbf{u} \rangle : \mathbf{u} \in \mathbb{C}^3 \}$. 
    For any $\mathbf{u} \in \mathbb{C}^3$,  since $\rho(M) \iny{\mathbf{F}} \in \casecsupp$, we have $\rho(M) \langle \iny{\mathbf{F}}, \mathbf{u} \rangle \in \casecsupp$ as well.
    This is because $\rho(M) \langle \iny{\mathbf{F}}, \mathbf{u} \rangle = \langle \rho(M) \iny{\mathbf{F}}, (\rho(M)^{-1})\transpose \mathbf{u} \rangle$ and $\casecsupp$ is closed under connecting a unary signature. For arity $n \ge 3$,
    we apply this operation on $\iny{\mathbf{F}}$ with $\mathbf{u} = \mathbf{e}_2$
    once
    and then $n-3$ times with $\mathbf{u} =  \mathbf{e}_1$. 
    The signature after $\rho(M)$ transform
    remains in $\casecsupp$.
    
    Consider the arity $n - 1$ signature $\iny{\mathbf{G}} = \langle \iny{\mathbf{F}}, \mathbf{e}_2 \rangle$. 
    Note that for any $x_1, \ldots, x_{n-1} \in \{\db, \dg, \dr\}$,
    $\iny{\mathbf{G}}(x_1, \ldots, x_{n-1}) = \iny{\mathbf{F}}(\dg, x_1, \ldots, x_{n-1})$.
    Then, we see that $\iny{\mathbf{G}}(x_1, \ldots, x_{n-1}) = 0$ unless $x_1, \ldots, x_{n-1} = \db$. Therefore, $\iny{\mathbf{G}} = \mathbf{e}_1^{\otimes n - 1}$.
    Then after connecting further $n - 3$ copies of $\mathbf{e}_1$ on $\iny{\mathbf{G}}$, we get $\iny{\mathbf{H}} = \mathbf{e}_1\tew$.
    
However, we have
 $\rho(M) \iny{\mathbf{H}} = (d^2, -c^2, \sqrt{2}c d)\tew \in \casecsupp$. Hence
 $c=0$.
\end{proof}

\begin{proposition}\label{prop:stabilizer-c-0-signature}
    Let $\iny{\mathbf{F}} \in \casecsupp$ be a nondegenerate signature of arity $n \ge 3$ such that $\iny{\mathbf{F}}$ is supported on $\db\dr$.
    Let $M = \begin{bsmallmatrix}
        a & b \\ c & d
    \end{bsmallmatrix} \in \speciallinear(2, \mathbb{C})$.
    If $\rho(M) \iny{\mathbf{F}} \in \casecsupp$, then $c = 0$.
\end{proposition}
\begin{proof}
Let $\iny{\mathbf{G}} = \rho(M) \iny{\mathbf{F}}$.
We have $\ann(\iny{\mathbf{F}}) = \spn(\mathbf{e}_2)$ and $\ann(\iny{\mathbf{G}}) = \spn(\mathbf{u})$ for $\mathbf{u} = (\rho(M)^{-1})\transpose \mathbf{e}_2$ by \cref{prop:annihilating-vector-under-gl-transform}.
Since $\rho(M)^{-1} = \rho(M^{-1}) = \rho(\begin{bsmallmatrix}
    d & -b \\ -c & a
\end{bsmallmatrix})$ we can compute that $\mathbf{u} = (-c^2, d^2, \sqrt{2} d c)$.

Since $\iny{\mathbf{G}} \in \casecsupp$ by assumption,
for $\langle \iny{\mathbf{G}}, \mathbf{u} \rangle = 0$,
we need 
\begin{align*}
& \mathbf{u}(\db) \iny{\mathbf{G}}(\db, \db, \ldots, \db, \dg)  +  \mathbf{u}(\dg) \iny{\mathbf{G}}(\dg, \db, \ldots, \db, \dg) + 
\mathbf{u}(\dg) \iny{\mathbf{G}}(\dr, \db, \ldots, \db, \dg)  \\ 
&=   (-c^2) \iny{\mathbf{G}}(\db, \db, \ldots, \db, \dg) + d^2 \cdot 0 + \sqrt{2} d c \cdot 0 \\
&=  (-c^2)  \iny{\mathbf{G}}(\db, \db, \ldots, \db, \dg) \\
&= 0 \, .
\end{align*}
If $ \iny{\mathbf{G}}(\db, \db, \ldots, \db, \dg) \ne 0$, then $c=0$.
If $ \iny{\mathbf{G}}(\db, \db, \ldots, \db, \dg) = 0$,
as  $\iny{\mathbf{G}} \in \casecsupp$, we have $\iny{\mathbf{G}}$ is supported on $\db\dr$, which means $\ann(\iny{\mathbf{G}}) = \spn(\mathbf{u}) = \spn(\mathbf{e}_2)$ because $\iny{\mathbf{G}}$ is nondegenerate.
Then $\mathbf{e}_2 \sim \mathbf{u}$, 
and we also get $c = 0$.
\end{proof}

\begin{corollary}
    \label{prop:casecsupp-fixed-under-rho-c-0}
    Let $M = \begin{bsmallmatrix}
        a & b \\ c & d
    \end{bsmallmatrix} \in \speciallinear(2, \mathbb{C})$.
    \begin{enumerate}
        \item If $c = 0$, then $\rho(M) \casecsupp \subseteq \casecsupp$.
        \item Let $\iny{\mathbf{F}} \in \casecsupp$ be  nondegenerate of arity 
        at least $3$.  If $\rho(M) \iny{\mathbf{F}} \in \casecsupp$  then $c = 0$.
    \end{enumerate}
\end{corollary}
\begin{proof}
Suppose $c = 0$. Then $d \ne 0$.
For such $M$, we have 
\begin{align*}
      \rho(M) = 
      \begin{bmatrix}
        d^2 & -b^2 & \sqrt{2}bd \\
        0 & a^2 & 0 \\
        0 & -\sqrt{2}ab & ad
      \end{bmatrix}
      &&
    \rho(M^{-1}) = 
      \begin{bmatrix}
        a^2 & -b^2 & -\sqrt{2}ba \\
        0 & d^2 & 0 \\
        0 & \sqrt{2}db & ad
      \end{bmatrix} \, .
\end{align*}
Suppose $\iny{\mathbf{F}} \in \casecsupp$ is of arity $n$.
Then, we may write $\iny{\mathbf{F}} = \lambda \sym(\mathbf{e}_1, \ldots, \mathbf{e}_1, \mathbf{e}_2) + \iny{\mathbf{F}}_0$ where $\langle \iny{\mathbf{F}}_0, \mathbf{e}_2 \rangle = 0$.
Then, $\rho(M) (\lambda \sym(\mathbf{e}_1, \ldots, \mathbf{e}_1, \mathbf{e}_2)) = \lambda' \sym(\mathbf{e}_1, \ldots, \mathbf{e}_1, (-b^2, a^2, -\sqrt{2}ab))$ for some $\lambda'$.
We see that the value of this function  is $0$ outside of $\db^n$, $\db^{n-1} \dg$ and $\db^{n-1} \dr$.
Also, $\rho(M) \iny{\mathbf{F}}_0$ is annihilated by $(\rho(M)^{-1})\transpose 
\mathbf{e}_2 = (0, d^2, 0)\transpose = d^2 \mathbf{e}_2$ by \cref{prop:annihilating-vector-under-gl-transform}, and thus
 $\rho(M) \iny{\mathbf{F}}_0$ is supported on $\db\dr$.
Therefore, $\rho(M) \iny{\mathbf{F}} \in \casecsupp$.

    The second statement follows from \cref{prop:stabilizer-c-1-signature,prop:stabilizer-c-0-signature}.
\end{proof}

\subsection{Representation of \texorpdfstring{$\stab(\casecsupp)$}{Stab(C)}}
\label{subsec:stabc-rep}
Thus, we have shown that to analyze the forms of $\Yinv \mathbf{F} \in \casecsupp$ for a \casec signature $\mathbf{F}$, we only need to analyze what forms $\iny{T} \iny{\mathbf{G}}$ can take for some $\iny{T} \in \stab(\casecsupp)$ and $\iny{\mathbf{G}}$ a signature taking a form in \cref{fig:each-casec-type}.
To do so, we exploit the fact that $\stab(\casecsupp)$ acts on $\casecsupp$ linearly and analyze the representation. 

Let $\casecsupp_3$ be the subset of ternary signatures of $\casecsupp$.
Consider the following basis of $\casecsupp_3$.
\begin{align*}
    \mathbf{x} &= \mathbf{e}_1\teh &&= [1; &&0, 0; &&0, 0, 0; &&0, 0, 0, 0]\\
    \mathbf{y} &= \sym(\mathbf{e}_1, \mathbf{e}_1, \mathbf{e}_3) &&= [0; &&0, 1; && 0, 0, 0; &&0, 0, 0, 0] \\
    \mathbf{z} &= \sym(\mathbf{e}_1, \mathbf{e}_3, \mathbf{e}_3) &&= [0; &&0, 0; &&0, 0, 1; &&0, 0, 0, 0] \\
    \mathbf{w} &= \mathbf{e}_3\teh &&= [0; &&0, 0; &&0, 0, 0; &&0, 0, 0, 1]\\
    \boldsymbol{\lambda} &= \sym(\mathbf{e}_1, \mathbf{e}_1, \mathbf{e}_2) &&= [0; &&1, 0; &&0, 0, 0; &&0, 0, 0, 0]
\end{align*}
 For $\rho(M) \in \stab(\casecsupp)$ where $M = \begin{bsmallmatrix}
    1/d & b \\ 0 & d
\end{bsmallmatrix}$ for $d \ne 0$, we can compute its matrix form on $\casecsupp_3$.
\begin{align*}
    \rho(M) \mathbf{x} &= [d^6; && 0, 0; && 0, 0, 0; && 0, 0, 0, 0] \\
    \rho(M) \mathbf{y} &= [3 \sqrt{2} b d^5; &&  0, d^4; && 0, 0, 0; && 0, 0, 0, 0] \\
    \rho(M) \mathbf{z} &= [6 b^2 d^4; && 0, 2 \sqrt{2} b d^3; && 0, 0, d^2; && 0, 0, 0, 0] \\
    \rho(M) \mathbf{w} &= [2 \sqrt{2} b^3 d^3; && 0, 2 b^2 d^2; && 0, 0, \sqrt{2} bd; && 0, 0, 0, 1] \\
    \rho(M) \boldsymbol{\lambda} &= [-3 b^2 d^4; &&  d^2, -\sqrt{2} bd^3; && 0, 0, 0; && 0, 0, 0, 0]
\end{align*}
This means that $\rho(M)$ acts on $\casecsupp_3 = \spn(\mathbf{x}, \mathbf{y}, \mathbf{z}, \mathbf{w}, \boldsymbol{\lambda})$ by 
\begin{equation}
\label{eq:representation-stab-c}
\psi(b, d) =  
\begin{bmatrix}
    d^6 & 3 \sqrt{2} b d^5 & 6 b^2 d^4 & 2 \sqrt{2} b^3 d^3 & -3b^2 d^4 \\
    0 & d^4 & 2 \sqrt{2} b d^3 & 2 b^2 d^2 & - \sqrt{2} b d^3 \\
    0 & 0 & d^2 & \sqrt{2} b d & 0 \\
    0 & 0 & 0 & 1 & 0 \\
    0 & 0 & 0 & 0 & d^2
\end{bmatrix} \,.
\end{equation}
We see that the matrix in \cref{eq:representation-stab-c} takes a particularly nice form, as it is upper triangular. 

\newcommand{\vz}{V^0}
\newcommand{\vobg}{V^1_2}
\newcommand{\vodz}{V^1_{\mathrm{z}}}
\newcommand{\vd}{V_d}
\begin{figure}
  \centering
  \begin{subfigure}[b]{0.4\textwidth}
\centering
\tersig{0.7}{* 0 * 0 0 * 0 0 0 *}
\caption{$\vz$}
\end{subfigure}
\begin{subfigure}[b]{0.4\textwidth}
\centering
\tersig{0.7}{* * * 0 0 0 0 0 0 0}
\caption{$\vobg$}
\end{subfigure} \\
\hfill
\\
  \begin{subfigure}[b]{0.4\textwidth}
\centering
\tersig{0.7}{* * * 0 0 * 0 0 0 0}
\caption{$\vodz$}
\end{subfigure}
\begin{subfigure}[b]{0.4\textwidth}
\centering
\tersig{0.7}{* 0 0 0 0 0 0 0 0 0}
\caption{$\vd$}
\end{subfigure}
\caption{Subrepresentations of $\stab(\casecsupp)$.}
\label{fig:subrepresentation-casecsupp}
\end{figure}
We obtain the following corollaries:
\begin{corollary}\label{cor:casecobg-form}
    Let $\mathbf{F}$ be a \casecobg signature such that $\Yinv \mathbf{F} = [x; \lambda, 0; 0, 0, 0; 0, 0, 0, 0]$ and $\lambda \ne 0$.
    If $\Yinv T \mathbf{F} = [x'; \lambda', y'; 0, 0, z'; 0, 0, 0, w'] \in \casecsupp$ for some $T \in \so(3, \mathbb{C})$, then $z', w' = 0$ and $\lambda' \ne 0$.

    Conversely, let $\mathbf{F}$ be a signature such that $\Yinv \mathbf{F} = [x; \lambda, y; 0, 0, 0; 0, 0, 0, 0]$ and $\lambda \ne 0$.
    Then, there exists some $T \in \so(3, \mathbb{C})$ such that $\Yinv T \mathbf{F} = [x'; \lambda', 0; 0, 0, 0; 0, 0, 0, 0]$.
\end{corollary}
\begin{proof}
By \cref{prop:casecsupp-fixed-under-rho-c-0}, it is sufficient to analyze the form of $\iny{T} \Yinv \mathbf{F}$. 

We prove the first statement.
    Note that $\Yinv \mathbf{F}$ is $(x, 0, 0, 0, \lambda)$ in the chosen basis for \cref{eq:representation-stab-c}. 
    We can compute that $(x', y', z', w', \lambda') = \psi(b, d) (x, 0, 0, 0, \lambda) = (*, *, 0, 0, d^2 \lambda)$, which is of the desired form. 

    For the second statement, we can compute that $\psi(b, d)(x, y, 0, 0, \lambda) = (*, d^4 y - \sqrt{2} bd^3 \lambda, 0, 0, d^2 \lambda)$. 
    We can set $b = d y/(\sqrt{2} \lambda)$ and $d = 1$.
\end{proof}

\begin{corollary}\label{cor:casecodz-form}
    Let $\mathbf{F}$ be a \casecodz signature such that $\Yinv \mathbf{F} = [x; \lambda, y; 0, 0, z; 0, 0, 0, 0]$ and $\lambda, z \ne 0$.
    If $\Yinv T \mathbf{F} = [x'; \lambda', y'; 0, 0, z'; 0, 0, 0, w'] \in \casecsupp$ for some $T \in \so(3, \mathbb{C})$, then $w' = 0$ and $z', \lambda' \ne 0$.
\end{corollary}
\begin{proof}
By \cref{prop:casecsupp-fixed-under-rho-c-0}, it is sufficient to analyze the form of $\iny{T} \Yinv \mathbf{F}$.
    Note that $\Yinv \mathbf{F}$ is $(x, y, z, 0, \lambda)$ in the chosen basis for \cref{eq:representation-stab-c}. 
    Then $(x', y', z', w', \lambda') = \psi(b, d) (x, y, z, 0, \lambda) = (*, *, d^2 z, 0, d^2 \lambda)$, which is of the desired form. 
\end{proof}

More precisely, 
\cref{eq:representation-stab-c} gives a nested chain of subrepresentations, and some of them reflect the classification of \casec subtypes.
Let $\vz, \vobg, \vodz, \vd \subset \casecsupp_3$ be the linear subspaces defined by signatures having the support pattern in \cref{fig:subrepresentation-casecsupp}.
We can easily check that $\vz, \vobg, \vodz, \vd$ are subrepresentations, with
$\vd \subset \vz \subset \casecsupp_3$ and $\vd \subset \vobg \subset \vodz \subset \casecsupp_3$. 
In fact, there are more subrepresentations, as depicted in \cref{fig:subrepresentation-c-lattice}.
\begin{figure}
\centering
\begin{tikzpicture}

\node (x) at (3, 0) {$\vd = \spn(\mathbf{x})$};

\node (xy) at (0, 2) {$\spn(\mathbf{x}, \mathbf{y})$};

\node (xyz) at (0, 4) {$\spn(\mathbf{x}, \mathbf{y}, \mathbf{z})$};

\node (xyzw) at (0, 6) {$\vz = \spn(\mathbf{x}, \mathbf{y}, \mathbf{z}, \mathbf{w})$};

\node (lxy) at (6, 3) {$\vobg = \spn(\mathbf{x}, \mathbf{y}, \mathbf{\lambda})$};

\node (lxyz) at (6, 5) {$\vodz = \spn(\mathbf{x}, \mathbf{y}, \mathbf{z}, \mathbf{\lambda})$};

\node (lxyzw) at (6, 7) {$\casecsupp_3 = \spn(\mathbf{x}, \mathbf{y}, \mathbf{z}, \mathbf{w}, \mathbf{\lambda})$};

\draw (x) -- (xy) -- (xyz) -- (xyzw);
\draw (x) -- (lxy) -- (lxyz) -- (lxyzw);
\draw (xy) -- (lxy);
\draw (xyz) -- (lxyz);
\draw (xyzw) -- (lxyzw);
\end{tikzpicture}
\caption{More subrepresentations of $\stab(\casecsupp)$.}
\label{fig:subrepresentation-c-lattice}
\end{figure}

We have so far shown that 
\begin{enumerate}
    \item If $\mathbf{F}$ is \casecz and $\Yinv \mathbf{F} \in \casecsupp$,
    then $\Yinv \mathbf{F} \in \vz$.
    \item If $\mathbf{F}$ is \casecobg and $\Yinv \mathbf{F} \in \casecsupp$,
    then $\Yinv \mathbf{F} \in \vobg$ and $(\Yinv \mathbf{F})(\db,\db,\dg) \ne 0$.
    \item If $\mathbf{F}$ is \casecodz and $\Yinv \mathbf{F} \in \casecsupp$,
    then $\Yinv \mathbf{F} \in \vodz$, $(\Yinv \mathbf{F})(\db,\db,\dg) \ne 0$, and $(\Yinv \mathbf{F})(\db,\dr,\dr) \ne 0$.
\end{enumerate}

We now consider types \casecosep and \casecogen.
Recall the seemingly strange definition of a \casecogen signature:
$\mathbf{F}$ is \casecogen if $\Yinv T \mathbf{F} 
= [x; \lambda, y; 0, 0, z;0, 0, 0, w]$ such that $\lambda, w \ne 0$ and $z^2 - z \lambda - yw \ne 0$ for some $T \in \so(3, \mathbb{C})$.
We show that the polynomial $P(X, Y, Z, W, \Lambda) = Z^2 - Z \Lambda - Y W$ on $\casecsupp$ is projectively invariant under the action of $\iny{T} \in \stab(C)$.
In other words (with a little bit of abuse of notation),
for $\iny{\mathbf{F}} = [x; \lambda, y; 0, 0, z; 0, 0, 0, w]$ identified as $(x, y, z, w, \lambda) \in \mathbb{C}^5$,
we show that $P(\iny{\mathbf{F}}) = z^2 - z \lambda - y w = 0$ if and only if $P(\iny{T}\iny{\mathbf{F}}) = 0$ for any $\iny{T} \in \stab(\casecsupp)$.

\begin{proposition}\label{prop:weird-polynomial-invariant-under-stab}
Let $P(X, Y, Z, W, \Lambda) = Z^2 - Z \Lambda - Y W$ be a polynomial on $\mathbb{C}^5$.
Then, 
\[P(\psi(b,d)(X, Y, Z, W, \Lambda))  =  d^4 P(X, Y, Z, W, \Lambda).\]
In particular, if $P(X, Y, Z, W, \Lambda) = 0$, then 
for any  $\rho(M) \in \stab(\casecsupp)$, $P( \psi(b, d) (X, Y, Z, W, \Lambda)) = 0$.
\end{proposition}
\begin{proof}
    Writing $(X', Y', Z', W', \Lambda') = \psi(b,d)(X, Y, Z, W, \Lambda)$, by \cref{eq:representation-stab-c}, we have
\begin{align*}
P(\psi(b,d)(X, Y, Z, W, \Lambda)) &= P(X', Y', Z', W', \Lambda') \\
 &= (Z')^2 -Z'\Lambda' - Y'W' \\
&= 
(d^2 Z + \sqrt{2} bd W)^2 - (d^2 Z + \sqrt{2} bd W) d^2 \Lambda \\
&\quad- (d^4 Y + 2\sqrt{2}b d^3 Z + 2 b^2 d^2 W - \sqrt{2} b d^3 \Lambda)W \\
&= d^4 Z^2 + 2\sqrt{2} b d^3 Z W + 2 b^2d^2W^2 - d^4 Z \Lambda 
-\sqrt{2} b d^3 W \Lambda \\
&\quad - d^4 Y W - 2\sqrt{2} b d^3 Z W - 2 b^2d^2 W^2 + \sqrt{2} b d^3 W \Lambda \\
&= d^4 Z^2  +  - d^4 Z \Lambda 
 - d^4 Y W   \\
&= d^4 (Z^2 - Z \Lambda - Y W) \\
&= d^4 P(X, Y, Z, W, \Lambda)
\end{align*}
Since $d \ne 0$, $P(X, Y, Z, W, \Lambda) = 0$ if and only if  $P( \psi(b, d) (X, Y, Z, W, \Lambda)) = 0$.
\end{proof}

For any signature $\iny{\mathbf{F}} \in \casecsupp$, we define the orbit under $\stab(\casecsupp)$ action: $\mathcal{O}(\iny{\mathbf{F}}) := \{ \iny{T} \iny{\mathbf{F}} : \iny{T} \in \stab(\casecsupp) \}$. 
\cref{prop:weird-polynomial-invariant-under-stab} shows that $P$ is identically $0$ on $\mathcal{O}(\iny{\mathbf{F}})$ or is never $0$. 
We further show that if $P(\iny{\mathbf{F}}) = 0$, then there is a signature of the form $[x'; \lambda', 0; 0, 0, 0; 0, 0, 0, w'] \in \mathcal{O}(\iny{\mathbf{F}})$.
\begin{proposition}\label{prop:if-P-0-then-c1-sep}
    Let $(x, y, z, w, \lambda) \in \mathbb{C}^5$ with $w \ne 0$.
    Then $P(x, y, z, w, \lambda) = 0$ if and only if
    there exists some $b, d \in \mathbb{C}^5$ with $d \ne 0$ such that 
    $\psi(b, d)(x, y, z, w, \lambda) = (x', 0, 0, w', \lambda')$ for some $x', w', \lambda'$. 
\end{proposition}
\begin{proof}
Let
$\psi(b, d)(x, y, z, w, \lambda) = (x', y', z', w', \lambda')$ for some $b, d$ to be determined later.
By \cref{eq:representation-stab-c},
\begin{enumerate}
    \item $y' = d^4 y + 2 \sqrt{2} b d^3 z + 2 b^2 d^2 w - \sqrt{2}bd^3 \lambda$ 
    \item $z' = d^2 z + \sqrt{2}bd w$
\end{enumerate}
We set $b = -\frac{dz}{\sqrt{2}w}$, and get  $z' = 0$.
Then, 
$y' = d^4 y - \frac{2d^4 z^2}{w} + \frac{d^4 z^2}{w} + \frac{d^4 z \lambda }{w} = -\frac{d^4}{w}(z^2 - z \lambda - y w) = 0$. 
Note that $d$ can be chosen arbitrarily.
\end{proof}

\begin{corollary}\label{cor:casecosep-form}
    Let $\mathbf{F}$ be a \casecosep signature such that $\Yinv \mathbf{F} = [x; \lambda, 0; 0, 0, 0; 0, 0, 0, w]$ and $w, \lambda \ne 0$.
    If $\Yinv T \mathbf{F} = [x'; \lambda', y'; 0, 0, z'; 0, 0, 0, w'] \in \casecsupp$ for some $T \in \so(3, \mathbb{C})$, then $w', \lambda' \ne 0$ and $(z')^2 - z' \lambda' - y'w' = 0$.

    Conversely, let $\mathbf{F}$ be a signature such that $\Yinv \mathbf{F} = [x; \lambda, y; 0, 0, z; 0, 0, 0, w]$ and $w, \lambda \ne 0$ and $z^2 - z \lambda - y w = 0$.
    Then, there exists some $T \in \so(3, \mathbb{C})$ such that $\Yinv T \mathbf{F} = [x'; \lambda', 0; 0, 0, 0; 0, 0, 0, w']$ and $w', \lambda' \ne 0$.
\end{corollary}
\begin{proof}
By \cref{prop:casecsupp-fixed-under-rho-c-0}, it is sufficient to analyze the form of $\iny{T} \Yinv \mathbf{F}$.

We prove the first statement.
Note that $P(\Yinv \mathbf{F}) = 0$. 
By \cref{prop:weird-polynomial-invariant-under-stab}, $P(\iny{T} \Yinv \mathbf{F}) = 0$. 
The fact that $w', \lambda' \ne 0$ follows from \cref{eq:representation-stab-c}.

The second statement follows from \cref{prop:if-P-0-then-c1-sep}.
\end{proof}

\begin{corollary}\label{cor:casecogen-form}
    Let $\mathbf{F}$ be a \casecogen signature such that $\Yinv \mathbf{F} = [x; \lambda, y; 0, 0, z; 0, 0, 0, w]$ and $w, \lambda \ne 0$ and $z^2 - z \lambda - yw \ne 0$.
    If $\Yinv T \mathbf{F} = [x'; \lambda', y'; 0, 0, z'; 0, 0, 0, w'] \in \casecsupp$ for some $T \in \so(3, \mathbb{C})$, then $w', \lambda' \ne 0$ and $(z')^2 - z' \lambda' - y'w' \ne 0$.
\end{corollary}
\begin{proof}
By \cref{prop:casecsupp-fixed-under-rho-c-0}, it is sufficient to analyze the form of $\iny{T} \Yinv \mathbf{F}$.
Note that $P(\Yinv \mathbf{F}) \ne 0$. 
By \cref{prop:weird-polynomial-invariant-under-stab}, $P(\iny{T} \Yinv \mathbf{F}) \ne 0$. 
The fact that $w', \lambda' \ne 0$ follows from \cref{eq:representation-stab-c}.
\end{proof}

\begin{corollary}\label{cor:classification-case-c-mutually-exclusive}
    The classification of nondegenerate \casec signatures into \casecz, \casecobg, \casecosep, \casecodz, and \casecogen is mutually exclusive and exhaustive.
\end{corollary}
\begin{proof}
Suppose $\mathbf{F}$ is a nondegenerate \casecz, \casecobg, \casecosep, \casecodz, or \casecogen signature.
Then, by definition, there exists some $T \in \so(3, \mathbb{C})$ such that
$\Yinv T \mathbf{F}$ takes the form in \cref{fig:each-casec-type}.
    \cref{prop:if-casec0-in-c-then-supp-br,cor:casecobg-form,cor:casecodz-form,cor:casecogen-form,cor:casecosep-form} show that for any $S \in \so(3, \mathbb{C})$ such that $\Yinv S T \mathbf{F} \in \casecsupp$, then it must take the form in \cref{fig:each-casec-type-just-in-C}. 
    We can easily see that the conditions in \cref{fig:each-casec-type-just-in-C} partitions $\casecsupp$. 
    For \casecobg and \casecosep, \cref{cor:casecobg-form} and \cref{cor:casecosep-form} show that a signature taking form in \cref{fig:each-casec-type-just-in-C} can be transformed to a form in \cref{fig:each-casec-type}. 
    Therefore, \cref{fig:each-casec-type} partitions nondegenerate signatures in $\casecsupp$ signatures as well.
    
    If $\mathbf{F}$ is a \casec signature, then there exists some $T \in \so(3, \mathbb{C})$ such that $\Yinv T \mathbf{F} \in \casecsupp$.
    Therefore, the types \casecz, \casecobg, \casecosep, \casecodz, or \casecogen partition \casec. 
\end{proof}

\subsection{A Combinatorial Interpretation of \texorpdfstring{$\stab(\casecsupp)$}{Stab(C)}}
We give a combinatorial interpretation of this interaction between $\so(3, \mathbb{C}, \X)$ and $\casecsupp$.
First, we notice that $\rho(M)$ acting on $\iny{\mathbf{F}} \in \casecsupp$ can be viewed as $\rho(M) \iny{\mathbf{F}} = (\rho(M) \X) \X \iny{\mathbf{F}}$.
If $M = \begin{bmatrix}
    a & b \\ 0 & d
\end{bmatrix}$, then
\[\rho(M) \X = \begin{bmatrix}
        -b^2 & d^2 & \sqrt{2}bd \\
        a^2  & 0 & 0 \\
        -\sqrt{2}ab & 0 & ad
      \end{bmatrix} \in \casecsupp\] 
and thus we can view $\rho(M) \X$ as a binary signature in $\casecsupp$, 
and $\rho(M) \iny{\mathbf{F}}$ as a gadget on the right side in the $\holts(\X \mid \rho(M) \X, \iny{\mathbf{F}})$ setting.
By the same argument as in \cref{sec:casec}, we can show that $\rho(M) \iny{\mathbf{F}} \in \casecsupp$.

Also, $\psi(b,d)$ in \cref{eq:representation-stab-c} is upper triangular.
A combinatorial interpretation is that the values of $(\rho(M) \iny{\mathbf{F}})(i, j, k)$ depend on those $\iny{\mathbf{F}}(i', j', k')$ such that the total number of $\db$ in $i', j', k'$ is at most
the total number of $\db$ in $i, j, k$.
For example, $\rho(M) \iny{\mathbf{F}}(\db,\dr,\dr)$ depends on $\iny{\mathbf{F}}(\db,\dr,\dr)$ and $\iny{\mathbf{F}}(\dr, \dr,\dr)$ but not on $\iny{\mathbf{F}}(\db,\db,\db), \iny{\mathbf{F}}(\db,\db,\dg), \iny{\mathbf{F}}(\db,\db,\dr)$.
And it can also be proved combinatorially for $\iny{\mathbf{G}}\X\iny{\mathbf{F}}$ any binary $\iny{\mathbf{G}} \in \casecsupp$.
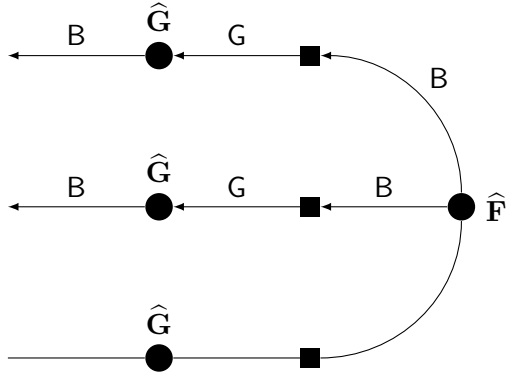
\begin{figure}
\centering
\begin{tikzpicture}
  \node[mycirc, label=right:{$\iny{\mathbf{F}}$}] (f) at (0, 0) {};
\node[mysq] (x1) at (-2, 2) {};
\node[mysq] (x2) at (-2, 0) {};
\node[mysq] (x3) at (-2, -2) {};

\node[mycirc, label=above:{$\iny{\mathbf{G}}$}] (g1) at (-4, 2) {};
\node[mycirc, label=above:{$\iny{\mathbf{G}}$}] (g2) at (-4, 0) {};
\node[mycirc, label=above:{$\iny{\mathbf{G}}$}] (g3) at (-4, -2) {};

\draw[-latex] (f) to[out=90, in=0] (x1);
\node at (-0.3, 1.7) {$\db$};
\draw[-latex] (f) -- (x2) node [midway, above] {$\db$};
\draw (f) to[out=-90, in=0] (x3);

\draw[latex-] (g1) -- (x1) node [midway, above] {$\dg$};
\draw[latex-] (g2) -- (x2) node [midway, above] {$\dg$};
\draw (g3) -- (x3);

\draw[-latex] (g1) -- (-6, 2) node[midway, above] {$\db$};
\draw[-latex] (g2) -- (-6, 0) node[midway, above] {$\db$};
\draw (g3) -- (-6, -2);

\end{tikzpicture}
\caption{Number of $\db$ outside must be more than inside.}
\label{fig:number-of-B}
\end{figure}
The argument is depicted in \cref{fig:number-of-B} where the square vertices are $\X$.
We see that whenever a $\db$ is used inside for an edge of $\iny{\mathbf{F}}$, it must propagate as the outer edge of $\iny{\mathbf{G}}\X\iny{\mathbf{F}}$ as $\db$.
Therefore, the value of $\iny{\mathbf{F}}(i', j', k')$ only affects those $(\iny{\mathbf{G}}\X\iny{\mathbf{F}})(i, j, k)$ where $(i, j, k)$ has at least as many $\db$ as $(i', j', k')$.

\subsection{A Combinatorial Interpretation of \texorpdfstring{$P$}{P}}
Holant signature grids can be viewed as defining polynomials invariant under holographic transformations~\cite{converse-holant-real-young,converse-holant-cai-young}. 
Recall the polynomial $P(X, Y, Z, W, \Lambda) = Z^2 - Z \Lambda - YW$. 
Why should it be $0$ on \casecosep?
We show that there is a Holant instance that realizes this polynomial such that it is easy to see that a \casecosep signature in the defining form evaluates to $0$.

Let $A_{p, q, r} = \begin{bsmallmatrix}
    0 & p & q \\
    -p & 0 & r \\
    -q & -r & 0
\end{bsmallmatrix}$.
For $M = \begin{bsmallmatrix}
    1/d & b \\ 0 & d
\end{bsmallmatrix} \in \speciallinear(2, \mathbb{C})$, we have
$(\rho(M)^{-1})\transpose A_{p, q, r} \rho(M)^{-1
} = A_{p', q', r'}$ for some $p', q', r'$ such that $q' = q/d^2$.
Consider the signature grid $\Omega(\iny{\mathbf{F}}, p, q, r)$ in \cref{fig:holant-instance-for-P}.
Interestingly, $A_{p,q,r}$ is not necessarily fixed under $\rho(M)$, but because of the cancellations when connected to $\casecsupp$ signatures, we can use it to define a projective invariant. 
\begin{figure}
\begin{subfigure}[b]{0.45\textwidth}
\centering
\begin{tikzpicture}
  \node[mycirc, label=right:{$\iny{\mathbf{F}}$}] (f) at (0, 0) {};
\node[draw, circle, label=above:{$A_{p, q, r}$}] (x1) at (-2, 2) {};
\node[draw, circle, label=above:{$\X$}] (x2) at (-2, 0) {};
\node[draw, circle, label=above:{$A_{p, q, r}$}] (x3) at (-2, -2) {};

\node[mycirc, label=left:{$\iny{\mathbf{F}}$}] (f2) at (-4, 0) {};

\draw (f) to[out=90, in=0] (x1);
\draw (f) -- (x2);
\draw (f) to[out=-90, in=0] (x3);

\draw (f2) to[out=90, in=180] (x1);
\draw (f2) -- (x2);
\draw (f2) to[out=-90, in=180] (x3);

\end{tikzpicture}
\caption{$\Omega(\iny{\mathbf{F}}, p, q, r)$}
\label{fig:holant-instance-for-P}
\end{subfigure}
\hfill
\begin{subfigure}[b]{0.45\textwidth}
\centering
\begin{tikzpicture}
  \node[mycirc, label=right:{$\iny{\mathbf{F}}$}] (f) at (0, 0) {};
\node[draw, circle, label=above:{$A_{p, q, r}$}] (x1) at (-2, 2) {};
\node[draw, circle, label=above:{$\X$}] (x2) at (-2, 0) {};
\node[draw, circle, label=above:{$A_{p, q, r}$}] (x3) at (-2, -2) {};

\node[mycirc, label=left:{$\iny{\mathbf{F}}$}] (f2) at (-4, 0) {};

\draw[latex-] (f) to[out=90, in=0] (x1);
\draw[latex-] (f) -- (x2) node [midway, above] {$\dg$};
\draw[latex-] (f) to[out=-90, in=0] (x3);
\node () at (-3.7, 1.7) {$\dg$};
\node () at (-3.7, -1.7) {$\db$};
\node () at (-0.1, -1.7) {$\dg, \dr$};

\draw[latex-] (f2) to[out=90, in=180] (x1);
\draw[-latex] (f2) -- (x2) node [midway, above] {$\db$};
\draw[-latex] (f2) to[out=-90, in=180] (x3);

\end{tikzpicture}
\caption{Cannot use $\dg$ next to $A_{p, q, r}$}
\label{fig:vanishing-a-pqr}
\end{subfigure}
\caption{Holant signature grid to define the projective invariant $P$.}
\end{figure}
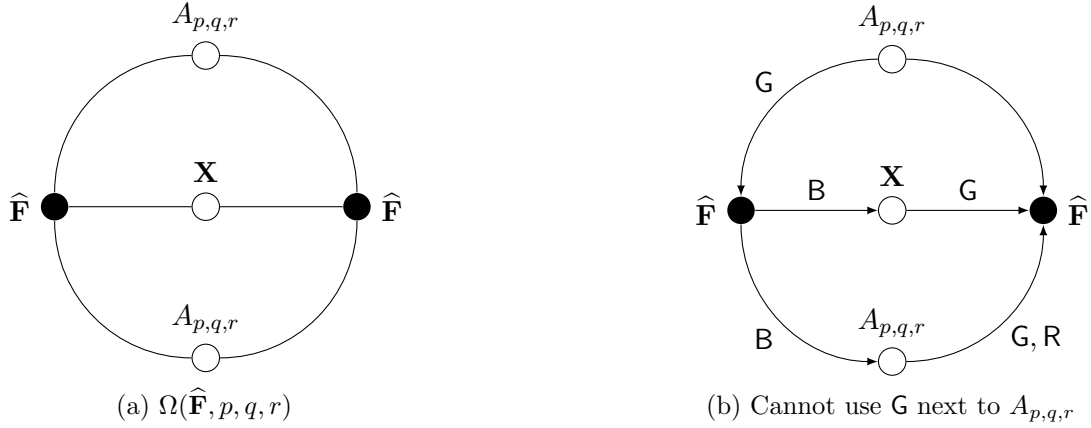

Let $H(\iny{\mathbf{F}}, p, q, r)$ be the Holant value of $\Omega(\iny{\mathbf{F}}, p, q, r)$.
\begin{proposition}
For any $\iny{\mathbf{F}} \in \casecsupp$,
$H(\iny{\mathbf{F}}, p, q, r) = H(\iny{\mathbf{F}}, 0, q, 0)$. 
\end{proposition}
\begin{proof}
We claim that any edge adjacent to $A_{p, q, r}$ cannot be assigned to $\dg$. 
See \cref{fig:vanishing-a-pqr} for reference.
Suppose the left edge to the $A_{p, q,r}$ at the top gets $\dg$.
Then, the other two edges out of $\iny{\mathbf{F}}$ on the left must get $\db$. 
Through $\X$, this propagates as $\dg$.
Through the bottom $A_{p,q,r}$, this propagates as $\dg$ or $\dr$.
We see that the right side always evaluates to $0$ because $\iny{\mathbf{F}}(\dg,\dg, *) = 0$ and  $\iny{\mathbf{F}}(\dg,\dr, *) = 0$. 
This means that the only surviving assignments use $\db$ or $\dr$ for edges adjacent to $A_{p, q, r}$, so 
the Holant value would not change if we replace $A_{p,q,r}$ with $A_{0, q, 0}$.
Therefore, $H(\iny{\mathbf{F}}, p, q, r) = H(\iny{\mathbf{F}}, 0, q, 0)$.
\end{proof}

The reason we choose this particular Holant instance is because we can easily see that if $\iny{\mathbf{F}} = [x; \lambda, 0; 0, 0, 0; 0, 0, 0, w]$, i.e., $\iny{\mathbf{F}}$ is \sepbg, then $H(\iny{\mathbf{F}}, 0, q, 0) = 0$. 
This is because if we look at just the two $\iny{\mathbf{F}}$ and $\X$ part of $\Omega(\iny{\mathbf{F}}, p, q, r)$, it is \sepbg as well because $\iny{\mathbf{F}}$ and $\X$ are \sepbg. 
But $A_{0, q, 0}$ is only supported on $\{\db\dr, \dr\db\}$, so $H(\iny{\mathbf{F}}, 0, q, 0) = 0$.

If we view $\iny{\mathbf{F}} \in \casecsupp$ as a signature with variables $X, Y, Z, W, \Lambda$, as $\iny{\mathbf{F}} = [X; \Lambda, Y; 0, 0, Z; 0, 0, 0, W]$, 
we can compute that $H(\iny{\mathbf{F}}, 0, q, 0) = -2q^2(Z^2 - Z \Lambda - YW) = - 2q^2 P(X, Y, Z, W, \Lambda)$. 
Consider applying a holographic transformation by $\rho(M)$ on $\Omega(\iny{\mathbf{F}}, p, q, r)$, viewing it as a $\holts(\X, A_{p, q, r} \mid \iny{\mathbf{F}})$ instance.
We get 
\[\holts( \X, (\rho(M)^{-1})\transpose  A_{p, q, r} \rho(M)^{-1} \mid \rho(M) \iny{\mathbf{F}}) = 
\holts( \X, A_{p', q', r'} \mid \rho(M) \iny{\mathbf{F}}) 
\] for some $p', q', r'$ such that $q' = q/d^2$.
Let $\iny{\mathbf{F}}' = \rho(M) \iny{\mathbf{F}} = [X'; \Lambda', Y'; 0, 0, Z'; 0, 0, 0, W']$.
The resulting instance is
$\Omega(\iny{\mathbf{F}}', p', q', r')$.
Since the Holant value is preserved, we have 
\begin{align*}
 P(X, Y, Z, W, \Lambda) &= -\frac{1}{2q^2} H(\iny{\mathbf{F}}, 0, q, 0) \\
 &= -\frac{1}{2q^2}  H(\iny{\mathbf{F}}', p', q', r') \\
 &= -\frac{1}{2q^2}  H(\iny{\mathbf{F}}', 0, q', 0) \\
 &=  -\frac{1}{2q^2}  (- 2 (q')^2 ((Z')^2 - Z' \Lambda' - Y'W')) \\
 &= \frac{1}{d^4} P(X', Y', Z', W', \Lambda') \\
 &=\frac{1}{d^4} P(\psi(b, d) (X, Y, Z, W, \Lambda))
 \end{align*}
giving an alternative proof of \cref{prop:weird-polynomial-invariant-under-stab}.

Of course, this does not show that $\Omega(\iny{\mathbf{F}}, p, q, r) = 0$ implies that we can transform $\iny{\mathbf{F}}$ to be \sepbg. 
In general, the converse of Holant theorem is not true when there are vanishing signatures~\cite{converse-holant-cai-young}.

\section{Tractability}\label{sec:tractable-cases}
We discuss the tractable cases and give an algorithm for each. 
The tractable cases given in \cref{subsec:tractable-case-binary,subsec:tractable-e,subsec:tractable-bg-r} can be viewed as a generalization of tractable cases $\mathscr{A}$, $\mathscr{B}$, and $\mathscr{D}$ given in~\cite{domain3-real-set}. 
We give a brief sketch of the algorithm for those cases. 
The tractable cases given in \cref{subsec:tractable-case-c,subsec:tractable-case-ebd} use the new algorithms we present in this paper.
A summary of the tractable cases can be found in \cref{thm:tractable-cases}.

One could expect that the number of tractable cases would significantly increase compared to that of the dichotomy of real-valued signatures in~\cite{domain3-real-set}, which gave $5$ cases each with a different algorithm.
However, some of the distinct tractable cases in~\cite{domain3-real-set}, in particular $\mathscr{C}$ and $\mathscr{E}$, can be combined using our new criteria for pairwise tractability for signatures with non-trivial annihilator. 
We remark that $\mathscr{C}$, $\mathscr{D}$ and $\mathscr{E}$ of~\cite{domain3-real-set} when every signature of arity $3$ or higher has a nonzero annihilator can be solved by one algorithm that we give in \cref{thm:algorithm-ebd-with-binary}. 
This algorithm also solves the examples given in \cref{subsec:tractable-gadget-meet}.
A downside is that we do not give an explicit expression for the binary signatures. 
However, the dichotomy criteria is still decidable.

Let $\mathcal{F}$ be a set of nondegenerate, complex-valued, symmetric signatures over domain $\{\db,\dg,\dr\}$ of arity $\ge 3$.

\subsection{Binary Signatures}
\label{subsec:tractable-case-binary}
Suppose every signature in $\mathcal{F}$ has arity $\le 2$. 
Then, $\holts(\mathcal{F})$ is tractable. 
This is a standard tractable case in the Holant theory, 
for which we can do a matrix multiplication for each connected component and take the trace. 

\subsection{\texorpdfstring{$\mathcal{E}$}{E} signatures}
\label{subsec:tractable-e}
Suppose there exists some $T \in \ortho(3, \mathbb{C})$ such that $T \mathcal{F} \subset \mathcal{E}$ or $\Yinv T \mathcal{F} \subset \mathcal{E}$. 
For the first, we have $\holts(\mathcal{F}) =_T \holts(T \mathcal{F}) \le_T \holts(\mathcal{E})$.
For the second, we have $\holts(\mathcal{F}) =_T \holts(\X \mid \Yinv T \mathcal{F}) \le_T \holts(\X \mid \mathcal{E}) \le_T \holts(\mathcal{E})$ because $\X \in \mathcal{E}$. 
Again, this is a well known tractable case, because any choice of assignment for a single edge in a given input grid must propagate uniquely to the whole connected component. 
We can use the same algorithm as $\mathscr{B}$ of~\cite{domain3-real-set} by noticing that the values being real or complex does not matter.
Alternatively, it can also be viewed as a higher domain generalization of the tractable case given in~\cite{holant-star-asymmetric}.

\subsection{Domain Separated Signatures}
\label{subsec:tractable-bg-r}
Suppose
there exists some $T \in \ortho(3, \mathbb{C})$ such that 
$T \mathcal{F}$ has the following properties:
\begin{enumerate}
    \item For all $\mathbf{F} \in T \mathcal{F}$ of arity $3$ or higher, $\mathbf{F}$ is \sepbg.
    \item For all binary $\mathbf{G} \in T \mathcal{F}$, either $\mathbf{G}$ is \sepbg or \swbg.
    \item $\holbs((T \mathcal{F})\domres{\db\dg})$ is tractable.
\end{enumerate}
This can be solved by the same algorithm as  $\mathscr{D}$ in~\cite{domain3-real-set} 
since the values being real or complex does not impact the algorithm. 
Note that in our case, the signatures of arity $3$ or higher could be \casecosep or \casecobg, which was not included in~\cite{domain3-real-set}.

We provide a sketch of the algorithm.
We may assume that the $\holts(T \mathcal{F})$ input grid $\Omega$ is connected.
If no \swbg binary signature is used, then every signature used in $\Omega$ is \sepbg. 
Then, for any nonzero edge assignment $\sigma$, $\sigma$ must only use $\{\db,\dg\}$ or only use $\dr$. 
The sum of the values for the assignment of the second kind can be easily computed as there is only one such assignment. 
The  sum of the values for the assignments of the first kind can be computed in polynomial time because it can be viewed as the Holant value of a $\holbs((T \mathcal{F})\domres{\db\dg})$ grid, which is assumed to be tractable. 
If a \swbg signature is used, then we can argue that any cycle containing a \swbg signature must use a \swbg signature even number of times for $\hol_\Omega$ to be nonzero.
This is because the edge assignment of the two adjacent edges of a \swbg signature must alternate between $\{\db, \dg\}$ and $\dr$. 
Therefore, if we view $\Omega$ as a graph with \sepbg signatures forming clusters and \swbg signatures as edges between the clusters, it must be bipartite. 
There are two kinds of nonzero assignments, that the clusters on the left side get assigned $\{\db, \dg\}$ and those on the right side get assigned $\dr$, or the opposite.
Noticing that a \swbg matrix $\begin{bsmallmatrix}
    0 & 0 & a \\ 0 & 0 & b \\ c & d & 0 
\end{bsmallmatrix}$ can be decomposed as $(a, b, 0) \otimes \mathbf{e}_3 + \mathbf{e}_3 \otimes (c, d, 0)$, 
we can compute the sum of the values for the assignments of the first kind by replacing each \swbg with the $(*, *, 0) \otimes \mathbf{e}_3$ component.
This separates $\Omega$ into left and right, and the left side clusters can be viewed as instances of $\holbs((T \mathcal{F}) \domres{\db\dg})$, while the right side gets all $\dr$. 
The sum of the values for the assignments of the second kind can be computed in polynomial time by a symmetric argument, and hence $\hol_{\Omega}$ can be computed in polynomial time.
See~\cite{domain3-real-set} for more details.

\subsection{\texorpdfstring{$\casecsupp$}{C} Signatures}
\label{subsec:tractable-case-c}
Suppose there exists some $T \in \ortho(3, \mathbb{C})$ such that
$\Yinv T \mathcal{F} \subset \casecsupp$.
Then, $\holts(\mathcal{F}) =_T \holts(\X \mid \Yinv T \mathcal{F}) \le_T \holts(\X \mid \casecsupp)$.
An algorithm for $\holts(\mathbf{X} \mid \casecsupp)$ is described in \cref{subsec:case-c-algorithm}. 

\subsection{Signatures with Annihilators}
\label{subsec:tractable-case-ebd}
The algorithm for this case is a little bit more involved. 
We explain the main intuition. 
First, we go back to the Boolean domain $\hols$ algorithms. 
We can separate them into two kinds.

Type $\typei(a, b)$ for $b \ne \pm 2 \i a$ and type $\typeii$ can be viewed as an algorithm by \emph{unique propagation}: after a suitable holographic transformation, we get $\holbs(\mathcal{E}_2)$ or $\holbs(\bne \mid \mathcal{E}_2)$, in which $\mathcal{E}_2$ is the Boolean domain analogue of $\mathcal{E}$. 
Then, assigning a domain element to one edge propagates uniquely to the whole connected component. 

Type $\typei(a, b)$ for $b = \pm 2 \i a$ can be viewed as a \emph{counting and vanishing argument}:
after a suitable holographic transformation, 
we get $\holbs(\bne \mid \mathcal{M})$, in which 
$\mathcal{M}$ is the set of signatures of the form $[*, *, 0, \ldots, 0]$. 
Then, unless every signature is binary after absorbing the unary signatures, in any edge assignment, there is a mismatch in the number of domain elements $0, 1$ on the edges, between the constraints imposed by the left side and the right side, resulting in a Holant value of $0$.

We define the following notion of a \emph{compatible pair} to state the algorithm. 
It will be proved later in \cref{lem:not-orthogonal-pair-dichotomy-ternary} that for a pair of ternary signatures $\mathbf{F}_1, \mathbf{F}_2$ with nonzero annihilators, if $(\mathbf{F}_1, \mathbf{F}_2)$ is not a compatible pair,
then $\holts(\mathbf{F}_1, \mathbf{F}_2)$ is \sph. 
\begin{definition}[Compatible Pairs]
\label{def:compatible-pairs}
  Let $\mathbf{F}_1, \mathbf{F}_2$ be nondegenerate symmetric signatures of arity $3$ or higher such that $\ann(\mathbf{F}_i) = \spn(\mathbf{n}_i) \ne 0$ for $i = 1, 2$.
  We say $(\mathbf{F}_1, \mathbf{F}_2)$ is a \emph{compatible pair} if one of the following holds.
\begin{enumerate}
  \item \label{itm:compatible-pair-orthogonal-ann} $\langle \mathbf{n}_1, \mathbf{n}_2 \rangle = 0$.
  \item \label{itm:compatible-pair-unique-prop} For some $\mathbf{u}_i, \mathbf{v}_i \in \mathbb{C}^3$, 
  $\mathbf{F}_i = \mathbf{u}_i^{\otimes m_i} + \mathbf{v}_i^{\otimes m_i}$ ($i=1,2$),
  satisfying
   $\langle \mathbf{u}_1, \mathbf{u}_2 \rangle =\langle \mathbf{v}_1, \mathbf{v}_2 \rangle = 0$.

  \item \label{itm:compatible-pair-vanish} For some $\mathbf{u}_i, \mathbf{v}_i \in \mathbb{C}^3$, and $\lambda_i \in \mathbb{C}$, $\mathbf{F}_i = \lambda_i \mathbf{u}_i^{\otimes m_i} + \sym(\mathbf{u}_i, \ldots, \mathbf{u}_i, \mathbf{v}_i)$ ($i=1,2$),  satisfying  $\langle \mathbf{u}_1, \mathbf{u}_2 \rangle = 0$. 
\end{enumerate}
\end{definition}

We can relate this definition with the above discussion of the Boolean domain tractable cases. 
For \cref{itm:compatible-pair-unique-prop} and \cref{itm:compatible-pair-vanish}, a similar constraint applies in the Boolean domain.
For example, if $\mathbf{F}_1, \mathbf{F}_2$ are type $\typei(a, b)$ signatures with $b \ne \pm 2 \i a$ 
and
$\holbs(\mathbf{F}_1, \mathbf{F}_2)$ is tractable,
then $\mathbf{F}_i = \mathbf{u}_i^{\otimes m_i} + \mathbf{v}_i^{\otimes m_i}$ with 
$\langle \mathbf{u}_1, \mathbf{u}_2 \rangle = \langle \mathbf{v}_1, \mathbf{v}_2 \rangle = 0$, for $\mathbf{u}_i, \mathbf{v}_i \in \mathbb{C}^2$.
However, the difference in $\mathbb{C}^2$ is that these $\mathbf{u}_i, \mathbf{v}_i$ are then essentially determined;
given $\mathbf{u}_1$ and $\mathbf{v}_1$, there is a unique solution up to scalar for $\mathbf{u}_2$ and $\mathbf{v}_2$. 
If $\mathcal{F}$ is a set of signatures of arity $3$ or higher containing $\mathbf{F}_1$ such that
$\holbs(\mathcal{F})$ is tractable, then
there exist $\mathbf{u}_0, \mathbf{v}_0 \in \mathbb{C}^2$ such that every signature in $\mathcal{F}$ can be written as $\mathbf{u}^{\otimes n} + \mathbf{v}^{\otimes n}$ for some $\{\mathbf{u}, \mathbf{v} \} \sim \{\mathbf{u}_0, \mathbf{v}_0 \}$.
This allows one to apply a transform to bring all $\mathcal{F}$ to $\mathcal{E}_2$, and we  have a unique propagation of the \emph{domain elements}.

In domain $3$, there are more degrees of freedom for the vectors in $\mathbb{C}^3$. 
The orthogonality requirements in condition 2 of \cref{def:compatible-pairs} do not imply the existence of a transform to bring every such signature simultaneously to $\mathcal{E}$. 
Some examples were given in \cref{subsec:tractable-gadget-meet}.
Instead, it is more helpful to view the computation of Holant value in a \emph{basis free} way; rather than trying to argue that there is some change of basis that gives a unique propagation of the \emph{choice of the domain elements}, we instead argue directly on the tensor summation expression that there is a unique propagation in the \emph{choice of the summand}. 
Indeed, this is why we can have a \caseatwo signature (which corresponds to $\typei(a, b)$ for $b \ne \pm 2 \i a$) and a \casebtwo signature (which corresponds to $\typeii$) interacting in a nondegenerate way but still tractable, because they are of the same unique propagation kind. 
Similarly, the orthogonality requirement in condition 3 can be seen as a sufficient condition for 
counting and vanishing of the choice of the summand. 

In other words, in the Boolean domain, the lack of degrees of freedom means there is a direct correspondence between the domain elements and the summands after a transform: for instance, if $\mathbf{u}_0, \mathbf{v}_0$ are non-isotropic and orthogonal, there exists some orthogonal $T$ such that $T \mathbf{u}_0 = \sbto$ and $T \mathbf{v}_0 = \sbtt$. 
This makes every signature in $T \mathcal{F}$ to be of the form $[*, 0, \ldots, 0, *] = x \sbto^{\otimes n} + y \sbtt^{\otimes n}$.
To solve $\holbs(T \mathcal{F})$, we can easily see that either all edges must get $0$ or all edges must get $1$ because the domain elements must propagate uniquely to the whole connected component. 
A different perspective is that if we choose the term $\sim \sbto^{\otimes n}$ for a vertex, we cannot choose the term $\sim \sbtt^{\otimes m}$ for the neighboring vertex because $\langle \sbto, \sbtt \rangle = 0$.
There is a corresponding choice in the $\holbs(\mathcal{F})$ grid,
in which if we choose $\sim \mathbf{u}_0^{\otimes n}$ for a vertex, we cannot choose $\sim \mathbf{v}_0^{\otimes m}$ for the neighboring vertex because $\langle \mathbf{u}_0, \mathbf{v}_0 \rangle = 0$.
Because it is possible to transform $\mathcal{F}$ into $\mathcal{E}_2$, it feels more natural to solve $\holbs(\mathcal{F})$ after the $T$ transform. 
In domain $3$, because of the additional degrees of freedom, it is simply not possible to solve certain tractable cases as a unique propagation of the domain elements, and a more natural algorithm is given by viewing it as a unique propagation of the summands. 

\cref{itm:compatible-pair-orthogonal-ann} of \cref{def:compatible-pairs} is related to the idea that some signatures interact only degenerately, and their interaction cannot give rise to \sphness, which leads to the next proposition. 
      \begin{proposition}\label{prop:orthogonal-normal-vector-edge-factors}
        Let $\mathbf{F}_1$ and $\mathbf{F}_2$ be nondegenerate symmetric ternary signatures on domain $3$
        with $\ann(\mathbf{F}_i) = \spn(\mathbf{n}_i) \ne 0$.
        Suppose $\langle \mathbf{n}_1, \mathbf{n}_2 \rangle = 0$.
        Then, in any signature grid $\Omega$ containing $\mathbf{F}_1$ and $\mathbf{F}_2$, an edge between $\mathbf{F}_1$ and $\mathbf{F}_2$ can be replaced with $\mathbf{v}_1 \otimes \mathbf{v}_2$ for some $\mathbf{v}_1, \mathbf{v}_2 \in \mathbb{C}^3$.
      \end{proposition}
      \begin{proof}
        Suppose $\mathbf{n}_1$ and $\mathbf{n}_2$ are not isotropic.
        Then, there exists an orthogonal matrix $T$ such that $T \mathbf{n}_i \sim \mathbf{e}_i$.
        We see that then $T \mathbf{F}_1$ is supported on $\dg\dr$ and $T \mathbf{F}_2$ is supported on $\db\dr$.
        Therefore, any connection between $T \mathbf{F}_1$ and $T \mathbf{F}_2$ factors as $\mathbf{e}_3$.
        Another way to view this is that we can write $I_3 = \mathbf{e}_1 \otimes \mathbf{e}_1 + \mathbf{e}_2 \otimes \mathbf{e}_2 + \mathbf{e}_3 \otimes \mathbf{e}_3$, and the only term that survives for an edge between $T \mathbf{F}_1$ and $T \mathbf{F}_2$ is $\mathbf{e}_3 \otimes \mathbf{e}_3$.

        Suppose $\mathbf{n}_1$ is not isotropic and $\mathbf{n}_2$ is isotropic.
        By Lemma 3.6 (arXiv version) of~\cite{domain3-single-signature}, there exists an orthogonal $T$ such that $T \mathbf{n}_1 \sim \mathbf{e}_3$ and $T \mathbf{n}_2 \sim \bzp$.
        When there is an edge between $T \mathbf{F}_1$ and $T \mathbf{F}_2$, we can write $I_3 = \bzp \otimes \bzm + \bzm \otimes \bzp + \mathbf{e}_3 \otimes \mathbf{e}_3$.
        We see that only $\bzp \otimes \bzm$ survives since the other terms kill one of $T \mathbf{F}_1$ or $T \mathbf{F}_2$.
        So, the edge factors degenerately as $\bzp \otimes \bzm$.

        Suppose $\mathbf{n}_1$ and $\mathbf{n}_2$ are both isotropic.
        Note that since $\mathbf{n}_1, \mathbf{n}_2 \in \mathbb{C}^3$, $\mathbf{n}_i$ both being isotropic and $\langle \mathbf{n}_1, \mathbf{n}_2 \rangle = 0$ implies $\mathbf{n}_1 \sim \mathbf{n}_2$, because $\begin{bsmallmatrix}
          \mathbf{n}_1 \\ \mathbf{n}_2
        \end{bsmallmatrix} \mathbf{x} = 0$ has solutions $\mathbf{n}_1$ and $\mathbf{n}_2$.
        Then, by Lemma 3.4 (arXiv version) of~\cite{domain3-single-signature}, there exists an orthogonal $T$ such that $T \mathbf{n}_i \sim \bzp $. 
        We can apply the same argument as above of writing $I_3 = \bzp \otimes \bzm + \bzm \otimes \bzp + \mathbf{e}_3 \otimes \mathbf{e}_3$, and only $\mathbf{e}_3 \otimes \mathbf{e}_3$ survives between $T \mathbf{F}_1$ and $T \mathbf{F}_2$.
      \end{proof}

We show how one example given in \cref{subsec:tractable-gadget-meet} can be solved using this idea compatible pairs.
Then, we generalize it and prove \cref{thm:ternary-set-ebd-algorithm}.

\begin{example}[\cref{ex:b2-a2-a2} revisited]\label{ex:b2-a2-a2-rev}
    Recall \cref{ex:b2-a2-a2}, in which we defined
$\mathbf{F}_1 
=  \begin{bsmallmatrix}
1 \\ \i \\0
\end{bsmallmatrix}\teh
+
\begin{bsmallmatrix}
1 \\ -\i \\0
\end{bsmallmatrix}\teh$,
$\mathbf{F}_2 = 
\begin{bsmallmatrix}
1 \\ \i \\p_2
\end{bsmallmatrix}\teh
+
\begin{bsmallmatrix}
1 \\ -\i \\q_2
\end{bsmallmatrix}\teh$
 and
$\mathbf{F}_3 = 
\begin{bsmallmatrix}
1 \\ \i \\p_3
\end{bsmallmatrix}\teh
+
\begin{bsmallmatrix}
1 \\ -\i \\q_3
\end{bsmallmatrix}\teh$,
such that $p_2 q_2 = p_3 q_3 = -2$ and $p_2/p_3 = e^{\pi \i / 3}$. 
Writing $\mathbf{F}_i = \mathbf{u}_i\teh + \mathbf{v}_i\teh$, we can see that 
$\langle \mathbf{u}_1, \mathbf{u}_2 \rangle = 0$, $\langle \mathbf{v}_1, \mathbf{v}_2 \rangle = 0$, 
$\langle \mathbf{u}_1, \mathbf{u}_3 \rangle = 0$, and $\langle \mathbf{v}_1, \mathbf{v}_3 \rangle = 0$. 
Therefore, $(\mathbf{F}_1, \mathbf{F}_2)$ and $(\mathbf{F}_1, \mathbf{F}_3)$ are compatible pairs.

Also note that $(\mathbf{F}_i, \mathbf{F}_i)$ is a compatible pair for $i = 1, 2, 3$.
For $\mathbf{F}_1$, $\langle \mathbf{u}_1, \mathbf{u}_1 \rangle = 0$ and $\langle \mathbf{v}_1, \mathbf{v}_1 \rangle = 0$. 
For $\mathbf{F}_2$, the orthogonality condition is true after reordering, by rewriting $\mathbf{u}_2\teh + \mathbf{v}_2\teh = \mathbf{v}_2\teh + \mathbf{u}_2\teh$.  

$(\mathbf{F}_2, \mathbf{F}_3)$ is also a compatible pair. 
Note that we do not necessarily have $\langle \mathbf{u}_2, \mathbf{u}_3 \rangle = 0$ or $\langle \mathbf{u}_2, \mathbf{v}_3 \rangle = 0$.
However, for $i = 2, 3$, we can compute that $\ann(\mathbf{F}_i) = \spn(\mathbf{n}_i)$ where $\mathbf{n}_i = (p_i + q_i, \i(q_i - p_i), 2)$. 
Then, 
\[\langle \mathbf{n}_2, \mathbf{n}_3 \rangle = 2 p_2 q_3 + 2 q_2 p_3 + 4 = -4(\frac{p_2}{p_3} + \frac{p_3}{p_2} - 1) = -4 (e^{\pi \i / 3} + e^{-\pi \i / 3} - 1) = 0 \, .\]
Therefore, every pair is a compatible pair.

Suppose $\Omega$ is a $\holts(\mathbf{F}_1, \mathbf{F}_2, \mathbf{F}_3)$ instance.
We illustrate the process in \cref{fig:unique-propagation}.
We can first replace every edge between $\mathbf{F}_2$ and $\mathbf{F}_3$ by two unary signatures by \cref{prop:orthogonal-normal-vector-edge-factors} because their annihilators are orthogonal. 
Let $C$ be a connected component.
Instead of computing the Holant value by the sum of product over all edge assignments to domain elements (\cref{eq:holant-value}), we may compute it by summing over all
$2^{\lvert V(C) \rvert}$ Holant instances where each one is obtained from an assignment $\sigma: V(C) \to \{\mathbf{u}, \mathbf{v}\}$ that chooses either $\mathbf{u}_i\teh$ or $\mathbf{v}_i\teh$ for each vertex $v$ with signature $\mathbf{F}_i$. 
The Holant value of each instance can be computed in polynomial time since it is solely composed of degenerate signatures, 
and we can identify the only two assignments that can
result in nonzero values, because of unique propagation:
E.g., if there is an edge between a $\mathbf{F}_1$ and $\mathbf{F}_2$ signatures,
and $\mathbf{u}_1$ is chosen for $\mathbf{F}_1$, then $\mathbf{v}_2$ must be chosen for $\mathbf{F}_2$ since $\langle \mathbf{u}_1, \mathbf{u}_2 \rangle = 0$.

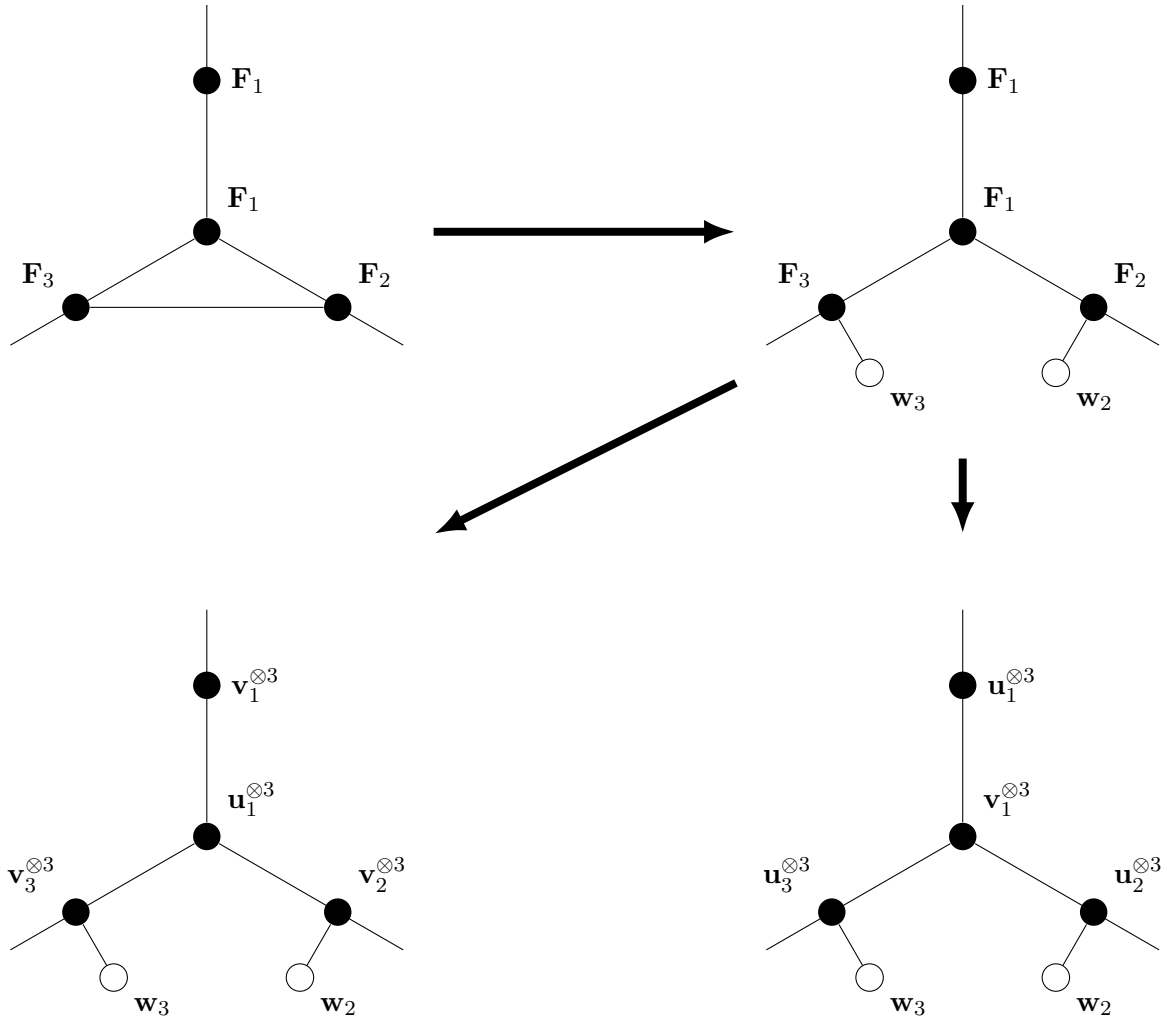
\begin{figure}
\centering
\begin{tikzpicture}
  \node[mycirc, label=above right:{$\mathbf{F}_1$}] (f1) at (0, 0) {};
  \draw (f1) --++ (330:2) node[mycirc, label=above right:{$\mathbf{F}_2$}] (f2) {} --++ (330:1);
  \draw (f1) --++ (90:2) node[mycirc, label=right:{$\mathbf{F}_1$}] (f11) {} --++ (90:1);
  \draw (f1) --++ (210:2) node[mycirc, label=above left:{$\mathbf{F}_3$}] (f3) {} --++ (210:1);
  \draw (f2) --++ (f3);

  \node[mycirc, label=above right:{$\mathbf{F}_1$}] (f1) at (10, 0) {};
  \draw (f1) --++ (330:2) node[mycirc, label=above right:{$\mathbf{F}_2$}] (f2) {} --++ (330:1);
  \draw (f1) --++ (90:2) node[mycirc, label=right:{$\mathbf{F}_1$}] (f11) {} --++ (90:1);
  \draw (f1) --++ (210:2) node[mycirc, label=above left:{$\mathbf{F}_3$}] (f3) {} --++ (210:1);
  \draw (f2) --++ (240:1) node[draw, circle, fill=white, label=below right:{$\mathbf{w}_2$}] () {};
  \draw (f3) --++ (300:1) node[draw, circle, fill=white, label=below right:{$\mathbf{w}_3$}] () {};
  
  \node[mycirc, label=above right:{$\mathbf{u}_1\teh$}] (f1) at (0, -8) {};
  \draw (f1) --++ (330:2) node[mycirc, label=above right:{$\mathbf{v}_2\teh$}] (f2) {} --++ (330:1);
  \draw (f1) --++ (90:2) node[mycirc, label=right:{$\mathbf{v}_1\teh$}] (f11) {} --++ (90:1);
  \draw (f1) --++ (210:2) node[mycirc, label=above left:{$\mathbf{v}_3\teh$}] (f3) {} --++ (210:1);
  \draw (f2) --++ (240:1) node[draw, circle, fill=white, label=below right:{$\mathbf{w}_2$}] () {};
\draw (f3) --++ (300:1) node[draw, circle, fill=white, label=below right:{$\mathbf{w}_3$}] () {};

  \node[mycirc, label=above right:{$\mathbf{v}_1\teh$}] (f1) at (10, -8) {};
  \draw (f1) --++ (330:2) node[mycirc, label=above right:{$\mathbf{u}_2\teh$}] (f2) {} --++ (330:1);
  \draw (f1) --++ (90:2) node[mycirc, label=right:{$\mathbf{u}_1\teh$}] (f11) {} --++ (90:1);
  \draw (f1) --++ (210:2) node[mycirc, label=above left:{$\mathbf{u}_3\teh$}] (f3) {} --++ (210:1);
  \draw (f2) --++ (240:1) node[draw, circle, fill=white, label=below right:{$\mathbf{w}_2$}] () {};
\draw (f3) --++ (300:1) node[draw, circle, fill=white, label=below right:{$\mathbf{w}_3$}] () {};

    \draw[line width=3pt,
             arrows = {-Latex[length=0pt 4]}] (3, 0) -- (7, 0);
    \draw[line width=3pt,
             arrows = {-Latex[length=0pt 4]}] (10, -3) -- (10, -4);
    \draw[line width=3pt,
             arrows = {-Latex[length=0pt 4]}] (7, -2) -- (3, -4);
  
\end{tikzpicture}
\caption{Illustration of the algorithm in \cref{ex:b2-a2-a2-rev}}
\label{fig:unique-propagation}
\end{figure}
\end{example}

\begin{theorem}\label{thm:ternary-set-ebd-algorithm}
Let $\mathcal{F}$ be a set of nondegenerate symmetric ternary signatures with nonzero annihilators,
such that every pair of signatures in $\mathcal{F}$ is a compatible pair.
Then, $\holts(\mathcal{F})$ is tractable.
\end{theorem}
We will eventually show in \cref{cor:ebd-set-ternary-dichotomy} that for such $\mathcal{F}$ (i.e., signatures with nonzero annihilators), failing to satisfy the conditions in \cref{thm:ternary-set-ebd-algorithm} results in \sphness of $\holts(\mathcal{F})$.
\begin{proof}
We give an explicit algorithm for solving an instance of $\holts(\mathcal{F})$.
Let $\Omega$ be a $\holts(\mathcal{F})$ grid.

Suppose $v_1, v_2$ are two vertices connected with an edge $e$ in $\Omega$ with $\mathbf{F}_i$ assigned to $v_i$.
 Let $\ann(\mathbf{F}_i) = \spn(\mathbf{n}_i) \ne 0$.
 If $\langle \mathbf{n}_1, \mathbf{n}_2 \rangle = 0$, then by \cref{prop:orthogonal-normal-vector-edge-factors}, we may remove $e$ and connect some $\mathbf{w}_i \in \mathbb{C}^3$ to $v_i$.
 By repeating this process between every such pair of vertices, we may assume that any edge is between two signatures with non-orthogonal annihilators.

 We now focus on a connected component $C$ having vertices $v_1, \ldots, v_n$ with $v_i$ assigned $\mathbf{F}_i$ with $n \ge 2$. 
Because $v_1$ is connected to some $v_i$, $(\mathbf{F}_1, \mathbf{F}_i)$ must be a compatible pair with non-orthogonal annihilators, so 
 $\mathbf{F}_1$ is of the form $\mathbf{u}_1\teh + \mathbf{v}_1\teh$ or $\lambda_1 \mathbf{u}_1\teh + \sym(\mathbf{u}_1, \mathbf{u}_1, \mathbf{v}_1)$.
Suppose $\mathbf{F}_1 = \mathbf{u}_1\teh + \mathbf{v}_1\teh$. 
 By \cref{prop:symmetric-rank-sym-sum}, because of the difference in the symmetric rank, $\mathbf{F}_1 \ne \lambda \mathbf{x}\teh + \sym(\mathbf{x}, \mathbf{x}, \mathbf{y})$ for any $\lambda \in \mathbb{C}$ and $\mathbf{x},\mathbf{y} \in \mathbb{C}^3$.
Then, repeatedly applying \cref{def:compatible-pairs} gives that all $\mathbf{F}_i$ are of the form $\mathbf{u}_i\teh + \mathbf{v}_i\teh$.
Similarly, if $\mathbf{F}_1 = \lambda_1 \mathbf{u}_1\teh + \sym(\mathbf{u}_1, \mathbf{u}_1, \mathbf{v}_1)$, then all $\mathbf{F}_i$ are of the form $\lambda_i \mathbf{u}_i\teh + \sym(\mathbf{u}_i, \mathbf{u}_i, \mathbf{v}_i)$.

 In the first case, we have a unique propagation. 
 We note that the representation of $\mathbf{F}_i = \mathbf{u}_i\teh + \mathbf{v}_i\teh$ is unique up to reordering and scaling of the terms by \cref{prop:uniqueness-sum-lin-ind-tensor}. 
 Therefore, 
 \cref{itm:compatible-pair-unique-prop} of \cref{def:compatible-pairs}
 imply that once we fix a representation of each $\mathbf{F}_i$ as $\mathbf{u}_i\teh + \mathbf{v}_i\teh$, then we have $\langle \mathbf{u}_i, \mathbf{u}_j \rangle = \langle \mathbf{v}_i, \mathbf{v}_j \rangle = 0$ or $\langle \mathbf{u}_i, \mathbf{v}_j \rangle = \langle \mathbf{v}_i, \mathbf{u}_j \rangle = 0$. 
 Suppose we choose $\mathbf{u}_i\teh$ for $\mathbf{F}_i$.
 Then, if $\mathbf{F}_j$ is connected to $\mathbf{F}_i$, the only nonzero contribution can come from choosing $\mathbf{u}_j\teh$  or $\mathbf{v}_j\teh$, depending on whether $\langle \mathbf{u}_i, \mathbf{v}_j \rangle = 0$ or $\langle \mathbf{u}_i, \mathbf{u}_j \rangle = 0$, respectively.
 We can apply this argument repeatedly, so there is only one way to result in possibly a nonzero value if we start with $\mathbf{u}_i\teh$ for $\mathbf{F}_i$.
 Choosing $\mathbf{v}_i\teh$ also results in a unique propagation, so we can compute the Holant value of $C$ in polynomial time.

 In the second case, we get an algorithm by a counting and vanishing argument.
  By \cref{prop:uniqueness-sym-form}, the representation $\mathbf{F}_i = \lambda_i \mathbf{u}_i^{\otimes n} + \sym(\mathbf{u}_i, \ldots, \mathbf{u}_i, \mathbf{v}_i)$ is unique up to scaling of $\mathbf{u}_i$. 
  Therefore, 
  \cref{itm:compatible-pair-vanish} of \cref{def:compatible-pairs}
  still hold even after fixing a representation of each $\mathbf{F}_i$.
 Note that connecting a unary $\mathbf{w}$ to $\mathbf{F}_i$ results in a binary signature $x_i \mathbf{u}_i \tew + y_i \sym(\mathbf{u}_i, \mathbf{v}_i)$ for some $x_i, y_i \in \mathbb{C}$.
 By recursively absorbing  unary signatures into  ternary (or binary) signatures, we may assume that $C$ consists of only ternary and binary signatures.
 If there is no ternary signature in $C$, then we can compute the Holant value by matrix multiplication.
 On the other hand, suppose  $C$  has some ternary signatures.
 Consider any assignment of choosing one of the terms in the summation of $\mathbf{F}_i$, i.e., one of $\lambda_i \mathbf{u}_i\teh$, $\mathbf{v}_i \otimes \mathbf{u}_i \otimes \mathbf{u}_i$, $\mathbf{u}_i \otimes \mathbf{v}_i \otimes \mathbf{u}_i$ or $\mathbf{u}_i \otimes \mathbf{u}_i \otimes \mathbf{v}_i$ for degree $3$ vertices or one of $x_i \mathbf{u}_i \otimes \mathbf{u}_i$, $y_i \mathbf{v}_i \otimes \mathbf{u}_i$ or $y_i \mathbf{u}_i \otimes \mathbf{v}_i$.
Counting the total number of $\mathbf{u}_i$ and $\mathbf{v}_i$ in this assignment,
we always have more $\mathbf{u}_i$ than $\mathbf{v}_i$, since there is a degree $3$ vertex.
This means that there must be an edge such that some $\mathbf{u}_i$ and $\mathbf{u}_j$ meet, which results in a $0$ evaluation. 
Therefore, the Holant value must be $0$.

Since we can compute the Holant value for each connected component in polynomial time, and the Holant value of $\Omega$ is the product of the values from the connnected components, $\holts(\mathcal{F})$ is tractable.
\end{proof}

The algorithmic idea of \cref{thm:ternary-set-ebd-algorithm} can be generalized to include binary signatures. 
The intuition is that we just need the binary signatures to preserve the orthogonality conditions in \cref{def:compatible-pairs}, since in a signature grid, binary signatures can only exist as a path between other signatures of higher arity. 
In the remainder of this section, we prove the following theorem. 
Note that we do not make the assumption that $\mathcal{B}$ consists of symmetric signatures. 
\begin{theorem}\label{thm:algorithm-ebd-with-binary}
Let $\mathcal{F} = \{ \mathbf{F}_1, \ldots, \mathbf{F}_n \}$ be a set of nondegenerate symmetric signatures of arity $3$ or higher, such that
$\ann(\mathbf{F}_i) = \spn(\mathbf{n}_i) \ne 0$ for $1 \le i \le n$.
Let $\mathcal{B}$ be a set of nondegenerate binary signatures. 
Then, $\holts(\mathcal{F} \cup \mathcal{B})$ is tractable if the following two conditions are satisfied.
\begin{enumerate}
    \item
    \label{enumitem:pairwise-tractable-condition}
    For all $\mathbf{F}_i, \mathbf{F}_j \in \mathcal{F}$, 
    $(\mathbf{F}_i, \mathbf{F}_j)$ is a compatible pair.
    \item 
    Let $\langle \mathcal{B} \rangle$ be the multiplicative closure of $\mathcal{B}$.
    For all $\mathbf{M} \in \langle \mathcal{B} \rangle$ and $\mathbf{F}_i, \mathbf{F}_j \in \mathcal{F}$, 
    $\mathbf{M} \mathbf{F}_i$ is degenerate or 
     $(\mathbf{M} \mathbf{F}_i, \mathbf{F}_j)$ is a compatible pair.
\end{enumerate}
\end{theorem}
We will eventually show in \cref{thm:dichotomy-ternary-binary} that for such $\mathcal{F}$ (i.e., signatures with nonzero annihilators) and $\mathcal{B}$, failing to satisfy the conditions in \cref{thm:algorithm-ebd-with-binary} results in \sphness of $\holts(\mathcal{F} \cup \mathcal{B})$.

To state the algorithm for $\holts(\mathcal{F} \cup \mathcal{B})$, it would be helpful to dissect the above conditions in more details and prove a generalization of \cref{prop:orthogonal-normal-vector-edge-factors}.

\subsubsection{Edge Factoring in General}
We first generalize \cref{prop:orthogonal-normal-vector-edge-factors} to asymmetric signatures in any domain size $q$.
\begin{lemma}
\label{lem:edge-factors-general-case}
    Let $\mathbf{F}, \mathbf{G}$ be signatures on domain $q$.
    Suppose $\mathbf{F} = \sum_{i = 1}^n \mathbf{v}_i \otimes \mathbf{F}_i$ and $\mathbf{G} = \sum_{j = 1}^m \mathbf{w}_j \otimes \mathbf{G}_j$ for $\mathbf{v}_i, \mathbf{w}_j \in \mathbb{C}^q$ and some signatures $\mathbf{F}_i, \mathbf{G}_j$ on domain $q$. 
    Let $V = \spn(\mathbf{v}_1, \ldots, \mathbf{v}_n)$ and $W = \spn(\mathbf{w}_1, \ldots, \mathbf{w}_n)$.
    Suppose there exist $\mathbf{u}_V, \mathbf{u}_W \in \mathbb{C}^q$ such that 
    $\langle \mathbf{v}, \mathbf{w} \rangle = \langle \mathbf{v}, \mathbf{u}_V \rangle \langle \mathbf{u}_W, \mathbf{w} \rangle $ for all $\mathbf{v} \in V, \mathbf{w} \in W$.
    Then, in any grid with an edge between $\mathbf{F}$ and $\mathbf{G}$ connecting 
    the respective first variables, i.e., the $\mathbf{v}_i$ and $\mathbf{w}_j$ parts, we may replace the edge with $\mathbf{u}_V \otimes \mathbf{u}_W$ without changing the Holant value.
\end{lemma}
\begin{proof}
    Suppose $\Omega$ is such a grid, and let $e = (v, w)$ be the edge connecting $v$ with $\mathbf{F}$ and $w$ with $\mathbf{G}$.
    We stratify $\hol_\Omega$ by the choice of a tensor term at $v$ and $w$.
    Then,
    $\hol_\Omega = \sum_{i, j}  c_{ij} \langle \mathbf{v}_i, \mathbf{w}_j \rangle$ where $c_{ij}$ is the sum over all choices for other vertices when the $i$th and $j$th terms are chosen for $v$ and $w$ to be $\mathbf{v}_i$ and $ \mathbf{w}_j$ respectively.
    Then, we have 
    $\hol_\Omega = \sum_{i, j}  c_{ij} \langle \mathbf{v}_i, \mathbf{w}_j \rangle = \sum_{i, j}  c_{ij} \langle \mathbf{v}_i, \mathbf{u}_V \rangle \langle \mathbf{u}_W, \mathbf{w}_j \rangle = \hol_{\Omega'}$
    where $\Omega'$ is the grid obtained by replacing $e$ with $\mathbf{u}_V \otimes \mathbf{u}_W$.
\end{proof}

Note that if $V$ or $W$ is dimension $1$, such $\mathbf{u}_V, \mathbf{u}_W$ always exists.
\begin{proposition}
    Let $V$ and $ W$ be linear subspaces of $\mathbb{C}^q$.
    Suppose $\dim V \le 1$ or $\dim W \le 1$.
    Then, there exist $ \mathbf{u}_V, \mathbf{u}_W \in \mathbb{C}^q$ such that 
    $\langle \mathbf{v}, \mathbf{w} \rangle = \langle \mathbf{v}, \mathbf{u}_V \rangle \langle \mathbf{u}_W, \mathbf{w} \rangle $ for all $\mathbf{v} \in V, \mathbf{w} \in W$.
\end{proposition}
\begin{proof}
By symmetry, we may assume  $\dim W \le 1$.
    If $\dim W = 0$, then we can choose $\mathbf{u}_V, \mathbf{u}_W = 0$.

    If $\dim W = 1$, then $W = \spn(\mathbf{w})$ for some nonzero $\mathbf{w} \in \mathbb{C}^q$.
    We let  $\mathbf{u}_V = \mathbf{w}$. Choose some $\mathbf{u}_W$ such that $\langle \mathbf{u}_W,  \mathbf{w} \rangle = 1$, which clearly exists. Then $\langle \mathbf{v}, c \mathbf{w} \rangle = c \langle \mathbf{v},\mathbf{w} \rangle =  \langle \mathbf{v}, \mathbf{u}_V \rangle \langle \mathbf{u}_W, c \mathbf{w} \rangle$.
\end{proof}

More generally, we prove that the existence of such $\mathbf{u}_V, \mathbf{u}_W$ depends on the interaction of the spaces $V$ and $W$ under the bilinear map $V \times W \to \mathbb{C}$ by $(\mathbf{v}, \mathbf{w}) \mapsto \mathbf{v}\transpose \mathbf{w}$.
\begin{proposition}\label{prop:pinning-vectors-iff-rank-1}
    Let $V = \spn(\mathbf{v}_1, \ldots, \mathbf{v}_n)$ and $W = \spn(\mathbf{w}_1, \ldots, \mathbf{w}_m)$ be subspaces of $\mathbb{C}^q$.
    Let $M_V$ be the matrix with columns $\mathbf{v}_i$ and $M_W$ be the matrix with columns $\mathbf{w}_i$.
    The following are equivalent:
    \begin{enumerate}
        \item There exist $\mathbf{u}_V, \mathbf{u}_W \in \mathbb{C}^q$ such that 
    $\langle \mathbf{v}, \mathbf{w} \rangle = \langle \mathbf{v}, \mathbf{u}_V \rangle \langle \mathbf{u}_W, \mathbf{w} \rangle $ for all $\mathbf{v} \in V, \mathbf{w} \in W$.
    \item $M_V\transpose M_W$ has rank $\le 1$.
    \item $M_V^T \restriction_W$ has rank $\le 1$.
    \item $M_W^T \restriction_V$ has rank $\le 1$.
    \end{enumerate}
\end{proposition}

\begin{proof}
3 and 4 are clearly equivalent to 2.
\begin{itemize}
    \item $1 \to 2$:
Suppose such $\mathbf{u}_V, \mathbf{u}_W$ exist.
Then,
\begin{align*}
M_V\transpose M_W
&= \begin{bmatrix}
    \langle \mathbf{v}_1, \mathbf{w}_1 \rangle & \cdots & \langle \mathbf{v}_1, \mathbf{w}_m \rangle \\
    \vdots & \ddots & \vdots \\
    \langle \mathbf{v}_n, \mathbf{w}_1 \rangle & \cdots & \langle \mathbf{v}_n, \mathbf{w}_m \rangle
\end{bmatrix} \\
&=
\begin{bmatrix}
    \langle \mathbf{v}_1, \mathbf{u}_V \rangle \langle \mathbf{u}_W, \mathbf{w}_1 \rangle & \cdots & \langle \mathbf{v}_1, \mathbf{u}_V \rangle \langle \mathbf{u}_W, \mathbf{w}_m \rangle\\
    \vdots & \ddots & \vdots \\
    \langle \mathbf{v}_n, \mathbf{u}_V \rangle \langle \mathbf{u}_W, \mathbf{w}_1 \rangle & \cdots & \langle \mathbf{v}_n, \mathbf{u}_V \rangle \langle \mathbf{u}_W, \mathbf{w}_m \rangle
\end{bmatrix} \\
&= 
M_V\transpose
\mathbf{u}_V \mathbf{u}_W\transpose
M_W
\end{align*}
Therefore, the product has rank $\le 1$.

\item $2 \to 1$:
Suppose $M_V\transpose M_W$ has rank $\le 1$.
Then, we may write $M_V\transpose M_W = \mathbf{u}_1 \mathbf{u}_2 \transpose$ for some $\mathbf{u}_1, \mathbf{u}_2$.
Further, we may assume that $\mathbf{v}_1, \ldots, \mathbf{v}_n$ is a basis of $V$ and $\mathbf{w}_1, \ldots, \mathbf{w}_m$ is a basis of $W$.
Then, there exists invertible $A_V, A_W \in \gl(q, \mathbb{C})$ obtained by completing the basis.
Note that $A_V\transpose A_W = \begin{bmatrix}
    M_V\transpose M_W & * \\
    * & *
\end{bmatrix} = \begin{bmatrix}
    \mathbf{u}_1 \mathbf{u}_2\transpose & * \\
    * & *
\end{bmatrix}$.
Let $\mathbf{u}_1', \mathbf{u}_2'$ be obtained from $\mathbf{u}_1, \mathbf{u}_2$ by appending enough $0$ so that $\mathbf{u}_1', \mathbf{u}_2' \in \mathbb{C}^q$.
Let $\mathbf{v} = a_1 \mathbf{v}_1 + \cdots a_n \mathbf{v}_n \in V$ and $\mathbf{w} = b_1 \mathbf{w}_1 + \cdots + b_m \mathbf{w}_m \in W$.
Write $\mathbf{a} = [a_1, \ldots, a_n]\transpose \in \mathbb{C}^n$ and $\mathbf{b} = [b_1, \ldots, b_m]\transpose \in \mathbb{C}^m$,
and $\mathbf{a}', \mathbf{b}'$ obtained by appending enough $0$ to $\mathbf{a}, \mathbf{b}$ so that $\mathbf{a}', \mathbf{b}' \in \mathbb{C}^q$.
we have

\begin{align*}
    \langle \mathbf{v}, \mathbf{w} \rangle &= 
        (\mathbf{a}')\transpose
        A_V\transpose A_W
        \mathbf{b}' \\
        &= \mathbf{a}\transpose
        M_V \transpose M_W
        \mathbf{b} \\
        &= \mathbf{a}\transpose
        \mathbf{u}_1 \mathbf{u}_2 \transpose
        \mathbf{b} \\
        &= (\mathbf{a}')\transpose
        \mathbf{u}_1' (\mathbf{u}_2') \transpose
        \mathbf{b}' \\
        &= (\mathbf{a}')\transpose 
        A_V\transpose (A_V\transpose)^{-1}
        \mathbf{u}_1' (\mathbf{u}_2') \transpose
        (A_W)^{-1} A_W
        \mathbf{b}' \\
        &= \mathbf{v}\transpose ((A_V\transpose)^{-1}
        \mathbf{u}_1') ((A_W\transpose)^{-1}
        \mathbf{u}_2')\transpose 
        \mathbf{w} \\
        &= \langle \mathbf{v}, (A_V\transpose)^{-1}
        \mathbf{u}_1' \rangle \; \langle (A_W\transpose)^{-1}
        \mathbf{u}_2', \mathbf{w} \rangle
\end{align*}
Therefore, we may take $\mathbf{u}_V = (A_V\transpose)^{-1}
        \mathbf{u}_1'$ and $\mathbf{u}_W = (A_W\transpose)^{-1}
        \mathbf{u}_2'$. \qedhere
\end{itemize}
\end{proof}

For a subspace $V \subset \mathbb{C}^q$,
let $\ann(V) = \{\mathbf{u} \in \mathbb{C}^q : \forall \mathbf{v} \in V \; \langle \mathbf{u}, \mathbf{v} \rangle = 0\}$.
If $V$ is a two dimensional subspace of $\mathbb{C}^3$,
then $\ann(V) = \spn(\mathbf{n}_V) \ne 0$.
For two dimensional subspaces $V, W \subseteq \mathbb{C}^3$, \cref{prop:pinning-vectors-iff-rank-1} is equivalent to 
$\mathbf{n}_V$ and $\mathbf{n}_W$ being orthogonal.
\begin{proposition}\label{prop:general-edge-factor-condition-domain-3}
    Let $V, W$ be two dimensional subspaces of $\mathbb{C}^3$.
    Let $\ann(V) = \spn(\mathbf{n}_V) \ne 0$ and $\ann(W) = \spn(\mathbf{n}_W) \ne 0$.
    The following are equivalent:
    \begin{enumerate}
        \item There exists $\mathbf{u}_V, \mathbf{u}_W \in \mathbb{C}^3$ such that $\langle \mathbf{v}, \mathbf{w} \rangle = \langle \mathbf{v}, \mathbf{u}_V \rangle \langle \mathbf{u}_W, \mathbf{w} \rangle $ for all $\mathbf{v}\in V, \mathbf{w} \in W$.
        \item $\langle \mathbf{n}_V, \mathbf{n}_W \rangle = 0$.
    \end{enumerate}
\end{proposition}
\begin{proof}
For a subspace $U$, we denote $U^\perp = \{\mathbf{v} \in \mathbb{C}^3 : \forall \mathbf{u} \in U \; \langle \mathbf{u}, \mathbf{v} \rangle = 0\}$.
Note that by the dimension, $V^\perp = \spn(\mathbf{n}_V)$ and $(V^\perp)^\perp = V$.

By \cref{prop:pinning-vectors-iff-rank-1}, 1.\ is equivalent to $M_V\transpose M_W$ having rank $\le 1$.
Note that $\ker(M_V\transpose) = V^\perp = \spn(\mathbf{n}_V)$.
Then, $M_V\transpose\restriction_W$ can be rank $\le 1$ if and only if $\mathbf{n}_V \in W = \spn(\mathbf{n}_W)^\perp$ if and only if $\langle \mathbf{n}_V, \mathbf{n}_W \rangle = 0$. 
\end{proof}

\subsubsection{Restatement of the Tractability Criteria and an Algorithm}
\label{subssec:ebd-tractability-criteria}
Now, we can state the conditions of \cref{thm:algorithm-ebd-with-binary} in a different way. 
Let $\mathcal{F}$ be a set of nondegenerate symmetric signatures of arity $3$ or higher such that every signature has a nonzero annihilator and 
let $\mathcal{B}$ be a set of nondegenerate binary signatures. 
To satisfy the conditions in \cref{thm:algorithm-ebd-with-binary}, $\mathcal{F}$ must be $\mathcal{F}_u \cup \mathcal{F}_v \cup \mathcal{G}$ where
\begin{enumerate}
    \item $\mathcal{F}_u = \{ \mathbf{U}_i = \mathbf{u}_i^{\otimes r_i} + \mathbf{v}_i^{\otimes r_i} : 1 \le i \le n \}$
    \item $\mathcal{F}_v = \{ \mathbf{V}_j = \lambda_i \mathbf{x}_i^{\otimes s_i} + \sym(\mathbf{x}_i, \ldots, \mathbf{x}_i, \mathbf{y}_i) : 1 \le i \le m \}$
    \item $\mathcal{G}$ is a set of \casecz signatures \textit{not} of the form $\mathbf{u}^{\otimes r} + \mathbf{v}^{\otimes r}$ or $\lambda \mathbf{u}^{\otimes r} + \sym(\mathbf{u}, \ldots, \mathbf{u}, \mathbf{v})$ for any $\mathbf{u}, \mathbf{v} \in \mathbb{C}^3$ and $r \ge 3$, all annihilated by a single nonzero isotropic vector $\boldsymbol{\beta}$.
\end{enumerate}
for some 
$(\mathbf{u}_i, \mathbf{v}_i ), (\mathbf{x}_j, \mathbf{y}_j) \in \mathbb{C}^3 \times \mathbb{C}^3$ for $1 \le i \le n$ and $1 \le j \le m$,  such that $\mathbf{u}_i \not \sim \mathbf{v}_i$ and $\mathbf{x}_j \not \sim \mathbf{y}_j$; $\lambda_i \in \mathbb{C}$; $r_i, s_i \ge 3$.
Note that $\mathcal{F}_u, \mathcal{F}_v, \mathcal{G}$ are disjoint.

 We note that the representation of $\mathbf{F}_i =\mathbf{u}_i^{\otimes n} +  \mathbf{v}_i^{\otimes n}$ is unique up to reordering and scaling of the terms by \cref{prop:uniqueness-sum-lin-ind-tensor}. 
 Therefore, once we fix a representation of each $\mathbf{F}_i$ as $\mathbf{u}_i^{\otimes n}  + \mathbf{v}_i^{\otimes n} $, 
 \cref{itm:compatible-pair-unique-prop} of \cref{def:compatible-pairs}
 imply that $\langle \mathbf{u}_i, \mathbf{u}_j \rangle  = \langle \mathbf{v}_i, \mathbf{v}_j \rangle = 0$ or $\langle \mathbf{u}_i, \mathbf{v}_j \rangle = \langle \mathbf{v}_i, \mathbf{u}_j \rangle = 0$ if $\langle \mathbf{n}_i, \mathbf{n}_j \rangle \ne 0$.
 Also, by \cref{prop:uniqueness-sym-form}, the representation $\mathbf{F}_i = \lambda_i \mathbf{x}_i^{\otimes n} + \sym(\mathbf{x}_i, \ldots, \mathbf{x}_i, \mathbf{y}_i)$ is unique up to scaling of $\mathbf{x}_i$.
 Therefore, 
after fixing a representation of each $\mathbf{F}_i$ as $\lambda_i \mathbf{x}_i^{\otimes n} + \sym(\mathbf{x}_i, \ldots, \mathbf{x}_i, \mathbf{y}_i)$,
 \cref{itm:compatible-pair-vanish} of \cref{def:compatible-pairs}
 imply that 
$\langle \mathbf{x}_i, \mathbf{x}_j \rangle = 0$ 
 if
 $\langle \mathbf{n}_i, \mathbf{n}_j \rangle \ne 0$.
 And $\mathbf{u}_i^{\otimes n} +  \mathbf{v}_i^{\otimes n} \ne \lambda_j \mathbf{x}_j^{\otimes n} + \sym(\mathbf{x}_j, \ldots, \mathbf{x}_j, \mathbf{y}_j)$ by the symmetric rank (\cref{prop:symmetric-rank-sym-sum}).
 Let $A_i$ be the matrix with columns being $\mathbf{u}_i, \mathbf{v}_i$
and $B_i$ be the matrix with columns being $\mathbf{x}_i, \mathbf{y}_i$. 
Let $\mathcal{M} = \langle I, \mathcal{B} \rangle$, the monoid generated by $\mathcal{B}$ and the identity matrix.
The conditions in \cref{thm:algorithm-ebd-with-binary} are equivalent to the following for all $\mathbf{M} \in  \mathcal{M}$, $1 \le i \le n$ and $1 \le j \le m$:
\begin{enumerate}
    \item $A_i\transpose \mathbf{M} B_i $ has rank $\le 1$.
    \item $A_i\transpose \mathbf{M} A_j$ has form $\begin{bsmallmatrix}
    * & 0 \\ 0 & *
\end{bsmallmatrix}$, or $\begin{bsmallmatrix} 
    0 & * \\ * & 0
\end{bsmallmatrix}$, or has rank $\le 1$.
    \item $B_i\transpose \mathbf{M} B_j$ has form  $\begin{bsmallmatrix}
    0 & * \\ * & *
\end{bsmallmatrix}$, or has rank $\le 1$.
\item 
If $\mathcal{G}$ is not empty, then 
$(\mathbf{M} A_i)\transpose \restriction_{\spn(\boldsymbol{\beta})^\perp}$ and
$(\mathbf{M} B_j)\transpose\restriction_{\spn(\boldsymbol{\beta})^\perp}$ have rank $\le 1$.
\end{enumerate}

Now we can state the algorithm for $\holts(\mathcal{F} \cup \mathcal{B})$.
The algorithm is a generalization of \cref{thm:ternary-set-ebd-algorithm} using \cref{lem:edge-factors-general-case} to handle the binary signatures. 

Let $\Omega$ be a signature grid for $\holts(\mathcal{F} \cup \mathcal{B})$.
We first reduce every chain of binary signatures by matrix multiplication, so that we have an equivalent
$\holts(\mathcal{M} \mid \mathcal{F})$ grid $\Omega'$.
Let $\mathbf{F}, \mathbf{G}$ be two signatures on the right side connected by $\mathbf{M} \in \mathcal{M}$ on the left. 
By the above conditions, if $\mathbf{F}$ and $\mathbf{G}$ are from distinct $\mathcal{F}_u, \mathcal{F}_v, \mathcal{G}$, then we can replace the $\mathbf{M}$ edge with some $\mathbf{u} \otimes \mathbf{v}$ by \cref{lem:edge-factors-general-case}.
We apply \cref{lem:edge-factors-general-case} for every case when the rank $\le 1$ condition holds.
Then, the remaining $\mathbf{M}$ edge between two signatures on the right side must be between
$\mathbf{U}_i, \mathbf{U}_j$ or $\mathbf{V}_i, \mathbf{V}_j$, meaning any connected component must use only $\mathcal{F}_u$ on the right side or only $\mathcal{F}_v$ on the right side, or is just a single $\mathcal{G}$ vertex on the right side.

Consider a connected component that contains $\mathbf{U}_i$ and $\mathbf{U}_j$ being connected via some $\mathbf{M} \in \mathcal{M}$.
Suppose we choose $\mathbf{u}_i^{\otimes r_i}$ for $\mathbf{U}_i$.
Then, for $\mathbf{U}_j$, only one of $\mathbf{u}_j^{\otimes r_j}$ or $\mathbf{v}_j^{\otimes r_j}$ can be chosen to result in a nonzero value 
because $\mathbf{v}_j\transpose \mathbf{M} \mathbf{u}_i = 0$ or $\mathbf{u}_j\transpose \mathbf{M} \mathbf{u}_i = 0$.
(These correspond to the 0 entries in $A_i\transpose \mathbf{M} A_j$.)
Therefore, the unique propagation argument in \cref{thm:ternary-set-ebd-algorithm} still applies in this connected component.

Consider a connected component $C$ that uses $\mathcal{F}_v$ signatures on the right side.
We may absorb all unary signatures so that $C$ consists of binary or higher signatures on the right side.
If all signatures in $C$ are binary, then we may compute its Holant value by matrix multiplication.
Consider any assignment of choosing one of the terms in the summation of $\mathbf{V}_i$, i.e.\ 
one of $\lambda_i \mathbf{x}_i^{\otimes k}$ or some permutation of $\mathbf{x}_i \otimes \cdots \otimes \mathbf{x}_i \otimes \mathbf{y}_i$. 
If we count the total number of $\mathbf{x}_i$ and $\mathbf{y}_i$ in this assignment,
then we see that we always have more $\mathbf{x}_i$ than $\mathbf{y}_i$ since there is at least one degree $\ge 3$ vertex. 
Then, somewhere, we must have a $\mathbf{x}_i$ meeting $\mathbf{x}_j$ via an $\mathbf{M} \in \mathcal{M}$ edge.
However, $\mathbf{x}_j\transpose \mathbf{M} \mathbf{x}_i = 0$, so the assignment results in $0$.
(These correspond to the 0 entry in $B_i\transpose \mathbf{M} B_j$.)
Since this is true for any choice, the Holant value must be $0$. 

\subsubsection{Decidability of the Tractability Criteria}
It is not \emph{a priori} clear
whether the tractability criteria given in \cref{subssec:ebd-tractability-criteria} is decidable, since it is an implicit condition on the multiplicative closure on a set of matrices. 
However, we can see that the signatures of arity $3$ or higher determine a finite number of polynomial conditions on $\mathcal{M}$. 
More precisely, every pair of such signatures impose some explicit polynomial conditions on $\mathcal{M}$. 
Then, $\mathcal{M}$ satisfying the tractability condition is equivalent to the Zariski closure of $\mathcal{M}$ being contained in the algebraic set that corresponds to the polynomial conditions. 
By~\cite{zariski-closure-finite-semigroup}, since the Zariski closure of a finitely generated semigroup of matrices is computable, this is decidable. 
It is an open problem whether  this criterion  can be decided in polynomial time (in the size
of  $\mathcal{F}$). Are there  methods for this decision problem other than appealing to the Zariski closure?

\subsection{Summary}
We collect the tractable cases.
\begin{theorem}\label{thm:tractable-cases}
Let $\mathcal{F}$ and $\mathcal{B}$  respectively be a set of nondegenerate, complex-valued, symmetric signatures over domain $\{\db,\dg,\dr\}$ of arity $\ge 3$ and arity 2.
Then, $\holts(\mathcal{F} \cup \mathcal{B})$ is tractable if there exists some $T \in \ortho(3, \mathbb{C})$ such that one of the following conditions holds.
\begin{enumerate}
    \item $\mathcal{F} = \emptyset$. 
    \item \label{itm:tractable-e}
    $T \mathcal{F} \cup T \mathcal{B} \subset \mathcal{E}$ 
    or 
    $\Yinv T \mathcal{F} \cup \Yinv T \mathcal{B} \subset \mathcal{E}$.
    \item \label{itm:tractable-bg-r}
    \begin{enumerate}
        \item $T \mathcal{F} \subset \sepbgmath$, and
        \item $T \mathcal{B} \subset \sepbgmath \cup \swbgmath$, and
        \item $\holbs((T \mathcal{F} \cup T \mathcal{B})\domres{\db\dg})$ is tractable.
    \end{enumerate}
    \item \label{itm:tractable-c}
      $\Yinv T \mathcal{F} \cup \Yinv T \mathcal{B} \subset \casecsupp$.
    \item \label{itm:tractable-ann}
    For all $\mathbf{F} \in \mathcal{F}$, $\ann(\mathbf{F}) \ne 0$, and 
    $\mathcal{F}$ and $\mathcal{B}$ satisfy the conditions of \cref{thm:algorithm-ebd-with-binary}.
\end{enumerate}
\end{theorem}

An algorithm for each case is given in \cref{subsec:tractable-case-binary}, \cref{subsec:tractable-e}, \cref{subsec:tractable-bg-r}, \cref{subsec:tractable-case-c}, and \cref{subsec:tractable-case-ebd} respectively.
\cref{thm:full-dichotomy} will prove 
that $\holts(\mathcal{F} \cup \mathcal{B})$ is \sph if none of the conditions is met.

\section{Frames and Shells}\label{sec:two-pieces}
Now we move to proving  \sphness. 
We want to 
show a $\holts(\mathcal{F})$ dichotomy for a set of complex-valued symmetric signatures $\mathcal{F}$.
As mentioned in \cref{subsec:proof-outline}, we decide to start with a set of ternary signatures.
We may assume that $\holts(\mathbf{F})$ is not \sph for every $\mathbf{F} \in \mathcal{F}$, therefore by the $\holts(\mathbf{F})$ dichotomy, it is  polynomial-time
computable and has one of the tractable types given in \cref{def:types-ternary}.
A natural starting point from the dichotomy of a single signature to  a set of signatures is  to analyze a pair of signatures. 
So, to analyze $\holts(\mathbf{F}, \mathbf{G})$, we may assume that both $\mathbf{F}$ and $\mathbf{G}$ take one of tractable types given in \cref{def:types-ternary} and try to characterize when we can conclude $\holts(\mathbf{F}, \mathbf{G})$ is \sph or tractable. 

How to approach this problem?
To show \sphness,
one can try to construct a direct reduction from a known \sph problem. 
But since we are dealing with $\hols$ problems and because there is a known $\holbs$ dichotomy, perhaps it is easier to show that $\holts(\mathbf{F}, \mathbf{G}) \ge_T \holbs(\mathcal{F}')$ for some set of Boolean domain signatures $\mathcal{F}'$ for which we know to be \sph. 
To do so, a natural approach is to show that $\mathbf{F}$ or $\mathbf{G}$ can realize \textit{domain restriction}. 
For example, if $\mathbf{F} = \mathbf{e}_1\teh + \mathbf{e}_2\teh + \mathbf{e}_3\teh$, we have $\langle \mathbf{F}, (1, 1, 0) \rangle = (=_{\db\dg})$, which is a binary signature that is an equality on $\db\dg$; $(=_{\db\dg})(i, j) = 1$ if $i = j = \db$ or $i = j = \dg$, and $0$ otherwise. 
Then, 
\[\holts(\mathbf{F}, \mathbf{G}) 
\; \ge_T \; \holts((=_{\db\dg}) \mathbf{F}, (=_{\db\dg}) \mathbf{G}) 
\; =_T \; \holbs(\mathbf{F}\domres{\db\dg}, \mathbf{G}\domres{\db\dg}) \, ,\]
where $(=_{\db\dg}) \mathbf{F}$ is the ternary signature  obtained by connecting 3 copies of
the binary $(=_{\db\dg})$ to $\mathbf{F}$, which effectively restricts the domain to 
$\db\dg$.
Then we can apply the known Boolean domain dichotomy. 
To be not \sph, the Boolean domain dichotomy imposes certain constraints on $\mathbf{G}$, in particular on the values $\mathbf{G}(\db,\db,\db)$, $\mathbf{G}(\db,\db,\dg)$, $\mathbf{G}(\db,\dg,\dg)$ and $\mathbf{G}(\dg,\dg,\dg)$. 
There are additional constraints that $\mathbf{F}$ imposes on $\mathbf{G}$, since $\mathbf{F}$ can realize $(=_{\db\dr})$ and $(=_{\dg\dr})$ as well.
As we will show later, the constraints that $\mathbf{F}$ imposes on $\mathbf{G}$ can be formulated as an algebraic set. 
We will call this algebraic set  a \emph{shell}. 
Then, checking whether $\mathbf{G}$ escapes the known Boolean domain dichotomy criteria can be formulated as whether $\mathbf{G}$ is contained in the shell of $\mathbf{F}$ or not.

However, this is not the easiest condition to check.
As we saw in \cref{subsec:proof-outline}, $\mathbf{G}$ is of the form $T \mathbf{G}'$ for some orthogonal $T$ and $\mathbf{G}'$ of the defining form in \cref{def:types-ternary}. 
The algebraic expression of the entries $\mathbf{G}$ is an unwieldy degree $4$ polynomial in terms of the entries of $T$ and $\mathbf{G}'$, each having some interdependency. 
Also, it could happen that $\mathbf{G}$ is contained in the shell of $\mathbf{F}$, but we still fail to identify an algorithm to solve $\holts(\mathbf{F}, \mathbf{G})$. 
Then, we would need to go back, construct some gadget $\mathbf{H}$, 
which is likely to have a more complicated expression,
then check whether $\mathbf{H}$ is in the shell and so on\ldots until we successfully identify an algorithm or \sphness. 
This approach quickly becomes infeasible.

This is where the idea of \emph{frame} comes in.
In \cref{sec:casec}, we saw that some \casec signatures can realize certain linear subspaces of signatures as a gadget or by polynomial interpolation. 
It turns out that other types of signatures also have such properties. 
Note that we only showed for example that a \casecodz signature $\mathbf{G}'$ of the defining form can realize $V = \spn(\mathbf{e}_1\teh, \sym(\mathbf{e}_1, \mathbf{e}_1, \mathbf{e}_3), \sym(\mathbf{e}_1, \mathbf{e}_1, \mathbf{e}_2) + \lambda \sym(\mathbf{e}_1, \mathbf{e}_3, \mathbf{e}_3))$ for some $\lambda$ on the right side 
when working in the $\holts(\X \mid \Yinv \mathbf{G}')$ setting.
If $\mathbf{G}$ is a \casecodz signature, then we have
$\mathbf{G} = T \mathbf{G}'$ for some orthogonal $T$.
Then, we have
\[\holts(\mathbf{G}) 
\; =_T \; 
\holts(\mathbf{G}')
\; =_T \;
\holts(\X \mid \Yinv \mathbf{G}') 
\; \ge_T \;
\holts(\X \mid \Yinv \mathbf{G}', V) \, .\]
The advantage of expressing the set of signatures as a span of basis vector is that $T \Y V = \spn(\mathbf{v}_1\teh, \sym(\mathbf{v}_1, \mathbf{v}_1, \mathbf{v}_3), \sym(\mathbf{v}_1, \mathbf{v}_1, \mathbf{v}_2) + \lambda \sym(\mathbf{v}_1, \mathbf{v}_3, \mathbf{v}_3))$ where $\mathbf{v}_i$ is the $i$th column of $T \Y$. 
Further, $\mathbf{G}$ can realize $T \Y V$ in the $\holts(\mathbf{G})$ setting:
\begin{align*}
    \holts(\mathbf{G}) &\ge_T \holts(\X \mid \Yinv \mathbf{G}', V) \\
    &=_T \holts(\X \mid \Yinv T^{-1} \mathbf{G}, V) \\
    &=_T \holts(T^{-1} \mathbf{G}, \Y V) \\
    &=_T \holts(\mathbf{G}, T \Y V) \,.
\end{align*}
Using the spin representation, the columns of $T \Y$ have a relatively simple expression.
Combined with the fact that the shells are relatively simple algebraic sets, essentially a combination of lines and twisted cubics in the projective space, rather than directly checking whether $\mathbf{G}$ is \emph{a point} of the shell,
we can instead check whether $T \Y V $ is \emph{a subset} of the shell. 
This is a lot easier since $T \Y V$ itself is an algebraic set, and in particular a linear subspace. 
Further, if we identify a large enough set of signatures that is realizable, then it is conceivable that this single check is sufficient to determine the computational complexity; we do not need to do the additional back and forth of gadget construction and checking shell containment.

We will define a \textit{frame of the type}, denoted $\ag(\mathfrak{G})$ 
for a tractable signature type $\mathfrak{G}$, which will be a union of linear subspaces of signatures. 
Then, similar to the above example, 
if $\mathbf{G}$ is of type $\mathfrak{G}$, then
checking if $\holts(\mathbf{F}, \mathbf{G})$ escapes the known \sph conditions boils down to checking if $T \ag(\mathfrak{G})$ is a subset of the shell of $\mathbf{F}$ for some orthogonal $T$.
This is a purely algebraic analysis. 

We give an informal justification of the name frame.
Literally, it is a set that captures a certain essence of the type, as a frame of a building. 
There is also an informal mathematical analogy. 
For a fixed type $\mathfrak{F}$, we can think of the set $\{(T, T \ag(\mathfrak{F})) : T \in \so(3, \mathbb{C})\}$ as a (trivial) bundle over $\so(3, \mathbb{C})$. 
Then, we are analyzing which point on the bundle is contained in the shell. 
One of the important features of $\ag(\mathfrak{F})$ is that it is written as a span of a easily manipulable basis, which can be said to be very loosely connected to the notion of a frame of a manifold.

With the intuition, we give a more formal definition. 
For each tractable type $\mathfrak{F}$ being one of \caseathree, \caseatwo, \casebthree, \casebtwo, \casecz, \casecobg, \casecosep, \casecodz, and \casecogen,
we define two sets, $\ag(\mathfrak{F})$ and $\agx(\mathfrak{F})$, in the following sections as the frame of the type $\mathfrak{F}$ in the following sections with certain properties. 
We define $\ag(\mathfrak{F})$ as a certain linear subspace or a
union of subspaces of $\symspace{3}{3}$, and show that
some $\mathbf{F}$ of type $\mathfrak{F}$ can realize all signatures
in $\ag(\mathfrak{F})$.
This $\mathbf{F}$ will be chosen to be of the defining form given in \cref{def:types-ternary} since they are easily analyzable. 
Then, we can show that
any signature $\mathbf{H}$ of type $\mathfrak{F}$   can realize in $\holts(\mathbf{H})$ the set $T \ag(\mathfrak{F})$ for some
$T \in \so(3, \mathbb{C})$, where $\mathbf{H} = T \mathbf{F}$.

Similarly, we define $\agx(\mathfrak{F})$,
and show that 
for some $\mathbf{F}$ of type  $\mathfrak{F}$,
$\iny{\mathbf{F}} = \Yinv \mathbf{F}$ can realize 
 all  signatures
in $\agx(\mathfrak{F})$  \textit{in the right side} in 
the bipartite setting $\holts(\X \mid \iny{\mathbf{F}})$. 
Let $\iny{\mathbf{H}} = \Yinv \mathbf{H}$.
This implies that 
any signature $\mathbf{H}$ of type $\mathfrak{F}$  can realize in $\holts(\X \mid \iny{\mathbf{H}})$ the set $\iny{T} \agx(\mathfrak{F})$ for $\iny{T} \in \so(3, \mathbb{C}, \X)$ such that $\iny{\mathbf{H}} = \iny{T} \iny{\mathbf{F}}$. 
Note that then, $\mathbf{H} = \Y \iny{T} \Yinv \mathbf{F} = T \mathbf{F}$ where $T = \Y \iny{T} \Yinv \in \so(3, \mathbb{C})$ by the isomorphism of $\so(3, \mathbb{C}) \cong \so(3, \mathbb{C}, \X)$.
In fact, we will define 
$\agx(\mathfrak{F}) = \Yinv \ag(\mathfrak{F})$. 
We will use whichever setting $\holts(\mathbf{F})$ or $\holts(\X \mid \iny{\mathbf{F}})$ is more convenient in a particular proof. 
In the case of \caseathree and \caseatwo, we analyze in the $\holts(\mathbf{F})$ setting and use $\ag$.
In the other cases, we analyze in 
the $\holts(\X \mid \iny{\mathbf{F}})$ setting, and use $\agx$.
For technical reasons explained in \cref{subsec:gc-case-c1dz}, for type \casecodz, we define $\ag(\casecodzmath, \lambda)$ and $\agx(\casecodzmath, \lambda)$ for $\lambda \ne 0$.

We also note the inclusion relation between $\ag$ sets in \cref{fig:lattice-ag}.
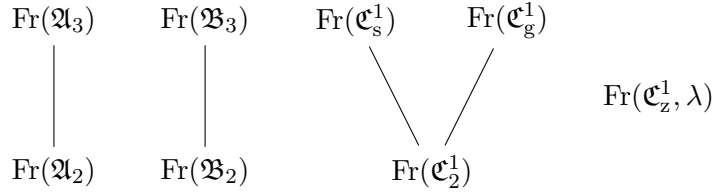
\begin{figure}
    \centering
\begin{tikzpicture}

    \node (a2) at (0, 0) {$\ag(\caseatwomath)$};
    \node (a3) at (0, 2) {$\ag(\caseathreemath)$};
    \node (b2) at (2, 0) {$\ag(\casebtwomath)$};
    \node (b3) at (2, 2) {$\ag(\casebthreemath)$};
    \node (c12) at (5, 0) {$\ag(\casecobgmath)$};
    \node (c1s) at (4, 2) {$\ag(\casecosepmath$)};
    \node (c1g) at (6, 2) {$\ag(\casecogenmath)$};
    \node (c1mv) at (8, 1) {$\ag(\casecodzmath, \lambda)$};

    \draw (a2) -- (a3);
    \draw (b2) -- (b3);
    \draw (c12) -- (c1s);
    \draw (c12) -- (c1g);
\end{tikzpicture}
    \caption{Lattice of frames by inclusion.}
    \label{fig:lattice-ag}
\end{figure}

We formalize the above idea in the following lemma, which can be proved
using the definitions of $\ag$ and $\agx$ for each types that we define in this section. 
\begin{lemma}
\label{lem:gadget-set-general-form}
    Let $\mathfrak{F}$ be a tractable type other than \casecodz. 
    Let $\mathbf{H}$ be a signature of type $\mathfrak{F}$.
    Then, there exists some $T \in \so(3, \mathbb{C})$ such that, 
    for any $\mathcal{F}$ containing $\mathbf{H}$,
    $\holts(\mathcal{F}, T \ag(\mathfrak{F})) \le_T \holts(\mathcal{F})$.
    Similarly, there exists some $\iny{T} \in \so(3, \mathbb{C}, \X)$ such that, 
    for any set $\mathcal{F}$ containing $\mathbf{H}$, 
    $\holts(\X \mid \Yinv \mathcal{F}, \iny{T} \agx(\mathfrak{F})) \le_T
    \holts(\X \mid \Yinv \mathcal{F})$.

    Let $\mathbf{H}$ be a signature of type \casecodz.
    Then, there exists some $T \in \so(3, \mathbb{C})$ and $\lambda \ne 0$ such that, 
    for any $\mathcal{F}$ containing $\mathbf{H}$,
    $\holts(\mathcal{F}, T \ag(\casecodzmath, \lambda)) \le_T \holts(\mathcal{F})$.
    Similarly, there exists some $\iny{T} \in \so(3, \mathbb{C}, \X)$ such that, 
    for any set $\mathcal{F}$ containing $\mathbf{H}$,
    $\holts(\X \mid \Yinv \mathcal{F}, \iny{T} \agx(\casecodzmath, \lambda)) \le_T \holts(\X \mid \Yinv \mathcal{F})$.
\end{lemma}
\begin{proof}
    Suppose $\mathfrak{F}$ is \caseathree.
    Then, $\mathbf{H} = T \mathbf{F}$ where $\mathbf{F} = a \mathbf{e}_1\teh + b \mathbf{e}_2\teh + c \mathbf{e}_3\teh$ for some orthogonal $T$ and $a, b, c \ne 0$.
    We may assume that $T \in \so(3, \mathbb{C})$.
    By \cref{subsec:gc-case-a-3}, $\mathbf{F}$ can realize $\ag(\caseathreemath)$ in $\holts(T^{-1} \mathcal{F})$. Thus, we have
    \begin{align*}
        \holts(\mathcal{F}) &=_T \holts(T^{-1} \mathcal{F}) \\
        &\ge_T \holts(T^{-1} \mathcal{F} ,\ag(\caseathreemath)) \\
        &=_T \holts(\mathcal{F} , T \ag(\caseathreemath)) \, .
    \end{align*}
    By definition, $\agx(\caseathreemath) = \Yinv \ag(\caseathreemath)$, then
    \begin{align*}
        \holts(\mathcal{F}) &\ge_T \holts(\mathcal{F} , T \ag(\caseathreemath)) \\
        &=_T \holts(\X \mid \Yinv \mathcal{F} ,  \Yinv T \ag(\caseathreemath)) \\
        &=_T \holts(\X \mid \Yinv \mathcal{F} ,  \iny{T} \Yinv \ag(\caseathreemath)) \\
        &=_T \holts(\X \mid \Yinv \mathcal{F} ,  \iny{T} \agx(\caseathreemath)) \, ,
    \end{align*}
    where $\iny{T} = \Yinv T \Y \in \so(3, \mathbb{C}, \X)$ by the isomorphism $\so(3, \mathbb{C}) \cong \so(3, \mathbb{C}, \X)$, provided by $\Yinv \so(3, \mathbb{C}) \Y = \so(3, \mathbb{C}, \X)$.

    The proofs for all other 
    eight
    types are similar. 
\end{proof}

We also prove a dichotomy on bipartite $\holbs$. 
This will be necessary for us to define the shells for certain types. 
A bipartite $\hols$ dichotomy on the Boolean domain  has not been stated explicitly before, but it is easily deducible from the $\holbs$ dichotomy in~\cite{holant-textbook}.
This form, especially~\cref{cor:bipartite-boolean-dichotomy-explicit}, will be convenient
for us.

\begin{lemma}
\label{lem:bipartite-boolean-dichotomy}
Let $\mathbf{F}$ and $\mathbf{G}$ be symmetric, non-degenerate, ternary, Boolean domain signatures.
Then, $\holbs(\mathbf{F} \mid \mathbf{G})$ is tractable if there is $M \in \gl_2(\mathbb{C})$ such that 
$\holbs(\mathbf{F} M^{-1}, M \mathbf{G})$ is tractable.
Otherwise, it is \sph.
\end{lemma}
\begin{proof}
  If such an $M$ exists, then
  \[
    \holbs(\mathbf{F} \mid \mathbf{G}) =_T \holbs(\mathbf{F}M^{-1} \mid M \mathbf{G}) \le_T \holbs(\mathbf{F}M^{-1} , M \mathbf{G}) \,,
  \]
  hence, $\holbs(\mathbf{F} \mid \mathbf{G})$ is tractable.

Assume $\holbs(\mathbf{F} \mid \mathbf{G})$ is not \sph.
Since $\mathbf{F}$ is non-degenerate, there exists some $\mathbf{u} \in \mathbb{C}^2$ such that $A = \langle \mathbf{F}, u \rangle$ is a symmetric rank $2$ matrix.
By a similar argument as in the proof of Lemma~2.14 in~\cite{holant-textbook}, there exists some $B \in \gl_2(\mathbb{C})$ such that $A = B\transpose B$.
Then, we can realize $\beq$ on the left side of $\holbs(\mathbf{F} B^{-1} \mid B \mathbf{G})$.
This implies that $\holbs(B \mathbf{G}) \le_T \holbs(\mathbf{F} \mid \mathbf{G})$.
By the $\holbs$ dichotomy, $B \mathbf{G}$ must be of a tractable form. 
We can show
that in all cases, $B \mathbf{G}$ can realize $\beq$, and so we have $\beq$ on both sides. 
Therefore, we have 
\[
  \holbs(\mathbf{F} \mid \mathbf{G}) =_T \holbs (\mathbf{F} B^{-1} \mid B \mathbf{G}) =_T \holbs(\mathbf{F} B^{-1}, B \mathbf{G}) \,,
\]
and we can finish the proof by applying the $\holbs$ dichotomy.
\end{proof}

We 
give an explicit form of the bipartite dichotomy.
Starting with $\mathbf{F}, \mathbf{G} \in \mathcal{E}$,
for $M = \begin{bsmallmatrix}
  a & b \\ c & d
\end{bsmallmatrix}$, which we may assume $\det(M) = 1$ so $M^{-1} = \begin{bsmallmatrix}
  d & -b \\
  -c & a
\end{bsmallmatrix}$,
we have, for some $x, y, u, v \in \mathbb{C}$,
\[
  M \mathbf{G} = x M \begin{bmatrix}
    1 \\ 0
    \end{bmatrix}\teh + y M \begin{bmatrix}
    0 \\ 1
  \end{bmatrix}\teh 
  = 
  x \begin{bmatrix}
  a \\ c 
  \end{bmatrix}\teh + y \begin{bmatrix}
  b \\ d
  \end{bmatrix}\teh
\]
and
\[
  \mathbf{F} M^{-1} = u \begin{bmatrix}
    1 & 0
    \end{bmatrix}\teh M^{-1} + v \begin{bmatrix}
    0 & 1
  \end{bmatrix}\teh M^{-1} 
  = u \begin{bmatrix}
    d & -b  
  \end{bmatrix}\teh + v \begin{bmatrix}
    -c & a
  \end{bmatrix} \,.
\]
Writing in the column vector form, $\mathbf{F}M^{-1}$ can be written as $(M^{-1})\transpose \mathbf{F}$.
We see that $M \mathbf{G}$ and $(M^{-1})\transpose \mathbf{F}$ satisfy the property $M \mathbf{G} = \mathbf{u}_2\teh + \mathbf{v}_2\teh$ and $(M^{-1})\transpose \mathbf{F} = \mathbf{u}_1\teh + \mathbf{v}_1\teh$ such that $\langle \mathbf{u}_1, \mathbf{u}_2 \rangle = 0$ and $\langle \mathbf{v}_1, \mathbf{v}_2 \rangle = 0$.

More generally, if we have $\mathbf{v}$ on the right side and $\mathbf{u}$ on the left side such that $\langle \mathbf{u}, \mathbf{v} \rangle = \mathbf{u}\transpose \mathbf{v} = 0$,
then the holographic transformation by $M$ preserves the annihilating condition, since $\langle (M ^{-1})\transpose \mathbf{u}, M \mathbf{v} \rangle = (\mathbf{u}\transpose M^{-1}) M \mathbf{v} = \mathbf{u}\transpose \mathbf{v} = 0$.
We can see it pictorially in \cref{fig:holo-keeps-orthogonal-left-right}.
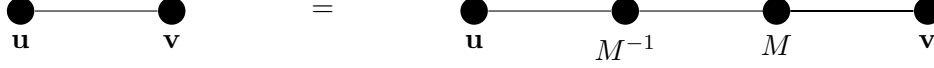
\begin{figure}
\centering
\begin{tikzpicture}
  \node[mycirc, label=below:{$\mathbf{u}$}] (g1) at (0, 0) {};

  \node[mycirc, label=below:{$\mathbf{v}$}] (g2) at (2, 0) {};

  \draw (g1) -- (g2);

  \node at (4, 0) {$=$};

  \node[mycirc, label=below:{$\mathbf{u}$}] (g1) at (6, 0) {};
  \node[mycirc, label=below:{$M^{-1}$}] (g3) at (8, 0) {};

  \node[mycirc, label=below:{$\mathbf{v}$}] (g2) at (12, 0) {};
  \node[mycirc, label=below:{$M$}] (g4) at (10, 0) {};

  \draw (g1) -- (g3) -- (g2) -- (g4);
\end{tikzpicture}
\caption{Holographic transformation preserves annihilation.}
\label{fig:holo-keeps-orthogonal-left-right}
\end{figure}

Applying the above calculation on each of the Boolean domain tractable cases, we obtain the following corollary. 
\begin{corollary}
  \label{cor:bipartite-boolean-dichotomy-explicit}
Let $\mathbf{F}$ and $\mathbf{G}$ be symmetric, non-degenerate, ternary, Boolean domain signatures.
Then $\holbs(\mathbf{F} \mid \mathbf{G})$ \sph unless for some $\mathbf{u}_i, \mathbf{v}_i \in \mathbb{C}^2$ and $\lambda_i \in \mathbb{C}$,
\begin{enumerate}
  \item $\mathbf{F} = \mathbf{u}_1\teh + \mathbf{v}_1\teh$ and $\mathbf{G} = \mathbf{u}_2\teh + \mathbf{v}_2\teh$ such that $\langle \mathbf{u}_1, \mathbf{v}_1 \rangle = 0$ and $\langle \mathbf{u}_2, \mathbf{v}_2 \rangle = 0$, or
  \item $\mathbf{F} = \lambda_1 \mathbf{u}_1\teh + \sym(\mathbf{u}_1, \mathbf{u}_1, \mathbf{v}_1)$ and $\mathbf{G} = \lambda_2 \mathbf{u}_2\teh + \sym(\mathbf{u}_2, \mathbf{u}_2, \mathbf{v}_2)$ such that $\langle \mathbf{u}_1, \mathbf{u}_2 \rangle = 0$.
\end{enumerate}
\end{corollary}

\subsection{\texorpdfstring{\casea}{A}}
Recall that $\mathbf{F}$ is type \casea if there exists some orthogonal $T$ such that $T \mathbf{F} = a \mathbf{e}_1\teh + b \mathbf{e}_2 \teh + c \mathbf{e}_3 \teh$ for some $a, b, c \in \mathbb{C}$.
Type \casea can be further separated by the existence of a nonzero annihilator (\cref{prop:nonzero-annihilator-types}).
We say $\mathbf{F}$ is  \caseatwo if it can be written $T \mathbf{F} = a \mathbf{e}_1\teh + b \mathbf{e}_2 \teh$ for nonzero $a, b \in \mathbb{C}$.
We say $\mathbf{F}$ is \caseathree if it can be written $T \mathbf{F} = a \mathbf{e}_1\teh + b \mathbf{e}_2 \teh + c \mathbf{e}_3 \teh$ for nonzero $a, b, c \in \mathbb{C}$.

\subsubsection{\texorpdfstring{\caseathree}{A3}}
\label{subsec:gc-case-a-3}
Let $\mathbf{F}$ be a type \caseathree signature in the form $\mathbf{F} = a \mathbf{e}_1\teh + b \mathbf{e}_2 \teh + c \mathbf{e}_3 \teh$ for nonzero $a, b, c \in \mathbb{C}$.
Replacing $\mathbb{R}$ with $\mathbb{C}$ in Proposition 6.1 (arXiv version) in~\cite{domain3-real-set} shows that $\spn\left(\mathbf{e}_1\teh, \mathbf{e}_2\teh, \mathbf{e}_3\teh \right)$ can be realized by $\mathbf{F}$ by gadget construction.
We also have $(=_{\db \dg}), (=_{\dg \dr}),$ and $(=_{\db \dr})$ realizable by $\mathbf{F}$.
This means that given another ternary symmetric signature $\mathbf{G}$, 
for $\holts(\mathbf{F}, \mathbf{G})$ to be not \sph,
it is necessary that $\holbs([1, 0, 0, 1], \mathbf{G}\domres{\db\dg})$, $\holbs([1, 0, 0, 1], \mathbf{G}\domres{\dg\dr})$, $\holbs([1, 0, 0, 1], \mathbf{G}\domres{\db\dr})$ are all not \sph.
Then the Boolean domain dichotomy
imposes strong restrictions on these Boolean domain signatures.

We may identify $\symspace{3}{3}$ as $\mathbb{C}^{10}$ with coordinates in \cref{fig:coordinates}.
\begin{figure}
  \centering
  \tersig{0.7}{{x_{\db\db\db}} {x_{\db\db\dg}} {x_{\db\db\dr}} {x_{\db\dg\dg}} {x_{\db\dg\dr}} {x_{\db\dr\dr}} {x_{\dg\dg\dg}} {x_{\dg\dg\dr}} {x_{\dg\dr\dr}} {x_{\dr\dr\dr}}}
  \caption{Coordinates for $\symspace{3}{3}$.}
  \label{fig:coordinates}
\end{figure}
By the Boolean domain dichotomy, 
for the problem $\holbs([1, 0, 0, 1], \mathbf{G}\domres{\db \dg})$ to be not \sph,
it must be the case that $\mathbf{G}$ is in the algebraic set
\begin{equation} \label{eq:set-e-bg}
  E_{\db \dg} = \{x_{\db \db \dg} = 0, x_{\db \dg \dg} = 0\} \cup D_{\db\dg} \, ,
\end{equation}
where $D_{\db \dg}$ is the set of symmetric ternary signatures $\mathbf{H}$ such that $\mathbf{H}\domres{\db\dg}$ is degenerate.
In particular, we can view $D_{\db\dg}$ as a projective variety on the coordinates $x_{\db\db\db}, x_{\db\db\dg}, x_{\db\dg\dg}, x_{\dg\dg\dg}$ that is the twisted cubic in $\mathbb{CP}^3$, which is the image of the Veronese embedding $\nu: \mathbb{CP}^1 \to \mathbb{CP}^3$ by $[x\mathrel{:}y] \mapsto [x^3 \mathrel{:} x^2 y \mathrel{:} x y^2 \mathrel{:} y^3]$.
Then, $D_{\db\dg}$ is the zero locus of the polynomials
\begin{align}
X_0X_2 = X_1 && X_1X_3 =X_2^2 && X_0X_3 = X_1X_2,
\end{align}
where $(X_0, X_1, X_2, X_3) = (x_{\db\db\db}, x_{\db\db\dg}, x_{\db\dg\dg}, x_{\dg\dg\dg})$.
Further, one can argue that $E_{\db\dg}$ is the zero locus of
\begin{align}
X_0X_2 = X_1 && X_1X_3 =X_2^2.
\end{align}

In summary, for an \caseathree signature $\mathbf{F} = a \mathbf{e}_1\teh + b \mathbf{e}_2 \teh + c \mathbf{e}_3 \teh$ for $a, b, c \ne 0$, we have 
\begin{enumerate}
  \item Frame:
  \item[] $\ag(\caseathreemath) := \spn\left(\mathbf{e}_1\teh, \mathbf{e}_2\teh, \mathbf{e}_3\teh\right)$
  \item[] $\agx(\caseathreemath) := \Yinv \ag(\caseathreemath)$

     \item
    Shell for $\holts(\mathbf{F}, \mathcal{F})$:
   \item[] $\mathcal{F} \subset \bigcap_{i \ne j \in \{\db, \dg, \dr\}} E_{ij}$
\end{enumerate}

\subsubsection{\texorpdfstring{\caseatwo}{A2}}\label{subsec:gc-case-a-2}

Similarly, for a type \caseatwo signature $\mathbf{F} = a \mathbf{e}_1\teh + b \mathbf{e}_2\teh$ for nonzero $a, b \in \mathbb{C}$, we have
\begin{enumerate}
  \item Frame:
  \item[] $\ag(\caseatwomath) := \spn\left(\mathbf{e}_1\teh, \mathbf{e}_2\teh \right)$
  \item[] $\agx(\caseatwomath) := \Yinv \ag(\caseatwomath)$

\item
    Shell for $\holts(\mathbf{F}, \mathcal{F})$: 
   \item[] $\mathcal{F} \subset E_{\db \dg}$
\end{enumerate}

\subsection{\texorpdfstring{\caseb}{B}}
Recall that $\mathbf{F}$ is a type \caseb signature if there exists some orthogonal $T$ such that 
$T \mathbf{F} = a \bzp \teh + b \bzm \teh + c\mathbf{e}_3\teh$ for some $a, b, c \in \mathbb{C}$.
Again, we separate type \caseb by the existence of a nonzero annihilator (\cref{prop:nonzero-annihilator-types}).
Let \casebthree be the case when $a, b, c \ne 0$.
Let \casebtwo be the case when $a, b \ne 0$ and $c = 0$.
Note that the case when $a$ or $b$ is $0$, for example, $a \bzp\teh + c \mathbf{e}_3\teh$ is a \casecz signature because it is annihilated by the isotropic $\bzp$, so we will not consider it among the subtypes of \caseb.

\subsubsection{\texorpdfstring{\casebthree}{B3}}\label{subsec:gc-case-b-3}
Let $\mathbf{F}$ be a type \casebthree signature of the form $\mathbf{F} = x \bzp \teh + y \bzm \teh + z \mathbf{e}_3\teh$ for nonzero $x, y, z \in \mathbb{C}$.
Proposition 6.2 (arXiv version) in~\cite{domain3-real-set} shows that $\spn(\bzp, \bzm, \mathbf{e}_3)$ is realizable by $\mathbf{F}$ by gadget construction.
Also, $(=_{\db\dg})$ is realizable by $\mathbf{F}$.
However, this case is more natural to work with if we use a non-orthogonal transformation. 
Consider applying $\Yinv$ to $\mathbf{F}$.
Let $\iny{\mathbf{F}} = \Yinv \mathbf{F}$.
We get 
\[\holts(\mathbf{F}) \; =_T \; \holts( \X \mid \iny{\mathbf{F}}) \; =_T \; \holts( \X \mid x \mathbf{e}_1\teh + y \mathbf{e}_2\teh + z \mathbf{e}_3 \teh) \,.\]
We have $V = \spn\left(\mathbf{e}_1\teh, \mathbf{e}_2\teh, \mathbf{e}_3\teh\right)$ realizable on the right side since $V = \Yinv\spn(\bzp\teh, \bzm\teh, \mathbf{e}_3\teh)$.

Since $\X V = V$, we have the following chain of reductions for any set of nondegenerate symmetric ternary signatures $\mathcal{F}$:
\begin{align*}
    \holts(\mathbf{F}, \mathcal{F}) 
                                 &=_T \holts(\X \mid \iny{\mathbf{F}}, \Yinv \mathcal{F} ) \\
                                 &\ge_T \holts(\spn\left(\mathbf{e}_1\teh, \mathbf{e}_2\teh, \mathbf{e}_3\teh\right), \X \mid \Yinv \mathcal{F})) \\
  &\ge_T \holbs([1, 0, 0, 1] \mid (\Yinv \mathcal{F})\domres{ij})
\end{align*}
for all $i \ne j \in \{\db, \dg, \dr\}$.
Then, by the bipartite Boolean domain dichotomy (\cref{cor:bipartite-boolean-dichotomy-explicit}), to escape \sphness, 
the symmetric ternary signatures in 
$(\Yinv \mathcal{F})\domres{ij}$ must be \geneq or degenerate. 

Also note that $\iny{\mathbf{F}}$ can realize on the right side the binary equality $\beq$ represented by the $3 \times 3$ identity matrix, by $\langle \mathbf{e}_1\teh + \mathbf{e}_2\teh + \mathbf{e}_3\teh, (1, 1, 1) \rangle$.
        This can be used to realize $\beq$ on the left side as well since $\X \beq \X = \beq$. 
        Then, 
        \begin{align*}
            \holts(\X \mid \iny{\mathbf{F}}, \Yinv \mathcal{F}) &\ge_T \holts(\X \mid \iny{\mathbf{F}}, \Yinv \mathcal{F}, \beq) \\
            &\ge_T \holts(\X, \beq \mid\iny{\mathbf{F}}, \Yinv \mathcal{F}, \beq) \\
            &\ge_T \holts(\X, \iny{\mathbf{F}}, \Yinv \mathcal{F})
        \end{align*}
        because having $\beq$ on both sides means we can remove the bipartite assumption. 
\begin{enumerate}
  \item Frame:
    \item[] $\ag(\casebthreemath) := \spn\left(\bzp\teh, \bzm\teh, \mathbf{e}_3\teh\right)$
    \item[] $\agx(\casebthreemath) :=  \spn\left(\mathbf{e}_1\teh, \mathbf{e}_2\teh, \mathbf{e}_3\teh\right)$
  \item  Shell for $\holts(\mathbf{F}, \mathcal{F})$: 
   \item[] $\Yinv \mathcal{F} \subset  \bigcap_{i \ne j \in \{\db, \dg, \dr\}} E_{ij}$
\end{enumerate}

We also compute the general form of \casebthree using the spin representation of $\speciallinear(2, \mathbb{C})$ parameterizing $\so(3, \mathbb{C})$.
Let $T$ be an orthogonal matrix with columns $\mathbf{v}_1, \mathbf{v}_2, \mathbf{v}_3 \in \mathbb{C}^3$.
Then, $T \ag(\casebthreemath) = \spn(\mathbf{u}_1\teh, \mathbf{u}_2\teh, \mathbf{u}_3\teh)$ where $\mathbf{u}_1 = \mathbf{v}_1 + i \mathbf{v}_2, \mathbf{u}_2 = \mathbf{v}_1 - i \mathbf{v}_2, \mathbf{u}_3 =  \mathbf{v}_3$.
A useful calculation here, using the 
(complex) spin representation of $\speciallinear(2, \mathbb{C})$ parameterizing $\so(3, \mathbb{C})$ and $\so(3, \mathbb{C}, \X)$ 
in~\cref{eqn:phiM-parameterization,eq:paramaterize-so3cx} and the isomorphism $\so(3, \mathbb{C}) \cong \so(3, \mathbb{C}, \X)$, is that 
\begin{align*}
  \begin{bmatrix}
    | & | & | \\
    \mathbf{u}_1 & \mathbf{u}_2 & \mathbf{u}_3 \\
    | & | & |
    \end{bmatrix} &= T \begin{bmatrix}
    1/\sqrt{2} & 1/\sqrt{2} & 0 \\ 
    \i/\sqrt{2} & -\i/\sqrt{2} & 0 \\ 
    0 & 0 & 1
  \end{bmatrix} \\
  &= T \Y \\
  &= \phi(M) \Y \\
  &= \Y \rho(M) \Yinv \Y \\
  &= \begin{bmatrix}
    \frac{1}{\sqrt{2}}(- c^2 + d^2) & \frac{1}{\sqrt{2}}(a^2 - b^2) & - a c + b d\\
    \frac{\i}{\sqrt{2}}(c^2 + d^2) & - \frac{\i}{\sqrt{2}}(a^2 + b^2) & \i (ac + bd) \\
    \sqrt{2} c d & - \sqrt{2} a b & b c + a d
  \end{bmatrix}
\end{align*}
where $M = \begin{bsmallmatrix} a & b \\ c & d \end{bsmallmatrix} \in \speciallinear(2, \mathbb{C})$,  $ad - bc = 1$. 
Therefore, we have
\begin{align}
  \label{eq:isotropic-after-orthogonal}
 \mathbf{u}_1 &= T \bzp = \frac{1}{\sqrt{2}} \begin{bmatrix}
    -c^2 + d^2 \\ \i(c^2 + d^2) \\ 2 cd
  \end{bmatrix} \nonumber \\ 
  \mathbf{u}_2 &= T \bzm = \frac{1}{\sqrt{2}} \begin{bmatrix}
      a^2 - b^2 \\ -\i (a^2 + b^2) \\ - 2 ab
    \end{bmatrix} \\
  \mathbf{u}_3 &=    T \mathbf{e}_3 = \begin{bmatrix}
        -a c + bd \nonumber \\ \i (ac + bd) \\ bc + ad
      \end{bmatrix} \, .
      \end{align}

      \subsubsection{\texorpdfstring{\casebtwo}{B2}}
\label{subsec:gc-case-b-2}
      Let $\mathbf{F}$ be a type \casebtwo signature of the form $\mathbf{F} = a \bzp \teh + b \bzm \teh $ for nonzero $a, b \in \mathbb{C}$.
      Proposition 6.2 (arXiv version) in~\cite{domain3-real-set} shows that $\spn(\bzp, \bzm)$ can be realized by $\mathbf{F}$.
      By a similar argument as above, we have

          \begin{align*}
        \holts(\mathbf{F}, \mathcal{F}) &\ge_T \holts(\X \mid \mathbf{e}_1\teh + \mathbf{e}_2\teh, \Yinv \mathcal{F}) \\
        &\ge_T \holts(\mathbf{e}_1\teh + \mathbf{e}_2\teh, \X \mid \Yinv \mathcal{F}) \\
        &\ge_T \holbs([1, 0, 0, 1] \mid (\Yinv \mathcal{F})\domres{\db\dg}) \,.
        \end{align*}
      \begin{enumerate}
        \item Frame:
            \item[] $\ag(\casebtwomath) := \spn\left(\bzp\teh, \bzm\teh\right)$
            \item[] $\agx(\casebtwomath) :=  \spn\left(\mathbf{e}_1\teh, \mathbf{e}_2\teh \right)$
  \item  Shell for $\holts(\mathbf{F}, \mathcal{F})$: 
   \item[] $\Yinv \mathcal{F} \subset  E_{\db\dg}$
      \end{enumerate}

      \subsection{\texorpdfstring{\casec}{C}}
      We summarize the analysis given in \cref{sec:casec}. 

      \subsubsection{\texorpdfstring{\casecz}{C0}}
      Let $\mathbf{F}$ be a type \casecz signature such that $\ann(\mathbf{F}) = \spn(\bzp)$.
      Then as $\Y\transpose \bzp = \mathbf{e}_2$, we have
       $\ann(\Yinv \mathbf{F}) = \spn(\mathbf{e}_2)$ by \cref{prop:annihilating-vector-under-Y}.
      So $\iny{\mathbf{F}} = \Yinv \mathbf{F}$ is supported on $\db\dr$.
      We have seen in \cref{subsec:casecz-characterization} that $\iny{\mathbf{F}}$ cannot make any useful connected gadget other than itself. 
      Expressing the problem in the bipartite setting, we may move $\iny{\mathbf{F}}$ to the left side, which becomes $\X \iny{\mathbf{F}}$ supported on $\dg\dr$.

      \begin{align*}
      \holts(\mathbf{F}, \mathcal{F}) &=_T \holts(\X \mid \iny{\mathbf{F}}, \Yinv \mathcal{F}) \\
      &\ge_T \holts(\X, \X \iny{\mathbf{F}} \mid \Yinv \mathcal{F}) \\
      &\ge_T  \holbs( (\X \iny{\mathbf{F}})\domres{\dg \dr} \mid (\Yinv \mathcal{F})\domres{\dg \dr}) \, .
      \end{align*}
      There are two possibilities here.
      If $(\X \iny{\mathbf{F}})\domres{\dg \dr}$ is not of the form $\mathbf{v}_1\teh + \mathbf{v}_2\teh$ or $\mathbf{v}_1\teh + \sym(\mathbf{v}_1, \mathbf{v}_1, \mathbf{v}_2)$, then to escape \sphness by the bipartite Boolean domain dichotomy (\cref{cor:bipartite-boolean-dichotomy-explicit}),
      the signatures in the right side must be 
      degenerate.
      Otherwise, the constraint is determined by the Boolean domain bipartite dichotomy. 
      \begin{enumerate}
        \item Frame:
        \item[] A \casecz signature cannot make any useful gadget by itself (other than itself).
        \item Shell for $\holts(\mathbf{F}, \mathcal{F})$: 
        \item[] $\Yinv \mathcal{F} \subset D_{\dg\dr} \cup A$
      \end{enumerate}
      where $A$ is the algebraic set determined by the Boolean domain bipartite dichotomy, possibly empty by the above discussion.
      We can write down the specific set $A$ if needed using \cref{cor:bipartite-boolean-dichotomy-explicit}.

      \subsubsection{\texorpdfstring{\casecobg}{C1bg}}
\label{subsec:gc-case-c1-bg}
      Let $\mathbf{F}$ be a type \casecobg signature such that $\iny{\mathbf{F}} = \Yinv \mathbf{F} = a \mathbf{e}_1\teh + \lambda \sym(\mathbf{e}_1, \mathbf{e}_1, \mathbf{e}_2)$ for any $a, \lambda \in \mathbb{C}$ and $\lambda \ne 0$. 
      We can show that in the setting $\holts(\X \mid \iny{\mathbf{F}})$, 
      $x \mathbf{e}_1\teh + y \sym(\mathbf{e}_1, \mathbf{e}_1, \mathbf{e}_2)$ is realizable on the right side
      for any $x, y \in \mathbb{C}$.

      \begin{align*}
          \holts(\mathbf{F}, \mathcal{F}) 
          &\ge_T \holts(\X \mid \sym(\mathbf{e}_1, \mathbf{e}_1, \mathbf{e}_2), \Yinv \mathcal{F})  \\
          &\ge_T \holts(\X(\sym(\mathbf{e}_1, \mathbf{e}_1, \mathbf{e}_2)), \X \mid \Yinv \mathcal{F}) \\
          &\ge_T \holts(\sym(\mathbf{e}_2, \mathbf{e}_2, \mathbf{e}_1) \mid \Yinv \mathcal{F}) \\
          &\ge_T \holbs([0, 0, 1, 0] \mid (\Yinv \mathcal{F})\domres{\db\dg}) \,.
      \end{align*}
      By the Boolean domain bipartite dichotomy (\cref{cor:bipartite-boolean-dichotomy-explicit}), 
      the symmetric ternary signatures in $(\Yinv \mathcal{F})\domres{\db\dg}$ must be of the form $[*,*, 0, 0]$ or degenerate.
      \begin{enumerate}
        \item Frame:
          \item[] $\ag(\casecobgmath) := \spn(\bzp\teh, \sym(\bzp, \bzp, \bzm))$
          \item[] $\agx(\casecobgmath) := \spn(\mathbf{e}_1\teh, \sym(\mathbf{e}_1, \mathbf{e}_1, \mathbf{e}_2))$
        \item Shell for $\holts(\mathbf{F}, \mathcal{F})$: 
        \item[] $\Yinv \mathcal{F} \subset M_{\db\dg}$
          where $M_{\db\dg} = D_{\db\dg} \cup \{x_{\db\dg\dg} = 0, x_{\dg\dg\dg} = 0\}$, 
          i.e.,  signatures when restricted to $\db\dg$ are degenerate or are   of the form $[*,*,0,0]$.
      \end{enumerate}

      \subsubsection{\texorpdfstring{\casecosep}{C1sep}}
    \label{subsec:gc-case-c1-sep}
      Let $\mathbf{F}$ be a type \casecosep signature such that $\iny{\mathbf{F}} = \Yinv \mathbf{F} = a \mathbf{e}_1\teh + \lambda \sym(\mathbf{e}_1, \mathbf{e}_1, \mathbf{e}_2) + d \mathbf{e}_3\teh$ for $a, d, \lambda \in \mathbb{C}$ with $d, \lambda \ne 0$.
      We can see that in $\holts(\X \mid \iny{\mathbf{F}})$, we can realize $\spn(\mathbf{e}_1\teh, \sym(\mathbf{e}_1, \mathbf{e}_1, \mathbf{e}_2), \mathbf{e}_3\teh)$ on the right side. 

      \begin{align*}
          \holts(\mathbf{F}, \mathcal{F}) 
          &\ge_T \holts(\X \mid \mathbf{e}_1\teh + \mathbf{e}_3\teh, \Yinv \mathcal{F})  \\
          &\ge_T \holts(\X(\mathbf{e}_1\teh + \mathbf{e}_3\teh), \mid \Yinv \mathcal{F}) \\
          &\ge_T \holts(\mathbf{e}_2\teh + \mathbf{e}_3\teh \mid \Yinv \mathcal{F}) \\
          &\ge_T \holbs([1, 0, 0, 1] \mid (\Yinv \mathcal{F}))\domres{\dg\dr}) \,.
      \end{align*}
      By the bipartite Boolean domain dichotomy (\cref{cor:bipartite-boolean-dichotomy-explicit}), every signature in
      $(\Yinv \mathcal{F})\domres{\dg\dr}$
      must be of the form $[*,0, 0, *]$ or degenerate.
      Also, since $\iny{\mathbf{F}}$ can realize $\sym(\mathbf{e}_1, \mathbf{e}_1, \mathbf{e}_2)$ on the right side, we additionally inherit the same constraints from \cref{subsec:gc-case-c1-bg}, that the symmetric ternary signatures in $(\Yinv \mathcal{F})\domres{\db\dg}$ must be $[*, *, 0, 0]$ or degenerate.
      \begin{enumerate}
        \item Frame:
        \item[] $\ag(\casecosepmath) := \spn(\bzp\teh, \sym(\bzp, \bzp, \bzm), \mathbf{e}_3\teh)$
          \item[] $\agx(\casecosepmath) :=  \spn(\mathbf{e}_1\teh, \sym(\mathbf{e}_1, \mathbf{e}_1, \mathbf{e}_2), \mathbf{e}_3\teh)$
        \item Shell for $\holts(\mathbf{F}, \mathcal{F})$: 
        \item[] $\Yinv \mathcal{F} \subset M_{\db\dg} \cap E_{\dg\dr}$

      \end{enumerate}

      \subsubsection{\texorpdfstring{\casecodz}{C1dz}}
      \label{subsec:gc-case-c1dz}
      Let $\mathbf{F}$ be a type \casecodz signature such that $\iny{\mathbf{F}} = \Yinv \mathbf{F} = [a; \kappa, b; 0, 0, c; 0, 0, 0, 0]$ for some $a, b, c, \kappa \in \mathbb{C}$ with $\kappa, c \ne 0$.
      Since replacing a signature with its nonzero scalar multiple does not change the complexity, we may assume that $\kappa = 1$.
      From \cref{subsec:case-c-codz-analysis}, $\holts(\X \mid \iny{\mathbf{F}})$ can realize on the right side using gadget construction and polynomial interpolation
      \[
            \bigcup_{k \ge 1, k \ne k_0} \spn\left(\mathbf{e}_1\teh, \quad \sym(\mathbf{e}_1, \mathbf{e}_1, \mathbf{e}_3), \quad \sym(\mathbf{e}_1, \mathbf{e}_1, \mathbf{e}_2) + c^{2k + 1} \sym(\mathbf{e}_1, \mathbf{e}_3, \mathbf{e}_3)\right)
          \]
          for some $k_0$.
      We see that the signatures realizable depend on the value of $c$. 
      For the purpose of the proofs we give, the exact value of $c$ does not actually matter. 
      Rather, we just need to use one linear subset $\spn\left(\mathbf{e}_1\teh, \sym(\mathbf{e}_1, \mathbf{e}_1, \mathbf{e}_3), \sym(\mathbf{e}_1, \mathbf{e}_1, \mathbf{e}_2) + \lambda \sym(\mathbf{e}_1, \mathbf{e}_3, \mathbf{e}_3)\right)$ for any $\lambda \ne 0$. 
      To make the definition well defined, we will define $\ag(\casecodzmath, \lambda)$ and $\agx(\casecodzmath, \lambda)$ for $\lambda \ne 0$. 
      Then, regardless of the values of $a, b, c$, if $c \ne 0$, there exists some $\lambda \ne 0$ such that 
      $\agx(\casecodzmath, \lambda)$ is realizable on the right side in $\holts(\X \mid [a; 1, b; 0, 0, c;0, 0, 0, 0])$.
      
      \begin{enumerate}
        \item Frame:
          \item[] $
            \ag(\casecodzmath, \lambda) := \spn\left(\bzp\teh, \sym(\bzp, \bzp, \mathbf{e}_3), \sym(\bzp, \bzp, \bzm) + \lambda \sym(\bzp, \mathbf{e}_3, \mathbf{e}_3)\right)
          $
                  \item[]     $
            \agx(\casecodzmath, \lambda) := \spn\left(\mathbf{e}_1\teh, \sym(\mathbf{e}_1, \mathbf{e}_1, \mathbf{e}_3), \sym(\mathbf{e}_1, \mathbf{e}_1, \mathbf{e}_2) + \lambda \sym(\mathbf{e}_1, \mathbf{e}_3, \mathbf{e}_3)\right)
          $

        \item Shell for $\holts(\mathbf{F}, \mathcal{F})$: 
        \item[] $\Yinv \mathcal{F} \subset M_{\dr\dg}$
        where $M_{\dr \dg} =  D_{\dg \dr} \cup \{x_{\dg\dg\dg} = 0, x_{\dg\dg\dr} = 0\}$, 
        i.e.,  signatures when restricted to $\dg\dr$ are degenerate or are  of the form $[0, 0, *, *]$.
      \end{enumerate}

      \begin{remark}
      If $c$ is not a root of unity, then by \cref{rem:c-not-root-of-unity-interpolate-away}, $\holts(\X \mid [*; 1, *; 0, 0, c; 0, 0, 0, 0])$ can realize on the right side
      \[
            \spn\left(\mathbf{e}_1\teh, \quad \sym(\mathbf{e}_1, \mathbf{e}_1, \mathbf{e}_3), \quad \sym(\mathbf{e}_1, \mathbf{e}_1, \mathbf{e}_2), \quad \sym(\mathbf{e}_1, \mathbf{e}_3, \mathbf{e}_3)\right)
      \]
      and impose the Polynomial Constraints
        $\Yinv \mathcal{F} \subset D_{\dg\dr} \cap M_{\db\dg}$.
      As we will see, for the proof of the dichotomy in this paper, 
      whether $c$ is a root of unity or not does not play a role.
      \end{remark}

      \subsubsection{\texorpdfstring{\casecogen}{C1gen}}
      \label{subsec:gc-case-c1gen}
      Let $\mathbf{F}$ be a type \casecogen signature such that $\iny{\mathbf{F}} = \Yinv \mathbf{F} = [a; \lambda, b; 0, 0, c; 0, 0, 0, d]$, in which we may assume $\lambda = 1$, $d \ne 0$ and $c^2 - c - bd \ne 0$.
      From \cref{subsec:case-c-gen-analysis}, we have 
      \begin{enumerate}
        \item Frame:
          \begin{align*}
          \ag(\casecogenmath) &:=
            \spn(\bzp\teh, \sym(\bzp, \bzp, \bzm)) \cup \bigcup_{u, v \in \mathbb{C}} \spn\left(\bzp\teh, (u \bzp + v \mathbf{e}_3)\teh\right) \\
            \\
            \agx(\casecogenmath) &:=
            \spn(\mathbf{e}_1\teh, \sym(\mathbf{e}_1,\mathbf{e}_1, \mathbf{e}_2)) \cup \bigcup_{u, v \in \mathbb{C}} \spn\left(\mathbf{e}_1\teh, (u\mathbf{e}_1 + v\mathbf{e}_3)\teh\right) 
          \end{align*}

        \item Shell for $\holts(\mathbf{F}, \mathcal{F})$: 
        \item[] $\Yinv \mathcal{F} \subseteq D_{\dg\dr} \cap M_{\db\dg}$ for some $T' \in \so(3, \mathbb{C}, \X)$.
      \end{enumerate}

\section{Dichotomy of a Pair of Ternary Signatures}
\label{sec:pair-dicho}
In this section, we prove a $\holts(\mathbf{F}, \mathbf{G})$ dichotomy when $\mathbf{F}$ and $\mathbf{G}$ are nondegenerate symmetric \emph{ternary} signatures.
We first list some results that will be useful in the analysis. 
Then, we  prove a dichotomy for $\holts(\mathbf{F}, \mathbf{G})$. 
This will be dealt with as follows. 
For each (tractable) signature type, suppose $\mathbf{F}$ is of this type; 
we anchor the analysis on $\mathbf{F}$ by assuming that $\mathbf{F}$ is of the defining form given in \cref{def:types-ternary}.

We will observe that 
the dichotomy theorems we prove come in two distinct flavors, in both the statements and the proof, depending on whether $\ann(\mathbf{F}) = 0$ or not.
If $\ann(\mathbf{F}) \ne 0$, then $\mathbf{F}$ is only supported on two domains after certain holographic transformation. 
This makes the signature very \emph{weak}: its frame is small and shell is coarse, as shown in \cref{sec:two-pieces}.
We overcome this difficulty by utilizing the properties of annihilators of $\mathbf{F}$ and $\mathbf{G}$. 
We show that either $\mathbf{F}$ and $\mathbf{G}$ do not interact in any meaningful way or one determines the form of the other, 
and $(\mathbf{F}, \mathbf{G})$ must be a compatible pair.
One interesting and unexpected consequence is that this characterization leads to a dichotomy of a set of ternary signatures for signatures with annihilators, and further we can add binary signatures with a little more work.

We first prove the dichotomy when $\ann(\mathbf{F}) \ne 0$ and $\ann(\mathbf{G}) \ne 0$ in \cref{subsec:two-ternary-dichotomy-a-2,subsec:two-ternary-dichotomy-C0,subsec:two-ternary-dichotomy-b-2,subsec:two-ternary-dichotomy-C-1-bg}.
This includes the cases when $\mathbf{F}$ is \caseatwo, \casebtwo, \casecobg, and \casecz by \cref{prop:nonzero-annihilator-types}. 
We will see the arguments for proving the dichotomy when $\ann(\mathbf{F}) \ne 0$ and $\ann(\mathbf{G}) \ne 0$ are very similar for different tractable types of $\mathbf{F}$; 
we analyze the relationship between the annihilators and apply \cref{prop:annihilator-uniquely-propagates-values}. 
For technical reasons, we also prove a dichotomy of $\holts(\mathbf{F}, \mathbf{G})$ when $\mathbf{G}$ is of type \casecodz and \casecogen. 
An intuition is that \casecodz signature has a large frame but a relatively coarse shell. 
This makes the analysis simpler by anchoring on a signature of a different type.

Then, we consider the cases when $\ann(\mathbf{F}) = 0$.
This includes the cases when $\mathbf{F}$ is \caseathree, \casebthree, \casecosep, \casecodz, and \casecogen. 
As we saw in \cref{sec:two-pieces}, often such signature can realize a signature with a nonzero annihilator. 
For instance, \caseathree can realize \caseatwo, \casebthree can realize \casebtwo, and \casecosep can realize \casecobg. 
So, the dichotomy for two signatures with annihilators can be applied when we prove the dichotomy for  signatures 
without annihilators. 
Also, we use the concept of frames and shells heavily here, combined with \cref{prop:subspace-e-bg} and \cref{prop:subspace-m-bg}, which allows us to easily check the containment of a frame in a shell by using basic algebraic geometry and the spin representation.

      \subsection{Some Useful Facts}
      \label{subsec:useful-facts}
      We use the following well known fact in algebraic geometry that algebraic sets can be written as irreducible components~\cite{hartshorne_algebraic_1977}. 
      \begin{proposition}\label{prop:algebraic-set-irreducible-components}
        Every algebraic set in the affine space $\mathbb{C}^n$ or the projective space
        $\mathbb{CP}^n$ can be expressed uniquely as a union of (irreducible) varieties, no one containing another.
      \end{proposition}
      As we described in \cref{sec:two-pieces}, 
      each signature defines polynomial constraints on
      other signatures 
      in order to escape \sphness, which we called shells. 
      For example, we defined the set $E_{\db \dg}$ in \cref{eq:set-e-bg}.
      Suppose $\mathbf{G}$ is of tractable type $\mathfrak{F}$, 
      If $T \ag(\mathfrak{F})$ is to be contained in $E_{\db\dg}$, 
      then for any
      linear subspace $V \subset T \ag(\mathfrak{F})$, we must have $V \subset E_{\db \dg}$.
      Since $E_{\db \dg}$ only places constraints on the domain $\{\db, \dg\}^*$, we can look at whether $V\domres{\db\dg} \subseteq E_1 \cup E_2$ where $E_1 = 
 \spn(\sbto\teh+ \sbtt\teh)$
            and $E_2 = \{\mathbf{v}\teh : \mathbf{v} \in \mathbb{C}^2\}$.
      The operator $(\blank)\domres{\db\dg}$ behaves particularly nicely with linear subspaces.
      For each $\mathbf{v}_1 \otimes \mathbf{v}_2 \otimes \mathbf{v}_3 \in (\mathbb{C}^3)^{\otimes 3}$, we have $(\mathbf{v}_1 \otimes \mathbf{v}_2 \otimes \mathbf{v}_3 )\domres{\db\dg} = (\proj_{12} \mathbf{v}_1) \otimes (\proj_{12} \mathbf{v}_2) \otimes(\proj_{12} \mathbf{v}_3)$ 
      and this extends linearly.
      Therefore, if $V = \spn( \mathbf{v}_1, \ldots, \mathbf{v}_k)$ then $V\domres{\db\dg} = \spn( \mathbf{v}_1\domres{\db\dg}, \ldots, \mathbf{v}_k\domres{\db\dg})$.

      We may also view the symmetric signatures as in $\mathbb{CP}^3$ instead of $\symspace{3}{2}$, since the constant multiples do not change the complexity of the problem, and the identically $0$ signature does not matter. 
      Denote by  $[x\mathpunct{:}y\mathpunct{:}z\mathpunct{:}w]$  the equivalence class $\{\lambda [x, y, z, w] : \lambda \in \mathbb{C}^* \} \in \mathbb{CP}^3$.
      Then, $E_1$ is the line connecting $[1\mathpunct{:}0\mathpunct{:}0\mathpunct{:}0]$ to $[0\mathpunct{:}0\mathpunct{:}0\mathpunct{:}1]$, and $E_2$ is the twisted cubic curve.
      Therefore, $E_1 \cup E_2$ is a union of a line and a curve.
      On the other hand $V\domres{\db\dg}$ is a linear variety in $\mathbb{CP}^3$.
      Suppose $V\domres{\db\dg} \subseteq E_1 \cup E_2$.
      If $V\domres{\db\dg}$ has (affine) dimension 
      $> 2$, then it cannot be a subset.
      If it has dimension $2$, then projectively, it must be a line, so by \cref{prop:algebraic-set-irreducible-components}, $V\domres{\db\dg} = E_1$.
      If it has dimension $\le 1$, then we can look at the spanning set of $V$ and deduce some linear dependence among its constituent vectors.
      Similar arguments work for the other permissible sets such as $M_{\db\dg}$ in \cref{sec:two-pieces}. 

      We can prove the following propositions by the above analysis and considering the symmetric rank using \cref{prop:symmetric-rank-sym-sum}.

      \begin{proposition}\label{prop:subspace-e-bg}
      Let $V$ be a linear subspace of $\symspace{3}{2}$.
      Assume $\holbs([1,0,0,1], V)$,  or
      $\holbs( \bne \mid [1,0,0, 1], V)$ or
      $\holbs([1,0,0, 1] \mid V)$ 
      is not \sph. Then $V \subseteq E$, where
      \[ E = \{[a, 0, 0, b] : a, b \in \mathbb{C}\} \cup \{ \mathbf{v} \teh : \mathbf{v} \in \mathbb{C}^2\}.\]
        Furthermore,
        \begin{itemize}
          \item $\dim V \le 2$, and
          \item If $\dim V = 2$, then $V = \spn(\sbto\teh, \sbtt\teh) = \{[a, 0, 0, b] : a, b \in \mathbb{C} \}$. 
          \item If $V = \spn(\mathbf{v}_1\teh, \mathbf{v}_2\teh)$, then 
            $\{\mathbf{v}_1, \mathbf{v}_2\} \sim \{\sbto, \sbtt\}$ or $\mathbf{v}_1 \sim \mathbf{v}_2$.
            \item If $V = \spn(\mathbf{v}_1\teh, \sym(\mathbf{v}_1, \mathbf{v}_1, \mathbf{v}_2))$, 
            then $\mathbf{v}_1 \sim \mathbf{v}_2$.
            \end{itemize}
      \end{proposition}

        \begin{proposition}\label{prop:subspace-m-bg}
        Let $V$ be a linear subspace of $\symspace{3}{2}$.
       Assume $\holbs(\bne \mid [0,1,0,0], V)$ or
       $\holbs([0, 0, 1, 0] \mid V)$ 
       is not \sph. 
       Then $V
       \subseteq M$, where
       \[M = \{[a, b, 0, 0] : a, b \in \mathbb{C}\} \cup \{ \mathbf{v} \teh : \mathbf{v} \in \mathbb{C}^2\}.\]
        Furthermore,
        \begin{itemize}
          \item $\dim V \le 2$, and
          \item If $\dim V = 2$, then $V = \spn(\sbto, \sym(\sbto, \sbto, \mathbf{u})) = \{ [a, b, 0, 0] : a, b \in \mathbb{C} \}$ for some $\mathbf{u} \in \mathbb{C}^2$.
          \item If $V = \spn(\mathbf{v}_1\teh, \mathbf{v}_2\teh)$, then 
            $\mathbf{v}_1 \sim \mathbf{v}_2$.
            \item If $V = \spn(\mathbf{v}_1\teh, \sym(\mathbf{v}_1, \mathbf{v}_1, \mathbf{v}_2))$, 
            then $\mathbf{v}_1 \sim \sbto$ or $\mathbf{v}_1 \sim \mathbf{v}_2$.
        \end{itemize}
      \end{proposition}

      We also list some facts about annihilators.
    \begin{proposition}
    \label{prop:annihilating-vector-under-gl-transform}
        Let $\mathbf{F}$ be a 
        symmetric signature on domain $3$.
        Let $M \in \gl(3, \mathbb{C})$.
        Then, $\ann(M \mathbf{F}) = (M^{-1})\transpose \ann(\mathbf{F})$.
    \end{proposition}
    \begin{proof}
    Suppose $\mathbf{u} \in \ann(\mathbf{F})$.
        By the same argument as in \cref{fig:holo-keeps-orthogonal-left-right} where we put $\mathbf{u}$ on the left and $\mathbf{F}$ on the right with $M$, 
        we see that $\mathbf{u}\transpose M^{-1}$ annihilates $M \mathbf{F}$ as a row vector.
        Therefore, $(M^{-1})\transpose \mathbf{u} \in \ann(M \mathbf{F})$. 
        Since $M \in \gl(3, \mathbb{C})$, we can apply the same argument to show that $\mathbf{u}' \in \ann(M \mathbf{F})$ implies $M\transpose \mathbf{u}' \in \ann(\mathbf{F})$. 
        Therefore, $\ann(M \mathbf{F}) = (M^{-1})\transpose \ann(\mathbf{F})$.
    \end{proof}
      \begin{proposition}
        \label{prop:annihilating-vector-under-Y}
        Let $\mathbf{F}$ be a nondegenerate symmetric signature on domain $3$
        with $\ann(\mathbf{F}) = \spn(\mathbf{u}) \ne 0$ where $\mathbf{u}$ is isotropic.
        Then, $\ann(\Yinv \mathbf{F}) = \spn(\mathbf{u}')$ where $\mathbf{u}' = \Y\transpose \mathbf{u} = (x, y, z)\transpose$ and $2 x y + z^2 = 0$.
      \end{proposition}
      \begin{proof}
        By \cref{prop:annihilating-vector-under-gl-transform}, $\ann(\Yinv \mathbf{F}) = \spn(\mathbf{u}')$.
        Recalling that $\mathbf{X} = \Y\transpose \Y$,
        we have 
        \[
          (\mathbf{u}')\transpose \mathbf{X} \mathbf{u}' = \mathbf{u}\transpose \Y \Y\transpose \Y \Y\transpose \mathbf{u} = \mathbf{u}\transpose \begin{bmatrix}
            1 & 0 & 0 \\
            0 & -1 & 0 \\
            0 & 0 & 1
          \end{bmatrix} 
          \begin{bmatrix}
            1 & 0 & 0 \\
            0 & -1 & 0 \\
            0 & 0 & 1
          \end{bmatrix}
          \mathbf{u} = \mathbf{u}\transpose \mathbf{u} = 0 \,.
        \]
        If we write $\mathbf{u}' = (x, y, z)\transpose$, then $(\mathbf{u}')\transpose \mathbf{X} \mathbf{u}' = 2 xy + z^2 = 0$.
      \end{proof}

The following two propositions are the main ingredients for proving dichotomy for a pair of ternary signatures with nontrivial annihilators. 
\cref{prop:annihilator-uniquely-propagates-values} will be applied 
when,  in addition to    $\mathbf{G}$, we have another signature  $\mathbf{F}$ with an annihilator 
(normalized to $\mathbf{e}_3$) which is not orthogonal to the annihilator of  
$\mathbf{G}$.
Otherwise, if the annihilators are orthogonal, we showed in 
\cref{prop:orthogonal-normal-vector-edge-factors} that the interaction between them cannot cause \sphness. 

      \begin{proposition}\label{prop:annihilator-uniquely-propagates-values}
        Let $\mathbf{G}$ be a symmetric ternary signature on domain $3$.
        Suppose $\ann(\mathbf{G}) = \spn(\mathbf{u})$ such that $\langle \mathbf{u}, \mathbf{e}_3 \rangle \ne 0$.
        Then, $\mathbf{G}$ is uniquely determined by the values of $\mathbf{G}\domres{\db\dg}$.

        In particular, without loss of generality, writing $\mathbf{u} = (x, y, -1)$, 
        \begin{enumerate}
          \item if $\mathbf{G}\domres{\db\dg} = [a, 0, 0, b]$, then $\mathbf{G} = a(1, 0, x)\teh + b(0, 1, y)\teh$;
        \item if $\mathbf{G}\domres{\db\dg} = [a, b, 0, 0]$, then $\mathbf{G} = a (1, 0, x)\teh + b \sym((1, 0, x), (1, 0, x), (0, 1, y))$;
         \item if  $\mathbf{G}$ is nondegenerate, then $\mathbf{G}\domres{\db\dg}$ is nondegenerate.
        \end{enumerate}
      \end{proposition}
      \begin{proof}
        Since $\langle \mathbf{u}, \mathbf{e}_3 \rangle \ne 0$, we may assume that $\mathbf{u} = (x, y, -1)$ since $\ann(\mathbf{G})$ is a linear subspace. 
        Then, $\langle \mathbf{G}, \mathbf{u} \rangle = 0$ implies that, for example, $\mathbf{G}(\db\db\dr) = x \mathbf{G}(\db\db\db) + y \mathbf{G}(\db\db\dg)$,
        i.e., $y_0 = x x_0 + y x_1$ in~\cref{fig:propagation-value} (see also \cref{fig:computing-f-u}). 
        In general, the values of  $\mathbf{G}\domres{\db\dg}$  propagate as depicted in \cref{fig:propagation-value}.
        \begin{figure}
          \centering
          \begin{tikzpicture}
            \node (bbb) at (0, 0) {$x_0$};
            \node (bbg) at (-1, -1) {$x_1$};
            \node (bbr) at (1, -1) {$y_0$};
            \node (bgg) at (-2, -2) {$x_2$};
            \node (bgr) at (0, -2) {$y_1$};
            \node (brr) at (2, -2) {$z_0$};
            \node (ggg) at (-3, -3) {$x_3$};
            \node (ggr) at (-1, -3) {$y_2$};
            \node (grr) at (1, -3) {$z_1$};
            \node (rrr) at (3, -3) {$w_0$};

            \draw[->] (bbb) -- (bbr);
            \draw[->] (bbg) -- (bbr);
            \draw[->] (bbg) -- (bgr);
            \draw[->] (bgg) -- (bgr);
            \draw[->] (bgg) -- (ggr);
            \draw[->] (ggg) -- (ggr);
            \draw[->] (bbr) -- (brr);
            \draw[->] (bgr) -- (brr);
            \draw[->] (bgr) -- (grr);
            \draw[->] (ggr) -- (grr);
            \draw[->] (brr) -- (rrr);
            \draw[->] (grr) -- (rrr);
          \end{tikzpicture}
          \caption{Propagation of the values for $\mathbf{G}$ such that $\langle \mathbf{G}, (x, y, -1) \rangle = 0$.}
          \label{fig:propagation-value}
        \end{figure}
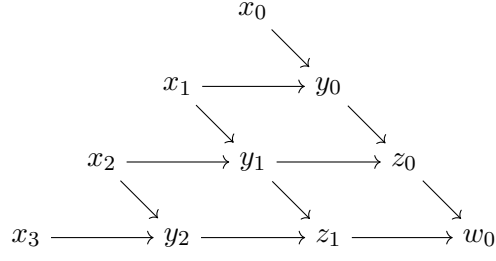
        The statements for the three listed tractable cases of $\mathbf{G}\domres{\db\dg} = [x_0, x_0, x_2, x_3]$  follow from uniqueness.
        See~\cref{fig:ann-e-on-bg,fig:casedz-m-on-bg,fig:ann-degen-on-bg}.
      \end{proof}

        \begin{figure}
        \centering
        \tersig{1}{a 0 {ax} 0 0 {a x^2} b {b y} {b y^2} {a x^3 + b y^3}}
        \caption{$\mathbf{G}$ satisfying $\langle \mathbf{G}, (x, y, -1) \rangle =0$
        and $\mathbf{G}\domres{\db\dg} = [a, 0, 0, b]$.}
        \label{fig:ann-e-on-bg}
      \end{figure}

      \begin{figure}
        \centering
        \tersig{1}{a b {ax + by} 0 bx {ax^2 + 2 bxy} 0 0 {bx^2} {ax^3 + 3bx^2y}}
        \caption{$\mathbf{G}$ satisfying $\langle \mathbf{G}, (x, y, -1) \rangle =0$
        and $\mathbf{G}\domres{\db\dg} = [a,b,0,0]$.}
        \label{fig:casedz-m-on-bg}
      \end{figure}

        \begin{figure}
        \centering
        \tersig{1.2}{{a^3} {a^2 b} {a^2 (a x + b y)} {a b^2} {ab(a x + b y)} {a(a x + by)^2} {b^3} {b^2(ax + by)} {b(a x + by)^2} {(a x + by)^3}}
        \caption{$\mathbf{G}$ satisfying $\langle \mathbf{G}, (x, y, -1) \rangle =0$
        and $\mathbf{G}\domres{\db\dg} = [a^3, a^2b, ab^2, b^3]$.}
        \label{fig:ann-degen-on-bg}
      \end{figure}

\subsection{\texorpdfstring{\caseatwo}{A2}}\label{subsec:A2}
      We start with a \caseatwo signature $\mathbf{F} = a \mathbf{e}_1\teh + b \mathbf{e}_2\teh$ for nonzero $a, b \in \mathbb{C}$.
      From \cref{subsec:gc-case-a-2}, $a \mathbf{e}_1\teh + b \mathbf{e}_2\teh$ and $\mathbf{e}_1\teh + \mathbf{e}_2\teh$ are inter-reducible to each other. 
      So, we may assume without loss of generality, $\mathbf{F} = \mathbf{e}_1\teh + \mathbf{e}_2\teh$.

      In this subsection, we prove two results. 
      We first prove the dichotomy of $\holts(\mathbf{F}, \mathbf{G})$ when $\ann(\mathbf{G}) \ne 0$.
      Then, we prove that $\holts(\mathbf{F}, \mathbf{G})$ is \sph when $\mathbf{G}$ is \casecodz or \casecogen. 
       The intuition for this choice is because as we saw in \cref{subsec:gc-case-c1dz}, the shell for \casecodz signature is relatively coarse while its frame is relatively complex. 
      This makes it easier to analyze \casecodz signatures by anchoring on a signature of a different type.
      
      Note that here we do not cover all possible cases for $\mathbf{G}$, as
      the types \caseathree, \casebthree and \casecosep are not considered.
      The reason is that it is easier to prove the dichotomy for those cases when the \caseathree, \casebthree or \casecosep signature is in its defining form,
      and we will anchor the proof on such a signature. 
      Those cases are covered in their respective sections (\cref{subsec:two-ternary-dichotomy-a-3,subsec:two-ternary-dichotomy-b-3,subsec:two-ternary-dichotomy-C-1-sep}).
      
      \label{subsec:two-ternary-dichotomy-a-2}
      \begin{proposition}
        \label{prop:two-ternary-dichotomy-a-2}
        Let $\mathbf{F} = \mathbf{e}_1\teh + \mathbf{e}_2\teh$.
        Let $\mathbf{G}$ be a nondegenerate symmetric ternary signature such that 
        $\ann(\mathbf{G}) = \spn(\mathbf{u}) \ne 0$.
        Then, $\holts(\mathbf{F}, \mathbf{G})$ is \sph unless
        \begin{itemize}
          \item $\langle \mathbf{u}, \mathbf{e}_3 \rangle = 0$, or
          \item $\langle \mathbf{u}, \mathbf{e}_3 \rangle \ne 0$ and 
            $\mathbf{G} = (*, 0, *)\teh + (0, *, *)\teh$, in particular, $\mathbf{G}$ must be \caseatwo, \casebtwo, or \casecz.
        \end{itemize}
      \end{proposition}
      \begin{proof}
        Assume $\holts(\mathbf{F}, \mathbf{G})$ is not \sph. Then in addition to
        $\mathbf{F}$, $\mathbf{G}$  must take a tractable form.
        Now suppose $\langle \mathbf{u}, \mathbf{e}_3 \rangle \ne 0$.
        Then, $\mathbf{F}$ allows us to take a Boolean domain restriction to 
        $\db\dg$, and we have the reduction
        $\holbs([1, 0, 0, 1], \mathbf{G}\domres{\db\dg})
        \le_T \holts(\mathbf{F}, \mathbf{G})$. Thus the former is also
        not \sph. By the Boolean domain dichotomy,  $\mathbf{G}\domres{\db\dg}$ is either degenerate or has the form $[*, 0, 0, *]$.
        By \cref{prop:annihilator-uniquely-propagates-values}, since $\mathbf{G}$ is  nondegenerate it must be of the form $(a, 0, b)\teh + (0, x, y)\teh$ for $a, x \ne 0$.
         As $\mathbf{G}$ has a nontrivial
         annihilator $\mathbf{u}$, by \cref{prop:nonzero-annihilator-types},  it %
         must be one of \caseatwo, \casebtwo, \casecobg, or \casecz. 
         However, since $\mathbf{G}$ has symmetric rank $2$, by \cref{prop:symmetric-rank-sym-sum} it cannot be \casecobg.
         It follows that $\mathbf{G}$  
         must be one of \caseatwo, \casebtwo,  or \casecz. 
      \end{proof}

      The proof strategy we will use in 
      \cref{prop:a2-hard-with-casecodz-casecosep} 
      will be applied repeatedly in this paper, in the following way. 
      We have two individually tractable signatures $\mathbf{F}$ and $\mathbf{G}$,
      and $\mathbf{F}$ can realize a domain restriction to $\db\dg$
      with $\mathbf{F}\domres{\db\dg} = [1, 0, 0, 1]$. 
      If $\mathbf{G}$ is of tractable type $\mathfrak{G}$, then 
      $\holts(\mathbf{F}, \mathbf{G}) \ge_T \holbs([1, 0, 0, 1], (T \ag(\mathfrak{G}))\domres{\db\dg})$ for some orthogonal $T$ by \cref{lem:gadget-set-general-form}.
      From this we can apply \cref{prop:subspace-e-bg} on $(T \ag(\mathfrak{G}))\domres{\db\dg}$ to analyze when $\mathbf{G}$ can escape \sphness.

      \begin{proposition}
        \label{prop:a2-hard-with-casecodz-casecosep}
        Let $\mathbf{F} = \mathbf{e}_1\teh + \mathbf{e}_2\teh$.
        Let $\mathbf{G}$ be a nondegenerate symmetric ternary signature of \casecodz or \casecogen.
        Then, $\holts(\mathbf{F}, \mathbf{G})$ is \sph.
      \end{proposition}
      \begin{proof}      
        For a contradiction suppose  $\holts(\mathbf{F}, \mathbf{G})$ is not \sph.
        By \cref{lem:gadget-set-general-form} and \cref{subsec:gc-case-a-2}, $\holts(\mathbf{F}, \mathbf{G}) \ge_T \holbs([1, 0, 0, 1], (T \agx(\mathfrak{F}))\domres{\db\dg})$
        for the tractable type $\mathfrak{F} = $ \casecodz or \casecogen of $\mathbf{G}$ and some $T \in \so(3, \mathbb{C})$.
          Then, we can apply \cref{prop:subspace-e-bg}. 
        
        Let $\mathbf{G}$ be of type \casecodz.
        By \cref{lem:gadget-set-general-form} and \cref{subsec:gc-case-c1dz}, $\holts(\mathbf{F}, \mathbf{G}) \ge_T \holts(\mathbf{F}, \mathbf{G}, V) \ge_T \holbs([1, 0, 0, 1], V\domres{\db\dg})$ where
        \[
        V = T \ag(\casecodzmath, \lambda) = T \spn((\bzp)\teh, \sym(\bzp, \bzp,  \mathbf{e}_3), \sym(\bzp,  \bzp,  \bzm) + \lambda \sym( \bzp,  \bzp,  \mathbf{e}_3))
        \]
        for some $T \in 
         \so(3, \mathbb{C})$ and some nonzero $\lambda \in \mathbb{C}$.
        Now we use the (complex) spin representation, i.e., the parametrization of 
        $\so(3, \mathbb{C})$
         by $\speciallinear(2, \mathbb{C})$ in \cref{eq:isotropic-after-orthogonal}. We can write 
        \[
          V\domres{\db\dg} = \spn(\mathbf{u}_1\teh, \sym(\mathbf{u}_1, \mathbf{u}_1, \mathbf{u}_3), \sym(\mathbf{u}_1, \mathbf{u}_1, \mathbf{u}_2) + \lambda \sym(\mathbf{u}_1, \mathbf{u}_3, \mathbf{u}_3))
        \]
        where $\mathbf{u}_1 = (-c^2 + d^2, \i(c^2 + d^2))$ and $\mathbf{u}_2 = (a^2 - b^2, -\i(a^2 + b^2)), \mathbf{u}_3 = ((-ac + bd), \i(ac + bd))$ where $ad - bc = 1$.
        By \cref{prop:subspace-e-bg} on $\spn(\mathbf{u}_1\teh, \sym(\mathbf{u}_1, \mathbf{u}_1, \mathbf{u}_3))$, we have $\mathbf{u}_1 \sim \mathbf{u}_3$.
        Since it cannot be the case that $\mathbf{u}_1 = 0$ because $ad - bc = 1$,
        we have $V\domres{\db\dg} = \spn(\mathbf{u}_1\teh, \sym(\mathbf{u}_1, \mathbf{u}_1, \mathbf{u}_2))$.
        Applying \cref{prop:subspace-e-bg} again, we get $\mathbf{u}_1 \sim \mathbf{u}_2$.
        However, this is impossible because 
        the span of $\{T \bzp, T \bzm, T \mathbf{e}_3\}$ has dimension $3$.

        Now let $\mathbf{G}$ be of type \casecogen.
        Then, by \cref{lem:gadget-set-general-form} and \cref{subsec:gc-case-c1gen},
        for $V = T \ag(\casecogenmath)$ and some $T \in \so(3, \mathbb{C})$, 
        $\holts(\mathbf{F}, \mathbf{G}) \ge_T \holbs([1, 0, 0, 1], V\domres{\db\dg})$, where
        \[
          V\domres{\db\dg} = \spn(\mathbf{u}_1\teh, \sym(\mathbf{u}_1, \mathbf{u}_1, \mathbf{u}_2)) \cup 
          \bigcup_{x, y \in \mathbb{C}}
        \spn(\mathbf{u}_1\teh, (x \mathbf{u}_1 + y \mathbf{u}_3)\teh) \]
        with the same $\mathbf{u}_1, \mathbf{u}_2, \mathbf{u}_3$ as above.
        By \cref{prop:subspace-e-bg} on $\spn(\mathbf{u}_1\teh, \sym(\mathbf{u}_1, \mathbf{u}_1, \mathbf{u}_2))$, we have $\mathbf{u}_1 \sim \mathbf{u}_2$.
        Then, $\mathbf{u}_3$ must be linearly independent from $\mathbf{u}_1$, by a simple calculation  due to $ad - bc = 1$.
        We apply \cref{prop:subspace-e-bg} on $\spn(\mathbf{u}_1\teh, (x \mathbf{u}_1 + y \mathbf{u}_2)\teh)$. 
        We see that there exist some $x, y$ such that 
        $\{\mathbf{u}_1, x \mathbf{u}_1 + y \mathbf{u}_3\} \not\sim \{\sbto, \sbtt\}$, 
        and $\mathbf{u}_1  \not\sim  x \mathbf{u}_1 + y \mathbf{u}_3$,
        which implies \sphness. 
      \end{proof}

      \subsection{\texorpdfstring{\casebtwo}{B2}}
      \label{subsec:two-ternary-dichotomy-b-2}
    We may assume the given \casebtwo signature $\mathbf{F}$ is $a \bzp\teh + b \bzm \teh$ for some nonzero $a, b \in \mathbb{C}$.
      From \cref{subsec:gc-case-b-2}, we know that $a \bzp\teh + b \bzm\teh$ and $\bzp\teh + \bzm\teh$ are inter-reducible to each other. 
      So, we may assume without loss of generality, $\mathbf{F} = \bzp\teh + \bzm\teh$.

Similar to~\cref{subsec:A2},
we prove two results. 
      We first prove the dichotomy of $\holts(\mathbf{F}, \mathbf{G})$ when $\ann(\mathbf{G}) \ne 0$.
      Then, we prove that $\holts(\mathbf{F}, \mathbf{G})$ is \sph when $\mathbf{G}$ is \casecodz or \casecogen. 
      Again, here we do not cover all possible cases for $\mathbf{G}$, as
      the types \caseathree, \casebthree and \casecosep are not considered.
      The reason is the same as in~\cref{subsec:A2}.

      \begin{proposition}
        \label{prop:two-ternary-dichotomy-b-2}
        Let $\mathbf{F} = \bzp\teh + \bzm\teh$.
        Let $\mathbf{G}$ be a nondegenerate symmetric ternary signature such that $\ann(\mathbf{G}) = \spn(\mathbf{u}) \ne 0$.
        Then, $\holts(\mathbf{F}, \mathbf{G})$ is \sph unless
        \begin{itemize}
          \item $\langle \mathbf{u}, \mathbf{e}_3 \rangle = 0$, or
          \item $\langle \mathbf{u}, \mathbf{e}_3 \rangle \ne 0$ and 
            $\mathbf{G} = x(1, \i, *)\teh + y(1, -\i, *)\teh$ for some $x, y \in \mathbb{C}$.
        \end{itemize}
      \end{proposition}
      \begin{proof}
      Let $\iny{\mathbf{F}} = \Yinv \mathbf{F} = \mathbf{e}_1\teh + \mathbf{e}_2\teh$ and $\iny{\mathbf{G}} = \Yinv \mathbf{G}$.
      The proof is similar to \cref{prop:two-ternary-dichotomy-a-2} but we work in
      the $\holts(\X \mid \iny{\mathbf{F}}, \iny{\mathbf{G}})$ setting.
        Assume $\holts(\mathbf{F}, \mathbf{G})$ is not \sph. Then in addition to
         $\mathbf{F}$, $\mathbf{G}$  must take a tractable form.
    Now suppose $\langle \mathbf{u}, \mathbf{e}_3 \rangle \ne 0$.
    By \cref{subsec:gc-case-b-2}, we have $\holts(\mathbf{F}, \mathbf{G}) \ge_T \holbs([1, 0, 0, 1] \mid \iny{\mathbf{G}}\domres{\db\dg})$.
        By the bipartite Boolean domain dichotomy
        \cref{cor:bipartite-boolean-dichotomy-explicit},
        this forces $\iny{\mathbf{G}}\domres{\db\dg}$ to be degenerate or $[*, 0, 0, *]$.
        Note that the annihilator of $\iny{\mathbf{G}}$ is $\Y\transpose \mathbf{u}$ by \cref{prop:annihilating-vector-under-Y} and $\langle \Y\transpose \mathbf{u}, \mathbf{e}_3 \rangle = \langle \mathbf{u}, \Y \mathbf{e}_3 \rangle = \langle \mathbf{u}, \mathbf{e}_3 \rangle \ne 0$.
        By \cref{prop:annihilator-uniquely-propagates-values}, since $\iny{\mathbf{G}}$ is  nondegenerate,  $\iny{\mathbf{G}}\domres{\db\dg}$ is also nondegenerate
        and so it takes the form $[*, 0, 0, *]$. Then it
        follows by \cref{prop:annihilator-uniquely-propagates-values} again, that $\iny{\mathbf{G}}$ must be of the form $(*, 0, *)\teh + (0, *, *)\teh$. 
        We get the desired form for $\mathbf{G}$ by applying $\Y$ on $\iny{\mathbf{G}}$.
      \end{proof}
      
      \begin{proposition}
        \label{prop:b2-hard-with-casecodz-casecosep}
        Let $\mathbf{F} = \bzp \teh + \bzm \teh$.
        Let $\mathbf{G}$ be a nondegenerate symmetric ternary signature of type \casecodz or 
        \casecogen.
        Then, $\holts(\mathbf{F}, \mathbf{G})$ is \sph.
      \end{proposition}
      \begin{proof}
            Let $\iny{\mathbf{F}} = \Yinv \mathbf{F} = \mathbf{e}_1\teh + \mathbf{e}_2\teh$ and $\iny{\mathbf{G}} = \Yinv \mathbf{G}$.
      The proof is similar to \cref{prop:a2-hard-with-casecodz-casecosep} but we work in the $\holts(\X \mid \iny{\mathbf{F}}, \iny{\mathbf{G}})$ setting.
          Assume $\holts(\mathbf{F}, \mathbf{G})$ is not \sph.
          By  \cref{lem:gadget-set-general-form} and  \cref{subsec:gc-case-b-2}, 
        $\holts(\mathbf{F}, \mathbf{G}) \ge_T \holbs([1, 0, 0, 1] \mid (\iny{T} \agx(\mathfrak{F}))\domres{\db\dg})$
        for the tractable type $\mathfrak{F}  = $ \casecodz or \casecogen  of $\mathbf{G}$ and some $\iny{T} \in \so(3, \mathbb{C}, \X)$.
          Then, we can apply \cref{prop:subspace-e-bg}.

        Suppose $\mathbf{G}$ is \casecodz. 
        By \cref{lem:gadget-set-general-form} 
        and \cref{subsec:gc-case-c1dz}, 
        $\holts(\mathbf{F}, \mathbf{G}) \ge_T \holbs([1, 0, 0, 1] \mid V\domres{\db\dg})$
        where $V = \iny{T} \agx(\casecodzmath, \lambda)$ for some $\iny{T} \in \so(3, \mathbb{C}, \X)$ and $\lambda \ne 0$.
        By the (complex) spin representation of $\speciallinear(2, \mathbb{C})$ parameterizing $\so(3, \mathbb{C}, X)$ (\cref{eq:paramaterize-so3cx}), 
        we may write $\iny{T} = \rho(M)$ for 
        $M = \begin{bsmallmatrix}
            a & b \\ c & d
        \end{bsmallmatrix} \in \speciallinear(2, \mathbb{C})$.
        Then, we have 
        \[
            V\domres{\db\dg} = \spn(\mathbf{u}_1\teh, \sym(\mathbf{u}_1, \mathbf{u}_1, \mathbf{u}_3), \sym(\mathbf{u}_1, \mathbf{u}_1, \mathbf{u}_2) + \lambda \sym(\mathbf{u}_1, \mathbf{u}_3, \mathbf{u}_3))
        \]
        where
        $\mathbf{u}_1 = (d^2, -c^2)$,
        $\mathbf{u}_2 = (-b^2, a^2)$,
        $\mathbf{u}_3 = (bd, -ac)$.
        We apply \cref{prop:subspace-e-bg} on 
        the subspace $\spn(\mathbf{u}_1\teh, \sym(\mathbf{u}_1,\mathbf{u}_1,\mathbf{u}_3))$, and get $\mathbf{u}_1 \sim \mathbf{u}_3$ by \cref{eq:paramaterize-so3cx}.
        Since it cannot be the case that $\mathbf{u}_1 = 0$ because $ad - bc = 1$, we have $V\domres{\db\dg} = \spn(\mathbf{u}_1\teh, \sym(\mathbf{u}_1, \mathbf{u}_1, \mathbf{u}_2))$. 
        Applying \cref{prop:subspace-e-bg} again, we get $\mathbf{u}_1 \sim \mathbf{u}_2$.
        However, this is impossible because the span of $\{\rho(M) \mathbf{e}_1, \rho(M) \mathbf{e}_2, \rho(M) \mathbf{e}_3\}$ has dimension $3$.

        Suppose $\mathbf{G}$ is \casecogen.
        Then, by \cref{lem:gadget-set-general-form} and \cref{subsec:gc-case-c1gen} and the same reasoning, 
        we have 
        $\holts(\mathbf{F}, \mathbf{G}) \ge_T \holbs([1, 0, 0, 1] \mid V\domres{\db\dg})$ for $V = \rho(M) \agx(\casecogenmath)$ with
        \[V\domres{\db\dg} = \spn(\mathbf{u}_1\teh, \sym(\mathbf{u}_1, \mathbf{u}_1, \mathbf{u}_2)) \cup \bigcup_{x, y \in \mathbb{C}} \spn(\mathbf{u}_1\teh, (x \mathbf{u}_1 + y \mathbf{u}_3))\teh )\] 
        for the same $\mathbf{u}_1, \mathbf{u}_2, \mathbf{u}_3$ as above. 
        By \cref{prop:subspace-e-bg} on $ \spn(\mathbf{u}_1\teh, \sym(\mathbf{u}_1, \mathbf{u}_1, \mathbf{u}_2)) $, we have $\mathbf{u}_1 \sim \mathbf{u}_2$.
        Then, $\mathbf{u}_3$ must be linearly independent from $\mathbf{u}_1$.
                We apply \cref{prop:subspace-e-bg} on $\spn(\mathbf{u}_1\teh, (x \mathbf{u}_1 + y \mathbf{u}_2)\teh)$. 
        We see that there exist some $x, y$ such that 
        $\{\mathbf{u}_1, x \mathbf{u}_1 + y \mathbf{u}_3\} \not\sim \{\sbto, \sbtt\}$, 
        and $\mathbf{u}_1  \not\sim  x \mathbf{u}_1 + y \mathbf{u}_3$, 
        which implies \sphness. 
      \end{proof}

      \subsection{\texorpdfstring{\casecobg}{C1bg}}
      \label{subsec:two-ternary-dichotomy-C-1-bg}
      Let $\mathbf{F}$ be a \casecobg signature such that $\Yinv \mathbf{F} = a \mathbf{e}_1\teh + \lambda \sym(\mathbf{e}_1, \mathbf{e}_1, \mathbf{e}_2)$ for some $a, \lambda \in \mathbb{C}$ with $\lambda \ne 0$.
      By \cref{subsec:gc-case-c1-bg}, we may assume without loss of generality that $a = 0$ and $\lambda = 1$,
      so $\mathbf{F} = \sym(\bzp, \bzp, \bzm)$.

      Similar to the last subsection, we prove two results. 
      We first prove the dichotomy of $\holts(\mathbf{F}, \mathbf{G})$ when $\ann(\mathbf{G}) \ne 0$.
      Then, we prove the dichotomy of $\holts(\mathbf{F}, \mathbf{G})$ when $\mathbf{G}$ is \casecodz or \casecogen. 
      Again, here we do not cover all possible cases for $\mathbf{G}$, as
      the types \caseathree, \casebthree and \casecosep are not considered.
      The reason is the same as in \cref{subsec:A2}.
      
      \begin{proposition}
        \label{prop:two-ternary-dichotomy-C-1-bg}

        Let $\mathbf{F} = \sym(\bzp, \bzp, \bzm)$.
        Let $\mathbf{G}$ be a nondegenerate symmetric ternary signature such that $\ann(\mathbf{G}) = \spn(\mathbf{u}) \ne 0$.
        Then, $\holts(\mathbf{F}, \mathbf{G})$ is \sph unless
        \begin{itemize}
          \item $\langle \mathbf{u}, \mathbf{e}_3 \rangle = 0$, or
          \item $\langle \mathbf{u}, \mathbf{e}_3 \rangle \ne 0$ and 
            $\mathbf{G} = \sym(\mathbf{v}, \mathbf{v}, \mathbf{w}) + \mathbf{v}\teh$,
            for some $\mathbf{v},  \mathbf{w} $ such that $\langle \bzp, \mathbf{v} \rangle = 0$.
        \end{itemize}
      \end{proposition}
            \begin{proof}
            Let $\iny{\mathbf{G}} = \Yinv \mathbf{G}$.
        Assume $\holts(\mathbf{F}, \mathbf{G})$ is not \sph. Then in addition to
         $\mathbf{F}$, $\mathbf{G}$  must take a tractable form.
        By \cref{subsec:gc-case-c1-bg}, 
        $\holts(\mathbf{F}, \mathbf{G}) \ge_T \holbs([0, 0, 1, 0] \mid \iny{\mathbf{G}}\domres{\db\dg})$, so 
        $\iny{\mathbf{G}}\domres{\db\dg}$ must be degenerate or $[*, *, 0, 0]$ by \cref{cor:bipartite-boolean-dichotomy-explicit}.
        Note that the annihilator of $\iny{\mathbf{G}}$ is $\Y\transpose \mathbf{u}$ by \cref{prop:annihilating-vector-under-Y} and $\langle \Y\transpose \mathbf{u}, \mathbf{e}_3 \rangle = \langle \mathbf{u}, \Y \mathbf{e}_3 \rangle = \langle \mathbf{u}, \mathbf{e}_3 \rangle \ne 0$.
        By \cref{prop:annihilator-uniquely-propagates-values}, since $\iny{\mathbf{G}}$ is  nondegenerate it must be of the form 
        given in \cref{prop:annihilator-uniquely-propagates-values}, 
         i.e.,  of the form $(*,0,*)\teh + \sym((*,0,*), (*,0,*), (0, *,*))$. 
        We get the desired form for $\mathbf{G}$ by applying $\Y$ on $\iny{\mathbf{G}}$.
      \end{proof}

      \begin{proposition}
        \label{prop:c1bg-vs-c1codz-c1gen}
        Let $\mathbf{F} = \sym(\bzp, \bzp, \bzm)$.
        Let $\mathbf{G}$ be a nondegenerate symmetric ternary signature of type \casecodz or \casecogen.
        Then, $\holts(\mathbf{F}, \mathbf{G})$ is \sph unless 
        $\Yinv\mathbf{G} \in \casecsupp$.
      \end{proposition}
      \begin{proof}
      Let $\iny{\mathbf{G}} = \Yinv \mathbf{G}$.
      The proof is similar to \cref{prop:a2-hard-with-casecodz-casecosep,prop:b2-hard-with-casecodz-casecosep}.
      Assume $\holts(\mathbf{F}, \mathbf{G})$ is not \sph.
      By  \cref{lem:gadget-set-general-form} and \cref{subsec:gc-case-c1-bg}, $\holts(\mathbf{F}, \mathbf{G}) \ge_T \holbs([0, 1, 0, 0] \mid \iny{T} \agx(\mathfrak{F}))$ for some $\iny{T} \in \so(3, \mathbb{C}, \X)$, where $\mathfrak{F}$ is the tractable type of $\mathbf{G}$. 
      Then, we apply \cref{prop:subspace-m-bg}.
          
          Suppose $\mathbf{G}$ is \casecodz. 
          Then, there exists some $\iny{T} \in \so(3, \mathbb{C}, \X)$ such that $\iny{T}^{-1} \iny{\mathbf{G}}$ has
          the form $\iny{T}^{-1} \iny{\mathbf{G}} \sim [x; 1, y; 0, 0, z; 0, 0, 0, 0] \in \casecsupp$ for some $x, y, z$ such that $z \ne 0$.
          By \cref{lem:gadget-set-general-form} and \cref{subsec:gc-case-c1dz}, 
          $\holts(\mathbf{F}, \mathbf{G}) \ge_T \holbs([0, 1, 0, 0] \mid V\domres{\db\dg})$ where
          $V = \iny{T} \agx(\casecodzmath, \lambda)$ for some $\lambda \ne 0$.
          By the (complex) spin representation of $\speciallinear(2, \mathbb{C})$ parameterizing $\so(3, \mathbb{C}, X)$ (\cref{eq:paramaterize-so3cx}),
          we may write 
          $\iny{T} = \rho(M)$ for $M = \begin{bsmallmatrix}
              a & b \\ c & d
          \end{bsmallmatrix} \in \speciallinear(2, \mathbb{C})$. 
        Then, we have 
        \[
            V \domres{\db\dg} = \spn(\mathbf{u}_1\teh, \sym(\mathbf{u}_1, \mathbf{u}_1, \mathbf{u}_3), \sym(\mathbf{u}_1, \mathbf{u}_1, \mathbf{u}_2) + \lambda \sym(\mathbf{u}_1, \mathbf{u}_3, \mathbf{u}_3))
        \]
        for
        $\mathbf{u}_1 = (d^2, -c^2)$,
        $\mathbf{u}_2 = (-b^2, a^2)$,
        $\mathbf{u}_3 = (bd, -ac)$.
        Applying \cref{prop:subspace-m-bg} on $\spn(\mathbf{u}_1, \sym(\mathbf{u}_1, \mathbf{u}_1, \mathbf{u}_3))$ gives $\mathbf{u}_1 \sim \sbto$ or $\mathbf{u}_1 \sim \mathbf{u}_3$.
        If $\mathbf{u}_1 \sim \sbto$, then $c = 0$, so $\iny{\mathbf{G}} \in \casecsupp$ by \cref{prop:casecsupp-fixed-under-rho-c-0}.
        If $\mathbf{u}_1 \sim \mathbf{u}_3$, then we see that $V\domres{\db\dg} = \spn(\mathbf{u}_1\teh, \sym(\mathbf{u}_1, \mathbf{u}_1, \mathbf{u}_2))$.
        Applying \cref{prop:subspace-m-bg} again gives 
        $\mathbf{u}_1 \sim \sbto$ or $\mathbf{u}_1 \sim \mathbf{u}_2$.
        But it cannot be the case that $\mathbf{u}_1 \sim \mathbf{u}_2$ since the span of $\{\rho(M) \mathbf{e}_1, \rho(M) \mathbf{e}_2, \rho(M) \mathbf{e}_3\}$ has dimension $3$.

        Suppose $\mathbf{G}$ is a \casecogen signature. 
        Then, there exists some $\iny{T} \in \so(3, \mathbb{C}, \X)$ such that 
        $\iny{T}^{-1} \iny{\mathbf{G}} \sim [x; 1, y; 0, 0, z; 0, 0, 0, w] \in \casecsupp$ for some $x, y, z, w$ such that $w, z^2 - z - y w \ne 0$. 
        By \cref{lem:gadget-set-general-form} and \cref{subsec:gc-case-c1gen} and the same reasoning, 
        $\holts(\mathbf{F}, \mathbf{G}) \ge_T \holbs([0, 1, 0, 0] \mid V\domres{\db\dg})$ for 
        $V = \rho(M) \agx(\casecogenmath)$ with
        \[
          V\domres{\db\dg} = \spn(\mathbf{u}_1\teh, \sym(\mathbf{u}_1, \mathbf{u}_1, \mathbf{u}_2)) \cup \bigcup_{p, q \in \mathbb{C}} \spn(\mathbf{u}_1\teh, (p \mathbf{u}_1 + q \mathbf{u}_3)\teh)
        \]
        for the same $\mathbf{u}_1, \mathbf{u}_2, \mathbf{u}_3$ as above.
        Applying \cref{prop:subspace-m-bg} on
        $\spn(\mathbf{u}_1\teh, \sym(\mathbf{u}_1, \mathbf{u}_1, \mathbf{u}_2)$ gives 
        $\mathbf{u}_1 \sim \sbto$ or $\mathbf{u}_1 \sim \mathbf{u}_2$.
        If $\mathbf{u}_1 \sim \sbto$, then $c = 0$ and $\iny{\mathbf{G}} \in \casecsupp$ by \cref{prop:casecsupp-fixed-under-rho-c-0}.
        Otherwise, if $\mathbf{u}_1 \sim \mathbf{u}_2$, then $\mathbf{u}_1 \not \sim \mathbf{u}_3$, and applying \cref{prop:subspace-m-bg} on $\spn(\mathbf{u}_1\teh, (p \mathbf{u}_1 + q \mathbf{u}_3)\teh)$ for any nonzero $p, q$ gives \sphness.
      \end{proof}

      \subsection{\texorpdfstring{\casecz}{C0}}
      \label{subsec:two-ternary-dichotomy-C0}
      Let $\mathbf{F}$ be a \casecz signature such that $\ann(\mathbf{F}) = \spn(\bzp)$.

      Similar to the last subsection, we prove two results. 
      We first prove the dichotomy of $\holts(\mathbf{F}, \mathbf{G})$ when $\ann(\mathbf{G}) \ne 0$.
      Then, we prove the dichotomy of $\holts(\mathbf{F}, \mathbf{G})$ when $\mathbf{G}$ is \casecodz or \casecogen. 
      Again, here we do not cover all possible cases for $\mathbf{G}$, as
      the types \caseathree, \casebthree and \casecosep are not considered.
        The reason is the same as in \cref{subsec:A2}.
      
      \begin{proposition}
        \label{prop:two-ternary-dichotomy-C-0}
        Let $\mathbf{F}$ and $\mathbf{G}$ be nondegenerate symmetric ternary signatures such that $\ann(\mathbf{F}) = \spn(\bzp)$
        and $\ann(\mathbf{G}) = \spn(\mathbf{u}) \ne 0$.
        Then, $\holts(\mathbf{F}, \mathbf{G})$ is \sph unless
        \begin{enumerate}
          \item $\langle \mathbf{u}, \bzp \rangle = 0$, or
          \item $\langle \mathbf{u}, \bzp \rangle \ne 0$ and $\mathbf{F} = \mathbf{u}_1\teh + \mathbf{v}_1\teh$ and $\mathbf{G} = \mathbf{u}_2\teh + \mathbf{v}_2\teh$ for some $\mathbf{u}_i, \mathbf{v}_i$ such that 
          $\langle \mathbf{u}_1, \mathbf{u}_2 \rangle = 0$ and $\langle \mathbf{v}_1, \mathbf{v}_2 \rangle = 0$
            , or
          \item $\langle \mathbf{u}, \bzp \rangle \ne 0$ and $\mathbf{F}= \lambda_1 \mathbf{u}_1\teh + \sym(\mathbf{u}_1, \mathbf{u}_1, \mathbf{v}_1)$ and $\mathbf{G} = \lambda_2 \mathbf{u}_2\teh + \sym(\mathbf{u}_2, \mathbf{u}_2, \mathbf{v}_2)$  for some $\lambda_i, \mathbf{u}_i, \mathbf{v}_i$, such that $\langle \mathbf{u}_1, \mathbf{u}_2 \rangle = 0$.
    \end{enumerate}
  \end{proposition}
  \begin{proof}
  Having a nonzero annihilator  $\mathbf{u}$ implies that
  $\mathbf{G}$ has type
  \caseatwo, \casebtwo, \casecz or \casecobg.
  We only need to consider the case when $\mathbf{G}$ is also \casecz 
    because the cases when one of the signatures is \caseatwo, \casebtwo, and \casecobg were already shown (\cref{prop:two-ternary-dichotomy-a-2,prop:two-ternary-dichotomy-b-2,prop:two-ternary-dichotomy-C-1-bg}). 
    
    Suppose $\langle \mathbf{u}, \bzp \rangle \ne 0$.
    If $\mathbf{u}$  is isotropic and $\langle \mathbf{u}, \bzp \rangle \ne 0$, then $\mathbf{u} \not \sim \bzp$, so by Lemma 3.5 (arXiv version) in~\cite{domain3-single-signature}, there exists some orthogonal $T$ such that $T \bzp \sim \bzp$ and $T \mathbf{u} \sim \bzm$.
    Then, note that $\Yinv T \mathbf{G}$ is supported on $\dg\dr$ because $\ann(\Yinv T \mathbf{G}) = ((\Yinv T)^{-1})\transpose \spn(\mathbf{u}) = \Y\transpose T \spn(\mathbf{u}) = \Y\transpose \spn(\bzm) = \spn(\mathbf{e}_1)
    $ 
    by \cref{prop:annihilating-vector-under-gl-transform}, and thus 
    $\X \Yinv T \mathbf{G}$ is supported on $\db\dr$.
    Since $\Yinv T \mathbf{F}$ is supported on $\db\dr$ as well, we have
    \begin{align*}
        \holts(\mathbf{F}, \mathbf{G}) &=_T \holts(T \mathbf{F}, T \mathbf{G}) \\
        &=_T \holts(\X \mid \Yinv T \mathbf{F}, \Yinv T \mathbf{G}) \\
        &\ge_T \holts( \X \Yinv T \mathbf{G} \mid \Yinv T \mathbf{F}) \\
        &=_T \holbs( (\X \Yinv T \mathbf{G})\domres{\db\dr} \mid (\Yinv T \mathbf{F})\domres{\db\dr}) \\
        &=_T \holbs((\Yinv T \mathbf{G})\domres{\dg\dr} \mid (\Yinv T \mathbf{F})\domres{\db\dr}) \,.
    \end{align*}
    We obtain the conclusions of the Proposition
    by applying \cref{cor:bipartite-boolean-dichotomy-explicit}.
    For example, one possibility is $\mathbf{F}' = \Yinv T \mathbf{F}$ 
    has the form $\begin{bsmallmatrix}
              a_1 \\ 0 \\ b_1 
              \end{bsmallmatrix}\teh
              +
              \begin{bsmallmatrix}
              a_2 \\ 0 \\ b_2
          \end{bsmallmatrix}\teh$,
          and $\mathbf{G}' = \Yinv T \mathbf{G} =
          \begin{bsmallmatrix}
              0 \\c_1  \\ d_1 
              \end{bsmallmatrix}\teh
              +
              \begin{bsmallmatrix}
              0 \\ c_2 \\ d_2
          \end{bsmallmatrix}\teh$,
          satisfying
          $a_1c_1 + b_1 d_1 =0$ and $a_2c_2 + b_2 d_2 =0$.
          Then,
           $\mathbf{F} = T^{-1} \Y \mathbf{F}'$ and $\mathbf{G} = T^{-1} \Y  \mathbf{G}'$
           have the corresponding forms.
           The respective dot product of $ T^{-1} \Y \begin{bsmallmatrix}
              a_1 \\ 0 \\ b_1 
              \end{bsmallmatrix}$ with
             $T^{-1} \Y \begin{bsmallmatrix}
              0 \\ c_1 \\ d_1 
              \end{bsmallmatrix}$ 
              is 
              $\langle \begin{bsmallmatrix}
              a_1 \\ 0 \\ b_1 
              \end{bsmallmatrix}, \X 
              \begin{bsmallmatrix}
              0 \\ c_1 \\ d_1 
              \end{bsmallmatrix} \rangle
               = a_1c_1 + b_1 d_1 =0$, and similarly
               for the other dot product
               $\langle \begin{bsmallmatrix}
              a_2 \\ 0 \\ b_2 
              \end{bsmallmatrix}, \X 
              \begin{bsmallmatrix}
              0 \\ c_2 \\ d_2 
              \end{bsmallmatrix} \rangle
              =0$.
  \end{proof}
   Note that if $\mathbf{F}$ 
   cannot be written in the form of $\mathbf{v}\teh + \mathbf{w}\teh$ or $\lambda \mathbf{v}\teh + \sym(\mathbf{v}, \mathbf{v}, \mathbf{w})$ for any $\mathbf{v}, \mathbf{w} \in \mathbb{C}^3$,
   then for $\holts(\mathbf{F}, \mathbf{G})$ to be not \sph, it must be the case that $\langle  \mathbf{u}, \bzp \rangle = 0$.

  \begin{proposition}\label{prop:c0-vs-c1dz}
   Let $\mathbf{F}$ and $\mathbf{G}$ be nondegenerate symmetric ternary signatures such that $\ann(\mathbf{F}) = \spn(\bzp)$
        and $\mathbf{G}$ of type \casecodz or \casecogen.
      Then $\holts(\mathbf{F}, \mathbf{G})$ is \sph unless $\Yinv \mathbf{G} \in \casecsupp$.
  \end{proposition}
  \begin{proof}
    Assume $\holts(\mathbf{F}, \mathbf{G})$ is not \sph.
      Let $\iny{\mathbf{F}} = \Yinv \mathbf{F}$ and $\iny{\mathbf{G}} = \Yinv \mathbf{G}$. 
      
      Suppose $\mathbf{G}$ is \casecodz.
      Then there exists some $\iny{T} \in \so(3, \mathbb{C}, \X)$ such that the signature $\iny{T}^{-1} \iny{\mathbf{G}} \sim [x; 1, y; 0, 0, z; 0, 0, 0, 0] \in \casecsupp$ for some $x, y, z$ such that $z \ne 0$.
      Note that $\iny{\mathbf{F}}$ is supported on $\db\dr$ since  $\ann(\iny{\mathbf{F}}) = \Y\transpose \spn(\bzp)
      =\spn(\mathbf{e}_2)$, 
      and thus $\X \iny{\mathbf{F}}$ is supported on $\dg\dr$.
      We note that $(\X \iny{\mathbf{F}})\domres{\dg\dr} = (\iny{\mathbf{F}})\domres{\db\dr}$ is nondegenerate,
      for otherwise $\iny{\mathbf{F}}$ has the form $(*, 0, *)\teh$ and would be degenerate, and so would be ${\mathbf{F}}$.
      By \cref{lem:gadget-set-general-form} and \cref{subsec:gc-case-c1dz}, we have
      \begin{align*}
          \holts(\mathbf{F}, \mathbf{G}) &=_T \holts(\X \mid \iny{\mathbf{F}}, \iny{\mathbf{G}}) \\
          &\ge_T \holts(\X \mid \iny{\mathbf{F}}, V) \\
          &\ge_T \holts(\X\iny{\mathbf{F}} \mid V) \\
          &\ge_T \holbs((\X \iny{\mathbf{F}})\domres{\dg\dr} \mid V\domres{\dg\dr})
      \end{align*}
              where $V = \iny{T} \agx(\casecodzmath, \lambda)$ for some $\lambda \ne 0$.
        By the (complex) spin representation of $\speciallinear(2, \mathbb{C})$ parameterizing $\so(3, \mathbb{C}, X)$ (\cref{eq:paramaterize-so3cx}), 
        we may write $\iny{T} = \rho(M)$ for 
        $M = \begin{bsmallmatrix}
            a & b \\ c & d
        \end{bsmallmatrix} \in \speciallinear(2, \mathbb{C})$.
        Then, we have 
        \[
            V\domres{\dg\dr} = \spn(\mathbf{u}_1\teh, \sym(\mathbf{u}_1, \mathbf{u}_1, \mathbf{u}_3), \sym(\mathbf{u}_1, \mathbf{u}_1, \mathbf{u}_2) + \lambda \sym(\mathbf{u}_1, \mathbf{u}_3, \mathbf{u}_3))
        \]
        for
        $\mathbf{u}_1 = (-c^2, \sqrt{2}cd)$,
        $\mathbf{u}_2 = (a^2, -\sqrt{2} a b)$
        and 
        $\mathbf{u}_3 = (-\sqrt{2} ac, ad + bc)$.
        If $c = 0$, then by \cref{prop:casecsupp-fixed-under-rho-c-0}, $\iny{\mathbf{G}} \in \casecsupp$, and we are done.
        Now assume $c \ne 0$. 
        Then, note that $\mathbf{u}_1 \not \sim \mathbf{u}_3$ since the determinant of the matrix $[\mathbf{u}_1, \mathbf{u}_3]$ is $-c^2(a d + b c) +2 a c^2 d = a c^2 d - b c^3 = c^2 (ad - bc) = c^2 \ne 0$.
        Then, $\{\mathbf{u}_1, \mathbf{u}_3\}$ forms a basis of $\mathbb{C}^2$, so
        $\mathbf{u}_2 = p \mathbf{u}_1 + q \mathbf{u}_3$ for some $p, q \in \mathbb{C}$.
        Then $\sym(\mathbf{u}_1, \mathbf{u}_1, \mathbf{u}_2) = 3 p \mathbf{u}_1\teh + q \sym(\mathbf{u}_1, \mathbf{u}_1, \mathbf{u}_3)$.
        Therefore, $V\domres{\dg\dr} = \spn(\mathbf{u}_1\teh, \sym(\mathbf{u}_1, \mathbf{u}_1, \mathbf{u}_3), \sym(\mathbf{u}_1, \mathbf{u}_3, \mathbf{u}_3))$. 
        Then, in $V\domres{\dg\dr}$, there are Boolean domain signatures $\iny{\mathbf{H}}_1 = \sym(\mathbf{u}_1, \mathbf{u}_1, \mathbf{u}_3)$ and $\iny{\mathbf{H}}_2 = \sym(\mathbf{u}_1, \mathbf{u}_3, \mathbf{u}_3)$.
        For $\holbs((\X \iny{\mathbf{F}})\domres{\dg\dr} \mid \iny{\mathbf{H}}_1)$ to be not \sph, \cref{cor:bipartite-boolean-dichotomy-explicit} implies that
        $(\X \iny{\mathbf{F}})\domres{\dg\dr} = \lambda_1 \mathbf{u}_1\teh + \sym(\mathbf{u}_1, \mathbf{u}_1, \mathbf{w}_1)$ for some $\lambda_1 \in \mathbb{C}$ and $\mathbf{w}_1 \in \mathbb{C}^3$ such that $\mathbf{w}_1 \not \sim \mathbf{u}_1$ because $(\X \iny{\mathbf{F}})\domres{\dg\dr} $ is nondegenerate.
        Similarly, 
        for $\holbs((\X \iny{\mathbf{F}})\domres{\dg\dr} \mid \iny{\mathbf{H}}_2)$ to be not \sph, \cref{cor:bipartite-boolean-dichotomy-explicit} implies that
        $(\X \iny{\mathbf{F}})\domres{\dg\dr} = \lambda_2 \mathbf{u}_3\teh + \sym(\mathbf{u}_3, \mathbf{u}_3, \mathbf{w}_2)$ for some $\lambda_2 \in \mathbb{C}$ and $\mathbf{w}_2 \in \mathbb{C}^3$ such that $\mathbf{w}_2 \not \sim \mathbf{u}_3$. 
        However, $\lambda_1 \mathbf{u}_1\teh + \sym(\mathbf{u}_1, \mathbf{u}_1, \mathbf{w}_1) \ne  \lambda_2 \mathbf{u}_3\teh + \sym(\mathbf{u}_3, \mathbf{u}_3, \mathbf{w}_2)$ by \cref{prop:uniqueness-sym-form} because $\mathbf{u}_1 \not \sim \mathbf{u}_3$.
        Therefore, $\holts(\mathbf{F}, \mathbf{G})$ is \sph in this case.

        Suppose $\mathbf{G}$ is a \casecogen signature. 
        Then, there exists some $\iny{T} \in \so(3, \mathbb{C}, \X)$ such that 
        $\iny{T}^{-1} \iny{\mathbf{G}} \sim [x; 1, y; 0, 0, z; 0, 0, 0, w] \in \casecsupp$ for some $x, y, z, w$ such that $w, z^2 - z - y w \ne 0$. 
        By \cref{lem:gadget-set-general-form} and \cref{subsec:gc-case-c1gen} and the same reasoning, 
        $\holts(\mathbf{F}, \mathbf{G}) \ge_T \holbs( (\X \iny{\mathbf{F}})\domres{\dg\dr} \mid V\domres{\dg\dr})$ for 
        $V = \rho(M) \agx(\casecogenmath)$ where
        \[
          V\domres{\dg\dr} = \spn(\mathbf{u}_1\teh, \sym(\mathbf{u}_1, \mathbf{u}_1, \mathbf{u}_2)) \cup \bigcup_{x, y \in \mathbb{C}} \spn(\mathbf{u}_1\teh, (x \mathbf{u}_1 + y \mathbf{u}_3)\teh)
        \]
        for the same $\mathbf{u}_1, \mathbf{u}_2, \mathbf{u}_3$ as above.
        Again, if $c = 0$, then $\iny{\mathbf{G}} \in \casecsupp$ by \cref{prop:casecsupp-fixed-under-rho-c-0}, so we may assume $c \ne 0$. 
        By the same argument as  above, $\mathbf{u}_1 \not \sim \mathbf{u}_3$. 
        Let $\iny{\mathbf{H}}_{p, q} = \mathbf{u}_1\teh + (p \mathbf{u}_1 + q \mathbf{u}_3)\teh \in V\domres{\dg\dr}$.
        For $\holbs( (\X \iny{\mathbf{F}})\domres{\dg\dr} \mid \iny{\mathbf{H}}_{0, 1})$ to be not \sph, we need $(\X \iny{\mathbf{F}})\domres{\dg\dr} = \mathbf{v}_1\teh + \mathbf{v}_2\teh$ such that $\langle \mathbf{v}_1, \mathbf{u}_1 \rangle = 0$ and $\langle \mathbf{v}_2, \mathbf{u}_3 \rangle = 0$ by \cref{cor:bipartite-boolean-dichotomy-explicit}.
        Since $(\X \iny{\mathbf{F}})\domres{\dg\dr}$ is nondegenerate, $\mathbf{v}_1 \not \sim \mathbf{v}_2$. 
        For $\holbs( (\X \iny{\mathbf{F}})\domres{\dg\dr} \mid \iny{\mathbf{H}}_{1, 1})$ to be not \sph, we need $\langle \mathbf{v}_2, \mathbf{u}_1 \rangle = 0$ or $\langle \mathbf{v}_2, \mathbf{u}_1 + \mathbf{u}_3 \rangle = 0$ by \cref{cor:bipartite-boolean-dichotomy-explicit}.
        Both imply that $\langle \mathbf{v}_2, \mathbf{u}_1 \rangle = 0$ because $\langle \mathbf{v}_2, \mathbf{u}_3 \rangle = 0$, but this implies $\mathbf{v}_2 \sim \mathbf{v}_1$ because the dimension of the kernel of map $\mathbb{C}^2 \to \mathbb{C}^2$ by $\mathbf{v} \mapsto \langle \mathbf{v}, \mathbf{u}_1\rangle$ is $1$. 
  \end{proof}

      \subsection{\texorpdfstring{\caseathree}{A3}}
      Assume we have an \caseathree signature $\mathbf{F} = a \mathbf{e}_1\teh + b \mathbf{e}_2\teh + c \mathbf{e}_3\teh$ for nonzero $a, b, c \in \mathbb{C}$.
      From \cref{subsec:gc-case-a-3}, we know that $a \mathbf{e}_1\teh + b \mathbf{e}_2\teh + c \mathbf{e}_3\teh$ and $\mathbf{e}_1\teh + \mathbf{e}_2\teh + \mathbf{e}_3\teh$ are inter-reducible to each other. 
      So, we may assume without loss of generality, $\mathbf{F} = \mathbf{e}_1\teh + \mathbf{e}_2\teh + \mathbf{e}_3\teh$.
      
      \label{subsec:two-ternary-dichotomy-a-3}
      \begin{proposition}
        \label{prop:two-ternary-dichotomy-a-3}
        Let $\mathbf{F} = \mathbf{e}_1\teh + \mathbf{e}_2\teh + \mathbf{e}_3\teh$.
        Let $\mathbf{G}$ be a nondegenerate symmetric ternary signature.
        Then, $\holts(\mathbf{F}, \mathbf{G})$ is \sph unless
        \begin{itemize}
          \item $\mathbf{G} \in \mathcal{E}$, or
          \item $\mathbf{G}$ is $ij|k$ and degenerate on $ij$ for some $i, j, k \in \{\db,\dg,\dr\}$. 

          In particular, $\mathbf{G}$ is $(a, b, 0)\teh + c \mathbf{e}_3\teh$, $(a, 0, b)\teh + c \mathbf{e}_2\teh$ or $(0, a, b)\teh + c \mathbf{e}_1\teh$ for some $a, b, c \in \mathbb{C}$ such that $c \ne 0$.
          If further $\mathbf{G} \notin \mathcal{E}$, then $a, b \ne 0$.
        \end{itemize}
      \end{proposition}
      \begin{proof}
      Assume $\holts(\mathbf{F}, \mathbf{G})$ is not \sph.
        We first consider the case when $\mathbf{G}$ has a nonzero annihilator. 
        By \cref{prop:nonzero-annihilator-types}, this means $\mathbf{G}$ is one of \caseatwo, \casebtwo, \casecobg, or \casecz. 
        We will show that in particular, such $\mathbf{G}$ must be \caseatwo or \casecz
        and take the forms given above
        to escape \sphness.
        
        Suppose $\ann(\mathbf{G}) = \spn(\mathbf{u}) \ne 0$.
        Since $\mathbf{F}$ is invariant under permutation of the domain, we may assume that $\langle \mathbf{u}, \mathbf{e}_3 \rangle \ne 0$.
        By \cref{subsec:gc-case-a-3}, $\holts(\mathbf{F}, \mathbf{G}) \ge_T \holts(\mathbf{e}_1\teh + \mathbf{e}_2\teh, \mathbf{G})$.
        Then, by \cref{prop:two-ternary-dichotomy-a-2},
        $\mathbf{G}$ must be of the form $(a, 0, b)\teh + (0, x, y)\teh$ for some $a, b, x, y \in \mathbb{C}$,
        with $a,  x \ne 0$, 
        and $\mathbf{G}$ does not have type \casecobg.
        By \cref{subsec:gc-case-a-3}, $\holts(\mathbf{F}, \mathbf{G}) \ge_T \holts(\mathbf{e}_2\teh + \mathbf{e}_3\teh, \mathbf{G}) \ge_T \holbs([1, 0, 0, 1], \mathbf{G}\domres{\dg\dr})$, so we can apply the Boolean domain dichotomy. 
        As $\mathbf{G}\domres{\dg\dr} = (0, b)\teh + (x, y)\teh$, 
        and $x \ne 0$,
        it must be that $b = 0$ or $y = 0$ to escape \sphness,
        from which we get the desired domain separated form $ij|k$
        for some  $i, j, k \in \{\db,\dg,\dr\}$.
        In particular, we see that $\mathbf{G}$ cannot be \casebtwo
        , since one of $(a, 0, b)$ and $(0, x, y)$ must be not isotropic.

        Now, we may assume that $\ann(\mathbf{G}) = 0$, so the remaining cases are when $\mathbf{G}$ is \caseathree, \casebthree, \casecosep, \casecodz or \casecogen. 

 Suppose $\mathbf{G}$ is \caseathree.
        Then, by \cref{lem:gadget-set-general-form} and \cref{subsec:gc-case-a-3}, $\holts(\mathbf{F}, \mathbf{G}) \ge_T \holts(\mathbf{F}, V)$ where $V = T \ag(\caseathreemath) = \spn((T \mathbf{e}_1)\teh, (T \mathbf{e}_2)\teh, (T \mathbf{e}_3)\teh)$ for some $T \in \so(3, \mathbb{C})$.
        For any permutation $\{i, j, k \} = \{1, 2, 3\}$, there is a \caseatwo signature $\mathbf{G}_{ij} = (T \mathbf{e}_i)\teh + (T \mathbf{e}_j)\teh \in V$ 
        with $\ann(\mathbf{G}_{ij}) = \spn(T \mathbf{e}_k)$.
        There must exist some $k$ such that $\langle T \mathbf{e}_k, \mathbf{e}_3 \rangle \ne 0$, so fix such $\{i, j, k\} = \{1, 2, 3\}$.
        Applying the above paragraph 
        on $\holts(\mathbf{F}, \mathbf{G}_{ij})$ implies that $\{T \mathbf{e}_i, T \mathbf{e}_j\} \sim \{(*, 0, *), (0, *, 0)\}$, $\{(*, *, 0), (0, 0, *)\}$ or $\{(*, 0, 0), (0, *, *)\}$.
        We analyze the case where $T \mathbf{e}_i = (a, 0, b)$ for some $a, b$ and 
        $T \mathbf{e}_j = (0, c, 0)$ with $c \ne 0$.
        The other cases are similar.
        Suppose both $a, b \ne 0$. 
        Since $\langle T \mathbf{e}_i, T \mathbf{e}_k \rangle = 0$ and $\langle T \mathbf{e}_j, T \mathbf{e}_k \rangle = 0$, it must be the case that $T \mathbf{e}_k \sim (-b, 0, a)$.
        By \cref{subsec:gc-case-a-3}, $\holts(\mathbf{F}, \mathbf{G}) \ge_T \holts(\mathbf{e}_1\teh + \mathbf{e}_3\teh, V) \ge_T \holbs([1, 0, 0, 1], V\domres{\db\dr})$,
        we can apply \cref{prop:subspace-e-bg} on $V\domres{\db\dr} = \spn( (a, b)\teh, (-b, a)\teh)$,
        which results in \sphness as $a, b \ne 0$.
        Therefore, one of $a, b$ must be $0$, meaning
        $T \mathbf{e}_i \sim \mathbf{e}_1$ or $\mathbf{e}_3$, 
        and thus 
        $\mathbf{G} = x \mathbf{e}_1\teh + y \mathbf{e}_2\teh + z \mathbf{e}_3\teh \in \mathcal{E}$.

        Suppose $\mathbf{G}$ is \casebthree. 
        Then, $\holts(\mathbf{F}, \mathbf{G}) \ge_T \holts(\mathbf{F}, T \ag(\casebthreemath))$ for some $T \in \so(3, \mathbb{C})$
        and $T \ag(\casebthreemath)$ contains a \casebtwo signature, resulting in \sphness,
        as proved earlier at the beginning of this proof.
        Similarly, if $\mathbf{G}$ is \casecosep, 
        $\holts(\mathbf{F}, \mathbf{G}) \ge_T \holts(\mathbf{F}, T \ag(\casecosepmath))$ for some $T \in \so(3, \mathbb{C})$
        and $T \ag(\casecosepmath)$ contains a \casecobg signature, resulting in \sphness.
        If $\mathbf{G}$ is \casecodz or \casecogen, then \sphness follows from \cref{prop:a2-hard-with-casecodz-casecosep}
        because $\holts(\mathbf{F}, \mathbf{G}) \ge_T \holts(\mathbf{e}_1\teh + \mathbf{e}_2\teh, \mathbf{G})$.

        We derive the specific forms for $\mathbf{G}$ in the conclusion. 
        Suppose $\mathbf{G}$ is \sepbg and degenerate on $\db\dg$. Other cases are similar.
        If $\mathbf{G}$ is \sepbg, then 
        it must be of the form
        $\mathbf{G}_0 + c \mathbf{e}_3\teh$ for some symmetric ternary $\mathbf{G}_0$ supported on $\db\dg$ and $c \in \mathbb{C}$. 
        By assumption, $\mathbf{G}\domres{\db\dg} = \mathbf{G}_0\domres{\db\dg}$ is degenerate, so $\mathbf{G}_0\domres{\db\dg} = (a, b)\teh$ for some $a, b \in \mathbb{C}$.
        Then, $\mathbf{G} = (a, b, 0)\teh + c \mathbf{e}_3\teh$. 
        If $c = 0$, then $\mathbf{G}$ is degenerate, so $c \ne 0$.
        If $a = 0$ or $b = 0$, then $\mathbf{G} \in \mathcal{E}$, so $a, b \ne 0$ if $\mathbf{G} \notin \mathcal{E}$.
      \end{proof}

      \subsection{\texorpdfstring{\casebthree}{B3}}
      \label{subsec:two-ternary-dichotomy-b-3}
      We start with a \casebthree signature $\mathbf{F} = a \bzp\teh + b \bzm\teh + c \mathbf{e}_3\teh$ for nonzero $a, b, c \in \mathbb{C}$.
      From \cref{subsec:gc-case-a-3}, we know that $a \bzp\teh + b \bzm\teh + c \mathbf{e}_3\teh$ and $\bzp\teh + \bzm\teh + \mathbf{e}_3\teh$ are inter-reducible to each other. 
      So, we may assume without loss of generality, $\mathbf{F} = \bzp\teh + \bzm\teh + \mathbf{e}_3\teh$.
      
      \begin{proposition}
        \label{prop:two-ternary-dichotomy-b-3}
        Let $\mathbf{F} = \bzp\teh + \bzm\teh + \mathbf{e}_3\teh$.
        Let $\mathbf{G}$ be a nondegenerate symmetric ternary signature.
        Then $\holts(\mathbf{F}, \mathbf{G})$ is \sph unless
        \begin{itemize}
          \item $\mathbf{G}$ is \sepbg and $\holbs(\mathbf{F}\domres{\db\dg}, \mathbf{G}\domres{\db\dg})$ is not \sph (and thus tractable).

          In particular, $\mathbf{G}$ must be $a \bzp\teh + b \bzm \teh + c \mathbf{e}_3\teh$ or $(a, b, 0)\teh + c \mathbf{e}_3\teh$ for some $a, b, c \in \mathbb{C}$.
        \end{itemize}
      \end{proposition}
      \begin{proof}

        Suppose  $\holts(\mathbf{F}, \mathbf{G})$ is not \sph.
        Let $\iny{\mathbf{F}} = \Yinv \mathbf{F} = \mathbf{e}_1\teh + \mathbf{e}_2\teh + \mathbf{e}_3\teh$ and $\iny{\mathbf{G}} = \Yinv \mathbf{G}$.
        We apply $\Yinv$ transform, which brings us to $\holts(\mathbf{X} \mid \iny{\mathbf{F}}, \iny{\mathbf{G}})$. 
        Also note that $\iny{\mathbf{F}}$ can realize on the right side the binary equality $\beq$ represented by the $3 \times 3$ identity matrix, by $\langle \mathbf{e}_1\teh + \mathbf{e}_2\teh + \mathbf{e}_3\teh, (1, 1, 1) \rangle$.
        This can be used to realize $\beq$ on the left side as well since $\X \beq \X = \beq$. 
        Then, 
        \begin{align*}
            \holts(\X \mid \iny{\mathbf{F}}, \iny{\mathbf{G}}) &\ge_T \holts(\X \mid \iny{\mathbf{F}}, \iny{\mathbf{G}}, \beq) \\
            &\ge_T \holts(\X, \beq \mid\iny{\mathbf{F}}, \iny{\mathbf{G}}, \beq) \\
            &\ge_T \holts(\X, \iny{\mathbf{F}}, \iny{\mathbf{G}})
        \end{align*}
        because having $\beq$ on both sides means we can remove the bipartite assumption. 
        We can now apply \cref{prop:two-ternary-dichotomy-a-3} on $\holts( \iny{\mathbf{F}}, \iny{\mathbf{G}})$,
        so $\iny{\mathbf{G}}$ must be either in $\mathcal{E}$, or $ij\mid k$ for some $i, j, k \in \{\db,\dg,\dr\}$ and degenerate on $ij$.

        If $\iny{\mathbf{G}} \in \mathcal{E}$, 
        since it is also symmetric, it has the form
        $x \mathbf{e}_1\teh + y \mathbf{e}_2\teh + z \mathbf{e}_3\teh$,
        for some $x,y, z$. 
        Applying $\Y$, we have $\mathbf{G} = x \bzp\teh + y \bzm \teh + z \mathbf{e}_3\teh$ which is \sepbg.
        Also, $\holbs(\mathbf{F}\domres{\db\dg}, \mathbf{G}\domres{\db\dg}) = \holbs((1, \i)\teh + (1, -\i)\teh, x (1, \i)\teh + y (1, -\i)\teh)$ which is Boolean domain tractable type $\typeii$.

        We now argue that if $\iny{\mathbf{G}}$ is $ij | k$ and degenerate on $ij$, then it must be \sepbg, i.e., $ij$ must be $\db\dg$.
        We  prove this by showing that 
        ($\neg$ \sepbg) $\wedge~ [ (ij = \dg\dr) \vee (ij= \db\dr)]$ leads to \sphness.
        Suppose  $\iny{\mathbf{G}}$ is not \sepbg but is \sepgr and degenerate on $\dg\dr$.
        Then, we have $\iny{\mathbf{G}} = (x, 0, 0)\teh + (0, y, z)\teh$ for $x, y, z \ne 0$. Here $x \ne 0$ because $\mathbf{G}$ and thus $\iny{\mathbf{G}}$
        is nondegenerate,
        and $y, z \ne 0$ because we assumed  $\iny{\mathbf{G}}$ is not \sepbg.
        Then, we see that $\ann(\iny{\mathbf{G}}) = \spn((0, -z, y))$. 
        By \cref{prop:nonzero-annihilator-types} and \cref{prop:symmetric-rank-sym-sum},  $\mathbf{G}$ must be \caseatwo, \casebtwo or \casecz. 
        By \cref{prop:annihilating-vector-under-Y}, 
        $\mathbf{G}$ cannot be annihilated by an isotropic vector because $2 \cdot 0 \cdot (-z) + y^2 \ne 0$, so it must be a \caseatwo or \casebtwo.
        If $\mathbf{G}$ were to be \caseatwo, then
        $\Y(x, 0, 0)$ and $\Y(0, y, z)$ must be orthogonal, but that is not the case:
        \[
          \langle \Y(x, 0, 0), \Y(0, y, z) \rangle 
          =\begin{bmatrix}
            x & 0 & 0
            \end{bmatrix} \Y\transpose \Y \begin{bmatrix}
            0 \\ y \\ z
          \end{bmatrix}
          = \begin{bmatrix}
            x & 0 & 0
            \end{bmatrix} \mathbf{X} \begin{bmatrix}
            0 \\ y \\ z
          \end{bmatrix} = xy \ne 0 \,.
        \]
        Similarly, if $\mathbf{G}$ were to be \casebtwo, $\Y(x, 0, 0)$ and $\Y(0, y, z)$ must be isotropic, but
        \[
          \langle \Y(0, y, z), \Y(0, y, z) \rangle 
          =\begin{bmatrix}
            0 & y & z
            \end{bmatrix} \Y\transpose \Y \begin{bmatrix}
            0 \\ y \\ z
          \end{bmatrix}
          = \begin{bmatrix}
            0 & y & z
            \end{bmatrix} \mathbf{X} \begin{bmatrix}
            0 \\ y \\ z
          \end{bmatrix} = z^2 \ne 0 \,.
        \]
        The case when  $\iny{\mathbf{G}}$ is not \sepbg but is \sepbr and degenerate on $\db\dr$ is similar.
        So, $\iny{\mathbf{G}}$ must be \sepbg, which implies $\mathbf{G}$ must be \sepbg as well.

        We derive the specific forms for $\mathbf{G}$ in the conclusion. 
        If $\mathbf{G}$ is \sepbg, then 
        it must be of the form
        $\mathbf{G}_0 + c \mathbf{e}_3\teh$ for some symmetric ternary $\mathbf{G}_0$ supported on $\db\dg$ and $c \in \mathbb{C}$. 
        Applying the Boolean domain dichotomy, 
        on $\holbs(\mathbf{F}\domres{\db\dg}, \mathbf{G}\domres{\db\dg}) =_T \holbs((1, \i)\teh + (1, -\i)\teh, \mathbf{G}_0\domres{\db\dg})$, 
        $\mathbf{G}_0\domres{\db\dg}$ must be either degenerate, which means $\mathbf{G}_0\domres{\db\dg} = (a, b)\teh$, or $a(1, \i)\teh + b(1, -\i)\teh$ for some $a, b \in \mathbb{C}$. 
        This gives $\mathbf{G} = (a, b, 0)\teh + c \mathbf{e}_3\teh$ or $a \bzp\teh + b \bzm \teh + c \mathbf{e}_3\teh$.
      \end{proof}

      \subsection{\texorpdfstring{\casecosep}{C1sep}}
      \label{subsec:two-ternary-dichotomy-C-1-sep}
            Let $\mathbf{F}$ be a \casecosep signature such that $\Yinv \mathbf{F} = a \mathbf{e}_1\teh + \lambda \sym(\mathbf{e}_1, \mathbf{e}_1, \mathbf{e}_2) + d \mathbf{e}_3\teh$ for some $a, d, \lambda \in \mathbb{C}$ with $d, \lambda \ne 0$.
      By \cref{subsec:gc-case-c1-sep}, we may assume without loss of generality that $a = 0$ and $d, \lambda = 1$,
      so $\mathbf{F} = \sym(\bzp, \bzp, \bzm) + \mathbf{e}_3\teh$.

      \begin{proposition}
        \label{prop:two-ternary-dichotomy-C-1-sep}
        Let $\mathbf{F}=\sym(\bzp, \bzp, \bzm) + \mathbf{e}_3\teh$.
        Let $\mathbf{G}$ be a nondegenerate symmetric ternary signature.
        Then, $\holts(\mathbf{F}, \mathbf{G})$ is \sph unless
        \begin{itemize}
          \item $\Yinv \mathbf{G} \in \casecsupp$, or
          \item $\mathbf{G}$ is \sepbg and $\holbs(\mathbf{F}\domres{\db\dg}, \mathbf{G}\domres{\db\dg})$ is not \sph (and thus tractable).

          In particular, $\mathbf{G} = a \bzp\teh + b\sym(\bzp, \bzp, \bzm) + c \mathbf{e}_3\teh$ or $(a, b, 0)\teh + c \mathbf{e}_3\teh$ for some $a, b, c \in \mathbb{C}$.
          If further $\Yinv \mathbf{G} \notin \casecsupp$, then
          $\Yinv \mathbf{G} = (a', b', 0)\teh + c'\mathbf{e}_3\teh$ for some $a', b', c' \in \mathbb{C}$ such that $b', c' \ne 0$.
        \end{itemize}
      \end{proposition}
      \begin{proof}
      Assume $\holts(\mathbf{F}, \mathbf{G})$ is not \sph.
      We have $\iny{\mathbf{F}} = \Yinv \mathbf{F} = \sym(\mathbf{e}_1, \mathbf{e}_1, \mathbf{e}_2) + \mathbf{e}_3\teh$. 
      Let $\iny{\mathbf{G}} = \Yinv \mathbf{G}$.
      By \cref{subsec:gc-case-c1-sep}, 
      $\holts(\mathbf{F}, \mathbf{G}) \ge_T \holbs([0, 0, 1, 0] \mid \iny{\mathbf{G}}\domres{\db\dg})$ and
      $\holts(\mathbf{F}, \mathbf{G}) \ge_T \holbs([1, 0, 0, 1] \mid \iny{\mathbf{G}}\domres{\dg\dr})$, so
    $\iny{\mathbf{G}}\domres{\db\dg}$ must be $[*, *, 0, 0]$ or degenerate,
    and
    $\iny{\mathbf{G}}\domres{\dg\dr}$ must be $[*, 0, 0, *]$ or degenerate.

      We first consider the case when $\ann(\iny{\mathbf{G}}) = \spn(\mathbf{u}) \ne 0$.
        We separate into cases where $\mathbf{u}$ orthogonal to $\mathbf{e}_3$ or not.
        Suppose $\langle \mathbf{u}, \mathbf{e}_3 \rangle = 0$,
        thus we have $\mathbf{u} = (x, y, 0)$.
        If $x = 0$, then $y \ne 0$ 
        so $\ann(\iny{\mathbf{G}}) = \spn(\mathbf{e}_2)$ which means $\iny{\mathbf{G}}$ is supported on $\db\dr$ and thus $\iny{\mathbf{G}} \in \casecsupp$, and we are done.
        So we may assume $x = -1$.
        We apply \cref{prop:annihilator-uniquely-propagates-values} with the roles of $\db$ and $\dr$ switched, which makes the values propagate upward from $\iny{\mathbf{G}}\domres{\dg\dr}$.
        If  $\iny{\mathbf{G}}\domres{\dg\dr}$ were degenerate, say $(a, b)\teh
        = [a^3, a^2b, ab^2, b^3]$, then $\iny{\mathbf{G}}$ would be $(ya, a, b)\teh$ 
        and thus also degenerate, a contradiction since the given $\mathbf{G}$ is nondegenerate.
        Hence $\iny{\mathbf{G}}\domres{\dg\dr}$ is nondegenerate, and it must be
        $[a, 0, 0, b]$ for some $a, b \ne 0$,
        giving $\iny{\mathbf{G}} = a (y, 1, 0)\teh + b \mathbf{e}_3\teh $, 
        which is \sepbg and is degenerate on $\db\dg$, namely 
        $\iny{\mathbf{G}}\domres{\db\dg} = a (y, 1)\teh$. 
        Then, $\mathbf{G} = \Y \iny{\mathbf{G}}$ is \sepbg and degenerate on $\db\dg$ as well, since $\Y$ is  \sepbg.
        Then, $\holts(\mathbf{F}, \mathbf{G}) =_T 
        \holbs( \mathbf{F}\domres{\db\dg}, \mathbf{G}\domres{\db\dg})
        =_T \holbs( \mathbf{F}\domres{\db\dg})$, which is P-time computable.

        If $\langle \mathbf{u}, \mathbf{e}_3 \rangle \ne 0$, then we can write $\mathbf{u} = (x, y, -1)$. 
        \cref{prop:annihilator-uniquely-propagates-values} applies; 
        since $\iny{\mathbf{G}}$ is nondegenerate, so 
        is $\iny{\mathbf{G}}\domres{\db\dg}$, and thus it takes the form
        $[a, b, 0, 0]$, with $b \ne 0$. 
         Such a $\iny{\mathbf{G}}$ satisfying $\langle \iny{\mathbf{G}},  (x, y, -1) \rangle = 0$ 
        must have the form in \cref{fig:casedz-m-on-bg} with $b \ne 0$.
        Again, $\iny{\mathbf{G}}\domres{\dg\dr} = [0, 0, bx^2, ax^3 + 3bx^2 y]$ must be degenerate or $[*, 0, 0, *]$, which can  only be true if $bx^2 = 0$ so $x = 0$, 
        implying that $\iny{\mathbf{G}} \in \casecsupp$, seen by its shape in~\cref{fig:casedz-m-on-bg}.

        Now, we consider the cases where $\ann(\iny{\mathbf{G}}) = 0$, equivalently $\ann(\mathbf{G}) = 0$. 
        Then, by \cref{prop:nonzero-annihilator-types},
        $\mathbf{G}$ must be of type \caseathree, \casebthree, \casecosep, \casecodz, or \casecogen. 
        
        Suppose $\mathbf{G}$ is \casecosep.
        Then, there exists some $\iny{T} \in \so(3, \mathbb{C}, \X)$ such that $\iny{T}^{-1} \Yinv \mathbf{G} \sim x \mathbf{e}_1\teh + \sym(\mathbf{e}_1, \mathbf{e}_1, \mathbf{e}_2) + w \mathbf{e}_3\teh$ for some $x, w$ such that $w \ne 0$.
        As $\iny{\mathbf{F}} = \Yinv \mathbf{F}$,
        by \cref{lem:gadget-set-general-form} and \cref{subsec:gc-case-c1-sep}, 
        $\holts(\mathbf{F}, \mathbf{G}) \ge_T \holts(\X \mid \iny{\mathbf{F}}, V)$ where $V = \iny{T} \agx(\casecosepmath)$.
        By the (complex) spin representation of $\speciallinear(2, \mathbb{C})$ parameterizing $\so(3, \mathbb{C}, X)$ (\cref{eq:paramaterize-so3cx}), 
        we have $\iny{T} = \rho(M)$ for some $M = \begin{bsmallmatrix}
              a & b \\ c & d 
          \end{bsmallmatrix} \in \speciallinear(2, \mathbb{C})$ and 
        \[
          V = \rho\left(M\right) \spn(\mathbf{e}_1\teh, \sym(\mathbf{e}_1, \mathbf{e}_1, \mathbf{e}_2), 
          \mathbf{e}_3\teh) = \spn(\mathbf{v}_1\teh, \sym(\mathbf{v}_1, \mathbf{v}_1, \mathbf{v}_2), \mathbf{v}_3\teh)
          )
        \]
          for
          $\mathbf{v}_1 = (d^2, -c^2, \sqrt{2}cd)$,
          $\mathbf{v}_2 = (-b^2, a^2, -\sqrt{2}ab)$, and
          $\mathbf{v}_3 = (\sqrt{2} bd, -\sqrt{2} ac, ad + bc)$. 
          By \cref{subsec:gc-case-c1-sep}, 
          we have
          $\holts(\mathbf{F}, \mathbf{G}) \ge_T \holbs([0, 0, 1, 0] \mid V\domres{\db\dg})$. 
          We can then apply \cref{prop:subspace-m-bg} on $\spn(\mathbf{v}_1\teh, \sym(\mathbf{v}_1, \mathbf{v}_1, \mathbf{v}_2))\domres{\db\dg}$, which gives $c = 0$ or $(d^2, -c^2) \sim (-b^2, a^2)$.
          If $c = 0$, then we have $\mathbf{G} \in \casecsupp$ by \cref{prop:casecsupp-fixed-under-rho-c-0}. 
          Otherwise, we can apply it again on $\spn(\mathbf{v}_1\teh, \mathbf{v}_3\teh)\domres{\db\dg}$, but this implies the span of $\{\proj_{12} \mathbf{v}_1, \proj_{12} \mathbf{v}_2, \proj_{12} \mathbf{v}_3\}$ has dimension $1$, which is impossible because the span of $\{\mathbf{v}_1, \mathbf{v}_2, \mathbf{v}_3\}$ has dimension $3$.
        
        We get \sphness for $\mathbf{G}$ in \caseathree and \casebthree from \cref{prop:two-ternary-dichotomy-a-3} and \cref{prop:two-ternary-dichotomy-b-3}.
        For $\mathbf{G}$ in \casecodz and \casecogen, we can apply \cref{prop:c1bg-vs-c1codz-c1gen} since
        $\holts(\mathbf{F}, \mathbf{G}) \ge_T \holts(\sym(\bzp, \bzp, \bzm), \mathbf{G})$.

        We derive the specific forms for $\mathbf{G}$ in the conclusion. 
Suppose $\mathbf{G}$ is \sepbg and the Boolean domain  problem  $\holbs(\mathbf{F}\domres{\db\dg}, \mathbf{G}\domres{\db\dg})$ is not \sph.
        If $\mathbf{G}$ is \sepbg, then 
        it must be of the form
        $\mathbf{G}_0 + c\mathbf{e}_3\teh$ for some symmetric ternary $\mathbf{G}_0$ supported on $\db\dg$ and $c \in \mathbb{C}$. 
        Since $\holbs(\mathbf{F}\domres{\db\dg}, \mathbf{G}\domres{\db\dg}) =_T \holbs(\sym((1, \i), (1, \i), (1, -\i)), \mathbf{G}_0\domres{\db\dg})$ must be tractable, by the Boolean domain dichotomy, 
        $\mathbf{G}_0\domres{\db\dg}$ must be either degenerate, which means $\mathbf{G}_0\domres{\db\dg} = (a, b)\teh$, or $a(1, \i)\teh + b\sym((1, \i), (1, \i), (1, -\i))$ for some $a, b \in \mathbb{C}$. 
        Now, further assume that $\Yinv \mathbf{G} \notin \casecsupp$.
In the second case
$\Yinv \mathbf{G} = a\mathbf{e}_1\teh + b \sym(\mathbf{e}_1, \mathbf{e}_1, \mathbf{e}_2) + c \mathbf{e}_3\teh$, but this is in $\casecsupp$.
So, $\Yinv \mathbf{G} = \Yinv (a, b, 0)\teh + c \Yinv \mathbf{e}_3\teh = (a', b', 0)\teh + c \mathbf{e}_3\teh$ where $(a', b', 0) = \Yinv (a, b, 0)$, as $\Y$ is \sepbg.
If $c = 0$, then $\mathbf{G}$ is degenerate, so $c \ne 0$.
If $b' = 0$, then $\Yinv \mathbf{G} \in \casecsupp$, so $b' \ne 0$.
      \end{proof}

      \subsection{\texorpdfstring{\casecodz}{C1dz}}
        Before we prove a dichotomy anchoring on a \casecodz signature,
        we prove a result that allows us to relate previously proven results about \casec signatures. 
        \begin{proposition}\label{prop:in-C-after-transform-then-before-transform}
            Let $\mathbf{F}$ be a nondegenerate symmetric ternary signature such that $\Yinv \mathbf{F} \in \casecsupp$.
            Let $\mathbf{G}$ be a nondegenerate symmetric ternary signature such that for some 
            $T \in \so(3, \mathbb{C})$,
            $\Yinv T \mathbf{G} \in \casecsupp$.
            If $\Yinv T \mathbf{F} \in \casecsupp$, then $\Yinv \mathbf{G} \in \casecsupp$.
        \end{proposition}
        \begin{proof}
            Let $\iny{T} = \Yinv T \Y \in \so(3, \mathbb{C}, \X)$.
            Then, $\Yinv T \mathbf{F} = \iny{T} \Yinv \mathbf{F}$.
            By the (complex) spin representation of $\speciallinear(2, \mathbb{C})$ parameterizing $\so(3, \mathbb{C}, X)$ (\cref{eq:paramaterize-so3cx}), 
            we may write $\iny{T} = \rho(M)$ for $M = \begin{bsmallmatrix}
                a & b \\ c & d
            \end{bsmallmatrix}$ for some $M \in \speciallinear(2, \mathbb(C))$.
            By \cref{prop:casecsupp-fixed-under-rho-c-0}, $c = 0$ since $\Yinv \mathbf{F} \in \casecsupp$ and $\iny{T} \Yinv \mathbf{F} \in \casecsupp$.
            Then, since $\iny{T}^{-1} = \rho(M^{-1}) = \rho\left(\begin{bsmallmatrix}
                d & -b \\ 0 & a
            \end{bsmallmatrix}\right)$,
            by \cref{prop:casecsupp-fixed-under-rho-c-0}, 
            $\iny{T}^{-1} \Yinv T \mathbf{G} = \Yinv T^{-1} \Y \Yinv T \mathbf{G} = \Yinv \mathbf{G} \in \casecsupp$.
        \end{proof}
      
      \label{subsec:two-ternary-dichotomy-C-1-dz}
      \begin{proposition}
        \label{prop:two-ternary-dichotomy-C-1-dz}
        Let $\mathbf{F}$ be a \casecodz signature such that $\Yinv \mathbf{F} = [p; 1, q; 0, 0, r; 0, 0, 0, 0]$ for some $p, q, r \in \mathbb{C}$ and $r \ne 0$.
        Let $\mathbf{G}$ be a nondegenerate symmetric ternary signature.
        Then, $\holts(\mathbf{F}, \mathbf{G})$ is \sph unless $\Yinv \mathbf{G} \in \casecsupp$.
      \end{proposition}
      \begin{proof}
      Assume  $\holts(\mathbf{F}, \mathbf{G})$ is not \sph.
      We have already shown that we get \sphness when a \casecodz signature 
      appears together with another signature of type \caseatwo, \caseathree, \casebtwo, or \casebthree.
      Therefore, we only need to analyze the case where $\mathbf{G}$ is \casecz, \casecobg, \casecosep, \casecodz or \casecogen. 
      Let $\iny{\mathbf{F}} = \Yinv \mathbf{F}$ and $\iny{\mathbf{G}} = \Yinv \mathbf{G}$.

Suppose $\mathbf{G}$ is \casecz.
Then, there exists some $T \in \so(3, \mathbb{C})$ such that 
$\ann(T \mathbf{G}) = \spn(\bzp)$.
Applying \cref{prop:c0-vs-c1dz} on $\holts(T \mathbf{G}, T \mathbf{F})$ 
(note we use $T \mathbf{G}$ for the $\mathbf{F}$ and $T \mathbf{F}$ for the $\mathbf{G}$ in the proposition statement)
we get $\Yinv T \mathbf{F} \in \casecsupp$. 
By \cref{prop:in-C-after-transform-then-before-transform}, 
$\Yinv \mathbf{G} \in \casecsupp$.

Suppose $\mathbf{G}$ is \casecobg or \casecosep.
Then, there exists some $T \in \so(3, \mathbb{C})$ such that 
$T \mathbf{G} = x \bzp \teh + y\sym(\bzp, \bzp, \bzm) + z \mathbf{e}_3\teh$ for some $y \ne 0$.
Applying \cref{prop:c1bg-vs-c1codz-c1gen} or \cref{prop:two-ternary-dichotomy-C-1-sep}
on $\holts(T \mathbf{G}, T \mathbf{F})$,
(note we use $T \mathbf{G}$ for the $\mathbf{F}$ and $T \mathbf{F}$ for the $\mathbf{G}$ in the proposition statement)
we get $\Yinv T \mathbf{F} \in \casecsupp$. 
By \cref{prop:in-C-after-transform-then-before-transform}, 
$\Yinv \mathbf{G} \in \casecsupp$ as well.

Suppose $\mathbf{G}$ is \casecogen.
Then, there exists some $T \in \so(3, \mathbb{C})$ such that 
the signature
$\Yinv T^{-1} \mathbf{G} \sim [x; 1, y; 0, 0, z; 0, 0, 0, w] \in \casecsupp$ for some $x, y, z, w$ such that $w, z^2 - z - y w \ne 0$. 
 By \cref{lem:gadget-set-general-form} and \cref{subsec:gc-case-c1gen}, 
$\holts(\mathbf{F}, \mathbf{G}) \ge_T \holts(\mathbf{F}, T \ag(\casecogenmath)) \ge_T \holts(\mathbf{F}, T \sym(\bzp, \bzp, \bzm))$.
Note that  $T \sym(\bzp, \bzp, \bzm)$ is \casecobg.
So by the above, $\Yinv \mathbf{G} \in \casecsupp$ if  $\holts(\mathbf{F}, \mathbf{G})$ is not \sph.

        Suppose $\mathbf{G}$ is \casecodz. Then, there exists some $\iny{T} \in \so(3, \mathbb{C}, \X)$ such that the signature $\iny{T}^{-1} \iny{\mathbf{G}} \sim [x; 1, y; 0, 0, z; 0, 0, 0, 0] \in \casecsupp$ for some $x, y, z$ such that $z \ne 0$.
         By \cref{lem:gadget-set-general-form} and \cref{subsec:gc-case-c1dz}, we have
        $\holts(\mathbf{F}, \mathbf{G}) \ge_T \holts(\X \mid \iny{\mathbf{F}}, V)$ where $V = \iny{T} \agx(\casecodzmath, \lambda)$ for some $\lambda \ne 0$.
        By the (complex) spin representation of $\speciallinear(2, \mathbb{C})$ parameterizing $\so(3, \mathbb{C}, X)$ (\cref{eq:paramaterize-so3cx}), 
        we can write $\iny{T} = \rho(M)$ and thus
        \[
          V = \rho(M) \spn(\mathbf{e}_1\teh, ~ \sym(\mathbf{e}_1, \mathbf{e}_1, \mathbf{e}_3), ~\sym(\mathbf{e}_1, \mathbf{e}_1, \mathbf{e}_2) + \lambda \sym(\mathbf{e}_1,  \mathbf{e}_3, \mathbf{e}_3))
        \]
        for some $M = \begin{bsmallmatrix}
            a & b \\ c & d
        \end{bsmallmatrix} \in \speciallinear(2, \mathbb{C})$. 
        Recall from \cref{subsec:gc-case-c1dz} that
        $\holts(\mathbf{F}, \mathbf{G}) \ge_T \holts(\X \mid \iny{\mathbf{F}}, V) \ge_T \holbs([0, 1, 0, 0] \mid V\domres{\dg\dr})$.
        Writing $\mathbf{v}_i = \rho(M) \mathbf{e}_i$, we have
        $\mathbf{v}_1 = (d^2, -c^2, \sqrt{2} c d)$,
        $\mathbf{v}_2 = (-b^2, a^2, - \sqrt{2} a b)$,
        $\mathbf{v}_3 = (bd, -ac, ad + bc)$.
        Let $\mathbf{u}_i = \proj_{23} \mathbf{v}_i$. 
        We can apply \cref{prop:subspace-m-bg} on $\spn(\mathbf{v}_1\teh, \sym(\mathbf{v}_1, \mathbf{v}_1, \mathbf{v}_3))\domres{\dg\dr} = 
        \spn(\mathbf{u}_1\teh, \sym(\mathbf{u}_1, \mathbf{u}_1, \mathbf{u}_3))$
        (note that we apply it with the role of the Boolean two domain elements, or $\sbto$ and $\sbtt$, changed).
        So, $\mathbf{u}_1 \sim \mathbf{u}_3$ or $\mathbf{u}_1 \sim \sbtt$. 
        If $c = 0$, then $\iny{\mathbf{G}} \in \casecsupp$ by \cref{prop:casecsupp-fixed-under-rho-c-0}, so we may assume $c \ne 0$, which implies $\mathbf{u}_1 \not \sim \sbtt$.
        Suppose $\mathbf{u}_1 \sim \mathbf{u}_3$.
        Then, 
        \[
            V\domres{\dg\dr} = \spn((\mathbf{u}_1)\teh, \sym(\mathbf{u}_1,\mathbf{u}_1, \mathbf{u}_2)) \,.
        \]
        Applying \cref{prop:subspace-m-bg} again, we obtain $\mathbf{u}_1 \sim \sbtt$ or $\mathbf{u}_1 \sim \mathbf{u}_2$.
        Both cases are not possible. 

      \end{proof}

      \subsection{\texorpdfstring{\casecogen}{C1gen}}
      \label{subsec:two-ternary-dichotomy-C-1-gen}
      \begin{proposition}
        \label{prop:two-ternary-dichotomy-C-1-gen}
        Let $\mathbf{F}$ be a \casecogen signature
        such that $\Yinv \mathbf{F} = [a;1,b;0,0,c;0,0,0,d]$ with $d \ne 0$ and $c^2 - c - bd \ne 0$.
        Let $\mathbf{G}$ be a nondegenerate symmetric ternary signature.
        Then, $\holts(\mathbf{F}, \mathbf{G})$ is \sph unless $\Yinv \mathbf{G} \in \casecsupp$.
      \end{proposition}
      \begin{proof}
       We have already shown that we get \sphness when a \casecogen signature 
       appears together with another signature of type \caseatwo, \caseathree, \casebtwo, or \casebthree signature.
      Therefore, we only need to analyze the case where $\mathbf{G}$ is \casecz, \casecobg, \casecosep, \casecodz or \casecogen. 
    Let $\iny{\mathbf{F}} = \Yinv \mathbf{F}$ and $\iny{\mathbf{G}} = \Yinv \mathbf{G}$.

      The fact that if $\mathbf{G}$ is \casecz, \casecobg, \casecosep, or \casecodz then $\iny{\mathbf{G}}$ must be in $\casecsupp$ can be shown in a similar way as in \cref{prop:two-ternary-dichotomy-C-1-dz} by applying 
      \cref{prop:in-C-after-transform-then-before-transform} after using \cref{prop:c0-vs-c1dz,prop:c1bg-vs-c1codz-c1gen,prop:two-ternary-dichotomy-C-1-sep,prop:two-ternary-dichotomy-C-1-dz}.

      Suppose $\mathbf{G}$ is \casecogen.
Then, there exists some $T \in \so(3, \mathbb{C})$ such that 
the signature
$\Yinv T^{-1} \mathbf{G} \sim [x; 1, y; 0, 0, z; 0, 0, 0, w] \in \casecsupp$ for some $x, y, z, w$ such that $w, z^2 - z - y w \ne 0$. 
 By \cref{lem:gadget-set-general-form} and \cref{subsec:gc-case-c1gen}, 
 $\holts(\mathbf{F}, \mathbf{G}) \ge_T \holts(\mathbf{F}, T \ag(\casecogenmath)) \ge_T \holts(\mathbf{F}, T \sym(\bzp, \bzp, \bzm))$.
Note that  $T \sym(\bzp, \bzp, \bzm)$ is \casecobg.
So by the above, $\Yinv \mathbf{G} \in \casecsupp$
if  $\holts(\mathbf{F}, \mathbf{G})$ is not \sph.

      \end{proof}

\section{Dichotomy of a Set of Ternary Signatures}
\label{sec:set-ternary}

In this section, we prove a $\holts(\mathcal{F})$ dichotomy when $\mathcal{F}$ is a set of nondegenerate symmetric \emph{ternary} signatures.
We may assume that for every $\mathbf{F} \in \mathcal{F}$, 
$\holts(\mathbf{F})$ is tractable.
Further, we may assume that $\holts(\mathbf{F}, \mathbf{G})$ is tractable for every $\mathbf{F}, \mathbf{G} \in \mathcal{F}$ by the pairwise dichotomy proved in the previous section. 

We split into two cases depending on whether there is $\mathbf{F} \in \mathcal{F}$ with $\ann(\mathbf{F}) = 0$ or not.
As mentioned before,  surprisingly, we get a simple proof of the dichotomy when every signature in $\mathcal{F}$ has a nontrivial annihilator. 
This is mainly because of the powerful \textit{basis free} algorithm in 
\cref{thm:ternary-set-ebd-algorithm},
which can be applied if every pair of signatures
is a compatible pair (\cref{def:compatible-pairs}).
We prove this result in \cref{subsec:proof-of-not-orthogonal-pair-dichotomy}.

Next,  assume  $\holts(\mathcal{F})$ is not \sph,
and  there is $\mathbf{F} \in \mathcal{F}$ with $\ann(\mathbf{F}) = 0$. We  anchor the proof on $\mathbf{F}$. 
However, unlike 
in the Boolean domain, 
one signature in $\mathcal{F}$ is not sufficient to completely determine the forms of other signatures in $\mathcal{F}$: 
it could happen that $\holts(\mathbf{F}, \mathbf{G})$ and 
$\holts(\mathbf{F}, \mathbf{H})$ are tractable, but $\holts(\mathbf{F}, \mathbf{G}, \mathbf{H})$ is \sph. 
We call this type of result \textit{non-mixing}, and prove some non-mixing results in \cref{subsec:non-mixing-ternary}.
After that, we use the pairwise dichotomy and non-mixing results to arrive at a dichotomy for a set of signatures in \cref{subsec:set-of-ternary}.

\subsection{Ternary Signatures with Annihilators}
\label{subsec:proof-of-not-orthogonal-pair-dichotomy}
A difficult part of proving a dichotomy for a set  $\mathcal{F}$ 
of ternary signatures 
 is when $\mathcal{F}$  consists of only \caseatwo, \casebtwo, \casecz, and \casecobg. 
This is because \caseatwo, \casebtwo, \casecz, and \casecobg can be viewed as Boolean domain signatures living in domain $3$, and a lot of effort was expended in~\cite{domain3-real-set} to analyze such cases.
We recall \cref{prop:orthogonal-normal-vector-edge-factors}.
If $\mathbf{F}_1, \mathbf{F}_2$ with $\ann(\mathbf{F}_i) = \spn(\mathbf{n}_i) \ne 0$ are such that $\mathbf{n}_1$ and $\mathbf{n}_2$ are orthogonal, then $\holts(\mathbf{F}_1, \mathbf{F}_2)$ is always tractable if $\holts(\mathbf{F}_i)$ are, because any edge between $\mathbf{F}_1$ and $\mathbf{F}_2$ factors degenerately.
For the cases when $\langle \mathbf{n}_1, \mathbf{n}_2 \rangle \ne 0$,
it
appears that there is no 
way to holographically transform $\mathbf{F}_i$ to get domain separation, and \cref{subsec:tractable-gadget-meet} shows that the interaction is more complicated when we consider more than two signatures. 
However, in that case, by \cref{prop:annihilator-uniquely-propagates-values}, $\mathbf{F}_1$ and $\mathbf{n}_2$ determine the form of $\mathbf{F}_2$ by the Boolean domain dichotomy. 
Then, the forms of $\mathbf{F}_1, \mathbf{F}_2$ are precisely the ones satisfying 
\cref{itm:compatible-pair-unique-prop} and \cref{itm:compatible-pair-vanish} of \cref{def:compatible-pairs}. 
The notion of compatible pairs allows us to state the dichotomy for signatures with nonzero annihilators very succinctly. 

\begin{lemma}
\label{lem:not-orthogonal-pair-dichotomy-ternary}
  Let $\mathbf{F}_1, \mathbf{F}_2$ be nondegenerate symmetric ternary signatures such that
   $\ann(\mathbf{F}_i) \ne 0$ for $i = 1, 2$.
  Then, $\holts(\mathbf{F}_1, \mathbf{F}_2)$ is \sph if 
  $(\mathbf{F}_1, \mathbf{F}_2)$ is not a compatible pair. 
\end{lemma}
\begin{proof}
We may assume that $\mathbf{F}_1$ and $\mathbf{F}_2$ are each of some tractable type, specifically, \caseatwo, \casebtwo, \casecobg, or \casecz by \cref{prop:nonzero-annihilator-types}.
We can prove that $\holts(\mathbf{F}_1, \mathbf{F}_2)$ is \sph if $(\mathbf{F}_1, \mathbf{F}_2)$ is not a compatible pair by combining \cref{prop:two-ternary-dichotomy-a-2,prop:two-ternary-dichotomy-b-2,prop:two-ternary-dichotomy-C-1-bg,prop:two-ternary-dichotomy-C-0}. 
\end{proof}

\begin{corollary}\label{cor:ebd-set-ternary-dichotomy}
    Let $\mathcal{F}$ be a set of nondegenerate symmetric ternary signatures with nonzero annihilators.
If every pair of signatures in $\mathcal{F}$ is a compatible pair, 
then $\holts(\mathcal{F})$ is tractable.
Otherwise, it is \sph.
\end{corollary}
\begin{proof}
We may assume that each signature is of some tractable type, specifically, \caseatwo, \casebtwo, \casecobg, or \casecz by \cref{prop:nonzero-annihilator-types}.

If every pair is a compatible pair, then the algorithm given in \cref{thm:ternary-set-ebd-algorithm} applies. 

If there exist $\mathbf{F}_1, \mathbf{F}_2 \in \mathcal{F}$ such that
$(\mathbf{F}_1, \mathbf{F}_2)$ is not a compatible pair, then $\holts(\mathcal{F}) \ge_T \holts(\mathbf{F}_1, \mathbf{F}_2)$ is \sph by \cref{lem:not-orthogonal-pair-dichotomy-ternary}. 
\end{proof}

\subsection{Non-mixing of Ternary Signatures}\label{subsec:non-mixing-ternary}

Unlike the Boolean domain $\hol^*$ dichotomy, the tractability is not transitive;
$\holts(\mathbf{F}, \mathbf{G})$ and $\holts(\mathbf{F}, \mathbf{H})$ being tractable does not imply that $\holts(\mathbf{F}, \mathbf{G}, \mathbf{H})$ is tractable. 
One reason to expect this is that the tractable cases given in 
\cref{thm:tractable-cases}
are not mutually exclusive, which means that one signature is not enough to rule out which algorithm cannot work on $\mathcal{F}$.
For example, a \caseathree signature can occur in 
\cref{itm:tractable-e} and \cref{itm:tractable-bg-r},
and a \casecosep signature can occur in 
\cref{itm:tractable-bg-r} and \cref{itm:tractable-c}.
Rather, we need to consider at least two signatures to completely determine the form of all other signatures in $\mathcal{F}$ for $\holts(\mathcal{F})$ to escape \sphness.

In this section, 
we actually show a stronger result, that
$\holts(\mathbf{G}, \mathbf{H})$ could be \sph even when $\holts(\mathbf{F}, \mathbf{G})$ and $\holts(\mathbf{F}, \mathbf{H})$ are tractable. 
This implies that
$\holts(\mathbf{F}, \mathbf{G}, \mathbf{H})$ is \sph.

\begin{proposition}\label{prop:case-c1-sep-ternary-non-mixing}
    Let $\mathbf{F} = \sym(\bzp, \bzp, \bzm) + \mathbf{e}_3\teh$.
    Let $\mathbf{G}$ and $\mathbf{H}$ be nondegenerate symmetric ternary signatures such that 
    $\holts(\mathbf{F}, \mathbf{G})$ and $\holts(\mathbf{F}, \mathbf{H})$ are both not \sph, and
    \begin{enumerate}
        \item $\mathbf{G}$ is not \sepbg and $\Yinv \mathbf{G} \in \casecsupp$.
        \item $\mathbf{H}$ is \sepbg and $\Yinv \mathbf{H} \notin \casecsupp$.
    \end{enumerate}
    Then 
    $\holts(\mathbf{G}, \mathbf{H})$ is \sph.
\end{proposition}
\begin{proof}
Assume $\holts( \mathbf{G}, \mathbf{H})$ is not \sph.
Let $\iny{\mathbf{F}} = \Yinv \mathbf{F}$, $\iny{\mathbf{G}} = \Yinv \mathbf{G}$, and $\iny{\mathbf{H}} = \Yinv \mathbf{H}$.
Note that $\mathbf{G}$ is not \sepbg if and only if $\iny{\mathbf{G}}$ is not \sepbg. 
By \cref{prop:two-ternary-dichotomy-C-1-sep}, 
we may assume that $\iny{\mathbf{H}} = (a, b, 0)\teh + (0, 0, c)\teh$ for some $a, b, c \in \mathbb{C}$ such that $b, c \ne 0$.

We consider each possible form of $\iny{\mathbf{G}} \in \casecsupp \setminus \sepbgmath$.
If $\mathbf{G}$ is \casecodz or \casecogen, 
then by \cref{prop:two-ternary-dichotomy-C-1-dz} and \cref{prop:two-ternary-dichotomy-C-1-gen}, $\holts(\mathbf{G}, \mathbf{H})$ is \sph, so we may assume otherwise.
The remaining cases are when $\mathbf{G}$ is \casecz, \casecobg, or \casecosep. 

Suppose $\mathbf{G}$ is \casecz.
Since $\iny{\mathbf{G}} \in \casecsupp \setminus \sepbgmath$, it must be supported on $\db\dr$ by \cref{prop:if-casec0-in-c-then-supp-br}.
From the given form of $\iny{\mathbf{H}}$ we see that $\ann(\iny{\mathbf{H}}) = \spn(\mathbf{u}')$ where $\mathbf{u}' = (-b, a, 0)$.
On the other hand, $\ann(\iny{\mathbf{G}}) = \spn(\mathbf{e}_2)$.
Applying $\Y$, we have
$\mathbf{H} = \frac{1}{\sqrt{2}}(a + b, \i(a - b), 0)\teh + (0, 0, c)\teh$ with
$\ann(\mathbf{H}) = \spn(\mathbf{u})$ for $\mathbf{u} = (a - b, \i(a + b), 0) \sim  (\Yinv)\transpose \mathbf{u}'$ by \cref{prop:annihilating-vector-under-gl-transform}. 
On the other hand, $\ann(\mathbf{G}) =  \spn((\Yinv)\transpose \mathbf{e}_2) = \spn(\bzp)$. 
Then, $\langle \bzp, \mathbf{u} \rangle = -2b \ne 0$ and we may apply \cref{lem:not-orthogonal-pair-dichotomy-ternary} on $\holts( \mathbf{G}, \mathbf{H})$.
For $\holts( \mathbf{G}, \mathbf{H})$ to be not \sph, we need $\mathbf{G} = \mathbf{v}_1\teh + \mathbf{v}_2\teh$ for some $\mathbf{v}_1, \mathbf{v}_2\in \mathbb{C}^3$ such that  $\langle \mathbf{v}_1, (0, 0, c) \rangle = 0$ and $\langle \mathbf{v}_2, (a + b, \i(a - b), 0) \rangle = 0$.
For $\ann(\mathbf{G}) = \spn(\bzp)$ to hold, we must have $\langle \mathbf{G}, \bzp \rangle = \langle \mathbf{v}_1, \bzp \rangle \mathbf{v}_1\tew + \langle \mathbf{v}_2, \bzp \rangle \mathbf{v}_2\tew = 0$.
Since $\mathbf{v}_1 \not \sim \mathbf{v}_2$ because $\mathbf{G}$ is nondegenerate, 
we need $\langle \mathbf{v}_1, \bzp \rangle = 0$ and $\langle \mathbf{v}_2, \bzp \rangle = 0$ because $\mathbf{v}_1\tew$ and $\mathbf{v}_2\tew$ are linearly independent.
These conditions can be  true only if both
\begin{align*}
  \begin{bmatrix}
    0 & 0 & c \\
    1 & \i & 0
  \end{bmatrix} \mathbf{v}_1 = 0 
  &&
  \begin{bmatrix}
    a + b & \i(a - b) & 0 \\
    1 & \i & 0
  \end{bmatrix} \mathbf{v}_2 = 0
\end{align*} hold, which imply that $\mathbf{v}_1 \sim (1, \i, 0)$ and $\mathbf{v}_2 \sim \mathbf{e}_3$.
Then, $\mathbf{G}$ is \sepbg, a contradiction.

Suppose $\mathbf{G}$ is \casecobg or \casecosep.
Then, there exist some $\iny{T} \in \so(3, \mathbb{C}, \X)$ 
such that
$\iny{\mathbf{G}} = \iny{T}(x \mathbf{e}_1\teh + \lambda \sym(\mathbf{e}_1, \mathbf{e}_1, \mathbf{e}_2) + w \mathbf{e}_3\teh)$ for some $x, \lambda, w \in \mathbb{C}$ and $\lambda  \ne 0$.
By the spin representation (\cref{eq:paramaterize-so3cx}) we have $\iny{T} = \rho(M)$ for some 
$M = \begin{bsmallmatrix}
    a' & b' \\ c' & d'
\end{bsmallmatrix} \in \speciallinear(2, \mathbb{C})$. 
Since $\iny{\mathbf{G}} \in \casecsupp$, 
by \cref{prop:casecsupp-fixed-under-rho-c-0},
we have $c'=0$. However, $b' \ne 0$, for otherwise $M$ is diagonal
and  so is  $\rho(M)$ by~\cref{eq:paramaterize-so3cx}, implying 
 $\iny{\mathbf{G}} \in \sepbgmath$ because 
$\iny{\mathbf{G}} = \rho(M) (x \mathbf{e}_1\teh + \sym(\mathbf{e}_1, \mathbf{e}_1, \mathbf{e}_2) + w \mathbf{e}_3\teh) = x' \mathbf{e}_1\teh + \lambda' \sym(\mathbf{e}_1, \mathbf{e}_1, \mathbf{e}_2) + w' \mathbf{e}_3\teh$ for some $x', w', \lambda' \in \mathbb{C}$, 
contrary to assumption.
As $M \in  \speciallinear(2, \mathbb{C})$, we have $a', d' \ne 0 $.
By \cref{lem:gadget-set-general-form} and \cref{subsec:gc-case-c1-bg} or \cref{subsec:gc-case-c1-sep} depending on the type $\mathfrak{G}$ of $\mathbf{G}$, we have
\begin{align*}
\holts(\mathbf{G}, \mathbf{H}) 
\ge_T 
\holts(\X \mid \iny{\mathbf{G}}, \iny{\mathbf{H}}) 
\ge_T 
\holts(\X \mid \rho(M) \agx(\mathfrak{G}), \iny{\mathbf{H}}) 
\ge_T 
\holts(\X \mid \iny{\mathbf{G}}', \iny{\mathbf{H}}) 
\end{align*}
in which $\iny{\mathbf{G}}'  = \rho(M) \sym(\mathbf{e}_1, \mathbf{e}_1, \mathbf{e}_2) \in \rho(M) \agx(\mathfrak{G})$ for both $\mathfrak{G} = \casecobgmath$ and $\mathfrak{G} = \casecosepmath$.
We see that in particular, $\ann(\iny{\mathbf{G}}') = \spn((0, \sqrt{2} b', a'))$, using~\cref{eq:paramaterize-so3cx}.
Then, $\mathbf{G}' = \Y \iny{\mathbf{G}}'$ and $\mathbf{H}$ have non-orthogonal annihilators.
To wit, note that $\ann(\mathbf{G}') = \spn((\Yinv)\transpose (0, \sqrt{2}b', a'))$ and
$\ann(\mathbf{H}) = \spn((\Yinv)\transpose (-b, a, 0))$ by \cref{prop:annihilating-vector-under-gl-transform}. 
Then, 
\begin{align*}
\left\langle (\Yinv)\transpose \begin{bmatrix}
    0 \\ \sqrt{2} b' \\ a'
\end{bmatrix}, (\Yinv)\transpose \begin{bmatrix}
    -b \\ a \\ 0
\end{bmatrix} \right\rangle 
&=
\begin{bmatrix}
    0 & \sqrt{2} b' & a'
\end{bmatrix} \Yinv (\Yinv)\transpose \begin{bmatrix}
    -b \\ a \\ 0
\end{bmatrix} \\
&= \begin{bmatrix}
    0 & \sqrt{2} b' & a'
\end{bmatrix} (\Y\transpose \Y)^{-1} \begin{bmatrix}
    -b \\ a \\ 0
\end{bmatrix} \\
&= \begin{bmatrix}
    0 & \sqrt{2} b' & a'
\end{bmatrix} (\X)^{-1} \begin{bmatrix}
    -b \\ a \\ 0
\end{bmatrix} \\
&= -\sqrt{2}bb' \\
&\ne 0 \,.
\end{align*}
Then $\holts(\mathbf{G}', \mathbf{H})$, and thus $\holts(\mathbf{G}, \mathbf{H})$, is \sph by \cref{lem:not-orthogonal-pair-dichotomy-ternary}. 
\end{proof}

\begin{proposition}\label{prop:case-a3-ternary-non-mixing}
    Let $\mathbf{F} = \mathbf{e}_1\teh + \mathbf{e}_2\teh + \mathbf{e}_3\teh$.
        Let $\mathbf{G}$ and $\mathbf{H}$ be nondegenerate symmetric ternary signatures 
        such that $\holts(\mathbf{F}, \mathbf{G})$ and $\holts(\mathbf{F}, \mathbf{H})$ are not \sph, and
    \begin{enumerate}
        \item $\mathbf{G}$ is \sepbg and not in $\mathcal{E}$.
        \item $\mathbf{H}$ is \sepbr and not in $\mathcal{E}$.
    \end{enumerate}
    Then $\holts(\mathbf{G}, \mathbf{H})$ is \sph.
\end{proposition}
\begin{proof}
By \cref{prop:two-ternary-dichotomy-a-3}, we may assume that $\mathbf{G} = (x, y, 0)\teh + (0, 0, z)\teh$ and $\mathbf{H} = (p, 0, q)\teh + (0, r, 0)\teh$ for 
some $x, y, z, p, q, r \in \mathbb{C}$ and
$z, r \ne 0$.
As they are  not in  $\mathcal{E}$ we also have
 $x, y,  p, q \ne 0$.
We see that $\ann(\mathbf{G}) = \spn(\mathbf{n}_\mathbf{G})$ for $\mathbf{n}_{\mathbf{G}} = (-y, x, 0)$ while $\ann(\mathbf{H}) = \spn(\mathbf{n}_\mathbf{H})$ for $\mathbf{n}_{\mathbf{H}} = (-q, 0, p)$. 
Then, $\langle \mathbf{n}_{\mathbf{G}}, \mathbf{n}_{\mathbf{H}} \rangle = yq \ne 0$.
By \cref{lem:not-orthogonal-pair-dichotomy-ternary}, $\holts(\mathbf{G}, \mathbf{H})$ is \sph.
\end{proof}

\subsection{A Set of Ternary Signatures}\label{subsec:set-of-ternary}
We are now ready to prove the dichotomy of a set of ternary signatures.

\begin{theorem}\label{thm:ternary-set-dichotomy}
Let $\mathcal{F}$ be a set of nondegenerate symmetric ternary signatures.
Then, $\holts(\mathcal{F})$ is tractable if 
it
satisfies 
the tractability criterion of~\cref{thm:tractable-cases}. 
Otherwise, it is \sph. 
\end{theorem}
\begin{proof}
We assume that $\holts(\mathcal{F})$ is not \sph.
Therefore, we may assume that every signature in $\mathcal{F}$ is of a tractable type.
We separate into the cases where $\mathcal{F}$ contains a signature of certain type.
Then we apply the dichotomy of pairs to narrow down the possibilities and  apply the non-mixing results to pinpoint the tractable case applicable for $\mathcal{F}$. 

Suppose $\mathcal{F}$ contains a \casecogen or \casecodz signature.
Then from the dichotomy of pair of signatures (\cref{prop:two-ternary-dichotomy-C-1-gen,prop:two-ternary-dichotomy-C-1-dz}), we see that there must exist some orthogonal $T$ such that $\Yinv T \mathcal{F} \subseteq \casecsupp$. 
Then, $\mathcal{F}$ satisfies \cref{itm:tractable-c} of \cref{thm:tractable-cases}.

Suppose $\mathcal{F}$ contains a \casecosep signature.
There exists some orthogonal $T$ such that $\Yinv T \mathcal{F}$ contains a signature of the form $\iny{\mathbf{F}} = [*;x, 0;0, 0, 0;0, 0, 0, y]$ with $x, y \ne 0$.
For any $\iny{\mathbf{G}} \in \Yinv T\mathcal{F}$, we see from \cref{prop:two-ternary-dichotomy-C-1-sep}, that either $\iny{\mathbf{G}} \in \casecsupp$ or $\iny{\mathbf{G}}$ is \sepbg with the further property that $\iny{\mathbf{G}}\domres{\db\dg}$ is either degenerate or $[*, *, 0, 0]$.

If $\Yinv T \mathcal{F} \subset \sepbgmath$, 
then $\mathcal{F}$ satisfies \cref{itm:tractable-bg-r} of \cref{thm:tractable-cases}.
If $\Yinv T \mathcal{F} \subset \casecsupp$, 
then $\mathcal{F}$ satisfies \cref{itm:tractable-c} of \cref{thm:tractable-cases}.
Otherwise, 
there are $\iny{\mathbf{G}}, \iny{\mathbf{H}} \in \Yinv T \mathcal{F}$ such that $\iny{\mathbf{G}} \in \casecsupp \setminus \sepbgmath$ and $\iny{\mathbf{H}} \in \sepbgmath \setminus \casecsupp$.
Then, $\holts(\mathcal{F})$ is \sph by \cref{prop:case-c1-sep-ternary-non-mixing}.

Suppose $\mathcal{F}$ contains a \casebthree signature.
Then, there exist some orthogonal $T$ and $\mathbf{F} \in T \mathcal{F}$ with the form $\mathbf{F} = a \bzp \teh + b \bzm\teh + c \mathbf{e}_3\teh$ for $a ,b, c \ne 0$.
By \cref{prop:two-ternary-dichotomy-b-3}, 
any other signature $\mathbf{G} \in T \mathcal{F}$ is \sepbg such that $\holbs(\mathbf{F}\domres{\db\dg}, \mathbf{G}\domres{\db\dg})$ is tractable.
This means $\mathbf{G}\domres{\db\dg}$ is either $x \bzp\teh + y \bzm \teh$ or degenerate by the Boolean domain dichotomy.
Then $\mathcal{F}$ satisfies \cref{itm:tractable-bg-r} of \cref{thm:tractable-cases}.

Suppose $\mathcal{F}$ contains a \caseathree signature,
so for some orthogonal $T$, $\mathbf{F} = a \mathbf{e}_1\teh + b \mathbf{e}_2\teh + c \mathbf{e}_3\teh \in T \mathcal{F}$ for $a, b, c \ne 0$.
Let $\mathcal{F}_{\mathcal{E}} = T \mathcal{F} \cap \mathcal{E}$.
If $\mathcal{F}_{\mathcal{E}} = T \mathcal{F}$ 
then $\mathcal{F}$ satisfies \cref{itm:tractable-e} of \cref{thm:tractable-cases}.
Otherwise, by \cref{prop:two-ternary-dichotomy-a-3}, any signature in $T \mathcal{F} \setminus \mathcal{F}_{\mathcal{E}}$ must be $ij|k$ and degenerate on $ij$ for some $i,j,k$ depending on the signature.
There must be a single choice of $i, j, k$ such that all signatures in $T \mathcal{F} \subset ij | k$ because otherwise we get \sphness from \cref{prop:case-a3-ternary-non-mixing}.
If $T \mathcal{F} \subset ij | k$, 
then $\mathcal{F}$ satisfies \cref{itm:tractable-bg-r} of \cref{thm:tractable-cases}.

If $\mathcal{F}$ does not contain any \caseathree, \casebthree, \casecosep, \casecodz, or \casecogen signature, then every signature in $\mathcal{F}$ has a nonzero annihilator. 
Then by \cref{cor:ebd-set-ternary-dichotomy}, $\mathcal{F}$ satisfies \cref{itm:tractable-ann} of \cref{thm:tractable-cases}.
\end{proof}

\section{Adding Binary Signatures}
\label{sec:binary}
We now consider binary signatures.
If every ternary signature has a nonzero annihilator, 
we can easily extend \cref{cor:ebd-set-ternary-dichotomy} to a dichotomy including binary signatures 
using \cref{lem:not-orthogonal-pair-dichotomy-ternary} and the algorithm in \cref{thm:algorithm-ebd-with-binary}. 
So, we consider the remaining cases. 
Let $\mathbf{F}$ be a signature of type \caseathree, \casebthree, \casecosep, \casecodz, or \casecogen.
We prove two kinds of results. 
\cref{prop:casebthree-binary,prop:casea3-binary,prop:casecosep-allowed-binary,prop:casecodz-binary,prop:casecogen-binary} show a dichotomy of $\holts(\mathbf{F}, \mathbf{G})$ for a symmetric binary signature $\mathbf{G}$. 
\cref{prop:casea3-binary-not-mix,prop:casebtwo-binary-not-mix,prop:casecosep-allowed-binary-not-mix} show non-mixing, a $\holts(\mathbf{F}, \mathbf{G}, \mathbf{H})$ dichotomy for symmetric binary signatures $\mathbf{G}$ and $\mathbf{H}$. 
Combining those results, we obtain a dichotomy for a set of binary and ternary signatures in \cref{thm:dichotomy-ternary-binary}.

In \cref{prop:casebthree-binary,prop:casebtwo-binary-not-mix,prop:casecosep-allowed-binary,prop:casecosep-allowed-binary-not-mix,prop:casecodz-binary,prop:casecogen-binary}, we characterize the binary signatures in the setting $\holts(\X \mid \blank)$  after the $\Yinv$ transform. 
This is because the binary signatures take a nicer form in that setting. 
Note that a binary signature in $\casecsupp$ in matrix form has the support pattern $\begin{bsmallmatrix}
  * & * & * \\ * & 0 & 0 \\ * & 0 & *
  \end{bsmallmatrix}$.

\subsection{Some Calculations}
Because the representation of a tensor is not unique, we first prove some results on the forms of the constituent vectors of a given signature. 

\begin{proposition} \label{prop:calculation-sym-in-C}
    Let $\iny{\mathbf{F}}  = \sym(\mathbf{v}_1, \mathbf{v}_1, \mathbf{v}_2)$ for $\mathbf{v}_1 \not \sim \mathbf{v}_2 \in \mathbb{C}^3$.
    Suppose $\iny{\mathbf{F}} \in \casecsupp$.
    Then, $\mathbf{v}_1 \sim \mathbf{e}_1$ or $\{\mathbf{v}_1, \mathbf{v}_2\} = \{(*, 0, *), (*, 0, *)\}$.
\end{proposition}
\begin{proof}
Consider $\mathbf{G}_{p, q} = p \mathbf{v}_1\tew + q \sym(\mathbf{v}_1, \mathbf{v}_2)$. 
Note that for any $p, q \in \mathbb{C}$, there exists some $\mathbf{u}$ such that $\langle \mathbf{F}, \mathbf{u} \rangle = \mathbf{G}_{p,q}$ by choosing $\mathbf{u}$ such that $\langle \mathbf{u}, \mathbf{v}_1 \rangle = q$ and $\langle \mathbf{u}, \mathbf{v}_2 \rangle = p$.
Since $\casecsupp$ is closed under connecting a unary signature, 
$\mathbf{G}_{p,q} \in \casecsupp$ for any $p, q \in \mathbb{C}$.
Writing $\mathbf{v}_1 = (a, b, c)$ and $\mathbf{v}_2 = (x, y, z)$,
we have
$\mathbf{G}_{p, q}(\dg\dg) = p b^2 + 2 q b y $
and $\mathbf{G}_{p, q}(\dg\dr) = p b c + q (b z + c y)$.
For $\mathbf{G}_{p, q}(\dg\dg)$ and $\mathbf{G}_{p, q}(\dg\dr)$ to be always $0$ for every $p, q$, it must be the case 
that $b, y = 0$ or $b , c = 0$.
\end{proof}

\begin{proposition} \label{prop:calculation-sym-bg-r}
    Let $\iny{\mathbf{F}} = \sym(\mathbf{v}_1, \mathbf{v}_1, \mathbf{v}_2)$ for $\mathbf{v}_1 \not \sim \mathbf{v}_2 \in \mathbb{C}^3$.
    Suppose $\iny{\mathbf{F}}$ is \sepbg.
    Then, $\{\mathbf{v}_1, \mathbf{v}_2\} = \{(*, *, 0), (*, *, 0)\}$.
    In particular, $\iny{\mathbf{F}}$ is supported on $\db\dg$.
\end{proposition}
\begin{proof}
Consider $\mathbf{G}_{p, q} = p \mathbf{v}_1\tew + q \sym(\mathbf{v}_1, \mathbf{v}_2)$. 
Note that for any $p, q \in \mathbb{C}$, there exists some $\mathbf{u}$ such that $\langle \mathbf{F}, \mathbf{u} \rangle = \mathbf{G}_{p,q}$ by choosing $\mathbf{u}$ such that $\langle \mathbf{u}, \mathbf{v}_1 \rangle = q$ and $\langle \mathbf{u}, \mathbf{v}_2 \rangle = p$.
Since \sepbg is closed under connecting a unary signature, 
$\mathbf{G}_{p,q}$ is \sepbg for any $p, q \in \mathbb{C}$.
Writing $\mathbf{v}_1 = (a, b, c)$ and $\mathbf{v}_2 = (x, y, z)$,
we have
$\mathbf{G}_{p, q}(\db\dr) = p a c + q(a z + c x)$
and $\mathbf{G}_{p, q}(\dg\dr) = p b c + q (b z + c y)$.
For $\mathbf{G}_{p, q}(\db\dr)$ and $\mathbf{G}_{p, q}(\dg\dr)$ to be always $0$ for every $p, q$, it must be the case that 
$c = 0$ or $a, b = 0$.

If $c = 0$, then we need $az, bz = 0$.
If further $z = 0$, then we are done. If $z \ne 0$, then $a, b = 0$ which makes $\mathbf{v}_1 = 0$ a contradiction.

If $a, b = 0$, then we need $cx, cy = 0$. If $c = 0$, then $\mathbf{v}_1 = 0$, a contradiction.
If $c \ne 0$, then $x, y = 0$, so $\mathbf{v}_1 \sim \mathbf{v}_2$, a contradiction again.
\end{proof}

\begin{proposition} \label{prop:calculation-tensor-2-sum-in-C}
    Let $\iny{\mathbf{F}} =\mathbf{v}_1\teh + \mathbf{v}_2\teh$ for $\mathbf{v}_1 \not \sim \mathbf{v}_2 \in \mathbb{C}^3$.
    Suppose $\iny{\mathbf{F}} \in \casecsupp$.
    Then, $\{\mathbf{v}_1, \mathbf{v}_2\} = \{(*, 0, *), (*, 0, *)\}$.
    In particular, $\iny{\mathbf{F}}$ is supported on $\db\dr$.
\end{proposition}
\begin{proof}
Consider $\mathbf{G}_{p, q} = p \mathbf{v}_1\tew + q \mathbf{v}_2\tew$. 
Note that for any $p, q \in \mathbb{C}$, there exists some $\mathbf{u}$ such that $\langle \mathbf{F}, \mathbf{u} \rangle = \mathbf{G}_{p,q}$ by choosing $\mathbf{u}$ such that $\langle \mathbf{u}, \mathbf{v}_1 \rangle = p$ and $\langle \mathbf{u}, \mathbf{v}_2 \rangle = q$.
Since $\casecsupp$ is closed under connecting a unary signature, 
$\mathbf{G}_{p,q} \in \casecsupp$ for any $p, q \in \mathbb{C}$.
Writing $\mathbf{v}_1 = (a, b, c)$ and $\mathbf{v}_2 = (x, y, z)$,
we have
$\mathbf{G}_{p, q}(\dg\dg) = p b^2 + q y^2$.
For $\mathbf{G}_{p, q}(\dg\dg)$ to be always $0$ for every $p, q$, it must be the case that $b = 0$ and $y = 0$.
\end{proof}

\begin{proposition} \label{prop:calculation-tensor-2-sum-in-BG-R}
    Let $\mathbf{F} =\mathbf{v}_1\teh + \mathbf{v}_2\teh$ for $\mathbf{v}_1 \not \sim \mathbf{v}_2 \in \mathbb{C}^3$.
    Suppose $\mathbf{F}$ is \sepbg.
    Then, 
    $\{\mathbf{v}_1, \mathbf{v}_2\}$ is $\{(*, *, 0), (0, 0, *)\}$ or $\{(*, *, 0), (*, *, 0)\}$.
\end{proposition}
\begin{proof}
Consider $\mathbf{G}_{p, q} = p \mathbf{v}_1\tew + q \mathbf{v}_2\tew$. 
Note that for any $p, q \in \mathbb{C}$, there exists some $\mathbf{u}$ such that $\langle \mathbf{F}, \mathbf{u} \rangle = \mathbf{G}_{p,q}$ by choosing $\mathbf{u}$ such that $\langle \mathbf{u}, \mathbf{v}_1 \rangle = p$ and $\langle \mathbf{u}, \mathbf{v}_2 \rangle = q$.
Since \sepbg is closed under connecting a unary signature, 
$\mathbf{G}_{p,q}$ is \sepbg for any $p, q \in \mathbb{C}$.
Writing $\mathbf{v}_1 = (a, b, c)$ and $\mathbf{v}_2 = (x, y, z)$,
we have
$\mathbf{G}_{p, q}(\db\dr) = p ac + q xz$
and
$\mathbf{G}_{p, q}(\dg\dr) = p bc + q y z$.
For $\mathbf{G}_{p, q}(\db\dr)$ and $\mathbf{G}_{p, q}(\dg\dr)$ to be always $0$,
we need $ac, bc, xz, yz = 0$.
If $c \ne 0$, then $a, b = 0$. Since $\mathbf{v}_1 \not \sim \mathbf{v}_2$, this means one of $x, y \ne 0$ so $z = 0$ giving the desired form.
If $c = 0$, then $z = 0$ or $x, y = 0$ again giving the desired form.
\end{proof}

\subsection{A Pair of Ternary and Binary Signatures}

\begin{proposition}\label{prop:casebthree-binary}
Let $\mathbf{F} = \bzp\teh + \bzm\teh + \mathbf{e}_3\teh$.
Let $\iny{\mathbf{G}}$ be a nondegenerate symmetric binary signature.
Then,
$\holts(\X \mid \Yinv \mathbf{F}, \iny{\mathbf{G}})$ is \sph unless
$\iny{\mathbf{G}}$ is \sepbg, \swbg, or in $\mathcal{E}$.

In particular,
if $\iny{\mathbf{G}}$ is \sepbg,
then it must be of the form 
$\begin{bsmallmatrix}
    a & 0 & 0 \\
    0 & b & 0 \\
    0 & 0 & c
\end{bsmallmatrix}$, $\begin{bsmallmatrix}
    0 & a & 0 \\
    a & 0 & 0 \\
    0 & 0 & c
\end{bsmallmatrix}$, or 
$\begin{bsmallmatrix}
    a^2 & ab & 0 \\
    ab & b^2 & 0 \\
    0 & 0 & c
\end{bsmallmatrix}$ for some $ a, b, c \in \mathbb{C}$. 
If further $\iny{\mathbf{G}} \notin \mathcal{E}$,
then $\iny{\mathbf{G}} = \begin{bsmallmatrix}
    a^2 & ab & 0 \\
    ab & b^2 & 0 \\
    0 & 0 & c
\end{bsmallmatrix}$ with $a, b, c \ne 0$.
\end{proposition}
\begin{proof}
Let $\iny{\mathbf{F}} = \Yinv \mathbf{F} = \mathbf{e}_1\teh + \mathbf{e}_2\teh + \mathbf{e}_3\teh$.
Assume $\holts(\X \mid  \iny{\mathbf{F}}, \iny{\mathbf{G}})$ is not \sph.
Let $\mathbf{v}_i$ be the $i$th column of the matrix form of $\iny{\mathbf{G}}$.
Let $\iny{\mathbf{H}}_{ij} = \mathbf{v}_i \teh + \mathbf{v}_j\teh$.
By \cref{lem:gadget-set-general-form}
and \cref{subsec:gc-case-b-3},
\begin{align*}
    \holts(\X \mid \iny{\mathbf{F}}, \iny{\mathbf{G}}) 
    &\ge_T \holts(\X \mid \iny{\mathbf{F}}, \iny{\mathbf{G}}, \agx(\casebthreemath)) \\
    &\ge_T \holts(\X \mid \iny{\mathbf{F}}, \iny{\mathbf{G}} \X \agx(\casebthreemath)) \\
    &\ge_T \holts(\X \mid \iny{\mathbf{F}}, \iny{\mathbf{H}}_{12}, \iny{\mathbf{H}}_{23}, \iny{\mathbf{H}}_{13})
\end{align*} where $\iny{\mathbf{G}} \X \agx(\casebthreemath) = \spn(\mathbf{v}_1\teh, \mathbf{v}_2\teh, \mathbf{v}_3\teh)$.
Since $\iny{\mathbf{G}}$ is nondegenerate, it must have two columns that are linearly independent.

\cref{prop:two-ternary-dichotomy-b-3} on $\holts(\X \mid \iny{\mathbf{F}}, \iny{\mathbf{H}}_{ij})$ implies that whenever $\mathbf{v}_i \not \sim \mathbf{v}_j$, it must be the case that $\{\mathbf{v}_i, \mathbf{v}_j\} \sim \{\mathbf{e}_1, \mathbf{e}_2\}$ or $\{(*, *, 0), (0, 0, *)\}$.
Note that the latter includes the case like 
$\{\mathbf{v}_i, \mathbf{v}_j\} \sim \{\mathbf{e}_1, \mathbf{e}_3\}$.
If $\mathbf{v}_3 = 0$, then as  $\iny{\mathbf{G}}$ is symmetric, it is supported on $\db\dg$, so \sepbg.
Hence, we may assume that $\mathbf{v}_3 \ne 0$.
In the following cases, we use the fact that $\iny{\mathbf{G}}$ is symmetric.
\begin{itemize}
    \item
Suppose $\mathbf{v}_1 \not\sim \mathbf{v}_2$. 
\begin{enumerate}
  \item $\mathbf{v}_1 \sim \mathbf{e}_1$, $\mathbf{v}_2 \sim \mathbf{e}_2$:
    Then, $\iny{\mathbf{G}} = \begin{bsmallmatrix}
      * & 0 & * \\
      0 & * & * \\
      0 & 0 & *
    \end{bsmallmatrix} = \begin{bsmallmatrix}
      * & 0 & 0 \\
      0 & * & 0 \\
      0 & 0 & *
    \end{bsmallmatrix}$. 
    So $\iny{\mathbf{G}} \in \mathcal{E}$.
  \item $\mathbf{v}_1 \sim \mathbf{e}_2$, $\mathbf{v}_2 \sim \mathbf{e}_1$:
    Then, $\iny{\mathbf{G}} = \begin{bsmallmatrix}
      0 & * & * \\
      * & 0 & * \\
      0 & 0 & *
    \end{bsmallmatrix} = \begin{bsmallmatrix}
      0 & * & 0 \\
      * & 0 & 0 \\
      0 & 0 & *
    \end{bsmallmatrix}$. So $\iny{\mathbf{G}} \in \mathcal{E}$.
  \item $\mathbf{v}_1 = (*, *, 0)$, $\mathbf{v}_2 = (0, 0, *)$:
    Then, $\iny{\mathbf{G}} = \begin{bsmallmatrix}
      * & 0 & * \\
      * & 0 & * \\
      0 & * & *
    \end{bsmallmatrix} = \begin{bsmallmatrix}
      * & 0 & 0 \\
      0 & 0 & * \\
      0 & * & *
    \end{bsmallmatrix}$.
    If $\mathbf{v}_2 \sim \mathbf{v}_3$ or $\mathbf{v}_2 = 0$, then $\iny{\mathbf{G}} = \begin{bsmallmatrix}
        * & 0 & 0 \\ 0 & 0 & 0 \\ 0 & 0 & *
    \end{bsmallmatrix} \in \mathcal{E}$.
    Otherwise, applying \cref{prop:two-ternary-dichotomy-b-3} on $\holts(\X \mid \iny{\mathbf{F}}, \iny{\mathbf{H}}_{23})$ implies 
    $\iny{\mathbf{H}}_{23} = \mathbf{v}_2\teh + \mathbf{v}_3\teh$ is \sepbg.
    By \cref{prop:calculation-tensor-2-sum-in-BG-R}, 
    since $\mathbf{v}_2$ is a nonzero multiple of $\mathbf{e}_3$, $\mathbf{v}_3 = (*, *, 0)$.
    But we are given $\mathbf{v}_3 = (0, *,*)$, so in fact  $\mathbf{v}_3 = (0, *,0)$.
    Then, $\iny{\mathbf{G}} \in \mathcal{E}$.
  \item $\mathbf{v}_1 = (0, 0, *)$, $\mathbf{v}_2 = (*, *, 0)$:
    Then, $\iny{\mathbf{G}} = \begin{bsmallmatrix}
      0 & * & * \\
      0 & * & * \\
      * & 0 & *
    \end{bsmallmatrix}= \begin{bsmallmatrix}
      0 & 0 & * \\
      0 & * & 0 \\
      * & 0 & *
    \end{bsmallmatrix}$,
    and applying a similar argument on 
    $\holts(\X \mid \iny{\mathbf{F}}, \iny{\mathbf{H}}_{13})$
    gives $\iny{\mathbf{G}} \in \mathcal{E}$.
\end{enumerate}

\item
Suppose $\mathbf{v}_1 \not \sim \mathbf{v}_3$.
\begin{enumerate}
  \item $\mathbf{v}_1 \sim \mathbf{e}_1$, $\mathbf{v}_3 \sim \mathbf{e}_2$:
    Then, $\iny{\mathbf{G}} = \begin{bsmallmatrix}
      * & * & 0 \\
      0 & * & * \\
      0 & * & 0
    \end{bsmallmatrix} =\begin{bsmallmatrix}
    * & 0 & 0 \\
    0 & * & * \\
    0 & * & 0
  \end{bsmallmatrix}$,
    and applying a similar argument on 
    $\holts(\X \mid \iny{\mathbf{F}}, \iny{\mathbf{H}}_{23})$, 
    gives $\iny{\mathbf{G}} \in \mathcal{E}$.
  \item $\mathbf{v}_1 \sim \mathbf{e}_2$, $\mathbf{v}_3 \sim \mathbf{e}_1$:
    Then, $\iny{\mathbf{G}} = \begin{bsmallmatrix}
      0 & * & * \\
      * & * & 0 \\
      0 & * & 0
    \end{bsmallmatrix} = \begin{bsmallmatrix}
      0 & * & 0 \\
      * & * & 0 \\
      0 & 0 & 0
    \end{bsmallmatrix}$,
    so $\iny{\mathbf{G}}$ is \sepbg.
  \item $\mathbf{v}_1 = (*, *, 0)$, $\mathbf{v}_3 = (0, 0, *)$:
    Then, $\iny{\mathbf{G}} = \begin{bsmallmatrix}
      * & * & 0 \\
      * & * & 0 \\
      0 & * & *
    \end{bsmallmatrix} = \begin{bsmallmatrix}
      * & * & 0 \\
      * & * & 0 \\
      0 & 0 & *
    \end{bsmallmatrix}$
    so $\iny{\mathbf{G}}$ is \sepbg.
  \item $\mathbf{v}_1 = (0, 0, *)$, $\mathbf{v}_3 = (*, *, 0)$:
    Then, $\iny{\mathbf{G}} = \begin{bsmallmatrix}
      0 & * & * \\
      0 & * & * \\
      * & * & 0
    \end{bsmallmatrix}= \begin{bsmallmatrix}
      0 & 0 & * \\
      0 & * & * \\
      * & * & 0
    \end{bsmallmatrix}$,
    and applying a similar argument on 
    $\holts(\X \mid \iny{\mathbf{F}}, \iny{\mathbf{H}}_{12})$, 
    gives either $\iny{\mathbf{G}} \in \mathcal{E}$ or \swbg.
\end{enumerate}

\item
The argument for $\mathbf{v}_2 \not \sim \mathbf{v}_3$ is similar.
\end{itemize}

We derive the specific forms for $\iny{\mathbf{G}}$. 
Suppose $\iny{\mathbf{G}}$ is \sepbg.
Then, we must have $\iny{\mathbf{G}} = \begin{bsmallmatrix}
    x & y & 0 \\
    y & z & 0 \\
    0 & 0 & c
\end{bsmallmatrix}$ for some $x, y, z, c \in \mathbb{C}$.
By \cref{subsec:gc-case-b-3}, we have 
\begin{align*}
    \holts(\X \mid \iny{\mathbf{F}}, \iny{\mathbf{G}}) &\ge_T \holts(\iny{\mathbf{F}}, \iny{\mathbf{G}}) \\
    &\ge_T \holbs(\iny{\mathbf{F}}\domres{\db\dg}, \iny{\mathbf{G}}\domres{\db\dg}) \\
    &=_T \holbs([1, 0, 0, 1], [x, y, z]) \,.
\end{align*} 
Applying the Boolean domain dichotomy, $[x, y, z]$ must be $[*, 0, *]$, $[0, *, 0]$ or degenerate.
Now, suppose further that $\iny{\mathbf{G}} \notin \mathcal{E}$.
The first two forms give $\iny{\mathbf{G}} \in \mathcal{E}$, so $[x, y, z]$ must be degenerate, giving $[x, y, z] = [a^2, ab, b^2]$ for some $a, b \in \mathbb{C}$.
If $a = 0$ or $b = 0$, then $\iny{\mathbf{G}} \in \mathcal{E}$, so $a, b \ne 0$. 
If $c = 0$, then $\iny{\mathbf{G}}$ is degenerate, so $c \ne 0$. 
\end{proof}

\begin{proposition}\label{prop:casea3-binary}
Let $\mathbf{F} = \mathbf{e}_1\teh + \mathbf{e}_2\teh + \mathbf{e}_3\teh$.
Let $\mathbf{G}$ be a nondegenerate symmetric binary signature. 
Then $\holts(\mathbf{F}, \mathbf{G})$ is 
\sph unless 
$\mathbf{G}$ is $ij|k$, or \swhelper{ij;k} or in $\mathcal{E}$,
for some $\{i,j,k\} = \{1,2,3\}$.

In particular,
\begin{enumerate}
\item 
If $\mathbf{G}$ is \sepbg but not in $\mathcal{E}$,
then $\mathbf{G} = \begin{bsmallmatrix}
    a^2 & ab & 0 \\
    ab & b^2 & 0 \\
    0 & 0 & c
\end{bsmallmatrix}$ with $a, b, c \ne 0$.
\item
If $\mathbf{G}$ is \sepgr but not in $\mathcal{E}$,
then $\mathbf{G} = \begin{bsmallmatrix}
    c & 0 & 0 \\
    0 & a^2 & ab \\
    0 & ab & b^2
\end{bsmallmatrix}$ with $a, b, c \ne 0$.
\item
If $\mathbf{G}$ is \sepbr but not in $\mathcal{E}$,
then $\mathbf{G} = \begin{bsmallmatrix}
    a^2 & 0 & ab \\
    0 & c & 0 \\
    ab & 0 & b^2
\end{bsmallmatrix}$ with $a, b, c \ne 0$.
\end{enumerate}
\end{proposition}
\begin{proof}
The proof is similar to \cref{prop:casebthree-binary}. Using $\mathbf{G} \ag(\caseathreemath) = \spn(\mathbf{v}_1\teh, \mathbf{v}_2\teh, \mathbf{v}_3\teh)$,
for all $i \ne j$ we have
$\holts(\mathbf{F}, \mathbf{G}) \ge_T \holts(\mathbf{F}, \mathbf{v}_i\teh
+ \mathbf{v}_j\teh)$.
By \cref{prop:two-ternary-dichotomy-a-3}, $\{\mathbf{v}_i, \mathbf{v}_j\} \sim \{\mathbf{e}_k, \mathbf{e}_\ell\}$ for some $k \ne \ell$,
or $\{(*, *, 0), (0, 0, *)\}$, $\{(*, 0, *), (0, *, 0)\}$, or $\{(0, *, *), (*, 0, 0)\}$.

We derive the specific forms for $\mathbf{G}$. 
Suppose $\mathbf{G}$ is \sepbg. 
The other cases are similar.
Then, we must have $\mathbf{G} = \begin{bsmallmatrix}
    x & y & 0 \\
    y & z & 0 \\
    0 & 0 & c
\end{bsmallmatrix}$ for some $x, y, z, c \in \mathbb{C}$.
By \cref{subsec:gc-case-a-3}, we have 
\begin{align*}
    \holts(\mathbf{F}, \mathbf{G}) &\ge_T 
    \holbs(\mathbf{F}\domres{\db\dg}, \mathbf{G}\domres{\db\dg}) \\
    &=_T \holbs([1, 0, 0, 1], [x, y, z]) \,.
\end{align*} 
Applying the Boolean domain dichotomy, 
$[x, y, z]$ must be $[*, 0, *]$, $[0, *, 0]$ or degenerate.
Suppose further that $\mathbf{G} \notin \mathcal{E}$.
The first two cases give $\mathbf{G} \in \mathcal{E}$, so $[x, y, z]$ must be degenerate, giving $[x, y, z] = [a^2, ab, b^2]$ for some $a, b \in \mathbb{C}$.
If $a = 0$ or $b = 0$, then $\mathbf{G}  \in \mathcal{E}$, so $a, b \ne 0$. 
If $c = 0$, then $\mathbf{G} $ is degenerate, so $c \ne 0$. 
\end{proof}

\begin{proposition}\label{prop:casecosep-allowed-binary}
Let $\mathbf{F} = \sym(\bzp, \bzp, \bzm) + \mathbf{e}_3\teh$.
Let $\iny{\mathbf{G}}$ be a nondegenerate symmetric binary signature.
Then, $\holts(\mathbf{X} \mid \Yinv \mathbf{F}, \iny{\mathbf{G}})$ is \sph unless 
$\iny{\mathbf{G}}$ is \sepbg, \swbg, or 
in $\casecsupp$.

In particular, 
    If $\iny{\mathbf{G}}$ is \sepbg,
    then it must be of the form 
    $\begin{bsmallmatrix}
    a & b & 0 \\
    b & 0 & 0 \\
    0 & 0 & c
\end{bsmallmatrix}$
or
    $\begin{bsmallmatrix}
    a^2 & ab & 0 \\
    ab & b^2 & 0 \\
    0 & 0 & c
\end{bsmallmatrix}$ for some $a, b, c \in \mathbb{C}$.
If further $\iny{\mathbf{G}} \notin \casecsupp$, then $\iny{\mathbf{G}} = \begin{bsmallmatrix}
    a^2 & ab & 0 \\
    ab & b^2 & 0 \\
    0 & 0 & c
\end{bsmallmatrix}$ with $b, c \ne 0$.
\end{proposition}
\begin{proof}
Let $\iny{\mathbf{F}} = \Yinv \mathbf{F}$.
Assume $\holts(\mathbf{X} \mid \iny{\mathbf{F}}, \iny{\mathbf{G}})$ is not \sph.
Let $\mathbf{v}_i$ be the $i$th column of $\iny{\mathbf{G}}$.
If $\mathbf{v}_2 = 0$, then $\iny{\mathbf{G}} \in \casecsupp$
as $\iny{\mathbf{G}} $ is also symmetric, so we may assume $\mathbf{v}_2 \ne 0$.

By \cref{lem:gadget-set-general-form} and \cref{subsec:gc-case-c1-sep},
\begin{align*}
    \holts(\X \mid \iny{\mathbf{F}}, \iny{\mathbf{G}}) &\ge_T \holts(\X \mid \iny{\mathbf{F}}, V) \\
    &\ge_T \holts(\X \mid \iny{\mathbf{F}}, \iny{\mathbf{H}}_{12}, \iny{\mathbf{H}}_{23})
\end{align*} where $V = \iny{\mathbf{G}} \X \agx(\casecosepmath) = \spn(\mathbf{v}_2\teh, \sym(\mathbf{v}_2, \mathbf{v}_2, \mathbf{v}_1), \mathbf{v}_3\teh)$,
$\iny{\mathbf{H}}_{12} = \sym(\mathbf{v}_2, \mathbf{v}_2, \mathbf{v}_1) \in V$ 
and
$\iny{\mathbf{H}}_{23} = \mathbf{v}_2\teh + \mathbf{v}_3\teh \in V$.
We use $\iny{\mathbf{H}}_{12}$ and $\iny{\mathbf{H}}_{23}$ to analyze what forms $\mathbf{v}_1, \mathbf{v}_2$ can take and $\mathbf{v}_2, \mathbf{v}_3$ can take.

Suppose $\mathbf{v}_1 \not \sim \mathbf{v}_2$.
By \cref{prop:two-ternary-dichotomy-C-1-sep} on $\holts(\X \mid \iny{\mathbf{F}}, \iny{\mathbf{H}}_{12})$,
$\iny{\mathbf{H}}_{12}$ must be \sepbg or in $\casecsupp$. 
\begin{enumerate}
\item 
If $\iny{\mathbf{H}}_{12}$ is \sepbg, then by \cref{prop:calculation-sym-bg-r}, $\{\mathbf{v}_1, \mathbf{v}_2\} = \{(*, *, 0), (*, *, 0) \}$.
Then, by the symmetry of $\iny{\mathbf{G}}$, $\iny{\mathbf{G}}$ is \sepbg, so we are done.
\item 
If $\iny{\mathbf{H}}_{12} \in \casecsupp$, then by \cref{prop:calculation-sym-in-C}, $\mathbf{v}_2 \sim \mathbf{e}_1$ or $\{\mathbf{v}_1, \mathbf{v}_2\} = \{(*, 0, *), (*, 0, *) \}$.
\begin{enumerate}
\item 
If $\mathbf{v}_2 \sim \mathbf{e}_1$, then $\iny{\mathbf{G}} \in \casecsupp$, so we are done. 
\item If $\{\mathbf{v}_1, \mathbf{v}_2\} = \{(*, 0, *), (*, 0, *) \}$, we analyze this case later. 
\end{enumerate}
\end{enumerate}
Therefore, the remaining cases from the analysis of $\iny{\mathbf{H}}_{12}$ is $\mathbf{v}_1 \sim \mathbf{v}_2$ or $\{\mathbf{v}_1, \mathbf{v}_2\} = \{(*, 0, *), (*, 0, *) \}$.

Suppose $\mathbf{v}_2 \not \sim \mathbf{v}_3$.
By \cref{prop:two-ternary-dichotomy-C-1-sep} on $\holts(\X \mid \iny{\mathbf{F}}, \iny{\mathbf{H}}_{23})$,
$\iny{\mathbf{H}}_{23}$ must be \sepbg or in $\casecsupp$. 
\begin{enumerate}
    \item 
If $\iny{\mathbf{H}}_{23} \in \casecsupp$, then by \cref{prop:calculation-tensor-2-sum-in-C}, $\{\mathbf{v}_2, \mathbf{v}_3\} = \{(*, 0, *), (*, 0, *) \}$.
Then, by the symmetry of $\iny{\mathbf{G}}$, $\iny{\mathbf{G}} \in \casecsupp$, so we are done.
\item 
If $\iny{\mathbf{H}}_{23}$ is \sepbg, then by \cref{prop:calculation-tensor-2-sum-in-BG-R}, $\{\mathbf{v}_2, \mathbf{v}_3\} = \{(*, *, 0), (*, *, 0) \}$ or $\{(*, *, 0), (0, 0, *) \}$. 
\begin{enumerate}
\item 
    If $\mathbf{v}_2 = (*, *, 0)$ and $\mathbf{v}_3 = (*, *, 0)$, then $\iny{\mathbf{H}}_{23}$ is supported on $\db\dg$. 
    Because $\mathbf{v}_2 \not \sim \mathbf{v}_3$, we have  $\proj_{23} \mathbf{v}_2 \not \sim \proj_{23} \mathbf{v}_3$, so $\iny{\mathbf{H}}_{23}\domres{\db\dg} = (\proj_{23}\mathbf{v}_2)\teh + (\proj_{23}\mathbf{v}_3)\teh$ is nondegenerate. 
    By \cref{prop:two-ternary-dichotomy-C-1-sep}, $\holts(\X \mid \iny{\mathbf{F}}, \iny{\mathbf{H}}_{23}) =_T \holts(\mathbf{F}, \Y \iny{\mathbf{H}}_{23})$ is \sph because
    \[\holbs(\mathbf{F}\domres{\db\dg}, (\Y \iny{\mathbf{H}}_{23})\domres{\db\dg}) =_T 
    \holbs(\bne \mid [0, 1, 0, 0], \iny{\mathbf{H}}_{23}\domres{\db\dg})\] 
    is \sph by the Boolean domain dichotomy.
\item 
    If $\mathbf{v}_2 = (*, *, 0)$ and $\mathbf{v}_3 = (0, 0, *)$, then 
    by the symmetry of $\iny{\mathbf{G}}$,
    $\iny{\mathbf{G}}$ is \sepbg. 
    \item If $\mathbf{v}_2 = (0, 0, *)$ and $\mathbf{v}_3 = (*, *, 0)$, we analyze later.
    \end{enumerate}
\end{enumerate}
Therefore, the remaining cases from the analysis of $\iny{\mathbf{H}}_{23}$ is $\mathbf{v}_2 \sim \mathbf{v}_3$
or [$\mathbf{v}_2 = (0, 0, *) \wedge \mathbf{v}_3 = (*, *, 0)$].

We now consider the four combinations of possibilities for $\{\mathbf{v}_1, \mathbf{v}_2 \}$ and $\{\mathbf{v}_2, \mathbf{v}_3\}$.
\begin{enumerate}
    \item $\mathbf{v}_1 \sim \mathbf{v}_2$ and $\mathbf{v}_2 \sim \mathbf{v}_3$:
    Since $\mathbf{v}_2 \ne 0$, this implies $\iny{\mathbf{G}}$ is degenerate, contrary to  assumption.
    \item $\mathbf{v}_1 \sim \mathbf{v}_2$, $\mathbf{v}_2 = (0, 0, *)$ and $\mathbf{v}_3 = (*, *, 0)$:
    Then, $\iny{\mathbf{G}} = \begin{bsmallmatrix}
        0 & 0 & * \\
        0 & 0 & * \\
        * & * & 0
    \end{bsmallmatrix}$, so $\iny{\mathbf{G}}$ is \swbg.
    \item $\mathbf{v}_1 = (*, 0, *)$, $\mathbf{v}_2 = (*, 0, *)$ and $\mathbf{v}_2 \sim \mathbf{v}_3$:
    Then, $\iny{\mathbf{G}} = \begin{bsmallmatrix}
        * & * & * \\
        0 & 0 & 0 \\
        * & * & *
    \end{bsmallmatrix} = \begin{bsmallmatrix}
        * & 0 & * \\
        0 & 0 & 0 \\
        * & 0 & *
    \end{bsmallmatrix}$, so $\iny{\mathbf{G}} \in \casecsupp$.
    \item $\mathbf{v}_1 = (*, 0, *)$, $\mathbf{v}_2 = (*, 0, *)$, $\mathbf{v}_2 = (0, 0, *)$ and $\mathbf{v}_3 = (*, *, 0)$:
    Then, $\iny{\mathbf{G}} = \begin{bsmallmatrix}
        * & 0 & * \\
        0 & 0 & * \\
        * & * & 0
    \end{bsmallmatrix} = \begin{bsmallmatrix}
        x & 0 & y \\
        0 & 0 & z \\
        y & z & 0
    \end{bsmallmatrix}$ for some $x, y, z \in \mathbb{C}$.
    Since $\mathbf{v}_2 \ne 0$, we have $z \ne 0$.
    If $x = 0$, then $\iny{\mathbf{G}}$ is \swbg, so we may assume $x \ne 0$.
    Then, we see that $\iny{\mathbf{H}}_{12} = \sym((0, 0, z), (0, 0, z), (x, 0, y))$ is in $\casecsupp$ but not \sepbg by \cref{prop:calculation-sym-bg-r}. 
    Also, we see that $\iny{\mathbf{H}}_{23} = (0, 0, z)\teh + (y, z, 0)\teh$ is \sepbg but not in $\casecsupp$ by \cref{prop:calculation-tensor-2-sum-in-C}. 
    Then, $\holts(\X \mid \iny{\mathbf{F}}, \iny{\mathbf{H}}_{12}, \iny{\mathbf{H}}_{23})$ is \sph by \cref{prop:case-c1-sep-ternary-non-mixing}.
\end{enumerate}

We derive the specific forms for $\iny{\mathbf{G}}$. 
Suppose $\iny{\mathbf{G}}$ is \sepbg.
Then, we must have $\iny{\mathbf{G}} = \begin{bsmallmatrix}
    x & y & 0 \\
    y & z & 0 \\
    0 & 0 & c
\end{bsmallmatrix}$ for some $x, y, z, c \in \mathbb{C}$.
By \cref{subsec:gc-case-c1-sep}, we have 
\begin{align*}
    \holts(\X \mid \iny{\mathbf{F}}, \iny{\mathbf{G}}) &\ge_T \holbs([0, 0, 1, 0], \iny{\mathbf{G}}\domres{\db\dg}) \\
    &=_T \holbs([0, 0, 1, 0] \mid [x, y, z]) \,.
\end{align*} 
Applying the Boolean domain dichotomy, $[x, y, z]$ must be $[*, *, 0]$ or degenerate.
Suppose further that $\iny{\mathbf{G}} \notin \casecsupp$.
The case $[x, y, z] = [*, *, 0]$, i.e., $z=0$, gives $\iny{\mathbf{G}} \in \casecsupp$, so $[x, y, z]$ must be degenerate, giving $[x, y, z] = [a^2, ab, b^2]$ for some $a, b \in \mathbb{C}$.
If $b = 0$, then $\iny{\mathbf{G}} \in \casecsupp$, so $b \ne 0$. 
If $c = 0$, then $\iny{\mathbf{G}}$ is degenerate, so $c \ne 0$. 
\end{proof}

\begin{proposition}
\label{prop:casecodz-binary}
  Let $\mathbf{F}$ be a \casecodz signature 
  such that $\Yinv \mathbf{F} = [x; 1, y; 0, 0, z; 0, 0, 0, 0]$ for some $x, y, z$ such that $z \ne 0$.
  Let $\iny{\mathbf{G}}$ be a nondegenerate symmetric binary signature.
  Then, $\holts(\mathbf{X} \mid \Yinv \mathbf{F}, \iny{\mathbf{G}})$ is \sph unless
  $\iny{\mathbf{G}} \in \casecsupp$.
\end{proposition}
\begin{proof}
Let 
$\iny{\mathbf{F}} = \Yinv \mathbf{F}$.
Assume $\holts(\mathbf{X} \mid \iny{\mathbf{F}}, \iny{\mathbf{G}})$ is not \sph.
  Let $\mathbf{v}_i$ be the $i$th column of $\iny{\mathbf{G}}$.
If $\mathbf{v}_2 = 0$, then $\iny{\mathbf{G}} \in \casecsupp$, so we may assume that $\mathbf{v}_2 \ne 0$.
By \cref{lem:gadget-set-general-form} and \cref{subsec:gc-case-c1dz},
\begin{align*} 
\holts(\X \mid \iny{\mathbf{F}}, \iny{\mathbf{G}}) 
&\ge_T \holts(\X \mid \iny{\mathbf{F}}, \iny{\mathbf{G}}, \agx(\casecodzmath, \lambda)) \\
&\ge_T \holts(\X \mid \iny{\mathbf{F}}, V)
\end{align*}
where for some $\lambda \ne 0$,
\[V = \iny{\mathbf{G}} \mathbf{X} \agx(\casecodzmath, \lambda) = \spn(\mathbf{v}_2\teh, \sym(\mathbf{v}_2, \mathbf{v}_2, \mathbf{v}_3), \sym(\mathbf{v}_2, \mathbf{v}_2, \mathbf{v}_1) + \lambda \sym(\mathbf{v}_2, \mathbf{v}_3, \mathbf{v}_3)) \,.\]
Let $\iny{\mathbf{H}}_{23} = \sym(\mathbf{v}_2, \mathbf{v}_2, \mathbf{v}_3) \in V$.

Applying \cref{prop:two-ternary-dichotomy-C-1-dz} on $\holts(\X \mid \iny{\mathbf{F}}, \iny{\mathbf{H}}_{23})$, 
$\iny{\mathbf{H}}_{23}$ must be degenerate or in $\casecsupp$.
Suppose $\mathbf{v}_2 \not \sim \mathbf{v}_3$. 
Then we have $\iny{\mathbf{H}}_{23} \in \casecsupp$. 
Then by \cref{prop:calculation-sym-in-C},
$\mathbf{v}_2 \sim \mathbf{e}_1$ or $\{\mathbf{v}_2, \mathbf{v}_3\} = \{(*, 0, *), (*, 0, *)\}$.
If $\mathbf{v}_2 \sim \mathbf{e}_1$, then $\iny{\mathbf{G}} \in \casecsupp$ so we are done.
If $\{\mathbf{v}_2, \mathbf{v}_3\} = \{(*, 0, *), (*, 0, *)\}$, by symmetry of $\iny{\mathbf{G}}$, $\mathbf{v}_2 \sim \mathbf{e}_1$ so $\iny{\mathbf{G}} \in \casecsupp$.

Suppose $\mathbf{v}_2 \sim \mathbf{v}_3$.
Then, $V = \spn(\mathbf{v}_2\teh, \sym(\mathbf{v}_2, \mathbf{v}_2, \mathbf{v}_1))$ and let $\iny{\mathbf{H}}_{12} = \sym(\mathbf{v}_2, \mathbf{v}_2, \mathbf{v}_1) \in V$.
Applying \cref{prop:two-ternary-dichotomy-C-1-dz} on $\holts(\X \mid \iny{\mathbf{F}}, \iny{\mathbf{H}}_{12})$ gives $\mathbf{v}_2 \sim \mathbf{v}_1$ or $\iny{\mathbf{H}}_{12} \in \casecsupp$.
Since $\mathbf{v}_2 \ne 0$, if $\mathbf{v}_2 \sim \mathbf{v}_1$, then together with
$\mathbf{v}_2 \sim \mathbf{v}_3$, 
$\iny{\mathbf{G}}$ is degenerate contrary to the assumption.
So, $\iny{\mathbf{H}}_{12} \in \casecsupp$, and by \cref{prop:calculation-sym-in-C}, 
$\mathbf{v}_2 \sim \mathbf{e}_1$ or $\{\mathbf{v}_2, \mathbf{v}_1\} = \{(*, 0, *), (*, 0, *)\}$.
If $\mathbf{v}_2 \sim \mathbf{e}_1$, then $\iny{\mathbf{G}} \in \casecsupp$.
If  $\{\mathbf{v}_2, \mathbf{v}_1\} = \{(*, 0, *), (*, 0, *)\}$, then by $\mathbf{v}_2 \sim \mathbf{v}_3$ and symmetry of $\iny{\mathbf{G}}$, $\mathbf{v}_2 = 0$ and thus $\iny{\mathbf{G}} \in \casecsupp$.
\end{proof}

\begin{proposition}
\label{prop:casecogen-binary}
  Let $\mathbf{F}$ be a \casecogen signature 
  such that $\Yinv \mathbf{F}= [x; 1, y; 0, 0, z; 0, 0, 0, w]$ for some $x, y, z, w$ such that $w \ne 0$ and $z^2 - z - yw \ne 0$.
  Let $\iny{\mathbf{G}}$ be a nondegenerate symmetric binary signature.
  Then, $\holts(\mathbf{X} \mid \Yinv \mathbf{F}, \iny{\mathbf{G}})$ is \sph unless
  $\iny{\mathbf{G}} \in \casecsupp$.
\end{proposition}
\begin{proof}
Let 
$\iny{\mathbf{F}} = \Yinv \mathbf{F}$.
Assume $\holts(\mathbf{X} \mid \iny{\mathbf{F}}, \iny{\mathbf{G}})$ is not \sph.
  Let $\mathbf{v}_i$ be the $i$th column of $\iny{\mathbf{G}}$.
If $\mathbf{v}_2 = 0$, then $\iny{\mathbf{G}} \in \casecsupp$, so we may assume that $\mathbf{v}_2 \ne 0$.
By \cref{lem:gadget-set-general-form} and \cref{subsec:gc-case-c1gen},
$\holts(\X \mid \iny{\mathbf{F}}, \iny{\mathbf{G}}) \ge_T \holts(\X \mid \iny{\mathbf{F}}, V)$
where 
\[V = \iny{\mathbf{G}} \X \agx(\casecogenmath) = 
\spn(\mathbf{v}_2, \sym(\mathbf{v}_2, \mathbf{v}_2, \mathbf{v}_1)) \cup \bigcup_{u, v \in \mathbb{C}} \spn(\mathbf{v}_2\teh, (u \mathbf{v}_2 + v \mathbf{v}_3)\teh) \,.
\]

Let $\iny{\mathbf{H}}_{23} = \mathbf{v}_2\teh + \mathbf{v}_3\teh \in V$.
Suppose $\mathbf{v}_2 \not \sim \mathbf{v}_3$.
Then, $\iny{\mathbf{H}}_{23}$ is nondegenerate, and by \cref{prop:two-ternary-dichotomy-C-1-gen} on $\holts(\X \mid \iny{\mathbf{F}}, \iny{\mathbf{H}}_{23})$, $\iny{\mathbf{H}}_{23} \in \casecsupp$.
By \cref{prop:calculation-tensor-2-sum-in-C}, we must have $\{\mathbf{v}_2, \mathbf{v}_3\} = \{(*, 0, *), (*, 0, *)\}$.
Then, by symmetry of $\iny{\mathbf{G}}$, $\mathbf{v}_2 \sim \mathbf{e}_1$ so $\iny{\mathbf{G}} \in \casecsupp$.

Suppose $\mathbf{v}_2 \sim \mathbf{v}_3$.
Let $\iny{\mathbf{H}}_{12} = \sym(\mathbf{v}_2, \mathbf{v}_2, \mathbf{v}_1) \in V$.
Applying \cref{prop:two-ternary-dichotomy-C-1-gen} on $\holts(\X \mid \iny{\mathbf{F}}, \iny{\mathbf{H}}_{12})$ gives $\mathbf{v}_2 \sim \mathbf{v}_1$ or $\iny{\mathbf{H}}_{12} \in \casecsupp$.
Since $\mathbf{v}_2 \ne 0$, if $\mathbf{v}_2 \sim \mathbf{v}_1$, then together with
$\mathbf{v}_2 \sim \mathbf{v}_3$, 
$\iny{\mathbf{G}}$ is degenerate contrary to  assumption.
So, $\iny{\mathbf{H}}_{12} \in \casecsupp$, and by \cref{prop:calculation-sym-in-C}, 
$\mathbf{v}_2 \sim \mathbf{e}_1$ or $\{\mathbf{v}_2, \mathbf{v}_1\} = \{(*, 0, *), (*, 0, *)\}$.
If $\mathbf{v}_2 \sim \mathbf{e}_1$, then $\iny{\mathbf{G}} \in \casecsupp$.
If  $\{\mathbf{v}_2, \mathbf{v}_1\} = \{(*, 0, *), (*, 0, *)\}$, then by $\mathbf{v}_2 \sim \mathbf{v}_3$ and symmetry of $\iny{\mathbf{G}}$, $\mathbf{v}_2 \sim 0$ and thus $\iny{\mathbf{G}} \in \casecsupp$.

\end{proof}

\subsection{Non-mixing of Ternary and Binary Signatures}

\begin{proposition}\label{prop:casebtwo-binary-not-mix}
Let $\mathbf{F} = \bzp\teh + \bzm\teh + \mathbf{e}_3\teh$. 
Let $\iny{\mathbf{G}}$ be a nondegenerate symmetric binary signature in $\mathcal{E}$ but not \sepbg and not \swbg.
Let $\iny{\mathbf{H}}$ be a nondegenerate symmetric binary signature not in $\mathcal{E}$.
Then, $\holts(\mathbf{X} \mid \Yinv \mathbf{F},\iny{\mathbf{G}}, \iny{\mathbf{H}})$ is \sph.
\end{proposition}
\begin{proof}
If $\iny{\mathbf{G}}$ is in $\mathcal{E}$ but not \sepbg and not \swbg, then it must be 
$\begin{bsmallmatrix}
  x & 0 & 0 \\
  0 & 0 & y \\
  0 & y & 0
\end{bsmallmatrix}$ or $\begin{bsmallmatrix}
  0 & 0 & y \\
  0 & x & 0 \\
  y & 0 & 0
\end{bsmallmatrix}$ for $x, y \ne 0$.
Assume $\iny{\mathbf{G}} = \begin{bsmallmatrix}
  x & 0 & 0 \\
  0 & 0 & y \\
  0 & y & 0
\end{bsmallmatrix}$. The proof for 
$\begin{bsmallmatrix}
  0 & 0 & y \\
  0 & x & 0 \\
  y & 0 & 0
\end{bsmallmatrix}$ is similar.

Assume $\holts(\mathbf{X} \mid \Yinv \mathbf{F},\iny{\mathbf{G}}, \iny{\mathbf{H}})$ is not \sph.
We may assume that $\iny{\mathbf{H}}$ is $\begin{bsmallmatrix}
  p^2 & pq & 0 \\
  pq & q^2 & 0 \\
  0 & 0 & r
\end{bsmallmatrix}$ (in \sepbg but not in $\mathcal{E}$) or $\begin{bsmallmatrix}
  0 & 0 & p \\
  0 & 0 & q \\
  p & q & 0
\end{bsmallmatrix}$ (in \swbg), for $p, q, r \ne 0$ by \cref{prop:casebthree-binary}.
Then, 
\[\holts(\X \mid \Yinv \mathbf{F}, \iny{\mathbf{G}}, \iny{\mathbf{H}}) 
\ge_T \holts(\X \mid \Yinv \mathbf{F}, \iny{\mathbf{G}}\mathbf{X} \iny{\mathbf{H}} \mathbf{X}\iny{\mathbf{G}})\]
where 
$\iny{\mathbf{G}}\mathbf{X} \iny{\mathbf{H}} \mathbf{X}\iny{\mathbf{G}}$ is either 
$\begin{bsmallmatrix}
  q^2 x^2 & 0 & p q x y \\
  0 & r y^2 & 0 \\
  p q x y & 0 & p^2 y^2
\end{bsmallmatrix}$ or $\begin{bsmallmatrix}
  0 & q x y & 0 \\
  q x y & 0 & p y^2 \\
  0 & p y^2 & 0
\end{bsmallmatrix}$
which are both rank $2$ and not of the tractable form given in \cref{prop:casebthree-binary}, a contradiction.
\end{proof}

\begin{proposition}\label{prop:casea3-binary-not-mix}
Let $\mathbf{F} = \mathbf{e}_1\teh + \mathbf{e}_2\teh + \mathbf{e}_3\teh$.
Let $\mathbf{G}$ be a nondegenerate symmetric binary signature in \sepbg or \swbg but not in $\mathcal{E}$.
Let $\mathbf{H}$ be a nondegenerate symmetric binary signature that is not \sepbg 
and  not \swbg.
Then $\holts(\mathbf{F}, \mathbf{G}, \mathbf{H})$ is \sph.
\end{proposition}
\begin{proof}
Assume $\holts(\mathbf{F}, \mathbf{G}, \mathbf{H})$ is not \sph.
$\mathbf{G}$ must be $\begin{bsmallmatrix}
  p^2 & pq & 0 \\
  pq & q^2 & 0 \\
  0 & 0 & r
\end{bsmallmatrix}$ or $\begin{bsmallmatrix}
  0 & 0 & p \\
  0 & 0 & q \\
  p & q & 0
\end{bsmallmatrix}$ for $p, q, r \ne 0$ by \cref{prop:casea3-binary}.
Then, $(p, q, 0)\teh + \mathbf{e}_3\teh \in \mathbf{G} \ag(\caseathreemath)$.
We have 
\begin{align*}
    \holts(\mathbf{F}, \mathbf{G}, \mathbf{H}) 
    &\ge_T \holts(\mathbf{F}, \ag(\caseathreemath), \mathbf{G}, \mathbf{H}) \\
    &\ge_T \holts(\mathbf{F}, \ag(\caseathreemath), \mathbf{G} \ag(\caseathreemath), \mathbf{H}) \\
    &\ge_T \holts(\mathbf{F}, (p, q, 0)\teh + \mathbf{e}_3\teh, \mathbf{H}, \mathbf{H} \ag(\caseathreemath) ) \,.
\end{align*}
By hypothesis, the given $\mathbf{H}$ is not \sepbg and not \swbg.
By \cref{prop:casea3-binary}, $\mathbf{H}$ must be $ij|k$ or \swhelper{ij;k} for $ij \ne \db\dg$ or in $\mathcal{E}$.
For example, if $\mathbf{H}$ is \sepgr or \swgr and is not \sepbg and not \swbg, 
then $\mathbf{H}$ is of the form 
$\begin{bsmallmatrix}
    c & 0 & 0 \\
    0 & a^2 & ab \\
    0 & ab & b^2
\end{bsmallmatrix}$, $\begin{bsmallmatrix}
    0 & a & b \\
    a & 0 & 0 \\
    b & 0 & 0
\end{bsmallmatrix}$, or  $\begin{bsmallmatrix}
    c & 0 & 0 \\
    0 & 0 & a \\
    0 & a & 0
\end{bsmallmatrix}$ for $a, b, c \ne 0$.
For the first two cases,
we see that $\mathbf{e}_1\teh + (0, a, b)\teh \in \mathbf{H} \ag(\caseathreemath)$ which is \sepgr and not in $\mathcal{E}$. Together with
$(p, q, 0)\teh + \mathbf{e}_3\teh$ which is \sepbg and not
$\mathcal{E}$, this leads to 
\sphness by \cref{prop:case-a3-ternary-non-mixing}, with the roles of
$\db$ and $\dg$ switched.
For the third case, we have 
$\mathbf{H} ((p, q, 0)\teh + \mathbf{e}_3\teh) = (cp, 0, aq)\teh + a \mathbf{e}_2\teh$.
Then, $\holts(\mathbf{F}, (p, q, 0)\teh + \mathbf{e}_3\teh, (cp, 0, aq)\teh + a \mathbf{e}_2\teh)$ is \sph by \cref{prop:case-a3-ternary-non-mixing} 
since 
$(p, q, 0)\teh + \mathbf{e}_3\teh$ is \sepbg but not in $\mathcal{E}$ 
and
$(cp, 0, aq)\teh + a \mathbf{e}_2\teh$ is \sepbr but not in $\mathcal{E}$ because $a, c, p, q \ne 0$ by assumption.

If $\mathbf{H} \in \mathcal{E}$ but it is not \sepbg or \swbg, then it must be 
$\begin{bsmallmatrix}
    a & 0 & 0 \\
    0 & 0 & b \\
    0 & b & 0
\end{bsmallmatrix}$ or 
$\begin{bsmallmatrix}
    0 & 0 & b \\
    0 & a & 0 \\
    b & 0 & 0
\end{bsmallmatrix}$ for $a, b \ne 0$. 
The proof for \sphness is similar to above. 
\end{proof}

\begin{proposition}\label{prop:casecosep-allowed-binary-not-mix}
Let $\mathbf{F} = \sym(\bzp, \bzp, \bzm) + \mathbf{e}_3\teh$.
Let $\iny{\mathbf{G}}$ be a nondegenerate symmetric binary signature
that is \sepbg or \swbg but not in $\casecsupp$.
Let $\iny{\mathbf{H}}$ be a nondegenerate symmetric binary signature in $\casecsupp$ 
but not \sepbg and not \swbg.
Then, $\holts(\X \mid \Yinv \mathbf{F}, \iny{\mathbf{G}}, \iny{\mathbf{H}})$ is \sph.
\end{proposition}
\begin{proof}
Let $\iny{\mathbf{F}} = \Yinv \mathbf{F}$.
Assume $\holts(\X \mid \iny{\mathbf{F}}, \iny{\mathbf{G}}, \iny{\mathbf{H}})$ is not \sph.
Then,  by \cref{prop:casecosep-allowed-binary},
$\iny{\mathbf{G}}$ must be $\begin{bsmallmatrix}
  p^2 & pq & 0 \\
  pq & q^2 & 0 \\
  0 & 0 & r
\end{bsmallmatrix}$ (\sepbg but not in $\casecsupp$)
or $\begin{bsmallmatrix}
  0 & 0 & p \\
  0 & 0 & q \\
  p & q & 0
\end{bsmallmatrix}$ (\swbg but not in $\casecsupp$)
for $q, r \ne 0$.
For both cases, we have $\iny{\mathbf{G}}\mathbf{X} \agx(\casecosepmath) = \spn((p, q, 0)\teh, \mathbf{e}_3\teh)$.
Let $\iny{\mathbf{A}} = (p, q, 0)\teh + \mathbf{e}_3\teh \in \iny{\mathbf{G}}\mathbf{X} \agx(\casecosepmath)$.
Since $\iny{\mathbf{H}} \in \casecsupp$, it has the form $\iny{\mathbf{H}} = \begin{bsmallmatrix}
  x & y & z \\
  y & 0 & 0 \\
  z & 0 & w
\end{bsmallmatrix}$. As $\iny{\mathbf{H}} \not \in $ \sepbg, we have $z \ne 0$.
\begin{align*}
\holts(\X \mid \iny{\mathbf{F}}, \iny{\mathbf{G}}, \iny{\mathbf{H}}) 
&\ge_T \holts(\X \mid \iny{\mathbf{F}}, \agx(\casecosepmath),  \iny{\mathbf{G}}, \iny{\mathbf{H}}) \\
&\ge_T \holts(\X \mid \iny{\mathbf{F}}, \iny{\mathbf{G}} \X \agx(\casecosepmath), \iny{\mathbf{H}}) \\
&\ge_T \holts(\X \mid \iny{\mathbf{F}}, \iny{\mathbf{A}}, \iny{\mathbf{H}}) \\
&\ge_T \holts(\X \mid \iny{\mathbf{F}}, \iny{\mathbf{A}}, \iny{\mathbf{B}})
\end{align*} 

where $\iny{\mathbf{B}} = \iny{\mathbf{H}} \X \iny{\mathbf{A}} = (y p + xq, yq, zq)\teh + (z, 0, w)\teh$. 
We can check that $\iny{\mathbf{B}} $ is nondegenerate: if $y \ne 0$, then it is clearly nondegenerate since $q \ne 0$; if $y = 0$, $\iny{\mathbf{B}}$ is degenerate if and only if $\iny{\mathbf{H}}$ is degenerate since $q \ne 0$.
By \cref{prop:calculation-tensor-2-sum-in-BG-R}, if $\iny{\mathbf{B}} $ were to be \sepbg, $\{(y p + xq, yq, zq), (z, 0, w)\}$ must be $\{(*, *, 0), (0, 0, *)\}$ or $\{((*, *, 0), (*, *, 0)\}$.
If $\{(y p + xq, yq, zq), (z, 0, w)\} = \{(*, *, 0), (0, 0, *)\}$, since $zq \ne 0$, we must have $yp + xq, yq, w = 0$.
Since $p, q \ne 0$, this gives $x, y, w = 0$, which means $\iny{\mathbf{H}}$ is \swbg, contrary to  assumption.
$\{(y p + xq, yq, zq), (z, 0, w)\} = \{(*, *, 0), (*, *, 0)\}$ is not possible because $zq \ne 0$. 
Therefore, $\iny{\mathbf{B}}$ is a nondegenerate signature not in  \sepbg.
By \cref{prop:two-ternary-dichotomy-C-1-sep} on $\holts(\X \mid \iny{\mathbf{F}}, \iny{\mathbf{B}} )$, $\iny{\mathbf{B}} $ must be in $\casecsupp$.
By \cref{prop:case-c1-sep-ternary-non-mixing} on $\holts(\X \mid \iny{\mathbf{F}}, \iny{\mathbf{A}}, \iny{\mathbf{B}} )$, 
since $\iny{\mathbf{A}}$ is \sepbg but not in $\casecsupp$ and
$\iny{\mathbf{H}}  \X \iny{\mathbf{A}}$ is not \sepbg but in $\casecsupp$, we get \sphness.
\end{proof}

\subsection{A Set of Ternary and Binary Signatures}
We are now ready to prove the dichotomy for a set of ternary and binary signatures.

\begin{theorem}\label{thm:dichotomy-ternary-binary}
    Let $\mathcal{F}$ and  $\mathcal{B}$ be respectively
    sets of nondegenerate symmetric ternary and binary signatures.
Then, $\holts(\mathcal{F} \cup \mathcal{B})$ is tractable if it satisfies
the tractability criterion of \cref{thm:tractable-cases}.
Otherwise, it is \sph. 
\end{theorem}
\begin{proof}
Assume $\holts(\mathcal{F} \cup \mathcal{B})$ is not \sph.
Then $\holts(\mathcal{F} \cup \mathcal{B})$ escapes all \sphness conditions we have proven so far. 
In particular, $\holts(\mathcal{F})$ is not \sph,  and so by \cref{thm:ternary-set-dichotomy},  
it must be in one of the tractable cases given in \cref{thm:tractable-cases}.
    We proceed in the same way as in the proof of \cref{thm:ternary-set-dichotomy}. 
    The idea is that the propositions in this section will show: 1) which binary signatures are allowed (\cref{prop:casebthree-binary,prop:casea3-binary,prop:casecosep-allowed-binary,prop:casecodz-binary,prop:casecogen-binary}); 2) some binary signatures cannot coexist (\cref{prop:casea3-binary-not-mix,prop:casebtwo-binary-not-mix,prop:casecosep-allowed-binary-not-mix}).
    
    If $\mathcal{F}$ contains a \casecogen or \casecodz signature, then there must exist some orthogonal $T$ such that $\Yinv T \mathcal{F} \subseteq \casecsupp$. 
    By \cref{prop:casecodz-binary,prop:casecogen-binary}, $\Yinv T \mathcal{B} \subseteq \casecsupp$ as well. 
    Then $\mathcal{F}$ and $\mathcal{B}$ satisfy \cref{itm:tractable-c} of \cref{thm:tractable-cases}.

    Suppose $\mathcal{F}$ contains a \casecosep signature.
    Then \cref{prop:casecosep-allowed-binary} and \cref{prop:casecosep-allowed-binary-not-mix} imply that there exists some orthogonal $T$ such that:
    \begin{enumerate}
        \item $T \mathcal{F} \subset \sepbgmath$ and $T \mathcal{B} \subset \sepbgmath \cup \swbgmath$.
        Then $\mathcal{F}$ and $\mathcal{B}$ satisfy \cref{itm:tractable-bg-r} of \cref{thm:tractable-cases}; or
        \item $\Yinv T \mathcal{F} \subseteq \casecsupp$ and $\Yinv T \mathcal{B} \subseteq \casecsupp$.
        Then $\mathcal{F}$ and $\mathcal{B}$ satisfy \cref{itm:tractable-c} of \cref{thm:tractable-cases}.
    \end{enumerate}
    
    Suppose $\mathcal{F}$ contains a \casebthree signature.
    Then \cref{prop:casebthree-binary} and \cref{prop:casebtwo-binary-not-mix} imply that there exists some orthogonal $T$ such that:
    \begin{enumerate}
        \item $T \mathcal{F} \subset \sepbgmath$ and $T \mathcal{B} \subset \sepbgmath \cup \swbgmath$.
        Then $\mathcal{F}$ and $\mathcal{B}$ satisfy \cref{itm:tractable-bg-r} of \cref{thm:tractable-cases};
        or
        \item $\Yinv T \mathcal{F} \subseteq \mathcal{E}$ and $\Yinv T \mathcal{B} \subseteq \mathcal{E}$.
        Then $\mathcal{F}$ and $\mathcal{B}$ satisfy \cref{itm:tractable-e} of \cref{thm:tractable-cases}.
    \end{enumerate}

    Suppose $\mathcal{F}$ contains a \caseathree signature.
    Then \cref{prop:casea3-binary} and \cref{prop:casea3-binary-not-mix} imply that there exists some orthogonal $T$
   such that:
    \begin{enumerate}
        \item $T \mathcal{F} \subset \sepbgmath$ and $T \mathcal{B} \subset \sepbgmath \cup \swbgmath$.
        Then $\mathcal{F}$ and $\mathcal{B}$ satisfy \cref{itm:tractable-bg-r} of \cref{thm:tractable-cases};
        or
        \item $T \mathcal{F} \subseteq \mathcal{E}$ and $T \mathcal{B} \subseteq \mathcal{E}$.
        Then $\mathcal{F}$ and $\mathcal{B}$ satisfy \cref{itm:tractable-e} of \cref{thm:tractable-cases}.
    \end{enumerate}

    Suppose $\mathcal{F}$ does not contain any \caseathree, \casebthree, \casecosep, \casecodz, or \casecogen signature.
    Then every signature in $\mathcal{F}$ has a nonzero annihilator. 
Let $\mathcal{M} = \langle \mathcal{B}, I \rangle$, the monoid generated by $\mathcal{B}$ under multiplication. 
We have $\holts(\mathcal{F} \cup \mathcal{B}) \ge_T \holts(\mathbf{M} \mathbf{F}_i, \mathbf{F}_j)$ for every $\mathbf{M} \in \mathcal{M}$, $\mathbf{F}_i, \mathbf{F}_j \in \mathcal{F}$. 
Then, $\mathbf{M} \mathbf{F}_i$ is degenerate or $(\mathbf{M} \mathbf{F}_i, \mathbf{F}_j)$ must be a compatible pair by \cref{lem:not-orthogonal-pair-dichotomy-ternary}.
Then $\mathcal{F}$ and $\mathcal{B}$ satisfy \cref{itm:tractable-ann} of \cref{thm:tractable-cases}.
\end{proof}

\section{Higher Arity}
\label{sec:higher-arity}
We now extend the dichotomy to higher arity signatures.
We show that higher arity signatures must also be of essentially the same form as the ternary signatures. 
Let $\mathbf{F}$ be a nondegenerate symmetric signature of arity $4$.
We analyze the subspace of signatures $\mathscr{F} = \{\langle \mathbf{F}, \mathbf{u} \rangle : \mathbf{u} \in \mathbb{C}^3\}$. 
Assume $\mathbf{F}$ is not \sph. 
Since $\holts(\mathbf{F}) \ge_T \holts(\mathscr{F})$, 
we may assume that $\holts(\mathscr{F})$ is not \sph. 
Since $\mathscr{F}$ is a set of ternary signatures, we can apply the dichotomy theorems we have proven so far. 
If there exists $\mathbf{G} \in \mathscr{F}$ such that $\ann(\mathbf{G}) = 0$, then it is 
easy to deduce the forms of the other signatures in  $\mathscr{F}$.
By the pairwise dichotomy, $\mathbf{G}$ puts strong constraint on $\mathscr{F}$ because $\mathscr{F}$ is a linear subspace: for any $\mathbf{H} \in \mathscr{F}$, $\holts(\mathscr{F}) \ge_T \holts(\mathbf{G}, \mathbf{H}, x \mathbf{G} + y \mathbf{H})$  for all $x, y \in \mathbb{C}$, and $\mathbf{H}$ must be of a specific form
in order that the latter is not \sph.

Otherwise, we need to analyze the case when every signature in $\mathscr{F}$ has a nontrivial annihilator. 
Let $\mathbf{G}, \mathbf{H} \in \mathscr{F}$ with $\ann(\mathbf{G}), \ann(\mathbf{H}) \ne 0$. 
If every signature in $\mathscr{F}$ has a nontrivial annihilator, then $\ann(x \mathbf{G} + y \mathbf{H}) \ne 0$ for all $x, y \in \mathbb{C}$. 
We first analyze when that can happen. 
Then, we analyze the complexity of $\holts(\mathbf{G}, \mathbf{H}, x \mathbf{G} + y \mathbf{H})$ assuming $\ann(x \mathbf{G} + y \mathbf{H}) \ne 0$ for all $x, y \in \mathbb{C}$. 
The reason for this seemingly roundabout approach instead of directly applying the  ternary
dichotomy on $\holts(\mathbf{G}, \mathbf{H}, x \mathbf{G} + y \mathbf{H})$ 
is that it is often difficult to check 
whether $\holts(x \mathbf{G} + y \mathbf{H})$ is tractable or not because of the non-uniqueness of the tensor summation form when there are multiple summands. 

With these preliminary results, we can analyze $\mathscr{F}$ and conclude that its support must be of a very restricted form: there must exist an orthogonal $T$ such that 
either $T \mathscr{F} \subset \mathcal{E}$, or $\Yinv T \mathscr{F} \subset \mathcal{E}$, or $\Yinv T \mathscr{F} \subset \casecsupp$. 
For example, if $T \mathscr{F} \subset \mathcal{E}$, 
then by looking at the triangular representation of $T \mathbf{F}$ and the subsignatures $\langle T \mathbf{F}, \mathbf{e}_i \rangle \in T \mathscr{F} \subset \mathcal{E}$, we can conclude that $T \mathbf{F} \in \mathcal{E}$. 

After proving the dichotomy for a single arity $4$ signature, we can inductively extend the dichotomy for arity $n \ge 4$. 
To do so, we first extend the definition of tractable types of ternary signatures (\cref{def:types-ternary}) to include higher arity signatures in \cref{def:types-higher-arity}.
Then, we show that a higher arity signature of certain type can realize a ternary signature of the same type in \cref{prop:reduce-arity}. 
This allows us to argue that the dichotomy for a set of arity $\le 4$ signatures will be essentially the same as the dichotomy for a set of ternary signatures, and we can apply the techniques of this section to inductively extend the dichotomy to a set of signatures of any arity.

\subsection{Some Preliminaries on \texorpdfstring{$\mathscr{F}$}{F}}

\begin{proposition}
  \label{prop:lin-comb-of-ebd-signatures}
Let $\mathbf{F}$ and $\mathbf{G}$ be nondegenerate symmetric ternary signatures such that 
$\supp \mathbf{F} \subseteq \{\db, \dg\}^*$ and $\supp \mathbf{G} \subseteq \{\dg,\dr\}^*$.
Then, unless $\mathbf{F}\domres{\db\dg} = [0,0, *, *]$ and $\mathbf{G}\domres{\dg\dr} = [*,*,0,0]$, 
there exists some nonzero $x, y \in \mathbb{C}$ such that 
$\ann(x \mathbf{F} + y \mathbf{G}) = 0$.
\end{proposition}
\begin{proof}
Under the assumptions, we can write $\mathbf{F}$ and $\mathbf{G}$ as in \cref{fig:lin-comb-mf-mg}.
\begin{figure}
  \centering
\begin{subfigure}[b]{0.45\textwidth}
    \centering
    \tersig{0.7}{a_1 a_2 0 a_3 0 0 a_4 0 0 0}
    \caption{$\mathbf{F}$}
  \end{subfigure}
\begin{subfigure}[b]{0.45\textwidth}
    \centering
    \tersig{0.7}{0 0 0 0 0 0 b_4 b_3 b_2 b_1}
    \caption{$\mathbf{G}$}
  \end{subfigure}
\par\bigskip
\begin{subfigure}[b]{1\textwidth}
    \centering
    \tersig{0.7}{{x a_1} {x a_2} 0 {x a_3} 0 0 {x a_4 + y b_4} {y b_3} {y b_2} {y b_1}}
    \caption{$x \mathbf{F} + y \mathbf{G}$}
  \end{subfigure}
  \caption{Linear combination of $\mathbf{F}$ and $\mathbf{G}$.}
  \label{fig:lin-comb-mf-mg}
\end{figure}
Since $\mathbf{F}$ and $ \mathbf{G}$ are nondegenerate, we must have that 
\begin{align}
  \label{eq:rank-mf-mg}
  \rank \begin{bmatrix}
    a_1 & a_2 \\ a_2 & a_3 \\ a_3 & a_4
  \end{bmatrix} = 2 && 
  \rank \begin{bmatrix}
    b_4 & b_3 \\ b_3 & b_2 \\ b_2 & b_1
  \end{bmatrix} = 2
\end{align}
since the only annihilator of $ \mathbf{F}$ and $ \mathbf{G}$ is $\mathbf{e}_3$ and $\mathbf{e}_1$ respectively.
Then, for $\mathbf{w}$ to annihilate $x  \mathbf{F} + y \mathbf{G}$, it must be the case that (see \cref{fig:computing-f-u})
\begin{equation}
  \label{eq:rank-xmf-ymg}
  \begin{bmatrix}
    x a_1 & x a_2 & 0 \\
    x a_2 & x a_3 & 0 \\
    x a_3 & x a_4 + y b_4 & y b_3 \\
    0 & y b_3 & y b_2 \\
    0 & y b_2 & y b_1
  \end{bmatrix} \mathbf{w} = 0 \,.
\end{equation}
Suppose $(a_1, a_2) \not \sim (a_2, a_3)$.
Then, since at least one of $b_1, b_2, b_3$ must be nonzero for \cref{eq:rank-mf-mg} to hold, the matrix in \cref{eq:rank-xmf-ymg} has rank $3$ because we can find three linearly independent rows.
Similarly, if $(b_3, b_2) \not \sim (b_2, b_1)$, the matrix is rank $3$.

If $(a_1, a_2) \sim (a_2, a_3)$ and $(b_3, b_2) \sim (b_2, b_1)$, we may assume that 
$(a_i, a_{i+1}) \not \sim (a_3, a_4)$
for $i = 1$ or $2$.  Similarly, 
  $(b_{j+1}, b_{j}) \not \sim (b_4, b_3)$, for $j = 1$ or $2$.
By choosing those rows and the third row in~\cref{eq:rank-xmf-ymg}, for the determinant
\begin{align*}
  \begin{vmatrix}
    x a_i & x a_{i+1} & 0 \\
    x a_3 & x a_4 + y b_4 & y b_3 \\
    0 & y b_{j+1} & y b_j \\
  \end{vmatrix} &= 
  x^2y \begin{vmatrix}
    a_i & a_{i+1} & 0 \\
    a_3 & a_4& b_3 \\
    0 & 0 & b_j \\
  \end{vmatrix}
  +
 xy^2 \begin{vmatrix}
    a_i & 0 & 0 \\
    a_3 & b_4 & b_3 \\
    0 & b_{j+1} & b_j \\
  \end{vmatrix} \\
  &=x^2y b_j \begin{vmatrix}
    a_i & a_{i+1} \\
    a_3 & a_4 \\
  \end{vmatrix}
  +
 xy^2 a_i  \begin{vmatrix}
    b_4 & b_3 \\
    b_{j+1} & b_j \\
  \end{vmatrix} \,
\end{align*}
to be zero for all nonzero $x, y$,  we have $b_j=0$ 
as $\lvert\begin{smallmatrix}
    a_i & a_{i+1} \\
    a_3 & a_4 \\
  \end{smallmatrix} \ne 0$, and  $a_i=0$
 as $\begin{vsmallmatrix}
    b_4 & b_3 \\
    b_{j+1} & b_j \\
  \end{vsmallmatrix}  \ne 0$. 
  Then, $i=2 \Rightarrow a_3^2 \ne 0$. Thus
  $0 = \begin{vsmallmatrix}
    a_1 & a_2 \\
    a_2 & a_3 \\
  \end{vsmallmatrix} = a_1 a_3 \Rightarrow a_1 = 0$. Thus $a_1=0$ regardless whether $i=1$ or
  $2$. Then we get $0 = \begin{vsmallmatrix}
    a_1 & a_2 \\
    a_2 & a_3 \\
  \end{vsmallmatrix} \Rightarrow a_2 = 0$. 
  So in fact $(a_i, a_{i+1}) \not \sim (a_3, a_4)$
for $i=2$. 
  Then $0 \ne \begin{vsmallmatrix}
    a_2 & a_3 \\
    a_3 & a_4 \\
  \end{vsmallmatrix} \Rightarrow a_3 \ne 0$.
Similarly, $b_1, b_2 = 0$ and $b_3 \ne 0$.
\end{proof}

\begin{corollary}
  \label{cor:lin-comb-ebd-signatures-cannot-be-a-b}
Let $\mathbf{F}, \mathbf{G}$ be nondegenerate symmetric ternary signatures such that 
$\ann(\mathbf{F}) = \spn(\mathbf{u}) \ne 0$ and $\ann(\mathbf{G}) = \spn(\mathbf{v}) \ne 0$ and $\mathbf{u} \not \sim \mathbf{v}$.
Suppose $\mathbf{F}$ is \caseatwo or \casebtwo.
Then, there exist some nonzero $x, y \in \mathbb{C}$ such that 
$\ann(x \mathbf{F} + y \mathbf{G}) = 0$.
\end{corollary}
\begin{proof}
Since $\mathbf{u} \not \sim \mathbf{v}$, there exists some $M \in \gl(3, \mathbb{C})$ (not necessarily orthogonal) such that 
$\langle M \mathbf{F}, \mathbf{e}_3 \rangle = 0$ and $\langle M \mathbf{G}, \mathbf{e}_1 \rangle = 0$.
Thus, $\supp M \mathbf{F} \subseteq \{\db, \dg\}^*$ and $\supp M \mathbf{G} \subseteq \{\dg,\dr\}^*$.
Also $(M \mathbf{F})\domres{\db\dg}$ has the form $\mathbf{x}\teh + \mathbf{y}\teh$
for some $\mathbf{x}, \mathbf{y} \in \mathbb{C}^2$, thus not of the form $[0,0,*,*]$ by \cref{prop:symmetric-rank-sym-sum}.
By \cref{prop:lin-comb-of-ebd-signatures}, for some 
 nonzero $x, y \in \mathbb{C}$, 
$\ann(x M \mathbf{F} + y M \mathbf{G}) = 0$.
Thus $\ann(x  \mathbf{F} + y  \mathbf{G}) = 0$.
\end{proof}

\begin{proposition}
  \label{prop:sum-casec1bg-else}
  Let $\iny{\mathbf{G}} = p \mathbf{e}_1\teh + q \sym(\mathbf{e}_1, \mathbf{e}_1, \mathbf{e}_2)$ for $q \ne 0$.
  Let $\iny{\mathbf{H}}$ be a nondegenerate symmetric ternary signature 
  such that $\ann(\iny{\mathbf{H}}) = \spn(\mathbf{u}) \ne 0$.
  If $\ann(\iny{\mathbf{G}} + \iny{\mathbf{H}}) \ne 0$ and
  $\holts(\mathbf{X} \mid \iny{\mathbf{G}}, \iny{\mathbf{H}}, \iny{\mathbf{G}} + \iny{\mathbf{H}})$ is not \sph, then the support of $\iny{\mathbf{H}}$ is a subset of $\db\db\db$, $ \db\db\dg$, $\db\db\dr$ and their permutations.
\end{proposition}
\begin{proof}
Note that $\iny{\mathbf{G}}$ and thus $\X\iny{\mathbf{G}}$ are supported on $\db\dg$. We have 
\begin{align*}
\holts(\mathbf{X} \mid \iny{\mathbf{G}}, \iny{\mathbf{H}}, \iny{\mathbf{G}} + \iny{\mathbf{H}}) 
&\ge_T \holts(\mathbf{X} \iny{\mathbf{G}} \mid  \iny{\mathbf{H}}, \iny{\mathbf{G}} + \iny{\mathbf{H}}) 
\\
&\ge_T 
\holbs([0, 0, q, p] \mid  \iny{\mathbf{H}}\domres{\db\dg}, (\iny{\mathbf{G}} + \iny{\mathbf{H}})\domres{\db\dg}) \; .
\end{align*}

Suppose $\mathbf{u}$ is not orthogonal to $\mathbf{e}_3$.
Then we can write $\mathbf{u} = (x, y, -1)$.
For $\holbs([0, 0, q, p] \mid \iny{\mathbf{H}}\domres{\db\dg})$ to be not \sph,
$\iny{\mathbf{H}}\domres{\db\dg} = [a, b, 0, 0]$ for some $a, b \in \mathbb{C}$ by \cref{cor:bipartite-boolean-dichotomy-explicit}.
Because $\iny{\mathbf{H}}$ is nondegenerate,
$\iny{\mathbf{H}}$ must have the form in \cref{fig:casedz-m-on-bg} with $b \ne 0$ by \cref{prop:annihilator-uniquely-propagates-values}.
Then, we have $\iny{\mathbf{G}} + \iny{\mathbf{H}}$ as in \cref{fig:G-plus-H}.
\begin{figure}
  \centering

\begin{subfigure}[b]{0.45\textwidth}
    \centering
    \tersig{1}{{p + a} { q + b} {ax + by} 0 bx {ax^2 + 2 bxy} 0 0 {bx^2} {ax^3 + 3bx^2y}}
    \caption{$ \iny{\mathbf{G}} + \iny{\mathbf{H}}$ for $\mathbf{u} = (x, y, -1)$}
    \label{fig:G-plus-H}
\end{subfigure}
    \hfill
\begin{subfigure}[b]{0.45\textwidth}
    \centering
    \tersig{1}{{p + ax^3} { q + ax^2 y} {bx^2} {a xy^2} {bxy} {cx} {ay^3} {by^2} {cy} {d}}
    \caption{$ \iny{\mathbf{G}} + \iny{\mathbf{H}}$ for $\mathbf{u} = (-y, x, 0)$}
    \label{fig:G-plus-H-xy0}
    \end{subfigure}
\caption{Forms of $\iny{\mathbf{G}} + \iny{\mathbf{H}}$ in the proof of \cref{prop:sum-casec1bg-else}.}
\end{figure}

Suppose $\langle \iny{\mathbf{G}} + \iny{\mathbf{H}}, \mathbf{w} \rangle = 0$ for some nonzero $\mathbf{w} \in \mathbb{C}^3$.
Then, it must be the case that (see \cref{fig:computing-f-u})
\[
\begin{bmatrix}
    (\iny{\mathbf{G}} + \iny{\mathbf{H}})(\db\db\dr) & 
    (\iny{\mathbf{G}} + \iny{\mathbf{H}})(\db\dg\dr) & 
    (\iny{\mathbf{G}} + \iny{\mathbf{H}})(\db\dr\dr) \\
    (\iny{\mathbf{G}} + \iny{\mathbf{H}})(\db\dg\dr) & 
    (\iny{\mathbf{G}} + \iny{\mathbf{H}})(\dg\dg\dr) &
    (\iny{\mathbf{G}} + \iny{\mathbf{H}})(\dg\dr\dr) \\
    (\iny{\mathbf{G}} + \iny{\mathbf{H}})(\db\db\dg) &
    (\iny{\mathbf{G}} + \iny{\mathbf{H}})(\db\dg\dg) &
    (\iny{\mathbf{G}} + \iny{\mathbf{H}})(\db\dg\dr)
\end{bmatrix} \mathbf{w}
=
  \begin{bmatrix}
    a x + by & bx & a x^2 + 2 b x y \\
    b x & 0 & b x^2 \\
    q + b & 0 & b x
  \end{bmatrix} \mathbf{w} = 0 \, .
\]
However, the determinant of the matrix is $-bx(b^2 x^2 - b x^2 (q + b)) =  q b^2 x^3$.
Since $b, q \ne 0$, it can only be $0$ if $x = 0$.
If $x = 0$, then the support of $\iny{\mathbf{H}}$ is a subset of $\db\db\db$, $ \db\db\dg$, $\db\db\dr$ and their permutations.

Suppose $\langle \mathbf{u}, \mathbf{e}_3 \rangle = 0$, and we may write $\mathbf{u} = (-y, x, 0)$.
Then we can write $\iny{\mathbf{G}} + \iny{\mathbf{H}}$ as in \cref{fig:G-plus-H-xy0} for some $a, b, c, d \in \mathbb{C}$. This follows from the known shape of $\iny{\mathbf{G}}$ and how the values of  $\iny{\mathbf{H}}$
are uniquely propagated when $\langle \iny{\mathbf{H}}, \mathbf{u} \rangle =0$.

Suppose $y = 0$. 
For $\ann(\iny{\mathbf{G}} + \iny{\mathbf{H}}) \ne 0$, 
the following determinant must be zero:
\[
\begin{vmatrix}
    (\iny{\mathbf{G}} + \iny{\mathbf{H}})(\db\db\db) & 
    (\iny{\mathbf{G}} + \iny{\mathbf{H}})(\db\db\dg) & 
    (\iny{\mathbf{G}} + \iny{\mathbf{H}})(\db\db\dr) \\
    (\iny{\mathbf{G}} + \iny{\mathbf{H}})(\db\db\dg) & 
    (\iny{\mathbf{G}} + \iny{\mathbf{H}})(\db\dg\dg) &
    (\iny{\mathbf{G}} + \iny{\mathbf{H}})(\db\dg\dr) \\
    (\iny{\mathbf{G}} + \iny{\mathbf{H}})(\db\dr\dr) &
    (\iny{\mathbf{G}} + \iny{\mathbf{H}})(\dg\dr\dr) &
    (\iny{\mathbf{G}} + \iny{\mathbf{H}})(\dr\dr\dr)
\end{vmatrix} =
\begin{vmatrix}
  p + ax^2 & q & b x^2 \\
  q & 0 & 0 \\
  cx & 0 & d
\end{vmatrix} = q^2 d = 0 \; .\]
Since $q \ne 0$, this implies $d = 0$.
We also need 
\[
\begin{vmatrix}
    (\iny{\mathbf{G}} + \iny{\mathbf{H}})(\db\db\db) & 
    (\iny{\mathbf{G}} + \iny{\mathbf{H}})(\db\db\dg) & 
    (\iny{\mathbf{G}} + \iny{\mathbf{H}})(\db\db\dr) \\
    (\iny{\mathbf{G}} + \iny{\mathbf{H}})(\db\db\dg) & 
    (\iny{\mathbf{G}} + \iny{\mathbf{H}})(\db\dg\dg) &
    (\iny{\mathbf{G}} + \iny{\mathbf{H}})(\db\dg\dr) \\
    (\iny{\mathbf{G}} + \iny{\mathbf{H}})(\db\db\dr) &
    (\iny{\mathbf{G}} + \iny{\mathbf{H}})(\db\dg\dr) &
    (\iny{\mathbf{G}} + \iny{\mathbf{H}})(\db\dr\dr)
\end{vmatrix} =
\begin{vmatrix}
  p + ax^2 & q & b x^2 \\
  q & 0 & 0 \\
  bx^2 & 0 & cx
\end{vmatrix} = q^2 c x = 0 \; .\]
Since $\mathbf{u} \ne 0$, we have $x \ne 0$.
Hence, $c = 0$, and thus the support of $\iny{\mathbf{H}}$ is a subset of $\db\db\db, \db\db\dg, \db\db\dr$ and their permutations. 

Suppose $y \ne 0$.
For $\holbs([0, 0, q, p] \mid (\iny{\mathbf{G}} + \iny{\mathbf{H}})\domres{\db\dg})$ to be not \sph, it must be the case that $a = 0$ by \cref{cor:bipartite-boolean-dichotomy-explicit}.
We consider if $\iny{\mathbf{G}} + \iny{\mathbf{H}}$ can have a nonzero annihilator. 
If $\ann(\iny{\mathbf{G}} + \iny{\mathbf{H}}) \ne 0$, the following determinant must be zero:
\[
\begin{vmatrix}
    (\iny{\mathbf{G}} + \iny{\mathbf{H}})(\db\db\db) & 
    (\iny{\mathbf{G}} + \iny{\mathbf{H}})(\db\db\dg) & 
    (\iny{\mathbf{G}} + \iny{\mathbf{H}})(\db\db\dr) \\
    (\iny{\mathbf{G}} + \iny{\mathbf{H}})(\db\db\dg) & 
    (\iny{\mathbf{G}} + \iny{\mathbf{H}})(\db\dg\dg) &
    (\iny{\mathbf{G}} + \iny{\mathbf{H}})(\db\dg\dr) \\
    (\iny{\mathbf{G}} + \iny{\mathbf{H}})(\db\dg\dg) &
    (\iny{\mathbf{G}} + \iny{\mathbf{H}})(\dg\dg\dg) &
    (\iny{\mathbf{G}} + \iny{\mathbf{H}})(\dg\dg\dr)
\end{vmatrix} =
\begin{vmatrix}
  p & q & b x^2 \\
  q & 0 & b x y \\
  0 & 0 & b y^2
\end{vmatrix} = - q^2 by^2 = 0 \; .\]
So, $b = 0$  because $q, y \ne 0$.
We also need
\[
\begin{vmatrix}
    (\iny{\mathbf{G}} + \iny{\mathbf{H}})(\db\db\db) & 
    (\iny{\mathbf{G}} + \iny{\mathbf{H}})(\db\db\dg) & 
    (\iny{\mathbf{G}} + \iny{\mathbf{H}})(\db\db\dr) \\
    (\iny{\mathbf{G}} + \iny{\mathbf{H}})(\db\db\dg) & 
    (\iny{\mathbf{G}} + \iny{\mathbf{H}})(\db\dg\dg) &
    (\iny{\mathbf{G}} + \iny{\mathbf{H}})(\db\dg\dr) \\
    (\iny{\mathbf{G}} + \iny{\mathbf{H}})(\db\dg\dr) &
    (\iny{\mathbf{G}} + \iny{\mathbf{H}})(\dg\dg\dr) &
    (\iny{\mathbf{G}} + \iny{\mathbf{H}})(\dg\dr\dr)
\end{vmatrix} =
\begin{vmatrix}
    p & q & 0 \\
    q & 0 & 0 \\
  0 & 0 & cy
\end{vmatrix} = - q^2 cy = 0 \; .\]
So $c = 0$ as well.
But then $a, b, c = 0$, so $\iny{\mathbf{H}}$ is degenerate, contrary to the assumption.

In all cases, the support of $\iny{\mathbf{H}}$ is a subset of $\db\db\db$, $ \db\db\dg$, $\db\db\dr$ and their permutations.
\end{proof}

\begin{proposition}\label{prop:linear-comb-boolean-degenerate-hard}
    Let $\mathbf{F} = \mathbf{v}_1\teh + \mathbf{v}_2\teh$ for some $\mathbf{v}_1, \mathbf{v}_2 \in \mathbb{C}^2$ such that $\mathbf{v}_1 \not \sim \mathbf{v}_2$.
    If $\mathbf{u} \not \sim \mathbf{v}_1$ and $\mathbf{u} \not \sim \mathbf{v}_2$, then 
    $\holbs(\mathbf{F}, \mathbf{F} + \mathbf{u}\teh)$ is \sph.

    Let $\mathbf{F} = \lambda \mathbf{v}_1\teh + \sym(\mathbf{v}_1, \mathbf{v}_1, \mathbf{v}_2)$ for some $\mathbf{v}_1, \mathbf{v}_2 \in \mathbb{C}^2$ and $\lambda \in \mathbb{C}$ such that $\mathbf{v}_1 \not \sim \mathbf{v}_2$.
    If $\mathbf{u} \not \sim \mathbf{v}_1$, then 
    $\holbs(\mathbf{F}, \mathbf{F} + \mathbf{u}\teh)$ is \sph.
\end{proposition}
\begin{proof}
The given $\mathbf{F}$ is nondegenerate.  Assume $\holbs(\mathbf{F}, \mathbf{F} + \mathbf{u}\teh)$ is not \sph.
Then certainly $\holbs(\mathbf{F})$ is not \sph.

    For the first case, from the Boolean domain dichotomy, 
    there exists some orthogonal $T$ such that 
    $T \mathbf{F} =  a\sbto\teh + b\sbtt\teh = [a, 0, 0, b]$ or $\Z^{-1} \mathbf{F} = [a, 0, 0,b]$ for some $a, b \ne 0$. 
    The assumption on the linear independence of $\mathbf{u}$ implies that $T \mathbf{u}$ or $\Z^{-1} \mathbf{u}$ is $(c, d)$ for some $c, d \ne 0$.
    the Boolean domain dichotomy implies that
    $\holbs([a, 0, 0, b], [a, 0, 0, b] + (c, d)\teh)$ and $\holbs(\bne \mid [a, 0, 0, b], [a, 0, 0, b] + (c, d)\teh)$ are \sph.

    For the second case, from the Boolean domain dichotomy, 
    $\Z^{-1} \mathbf{F} = [a, b, 0, 0]$ or $[0, 0, b, a]$ for some $b \ne 0$. 
    Suppose $\Z^{-1} \mathbf{F} = [a, b, 0, 0]$.
    The assumption implies that $\Z^{-1} \mathbf{u} = (c, d)$ for some $d \ne 0$.
    The Boolean domain dichotomy implies that
    $\holbs(\bne \mid [a, b, 0, 0], [a, b, 0, 0] + (c, d)\teh)$ is \sph.
    The proof for the case $\Z^{-1} \mathbf{F} = [0, 0, b, a]$ is similar. 
\end{proof}

\begin{proposition}\label{prop:linear-comb-ann-degenerate}
    Let $\mathbf{F}$ be a nondegenerate symmetric ternary signature such that $\ann(\mathbf{F}) = \spn(\mathbf{n}) \ne 0$.
    Let $\mathbf{u} \in \mathbb{C}^3$ such that $\langle \mathbf{n}, \mathbf{u} \rangle \ne 0$.
    Then, $\ann(\mathbf{F} + \mathbf{u}\teh) = 0$.
\end{proposition}
\begin{proof}
Note that $\ann(\mathbf{u}\teh) = \{ \mathbf{v} \in \mathbb{C}^3 : \langle \mathbf{u}, \mathbf{v} \rangle = 0\}$ has dimension $2$.
Since $\mathbf{n} \notin \ann(\mathbf{u}\teh)$, 
we have $\spn(\mathbf{n}) \cap \ann(\mathbf{u}\teh) = 0$
and $\mathbb{C}^3 = \spn(\mathbf{n}) \oplus \ann(\mathbf{u}\teh)$.
Note that $\mathbf{v} \not \in \ann(\mathbf{F})$ for any $\mathbf{v} \in \ann(\mathbf{u}\teh)$. 
Then, we can write any vector in $\mathbb{C}^3$ as $a \mathbf{v} + b \mathbf{n}$ for some $\mathbf{v} \in \ann(\mathbf{u}\teh)$. 
Let 
$\mathbf{G}_{a, b} = \langle \mathbf{F} + \mathbf{u}\teh, a \mathbf{v} + b \mathbf{n} \rangle = a \langle \mathbf{F}, \mathbf{v} \rangle + b \langle \mathbf{u}\teh, \mathbf{n} \rangle$.
Then, $\langle \mathbf{G}_{a, b}, \mathbf{n} \rangle = b \langle \mathbf{u}, \mathbf{n} \rangle^2 \mathbf{u} \ne 0$ if $b \ne 0$. In particular, $\mathbf{G}_{a, b} \ne 0$ for $b \ne 0$.
So, $a \mathbf{v} + b \mathbf{n} \notin \ann(\mathbf{F} + \mathbf{u}\teh)$ if $b \ne 0$.
But if $b = 0$, then for any nonzero $a \mathbf{v} \in \ann(\mathbf{u}\teh)$, we have $\mathbf{G}_{a, 0} = a \langle \mathbf{F}, \mathbf{v} \rangle \ne 0$ as well, 
so $\ann(\mathbf{F} + \mathbf{u}\teh) = 0$.
\end{proof}

\begin{figure}
\centering
  \begin{tikzpicture}[scale=1, every node/.style={scale=1}]
  \node at (0, 0) {$a$};
  \node at (-1, -1) {$b$};
  \node at (1, -1) {$c$};
  \node at(-2, -2) {$d$};
  \node at (0, -2) {$e$};
  \node at (2, -2) {$f$};
  \node at (-3, -3) {$g$};
  \node at (-1, -3) {$h$};
  \node at (1, -3) {$j$};
  \node at (3, -3) {$k$};
  \node at (-4, -4) {$l$};
  \node at (-2, -4) {$m$};
  \node at (0, -4) {$n$};
  \node at (2, -4) {$o$};
  \node at (4, -4) {$p$};

  \draw[dashed] (-1, 0) -- (3.5, -4.5) -- (-5.5, -4.5) -- cycle;
  \draw (0, 1.4) -- (6.1, -4.7) -- (-6.1, -4.7) -- cycle;

\end{tikzpicture}
\caption{$\mathbf{F}$ and $\langle \mathbf{F}, \mathbf{e}_2 \rangle$.} \label{fig:arity-4-signature}
\end{figure}
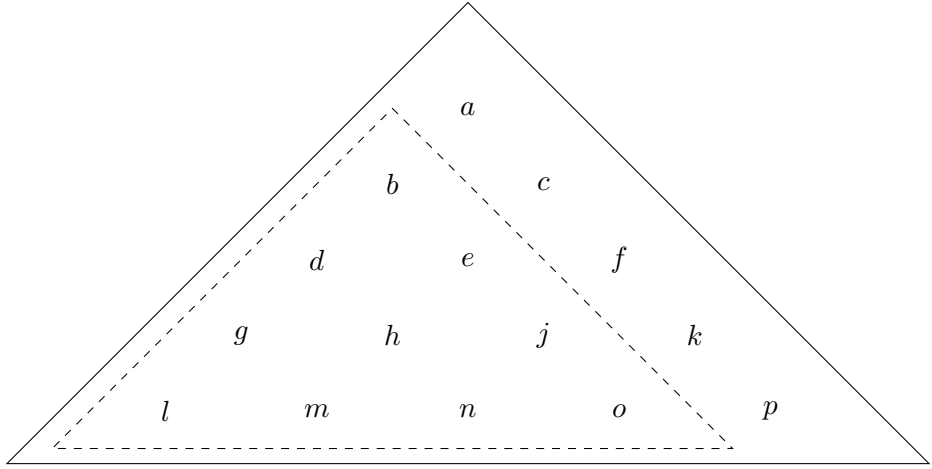

\subsection{A Single Arity 4 Signature}
\begin{lemma}\label{lem:arity-4}
Let $\mathbf{F}$ be a nondegenerate symmetric signature of arity $4$.
Then, $\holts(\mathbf{F})$ is \sph unless there exists some orthogonal $T$ such that
\begin{enumerate}
  \item $T \mathbf{F} = a \mathbf{e}_1^{\otimes 4} + b \mathbf{e}_2^{\otimes 4} + c \mathbf{e}_3^{\otimes 4}$, or
  \item $T \mathbf{F} = a \bzp^{\otimes 4} + b\bzm^{\otimes 4} + c \mathbf{e}_3^{\otimes 4}$, or
  \item $T \mathbf{F} = a \sym(\bzp, \bzp, \bzp, \bzm) + \mathbf{F}_0$
\end{enumerate}
for some $a, b, c \in \mathbb{C}$ and a symmetric arity $4$ signatrure $\mathbf{F}_0$ such that $\langle \mathbf{F}_0, \bzp \rangle = 0$.
Equivalently,
\begin{enumerate}
  \item $T \mathbf{F} \in \mathcal{E}$, or
  \item $\Yinv T \mathbf{F} \in \mathcal{E}$, or
  \item $\Yinv T \mathbf{F} \in \casecsupp$.
\end{enumerate}
\end{lemma}
\begin{proof}

Suppose $\holts(\mathbf{F})$ is not \sph.
Let $\mathscr{F} = \{\langle \mathbf{F}, \mathbf{u} \rangle : \mathbf{u} \in \mathbb{C}^3\}$.
Since $\holts(\mathscr{F}) \le_T \holts(\mathbf{F})$, $\holts(\mathscr{F})$ is not \sph.
Also, note that $\mathscr{F}$ is a vector subspace, so if $\mathbf{G}, \mathbf{H} \in \mathscr{F}$, then $x \mathbf{G} + y \mathbf{H} \in \mathscr{F}$ for any $x, y \in \mathbb{C}$. 

We briefly summarize an argument that we will use repeatedly, which is from Lemma 8.9 (arXiv version) of~\cite{domain3-real-set}. 
We will call it \textit{the triangular representation argument}.
Intuitively, it shows that if the support of $\mathscr{F}$ becomes restricted enough,
then we can use it to deduce the support of $\mathbf{F}$. 
Lemma 8.9 (arxiv version) of~\cite{domain3-real-set} shows that if $\mathscr{F} \subset \mathcal{E}$, then $\mathbf{F} \in \mathcal{E}$.
We show that a similar argument works to show that $\mathscr{F} \subset \casecsupp$ implies $\mathbf{F} \in \casecsupp$.
This becomes clear if one looks at the triangular representation of $\mathbf{F}$ (\cref{fig:arity-4-signature}).
If $\mathscr{F} \subseteq \casecsupp$, then in particular, the sub-triangles $\mathbf{F}_1 = \langle \mathbf{F}, \mathbf{e}_1 \rangle$, $\mathbf{F}_2 = \langle \mathbf{F}, \mathbf{e}_2 \rangle$, and $\mathbf{F}_3 = \langle \mathbf{F}, \mathbf{e}_3 \rangle$ must be in $\casecsupp$. 
Let $\mathbf{F}_2$ be as in \cref{fig:arity-4-signature}.
Since $\mathbf{F}_2 \in \casecsupp$, we have $g, h, l, m, n = 0$.
Similarly, $\mathbf{F}_1 \in \casecsupp$ implies $d, e, j = 0$ and
$\mathbf{F}_3 \in \casecsupp$ implies $o = 0$.
Therefore, $\mathbf{F} \in \casecsupp$.

First, consider the case where $\mathscr{F}$ contains a 
nondegenerate signature $\mathbf{G}$ such that 
$\ann(\mathbf{G}) = 0$.
Then, we may assume that $\mathbf{G}$ is one of \caseathree, \casebthree, \casecosep, \casecodz, or \casecogen
by \cref{prop:nonzero-annihilator-types}.

If $\mathbf{G}$ is \caseathree, then there exists some orthogonal $T$ such that $T \mathbf{G} = a \mathbf{e}_1\teh + b \mathbf{e}_2\teh + c \mathbf{e}_3\teh$ for $a, b, c \ne 0$.
    Consider any $T \mathbf{H} \in T \mathscr{F}$. 
    We may assume that $T \mathbf{H}$ is degenerate or of the form in \cref{prop:two-ternary-dichotomy-a-3}.
    Suppose $T \mathbf{H} \notin \mathcal{E}$. 
    If $T \mathbf{H}$ is degenerate, we may assume after permuting the domains that $T \mathbf{H} = (x, y, z)\teh$ where $x, y \ne 0$. 
    If $T \mathbf{H}$ is not degenerate, we may assume after permuting the domains that $T \mathbf{H} = (x, y, 0)\teh + z \mathbf{e}_3\teh$ for $x, y, z \ne 0$. 
    In both cases, by \cref{subsec:gc-case-a-3},
    \[\holts(T \mathcal{F}) 
    \; \ge_T\; 
    \holts(T \mathbf{G}, T \mathbf{G} + T \mathbf{H}) 
    \;\ge_T\; 
    \holbs([1, 0, 0, 1], [a, 0, 0, b] + (x, y)\teh)\]
    which is \sph by \cref{prop:linear-comb-boolean-degenerate-hard}.
    Therefore, $T \mathscr{F} \subseteq \mathcal{E}$.
This can  only be true if $T \mathbf{F} \in \mathcal{E}$ by the 
triangular representation argument.

For the other cases of \casebthree, \casecosep, \casecodz, \casecogen, we can apply a similar argument by using the respective pairwise dichotomy and looking at the support of the signatures after bringing $\mathbf{G}$ to the defining form. 
If $\mathscr{F}$ contains a \casebthree signature,
there exists some orthogonal $T$ such that $\Yinv T \mathscr{F}$ contains $a \mathbf{e}_1\teh + b \mathbf{e}_2\teh + c\mathbf{e}_3\teh$. 
A similar argument show that $\Yinv T \mathscr{F} \subset \mathcal{E}$ using the pairwise dichotomy, and the triangular representation argument shows $\Yinv T \mathbf{F} \in \mathcal{E}$.
If $\mathscr{F}$ contains a \casecosep, \casecodz or \casecogen signature,
then a similar argument shows that there exists some orthogonal $T$ such that 
$\Yinv T \mathscr{F} \subset \casecsupp$. 
By the triangular representation argument, $\Yinv T \mathbf{F} \in \casecsupp$.

Now, we may assume every signature in $\mathscr{F}$ has a nonzero annihilator.
If $\mathscr{F}$  contains only degenerate signatures, 
then we can show that $\mathbf{F}$ is degenerate as well 
using  the proof of Lemma 8.9 (arXiv version) of~\cite{domain3-real-set}.
Otherwise, we may assume that there is a nondegenerate $\mathbf{G} \in \mathscr{F}$ with $\ann(\mathbf{G}) = \spn(\mathbf{u}_{\mathbf{G}}) \ne 0$.
We consider each tractable case for $\mathbf{G}$, which can be \caseatwo, \casebtwo, \casecobg, and \casecz by \cref{prop:nonzero-annihilator-types}.

Suppose $\mathbf{G}$ is \caseatwo. We may apply some orthogonal transform $T$ so that $T \mathbf{G} = a \mathbf{e}_1\teh + b \mathbf{e}_2\teh$, with $a, b \ne 0$. 
Then $\ann(T \mathbf{G}) = \spn(\mathbf{e}_3)$.
If there is some nondegenerate $T \mathbf{H} \in T \mathscr{F}$ such that $\ann(T \mathbf{H}) = \spn(\mathbf{u})$ with $\mathbf{e}_3 \not \sim \mathbf{u}$, then $\ann(x T \mathbf{G} + y T \mathbf{H}) = 0$ for some nonzero $x, y\in \mathbb{C}$ by~\cref{{cor:lin-comb-ebd-signatures-cannot-be-a-b}}, 
a contradiction.
So, any nondegenerate signature in $T \mathscr{F}$ must be supported on $\db\dg$.
Suppose $\mathbf{v}\teh \in T \mathscr{F}$,
If $\langle \mathbf{v}, \mathbf{e}_3 \rangle
\ne 0$, then $\ann(T \mathbf{G} + \mathbf{v}\teh) = 0$
by \cref{prop:linear-comb-ann-degenerate},  a contradiction.
Thus, $\langle \mathbf{v}, \mathbf{e}_3 \rangle
= 0$, and so any degenerate signature in $T \mathscr{F}$ is also supported on $\db\dg$.
Combined with the Boolean domain dichotomy and \cref{prop:linear-comb-boolean-degenerate-hard}, this implies that $T\mathscr{F} \subset \spn(\mathbf{e}_1\teh, \mathbf{e}_2\teh) \subset \mathcal{E}$. 
We conclude that $T \mathbf{F} = x \mathbf{e}_1^{\otimes 4} + y \mathbf{e}_2^{\otimes 4} \in \mathcal{E}$ 
by the triangular representation argument.

The case for when $\mathbf{G}$ is \casebtwo can be proved similarly by anchoring on $\Yinv T \mathbf{G} = a \mathbf{e}_1\teh + b \mathbf{e}_2\teh$ for some orthogonal $T$.

We now assume that $\mathbf{G}$ is \casecobg. 
We may apply $\Yinv T$ for some orthogonal $T$ to get $\Yinv T \mathbf{G} = p \mathbf{e}_1\teh + q \sym(\mathbf{e}_1, \mathbf{e}_1, \mathbf{e}_2)$ for some $q \ne 0$.
By \cref{prop:sum-casec1bg-else}, the support of any other nondegenerate signature in $\Yinv T \mathscr{F}$ must be a subset $\db\db\db, \db\db\dg, \db\db\dr$ and their permutations. 
By \cref{prop:linear-comb-ann-degenerate}, any degenerate signature in  $\Yinv T \mathscr{F}$  must be supported on $\db\dg$ by the same proof idea as in the  \caseatwo case, 
and by \cref{prop:linear-comb-boolean-degenerate-hard} it must be $\sim \mathbf{e}_1\teh$. 
Therefore, the support of every signature in $\Yinv T \mathscr{F}$ must be a subset of $\db\db\db, \db\db\dg, \db\db\dr$ and their permutations. 
In particular, $\Yinv T \mathscr{F} \subset \casecsupp$.
By the triangular representation argument, $\Yinv T \mathbf{F} \in \casecsupp$.

Now, we may assume all signatures in $\mathscr{F}$ are \casecz or degenerate. 
Suppose $\mathbf{G}$ is \casecz.
Then $\ann(\mathbf{G}) = \spn(\mathbf{u}_{\mathbf{G}}) \ne 0$ and is isotropic.
Suppose there is a nondegenerate $\mathbf{H} \in \mathscr{F}$ with 
$\ann(\mathbf{H}) = \spn(\mathbf{u}_{\mathbf{H}})$ such that $\mathbf{u}_{\mathbf{G}} \not \sim \mathbf{u}_{\mathbf{H}} \ne 0$ and $\mathbf{u}_{\mathbf{H}}$ is isotropic.
Then, by Lemma 3.5 (arXiv version) in~\cite{domain3-single-signature}, there exists some orthogonal $T$ such that 
$\ann(\Yinv T \mathbf{G}) = \Y\transpose \ann(T \mathbf{G}) = \Y\transpose \spn(\bzp) = \spn(\mathbf{e}_2)$ and $\ann(\Yinv T \mathbf{H}) = \spn(\mathbf{e}_1)$.
Hence,
$\Yinv T \mathbf{G}$ is supported on $\db\dr$ 
and $\Yinv T \mathbf{H}$ is supported on $\dg\dr$.
By \cref{prop:lin-comb-of-ebd-signatures} applied with domains $\db\dr$ and $\dg\dr$
(where the common domain element is $\dr$), 
we must have $(\Yinv T \mathbf{G})\domres{\db\dr} = [0, 0, a, b]$ and
$(\Yinv T \mathbf{H})\domres{\dg\dr} = [0, 0, c, d]$ for some $a, b, c, d \in \mathbb{C}$.
Note that $a, c \ne 0$ because otherwise $\mathbf{G}$ or $\mathbf{H}$ is degenerate. 
However, 
\begin{align*}
\holts(\mathbf{X} \mid \Yinv T \mathbf{G}, \Yinv T \mathbf{H}) &\ge_T \holts(\mathbf{X}\Yinv T \mathbf{G} \mid \Yinv T \mathbf{H}) \\
&\ge_T \holbs((\Yinv T \mathbf{G})\domres{\db\dr} \mid (\Yinv T \mathbf{H})\domres{\dg\dr}) \\
&\ge_T \holbs([0,0, a,b] \mid [0,0,c, d])
\end{align*}
which is \sph by \cref{cor:bipartite-boolean-dichotomy-explicit}.
By \cref{prop:linear-comb-ann-degenerate}, any degenerate signature in $\mathscr{F}$ 
must also be annihilated by $\mathbf{u}_{\mathbf{G}}$. 
So, every signature in $\mathscr{F}$ must be annihilated by the same isotropic vector $\mathbf{u}_{\mathbf{G}}$.
Then, every signature in $\Yinv T \mathscr{F}$ is supported on $\db\dr$, which means $\Yinv T \mathscr{F} \subset \casecsupp$.
By the triangular representation argument, $\Yinv T \mathbf{F} \in \casecsupp$.
\end{proof}

\subsection{Beyond Arity 4}

With all the preceding preparatory work, we may extend the dichotomy to include higher arity signatures by  an induction on arity.

We extend the 9 tractable types given in \cref{def:types-ternary} to include higher arity signatures.
\begin{definition}\label{def:types-higher-arity}
Let $\mathbf{F}$ be a nondegenerate, symmetric, artiy $n \ge 3$ signature.
We say $\mathbf{F}$ is one of the following types if there exists $T \in O(3, \mathbb{C})$ that satisfies the given condition. 
We say $\mathbf{F}$ is of the \emph{defining form} of the type if $T$ can be chosen to be the identity matrix. 
\begin{enumerate}
    \item \caseathree: $T \mathbf{F} = a \mathbf{e}_1^{\otimes n} + b \mathbf{e}_2^{\otimes n} + c \mathbf{e}_3^{\otimes n}$ for nonzero $a, b, c \in \mathbb{C}$.
    \item \caseatwo: $T \mathbf{F} = a \mathbf{e}_1^{\otimes n} + b \mathbf{e}_2^{\otimes n}$ for nonzero $a, b \in \mathbb{C}$.
    \item \casebthree: $T \mathbf{F} = a \bzp^{\otimes n} + b \bzm^{\otimes n} + c \mathbf{e}_3^{\otimes n}$ for nonzero $a, b, c \in \mathbb{C}$.
    \item \casebtwo: $T \mathbf{F} = a \bzp^{\otimes n} + b \bzm^{\otimes n}$ for nonzero $a, b \in \mathbb{C}$.
    \item \casecz: $T \mathbf{F}$ with the condition that $\langle T \mathbf{F}, \bzp \rangle = 0$.
    \item \casecobg: 
    $\Yinv T \mathbf{F} = a \mathbf{e}_1^{\otimes n} + \lambda \sym(\mathbf{e}_1, \ldots, \mathbf{e}_1, \mathbf{e}_2)$ for nonzero $\lambda \in \mathbb{C}$ and any $a \in \mathbb{C}$.
    \item \casecosep: 
    $\Yinv T \mathbf{F} = a \mathbf{e}_1^{\otimes n} + \lambda \sym(\mathbf{e}_1, \ldots, \mathbf{e}_1, \mathbf{e}_2) + d \mathbf{e}_3^{\otimes n}$ for nonzero $\lambda, d \in \mathbb{C}$ and any $a \in \mathbb{C}$.
    \item \casecodz: The support of $\Yinv T \mathbf{F}$ is a subset of $\db^n$, $\db^{n-1} \dg$, $\db^{n-1} \dr$, $\db^{n-2} \dr^2$ and their permutations, $\Yinv T \mathbf{F}(\db^{n-1} \dg) \ne 0$, and $\Yinv T \mathbf{F}(\db^{n-2} \dr^2) \ne 0$.
    \item \casecogen: $\Yinv T \mathbf{F} \in \casecsupp$ and $\mathbf{F}$ is not \casecz, \casecobg, \casecosep, or \casecodz. 
\end{enumerate}
\end{definition}
We can extend the arguments from \cref{lem:types-distinct} to show that these are mutually exclusive. 
We will show in \cref{prop:reduce-arity} that if $\mathbf{F}$ is a signature of type $\mathfrak{F}$, then we can produce a lower arity signature of the same type $\mathfrak{F}$ by connecting unary signatures. 
For the types other than \casecogen, the argument is straightforward, but the proof is a little more involved for \casecogen because of the implicit definition. 
We prove an invariant of \casecosep to separate it from \casecogen, generalizing \cref{cor:casecosep-form}. 

\begin{proposition}\label{prop:full-invariant-c1-sep}
    Let $\mathbf{F} \in \casecsupp$ of arity $n \ge 3$. 
    Let $\iny{\mathbf{F}} = \Yinv \mathbf{F}$ and $\lambda = \iny{\mathbf{F}}(\db^{n-1} \dg) \ne 0$ 
and $x_i = \iny{\mathbf{F}}(\db^{n-i} \dr^i)$.
    Then, $\mathbf{F}$ is \casecosep if and only if all of the following conditions are satisfied:
    \begin{enumerate}
        \item 
        $P_n(\iny{\mathbf{F}}) = x_{n-1} x_2 - x_{n-1} \lambda - x_1 x_n = 0$.
        \item If $n \ge 4$, then $[x_2, x_3, \ldots, x_n]$ is degenerate.
        \item $x_n \ne 0$.
    \end{enumerate}
\end{proposition}
\begin{proof}
If $n = 3$, then \cref{cor:casecosep-form} proves the equivalence, so we may assume $n \ge 4$.

We prove the forward direction. 
Suppose $\mathbf{F}$ is \casecosep.
Then, by \cref{prop:casecsupp-fixed-under-rho-c-0}, there exists some $M = \begin{bsmallmatrix}
        a & b \\ 0 & d
    \end{bsmallmatrix} \in \speciallinear(2, \mathbb{C})$
    such that 
    $\iny{\mathbf{F}} =  \rho(M) (x \mathbf{e}_1^{\otimes n} + \kappa \sym(\mathbf{e}_1, \ldots, \mathbf{e}_1, \mathbf{e}_2) + z \mathbf{e}_3^{\otimes n})$ for some $x, z, \kappa \in \mathbb{C}$ such that $\kappa, z \ne 0$.
    Recall that 
$
      \rho(M) = 
      \begin{bsmallmatrix}
        d^2 & -b^2 & \sqrt{2}bd \\
        0 & 1/d^2 & 0 \\
        0 & -\sqrt{2}b/d & 1
      \end{bsmallmatrix}
$.
We analyze each term in $\iny{\mathbf{F}}$.
The first term
$x (\rho(M) \mathbf{e}_1)^{\otimes n} = x d^{2n}\mathbf{e}_1^{\otimes n}$
only contributes to $\iny{\mathbf{F}}(\db^n)$.
The second term 
$\kappa \rho(M)\sym(\mathbf{e}_1, \ldots, \mathbf{e}_1, \mathbf{e}_2) = \kappa \sym(d^2\mathbf{e}_1, \ldots, d^2 \mathbf{e}_1, \rho(M) \mathbf{e}_2)$ contributes to $\iny{\mathbf{F}}(\db^i \dg^j \dr^k)$ for $i \ge n - 1$,
and $i+j+k=n$. 
Therefore, $\iny{\mathbf{F}}(\db^{n - i}\dr^i) = (z (\rho(M) \mathbf{e}_3)^{\otimes n})(\db^{n - i} \dr^i)$ for $i \ge  2$ and thus $[x_2, x_3, \ldots, x_n]$ is degenerate.
We can calculate that
\begin{align*}
\lambda &= \iny{\mathbf{F}}(\db^{n-1} \dg) &&= \kappa (d^2)^{n-1}/d^2 &&= \kappa d^{2n - 4} \\
x_1 &= \iny{\mathbf{F}}(\db^{n-1} \dr) 
&&= z (\sqrt{2} b d)^{n-1} + \kappa (d^2)^{n-1} (- \sqrt{2} b/d) 
&&= z (\sqrt{2} b d)^{n-1} - \kappa \sqrt{2} b d^{2n - 3} \\
x_2 &= \iny{\mathbf{F}}(\db^{n-2} \dr^2) &&= z (\sqrt{2} b d)^{n-2} \\
x_{n-1} &= \iny{\mathbf{F}}(\db \dr^{n-1}) &&= z \sqrt{2} b d \\
x_n  &= \iny{\mathbf{F}}(\dr^n) &&= z \ne 0
\end{align*}
Then,
\begin{align*}
    &\quad x_{n-1} x_2 - x_{n-1} \lambda - x_1 x_n \\
    &= z \sqrt{2} b d z (\sqrt{2} b d)^{n-2} -  \kappa d^{2n - 4} z \sqrt{2} b d - (z (\sqrt{2} b d)^{n-1} - \kappa \sqrt{2} b d^{2n - 3})z \\
    &= z^2 (\sqrt{2} b d)^{n-1} - \kappa z \sqrt{2} b d^{2n - 3} - z^2(\sqrt{2} bd)^{n-1} + \kappa z \sqrt{2} b d^{2n - 3} \\
    &= 0 \,.
\end{align*}

We now prove the other direction. 
We may assume $x_n = 1$ by rescaling $\iny{\mathbf{F}}$.
Note that $P_n$ is homogeneous, so $P_n(\iny{\mathbf{F}}) = 0$ after rescaling. 

Suppose $x_{n-1} = 0$.
Then, $[x_2, x_3, \ldots,  x_n]$ being degenerate implies $x_2, \ldots, x_{n-1} = 0$.
Then, by $x_{n - 1} x_2 - x_{n-1} \lambda - x_1 x_n = 0$ implies $x_1 = 0$ as well.
Then, $\iny{\mathbf{F}}$ is \sepbg, so $\mathbf{F}$ is \casecosep. 

Suppose $x_{n-1} \ne 0$.
Then, we can write $[x_2, x_3, \ldots, x_n] = [p^{n-2}, p^{n-3}, \ldots, 1]$ for some $p \ne 0$.
Then, $P(\iny{\mathbf{F}}) = 0$ implies 
$
x_1 = x_{n-1} x_2 - x_{n-1} \lambda
= 
p^{n-1} - p \lambda
$. 
Then, we have $\iny{\mathbf{F}} = q \mathbf{e}_1^{\otimes n} + \lambda \sym(\mathbf{e}_1, \mathbf{e}_1, \mathbf{v}) + (p, 0, 1)^{\otimes n}$ for some $q \in \mathbb{C}$ and $\mathbf{v} = (0, 1, -p)$.
Then, for $ M = \begin{bsmallmatrix}
    1 & -p/\sqrt{2} \\ 0 & 1
\end{bsmallmatrix}$, $\rho(M) \iny{\mathbf{F}} = q \mathbf{e}_1^{\otimes n} + \lambda \sym(\mathbf{e}_1, \mathbf{e}_1, (p^2/2, 1, 0)) + \mathbf{e}_3^{\otimes n}$ which is \sepbg. 
In particular, we have $\rho(M) \iny{\mathbf{F}} =  q' \mathbf{e}_1^{\otimes n} + \lambda \sym(\mathbf{e}_1, \mathbf{e}_1, \mathbf{e}_2) + \mathbf{e}_3^{\otimes n}$ for some $q' \in \mathbb{C}$.
Therefore, $\mathbf{F}$ is \casecosep. 
\end{proof}

\begin{proposition}\label{prop:c1-gen-reduce-arity}
    Let $\mathbf{F}$ be a \casecogen signature of arity $n \ge 4$ such that $\Yinv \mathbf{F} \in \casecsupp$. 
    Then, there exists $\mathbf{u} \in \mathbb{C}^3$ such that $\mathbf{G} = \langle \mathbf{F}, \mathbf{u} \rangle$ is a \casecogen signature and $\Yinv \mathbf{G} \in \casecsupp$.
\end{proposition}
\begin{proof}
Let $\iny{\mathbf{F}} = \Yinv \mathbf{F}$,
$\lambda = \iny{\mathbf{F}}(\db^{n-1} \dg)$ and
$x_i = \iny{\mathbf{F}}(\db^{n-i} \dr^i)$.
By the definition of \casecogen, $\lambda \ne 0$ 
because otherwise $\mathbf{F}$ is \casecz, so we may assume $\lambda = 1$ after rescaling. 
Also, there exists $i \ge 3$ such that $x_i \ne 0$ because otherwise $\mathbf{F}$ is \casecodz.

Consider $\iny{\mathbf{G}}_p = \langle \iny{\mathbf{F}}, (1, 0, p) \rangle$ for some $p \in \mathbb{C}$ to be determined later.
We have $\iny{\mathbf{G}}(\db^{n-2} \dg) = 1$
and $y_i = \iny{\mathbf{G}}(\db^{n-1-i} \dg^i) = x_i + p x_{i + 1}$. 
We would like to argue that there is some $p$ such that $\Y \iny{\mathbf{G}}_p$ is \casecogen. 

Because there exists $k \ge 3$ such that $x_k \ne 0$, 
$y_i = x_i + p x_{i+1}$ is a nonzero polynomial in $p$ for at least one $i \ge 3$.
Then, except for possibly one value of $p \in \mathbb{C}$, $\Y \iny{\mathbf{G}}_p$ is not \casecobg or \casecodz. 
If $\Y \iny{\mathbf{G}}_p$ were to be \casecosep,
it must satisfy the conditions in \cref{prop:full-invariant-c1-sep}, which requires $y_{n - 2} y_2 - y_{n-2} - y_1 y_{n-1} = 0$, $[y_2, y_3, \ldots, y_{n -1}]$ is degenerate and $y_{n-1} \ne 0$. 
Note that $[y_2, y_3, \ldots, y_{n -1}]$ is degenerate if and only if $y_2, \ldots, y_{n-1}$ satisfy a finite set of polynomials
$y_i y_{j+1} - y_{i + 1} y_j = 0$ for $2 \le i, j \le n - 2$ because the set of degenerate signatures can be viewed as the rational normal curve in the projective space. 
We have 
\begin{align*}
    y_i y_{j+1} - y_{i + 1} y_j &= (x_i + p x_{i + 1})(x_{j + 1} + p x_{j + 2}) - 
(x_{i + 1} + p x_{i + 2})(x_j + p x_{j + 1}) \\
&= (x_i x_{j+1} - x_{i + 1} x_j) + p( x_i x_{j+2} - x_{i + 2} x_j) 
+ p^2(x_{i+1} x_{j+2} - x_{i + 2} x_{j + 1}) \;.
\end{align*}
If $[x_2, x_3, \ldots, x_n]$ is nondegenerate, there is some $2 \le i', j' \le n - 1$ such that  $x_{i'} x_{j'+1} - x_{i'+1} x_{j'} \ne 0$. 
Then, at least one of $y_i y_{j + 1} - y_{i + 1} y_j$ is a nonzero polynomial in $p$. 

If $[x_2, x_3, \ldots, x_n]$ is degenerate, then for $\mathbf{F}$ to be  \casecogen,
we must have $x_{n-1} x_2 - x_{n-1} - x_1 x_n \ne 0$ or $x_n = 0$.
If $x_n = 0$, then $[x_2, x_3, \ldots, x_n] = [1, 0, \ldots, 0]$ by degeneracy. However, this means $\mathbf{F}$ is \casecodz, contrary to the assumption.
So, we may assume $x_{n-1} x_2 - x_{n-1} - x_1 x_n \ne 0$.
Then, we have 
\begin{align*}
y_{n-2} y_2 - y_{n-2} - y_1 y_{n-1} &= 
(x_{n-2} + p x_{n-1})(x_2 + p x_3) - (x_{n-2} + p x_{n-1}) - (x_1 + p x_2) (x_{n-1} + p x_n) \\
&= (x_{n-2} x_2 - x_{n-2} - x_1 x_{n-1}) + p(x_{n-1} x_2 - x_{n-1} - x_1 x_n) \\
&\quad  - p(x_2 x_{n-1} - x_3 x_{n-2}) + p^2(x_{n-1} x_3 - x_2 x_n) 
\end{align*}
Because $[x_2, x_3, \ldots, x_n]$ is degenerate, $x_2 x_{n-1} - x_3 x_{n-2} = 0$, so the above is a nonzero polynomial in $p$. 

Therefore,
for $\Y \iny{\mathbf{G}}_p$ to be a \casecosep, 
a finite number of nonzero polynomials in $p$ must evaluate to $0$. 
There could be only finitely many such $p$, so there exists
$p \in \mathbb{C}$ such that $\Y \iny{\mathbf{G}}_p$ is \casecogen and $\iny{\mathbf{G}}_p \in \casecsupp$ as desired. 
\end{proof}

\begin{proposition}\label{prop:reduce-arity}
Let $\mathfrak{F}$ be a tractable type \caseatwo, \caseathree, \casebtwo, \casebthree, \casecz, \casecobg, \casecosep, \casecodz, or \casecogen. 
    Let $\mathbf{F}$ be a nondegenerate, symmetric, arity $n \ge 3$ signature of type $\mathfrak{F}$. 
    Then, there exists $\mathbf{u}_1, \ldots, \mathbf{u}_k \in \mathbb{C}^3$ for $k = n - 3$ such that connecting $\mathbf{u}_1, \ldots, \mathbf{u}_k$ to $\mathbf{F}$ results in a ternary signature $\mathbf{G}$ of type $\mathfrak{F}$. 

    Further, if $T$ is an orthogonal matrix such that $T \mathbf{F}$ is of the defining form, then $T \mathbf{G}$ is of the defining form as well.
\end{proposition}
\begin{proof}
We will assume that $\mathbf{F}$ is given in the defining form and show that we can realize a ternary signature of the defining form by connecting unary signatures. 
Let $\iny{\mathbf{F}} = \Yinv \mathbf{F}$.

Suppose $\mathfrak{F}$ is \caseathree. 
Then, $\mathbf{F} = a \mathbf{e}_1^{\otimes n} + b \mathbf{e}_2^{\otimes n} + c \mathbf{e}_3^{\otimes n}$ for some nonzero $a, b, c \in \mathbb{C}$. 
Then, connecting $(1, 1, 1)$ to $\mathbf{F}$ $n-3$ times results in $  \mathbf{G} = a \mathbf{e}_1^{\otimes 3} + b \mathbf{e}_2^{\otimes 3} + c \mathbf{e}_3^{\otimes 3}$ which is a ternary \caseathree signature of the defining form. 
We can argue similarly for \caseatwo. 
For \casebthree and \casebtwo, we can argue similarly since $\iny{\mathbf{F}} = a \mathbf{e}_1^{\otimes n} + b \mathbf{e}_2^{\otimes n} + c \mathbf{e}_3^{\otimes n}$ for some $a, b, c \in \mathbb{C}$.

Suppose $\mathfrak{F}$ is \casecz. 
Then, $\iny{\mathbf{F}}$ is supported on $\db\dr$.
Since it is nondegenerate, there exists some $\mathbf{u}_1 \in \mathbb{C}^3$ such that $\iny{\mathbf{G}}_1 = \langle \iny{\mathbf{F}}, \mathbf{u}_1 \rangle$ is nondegenerate.
Note that $\iny{\mathbf{G}}_1$ is supported on $\db\dr$.
Then, we can inductively find such $\mathbf{u}_2, \ldots, \mathbf{u}_k$ so that the resulting ternary signature is nondegenerate and is supported on $\db\dr$.

Suppose $\mathfrak{F}$ is \casecosep.
Then, 
$\iny{\mathbf{F}} = a \mathbf{e}_1^{\otimes n} + \lambda \sym(\mathbf{e}_1, \ldots, \mathbf{e}_1, \mathbf{e}_2) + d \mathbf{e}_3^{\otimes n}$ for nonzero $\lambda, d \in \mathbb{C}$ and any $a \in \mathbb{C}$.
 Then, connecting $(1, 0, 1)$ to $\iny{\mathbf{F}}$ $n- 3$ times results in $\iny{\mathbf{G}} = a \mathbf{e}_1^{\otimes 3} + \lambda \sym(\mathbf{e}_1, \mathbf{e}_1, \mathbf{e}_2) + d \mathbf{e}_3^{\otimes 3}$ and $\Y \iny{\mathbf{G}}$ is \casecosep of the defining form. 
We can argue similarly for \casecobg.

Suppose $\mathfrak{F}$ is \casecodz. 
The support of $\iny{\mathbf{F}}$ is a subset of on $\db^n$, $\db^{n-1} \dg$, $\db^{n-1} \dr$, $\db^{n-2} \dr^2$ and their permutations, $\iny{\mathbf{F}}(\db^{n-1} \dg) \ne 0$, and $\iny{\mathbf{F}}(\db^{n-2} \dr^2) \ne 0$.
Then, connecting $\mathbf{e}_1$ to $\iny{\mathbf{F}}$ $n - 3$ times results in 
$\iny{\mathbf{G}}$ such that 
$\iny{\mathbf{G}}(\db\db\dg) = \iny{\mathbf{F}}(\db^{n-1} \dg) \ne 0$, and $\iny{\mathbf{G}}(\db\dr\dr) = \iny{\mathbf{F}}(\db^{n-2} \dr \dr) \ne 0$.
This can be easily seen using the triangular representation, because connecting $\mathbf{e}_1$ just results in a signature of one less arity that is represented by the triangle with the bottom row removed (see \cref{fig:arity-4-signature}). Therefore, $\Y \iny{\mathbf{G}}$ is \casecodz of the defining form. 

If $\mathfrak{F}$ is \casecogen, then we can apply \cref{prop:c1-gen-reduce-arity} inductively to produce such $\mathbf{G}$.

\end{proof}

\cref{lem:arity-4} tells us that if $\mathbf{F}$ is a nondegenerate symmetric arity $4$ signature such that $\holts(\mathbf{F})$ is not \sph, then it must be one of the types in \cref{def:types-higher-arity}.
Let $\mathcal{F}$ be a set of nondegenerate symmetric arity $\le 4$ signatures. 
For each $\mathbf{F}_4 \in \mathcal{F}$ of arity $4$ signature we have $\holts(\mathcal{F}) \ge_T \holts(\mathcal{F}, \mathbf{F}_3)$, in which $\mathbf{F}_3$ is the ternary signature of the same type as $\mathbf{F}_4$ given by \cref{prop:reduce-arity}.
Further, the same $T$ brings both $\mathbf{F}_3$ and $\mathbf{F}_4$ to the defining form. 
Then, we can easily deduce the dichotomy of $\holts(\mathcal{F})$ from \cref{thm:dichotomy-ternary-binary}, which would state that $\holts(\mathcal{F})$ is tractable if it satisfies one of the conditions given in 
\cref{thm:tractable-cases}
and \sph otherwise. 
We can repeat the arguments in this section to inductively obtain a dichotomy of signatures of arbitrary arity. Therefore, we have proven the main goal of this paper.
\begin{theorem}\label{thm:full-dichotomy}
Let $\mathcal{F}$ and $\mathcal{B}$ be respectively
    sets of nondegenerate symmetric arity $\ge 3$ and binary signatures.
Then, $\holts(\mathcal{F} \cup \mathcal{B})$ is tractable if it satisfies the tractability criterion of \cref{thm:tractable-cases}.
Otherwise, it is \sph. 
\end{theorem}

\printbibliography

\end{document}